\documentclass[11 pt]{article}
\usepackage{hyperref}
\usepackage{graphicx,color,xcolor} 
\usepackage{amsmath,amssymb,amsfonts,amsthm}
\usepackage[title]{appendix}

\newtheorem{thm}{Theorem}
\newtheorem{lemma}[thm]{Lemma}
\newtheorem{prop}[thm]{Proposition}
\newtheorem{defi}[thm]{Definition}

\newtheorem{coro}[thm]{Corollary} 
\newtheorem{remark}[thm]{Remark}

\newtheorem{assumption}[thm]{Assumption}

\counterwithin{thm}{section}
\counterwithin{equation}{section}
\counterwithin{figure}{section}

\newcommand{\HmhalfpD}{H^{-1/2}(\partial D_*)}
\newcommand{\HhalfpD}{H^{1/2}(\partial D_*)}
\newcommand{\HhalfG}{\cH^{1/2}(\Gamma)}
\newcommand{\HmhalfG}{\cH^{-1/2}(\Gamma)}
\newcommand{\HhalfGa}{\cH^{1/2,a}(\Gamma^a)}
\newcommand{\HmhalfGa}{\cH^{-1/2,a}(\Gamma^a)}

\newcommand{\dpt}{d\ell}

\newcommand{\im}{\textrm i} 
\renewcommand{\ker}{\text{Ker}}

\newcommand{\chomo}{\cC_z}
\newcommand{\cunp}{\cC_z\backslash D}
\newcommand{\cpeps}{\cC_z\backslash D^\eps}
\newcommand{\Gb}{\Gamma_{\text{b}}}
\newcommand{\Gt}{\Gamma_{\text{t}}}
\newcommand{\Gl}{\Gamma_{\text{l}}}
\newcommand{\Gr}{\Gamma_{\text{r}}}
\newcommand{\nuGl}{\nu_1}
\newcommand{\nuGb}{\nu_2}
\newcommand{\rflc}{F}

\newcommand{\ssubset}{\subset\joinrel\subset}

\newcommand{\hrad}{\mathfrak d|\frac{t_*}{\gamma_*}|}

\newcommand{\eps}{\varepsilon}
\newcommand{\brho}{\boldsymbol\rho}
\newcommand{\bbeta}{\boldsymbol\beta}

\newcommand{\Sepsevan}{\mathbb S^{\eps,\text{evan}}}
\newcommand{\Sevan}{\mathbb S^{0,\text{evan}}}

\newcommand{\Kone}{K}
\newcommand{\Ktwo}{K'}
\newcommand{\sgn}{\text{sgn}}
\newcommand{\kp}{k_{\parallel}}

\newcommand{\bbR}{\mathbb R}
\newcommand{\bbZ}{\mathbb Z}

\newcommand{\be}{\mathbf e}

\newcommand{\bm}{\mathbf m} \newcommand{\bn}{\mathbf n}
 \newcommand{\bp}{\mathbf p}
\newcommand{\bq}{\mathbf q}

 \newcommand{\bx}{\mathbf x} 
\newcommand{\by}{\mathbf y}

 \newcommand{\cB}{\mathcal B}
\newcommand{\cC}{\mathcal C} \newcommand{\cD}{\mathcal D} 
 
 \newcommand{\cH}{\mathcal H}
\newcommand{\cI}{\mathcal I} 
\newcommand{\cK}{\mathcal K} \newcommand{\cL}{\mathcal L}
\newcommand{\cM}{\mathcal M} \newcommand{\cN}{\mathcal N}
\newcommand{\cO}{\mathcal O}  
 \newcommand{\cR}{\mathcal R}
\newcommand{\cS}{\mathcal S}

\title{Interface Modes in Honeycomb Topological Photonic Structures with Broken Reflection Symmetry\thanks{J. Lin was partially supported by the NSF grant DMS-2011148,
and H. Zhang was partially supported by Hong Kong RGC grant GRF 16304621 and NSFC grant 12371425}}

\author{
    Wei Li\thanks{Department of Mathematical Sciences, DePaul University, Chicago, IL 60614.
    \tt wei.li@depaul.edu.} \,\,\,  
     Junshan Lin
  \thanks{Department of Mathematics and Statistics, Auburn University, Auburn, AL 36849.  \tt jzl0097@auburn.edu.}\,\,\,
  Jiayu Qiu\thanks{Department of Mathematics, 
 HKUST,  Clear Water Bay, Kowloon, Hong Kong S.A.R., China.
    \tt jqiuaj@connect.ust.hk.}\,\,\,
  Hai Zhang
  \thanks{Department of Mathematics, 
 HKUST,  Clear Water Bay, Kowloon, Hong Kong S.A.R., China.
    \tt haizhang@ust.hk.}}
    
\begin{document}

\maketitle

\begin{abstract}
In this work, we present a mathematical theory for Dirac points and interface modes in honeycomb topological photonic structures consisting of impenetrable obstacles. Starting from a honeycomb lattice of obstacles attaining $120^\circ$-rotation symmetry and horizontal reflection symmetry, we apply the boundary integral equation method to show the existence of Dirac points for the first two bands at the vertices of the Brillouin zone.
We then study interface modes in a joint honeycomb photonic structure, which consists of two periodic lattices obtained by perturbing the honeycomb one with Dirac points differently. The perturbations break the reflection symmetry of the system, as a result, they annihilate the Dirac points and generate two structures with different topological phases, which mimics the quantum valley Hall effect in topological insulators. We investigate the interface modes that decay exponentially away from the interface of the joint structure in several configurations with different interface geometries, including the zigzag interface, the armchair interface, and the rational interfaces.
Using the layer potential technique and asymptotic analysis, we first characterize the band-gap opening for the two perturbed periodic structures and derive the asymptotic expansions of the Bloch modes near the band gap surfaces. By formulating the eigenvalue problem for each joint honeycomb structure using boundary integral equations over the interface and analyzing the characteristic values of the associated boundary integral operators, we prove the existence of interface modes when the perturbation is small.


\end{abstract}

\section{Introduction and outline}

\subsection{Background and motivation}
Photonic and phononic materials with band gaps can be used to localize and confine waves, which have wide applications in the transportation and manipulation of wave energy  \cite{Joanno-11, Khelif-15}. In a gapped photonic or phononic crystal, a localized wave mode with frequency in the band gap can be created by introducing a local perturbation in the periodic structure, such as a point or line defect \cite{Joanno-11}. Such a wave mode is called a defect mode and it is confined near the defect. Mathematically, a defect mode and its frequency correspond to an eigenpair of a locally perturbed periodic operator for the acoustic wave equation or Maxwell's equations. The existence of point defect modes and line defect modes  
was proved in \cite{Ammari-18,Ammari-22, Brown-17, Brown-20, Figotin-98, Figotin-97-2, kamotski-18, Santosa-Ammari-04}  for several different configurations of periodic acoustic and electromagnetic media, including the periodic dielectric media,  high contrast media, and bubbly media, etc. Besides the deterministic approaches, random media also allows for wave localization. One well-known strategy is the Anderson localization, wherein a periodic medium is randomly perturbed in the whole spatial domain \cite{Figotin-94, Figotin-96, Figotin-97, Sheng-90}. 

The recent development in topological insulators (cf. \cite{Hasan-Kane-10, bernevig-13, Qi-Zhang-11}) opens up new avenues for wave localization and confinement in photonic and phononic materials. 
The concept of topological phases for classical waves was proposed in the seminal work \cite{raghu-08}, when it was realized that topological band structures are a ubiquitous property of waves for periodic media, regardless of the classical or quantum nature of the waves. Therefore, the concepts in topological insulators can be parallelly extended to periodic wave media, and remarkably, extensive research work has been sparked in pursuit of topological acoustic, electromagnetic, and mechanical insulators to manipulate the classical wave in the same way as solids modulating electrons \cite{Khanikaev-12,  Lu-14, Chan-19, Ozawa-19, Yang-15}.
Briefly speaking, there are mainly two strategies to realize topological structures for classical waves. The first strategy mimics the quantum Hall effect in topological insulators using active components to break the time-reversal symmetry of the system \cite{Khanikaev-15, Wang-09}.
The second strategy relies on an analog of the quantum spin Hall effect or quantum valley Hall effect, and it uses passive components to break the spatial symmetry of the system \cite{Ma_Shvets-15, Wu_Hu-15}.

Wave localization in topological structures is achieved by gluing together two periodic media with distinct topological invariants. The topological phase transition at the interface of two media gives rise to the so-called interface modes, which propagate parallel to the interface but localize in the direction transverse to the interface. Recently there has been intensive mathematical research investigating the interface modes in topological insulators from different perspectives. In particular, the existence of interface modes was proved in
\cite{Drouot-Wenstein-20, Fefferman-Thorp-Weinsein-14, Fefferman-Thorp-Weinsein-16, Lee-Thorp-Weinstein-Zhu-19} for the Schr\"{o}dinger operator and several other elliptic operators, wherein the interfaces are modeled by smooth domain walls. In addition, the spectra of interface modes are closely related to the topological nature of the bulk media. In general, the net number of interface modes is equal to
the difference of the bulk topological invariants across the interface, which is known as the \textit{bulk-edge correspondence} \cite{Hasan-Kane-10, Ozawa-19}. We refer to \cite{Bal-17, Bal-19, EGS-05, Graf-02, Hatsugai-93} for the studies of the bulk-edge correspondence in discrete electron models and 
 \cite{Druout-20, Druout-20-3} for the bulk-edge correspondence in several elliptic PDE models.

In this work, we study the interface modes in a joint honeycomb photonic structure, where two periodic lattices separated by an interface are obtained by perturbing a honeycomb lattice with Dirac points differently. Such perturbations break the reflection symmetry of the system, as a result, they annihilate the Dirac points and generate two structures with different topological phases. This mimics the quantum valley Hall effect in topological insulators \cite{Ma_Shvets-15, Wu_Hu-15}.
A one-dimensional joint structure with a similar setup was investigated \cite{LZ21, Thiang-Zhang-22} using the transfer matrix method and the oscillatory theory for Sturm-Liouville operators.
In contrast to the studies of interface modes in \cite{Drouot-Wenstein-20, Fefferman-Thorp-Weinsein-14, Fefferman-Thorp-Weinsein-16, Lee-Thorp-Weinstein-Zhu-19}, where two bulk media are ``connected'' adiabatically over a length scale that is much larger than the period of the structure to form a joint photonic structure and the interface is modeled by a smooth domain wall extending to the whole spatial domain, we consider more realistic models where two periodic media are connected directly such that the medium coefficient attains a jump across the interface. Therefore, we have to address the new challenges in the spectral analysis brought by the discontinuities of coefficients in the PDE model. Beyond that, we consider the model with more general shapes of the interface that separates two bulk media. The goal of this work is to develop a mathematical framework based on a combination of layer potential theory, asymptotic analysis, and the generalized Rouch\'e theorem to examine the existence of the interface modes in such settings. The mathematical framework can be extended to study localized modes in other contexts.

\subsection{Outline}
\subsubsection{The honeycomb lattice and Dirac points}
We start from a honeycomb lattice consisting of a two-dimensional array of impenetrable obstacles and examine the existence of Dirac points in the band structure of the lattice. A schematic plot of the periodic structure and its band structure is shown in Figure \ref{fig:honeycomb}. The honeycomb lattice is a natural choice for the photonic structure, as it attains the desired symmetry to create Dirac points over the vertices of the Brillouin zone \cite{ Berkolaiko-18, Fefferman-Weinsein-12}.

\begin{figure}[!htbp]
\begin{center}
\includegraphics[height=5cm]{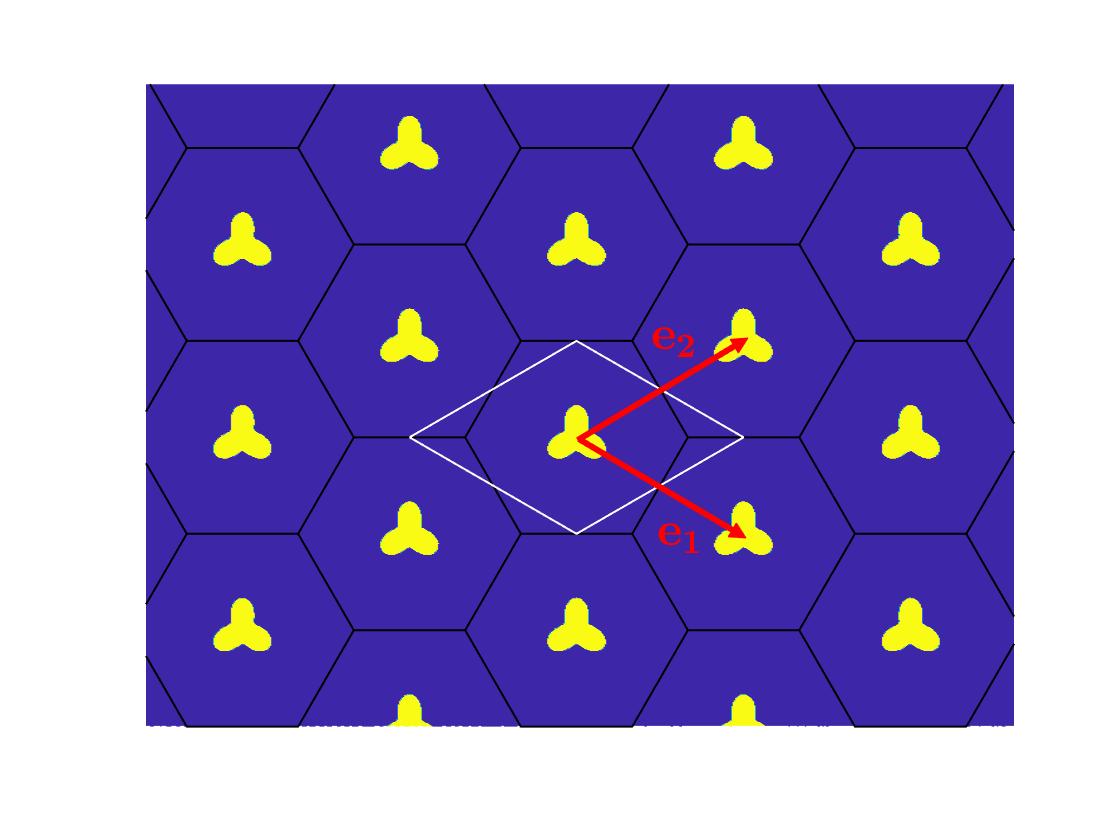} 
\includegraphics[height=5cm]{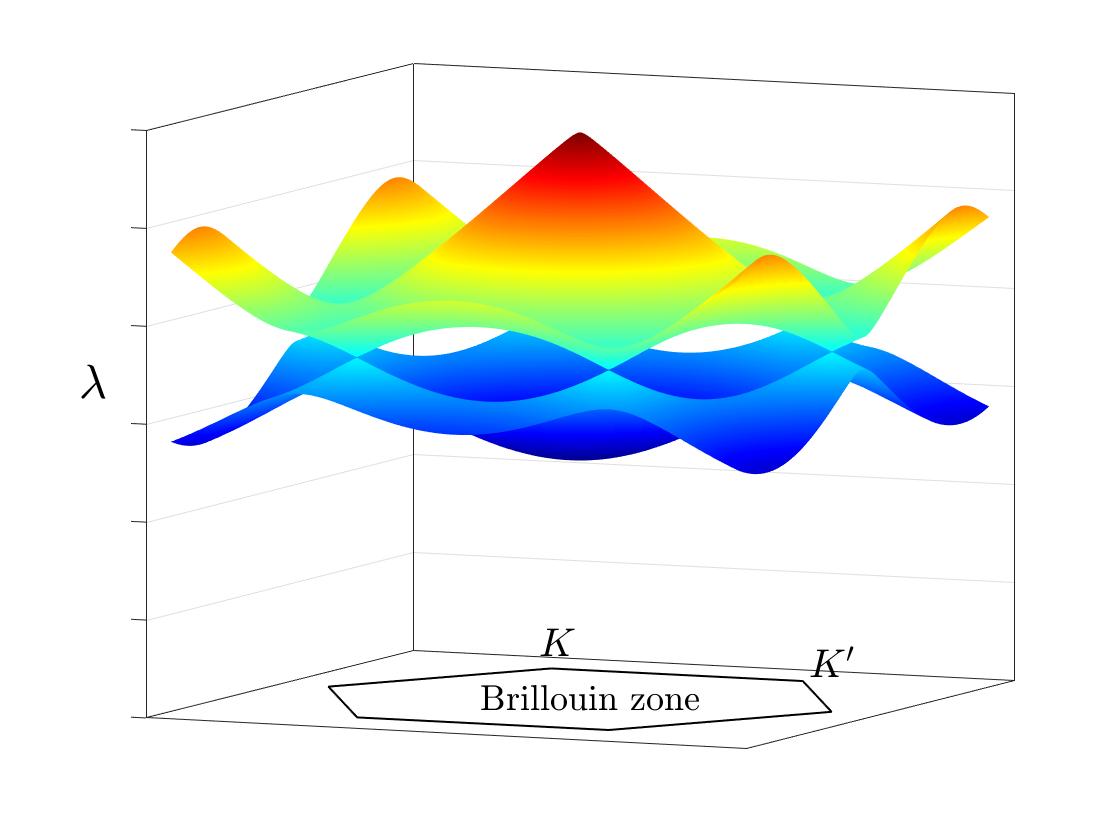} 
\caption{An infinite array of impenetrable obstacles are arranged over the honeycomb lattice (left) and its band structure (right). The obstacle in each periodic cell attains  $120^\circ$-rotation symmetry and horizontal reflection symmetry. }
\label{fig:honeycomb}
\end{center}
\end{figure}

Dirac points refer to conical intersections of two dispersion surfaces in the band structure. They are the degenerate points in the spectrum where the topological phases of the material may change. More specifically, at a Dirac point $(\bp^*, \lambda^*)$, the eigenspace of the associated partial differential operator spans a two-dimensional space. In addition,
the two dispersion surfaces forming the Dirac point attain the following expansion
 \begin{equation}\label{eq:Dirac_point}
  \lambda^\pm(\bp) = \lambda^*\pm \alpha |\bp-\bp^*| +O(|\bp-\bp^*|^2),
\end{equation}
wherein $\alpha\neq0$ denotes the slope of the linear dispersion relation near the Dirac point. Due to the surging interest in topological insulators, Dirac points were investigated for a broad class of PDE operators recently, especially for the Schr\"{o}dinger operator over the honeycomb lattice, the Helmholtz operator with high-contrast medium and resonant bubbles, etc \cite{Ammari-20-4, Cassier-Weinstein-21, Fefferman-Weinsein-12, Berkolaiko-18, Fefferman-Thorp-Weinsein-18, Lee-16, Lee-Thorp-Weinstein-Zhu-19}. 
In general, Dirac points exist at the vertices of the Brillouin zone when the medium coefficients in the honeycomb lattice attain suitable symmetry, such as inversion and $120^\circ$-rotation symmetry, horizontal reflection and $120^\circ$-rotation symmetry, etc. The degeneracy at a Dirac point
can be deduced from the representation of the relevant symmetry group, and the conical shape of the dispersion relation is obtained from its invariance under the rotation symmetry \cite{Berkolaiko-18}.

In this work, 
we apply the boundary integral equation method to show the existence of Dirac points for the first two bands at the vertices of the Brillouin zone, assuming that the shape of each obstacle in the lattice attains $120^\circ$-rotation symmetry and horizontal reflection symmetry as shown in Figure \ref{fig:honeycomb}. The existence of Dirac points for obstacles with other symmetries can be examined similarly using this method.

\subsubsection{The perturbed honeycomb lattices: spectral gap and topological phase}
We then break the spatial reflection symmetry of the honeycomb lattice by rotating the obstacles in opposite directions to obtain two photonic structures in Figure \ref{fig:honeycomb_perterb}, which will create spectral gaps at the Dirac point as shown in Figure \ref{fig:band_honeycomb_perterb} so that wave propagation is prohibited for frequencies located in the gap interval. We carry out the asymptotic analysis for the spectrum of each perturbed periodic operator using the layer potential technique and prove that a spectral gap is opened at the Dirac point when the perturbation is small. Furthermore, we prove that the eigenspaces at the band edges are swapped for the two perturbed periodic operators, which demonstrates the topological phase transition of the medium at the Dirac point. The topological phase difference between the two lattices can also be manifested through the Berry phase, which describes the phase evolution of eigenfunctions in the momentum space \cite{Asboth-16}. As demonstrated in Figure \ref{fig:Berry_curvature}, the Berry curvatures associated with the first bands of the two perturbed lattices attain opposite values.

\begin{figure}[!htbp]
\begin{center}
\includegraphics[height=5cm]{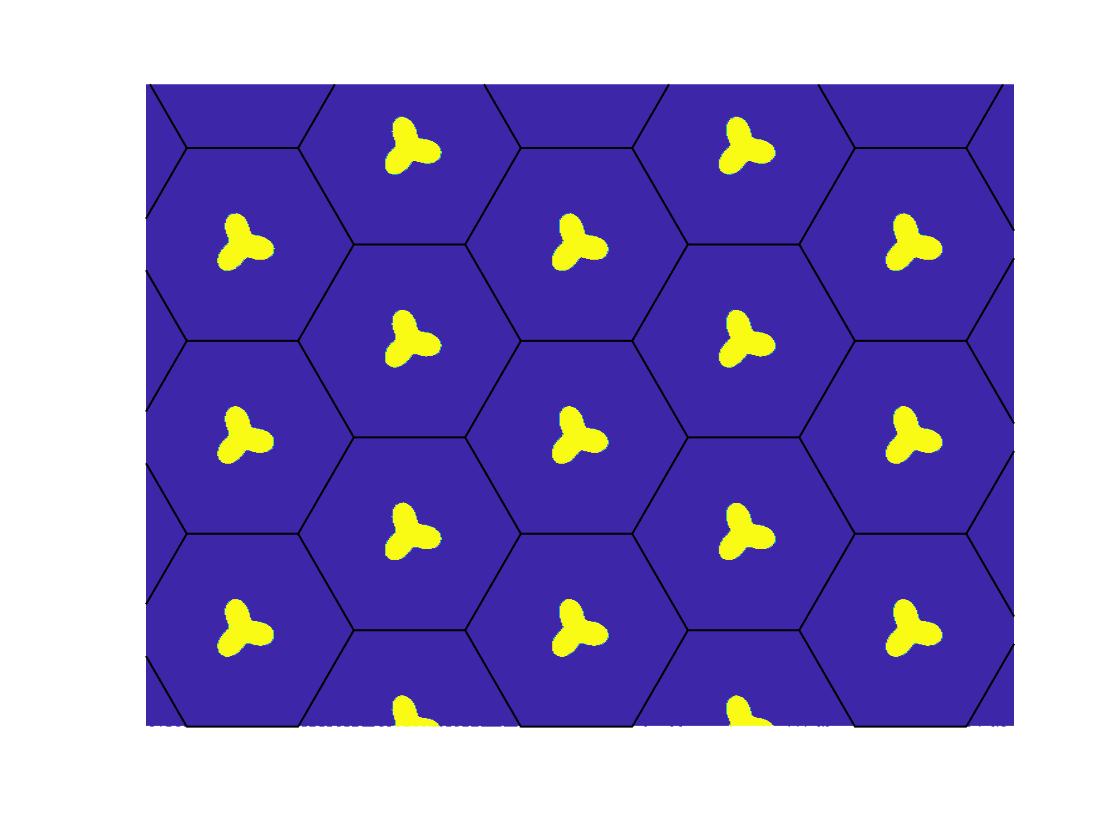}
\includegraphics[height=5cm]{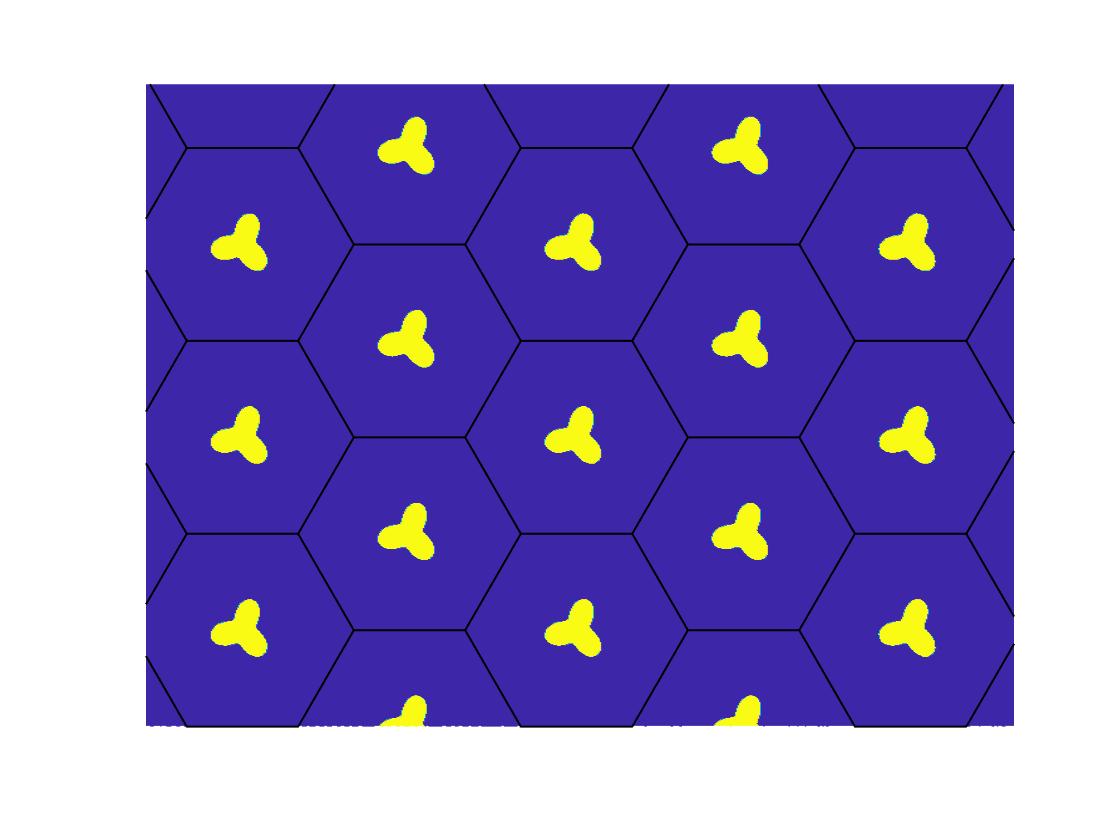} \vspace*{-0.5cm}
\caption{The two perturbed honeycomb lattices by rotating the obstacles counter-clockwisely and clockwisely, respectively. }
\label{fig:honeycomb_perterb}
\end{center}
\end{figure}

\begin{figure}[!htbp]
\begin{center}
\vspace*{-0.5cm}
\includegraphics[height=5cm]{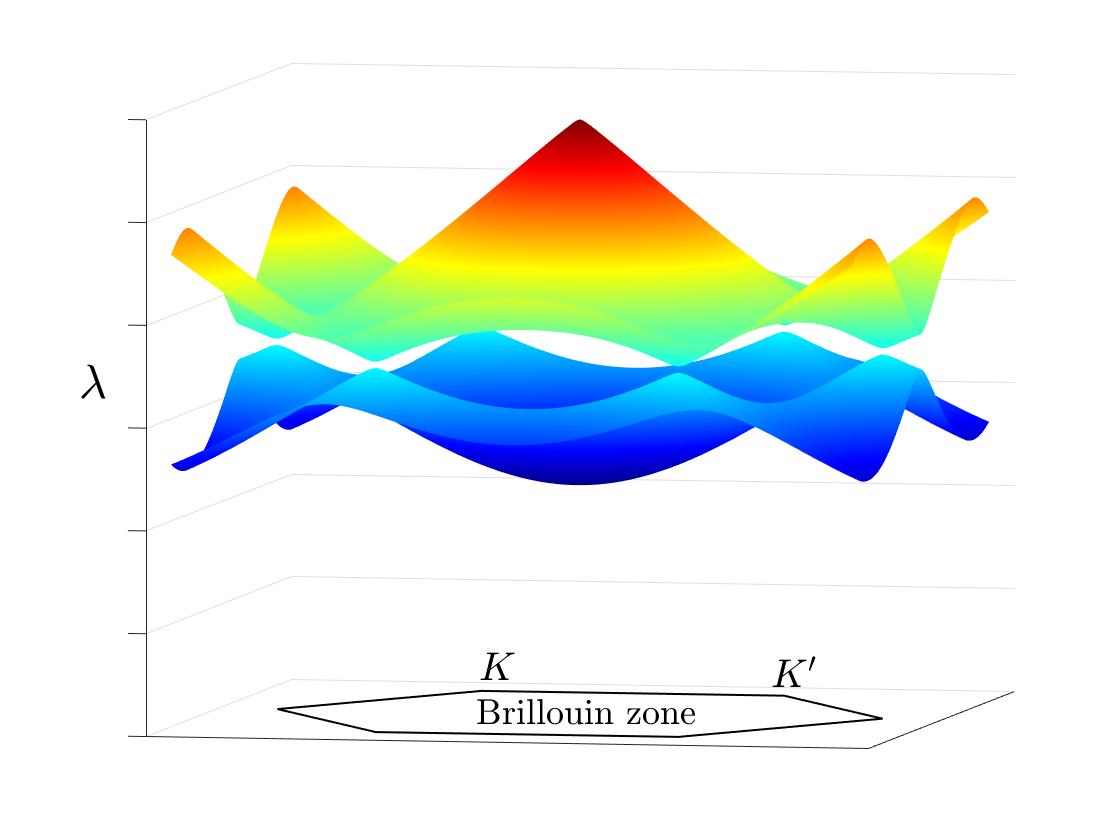}
\includegraphics[height=5cm]{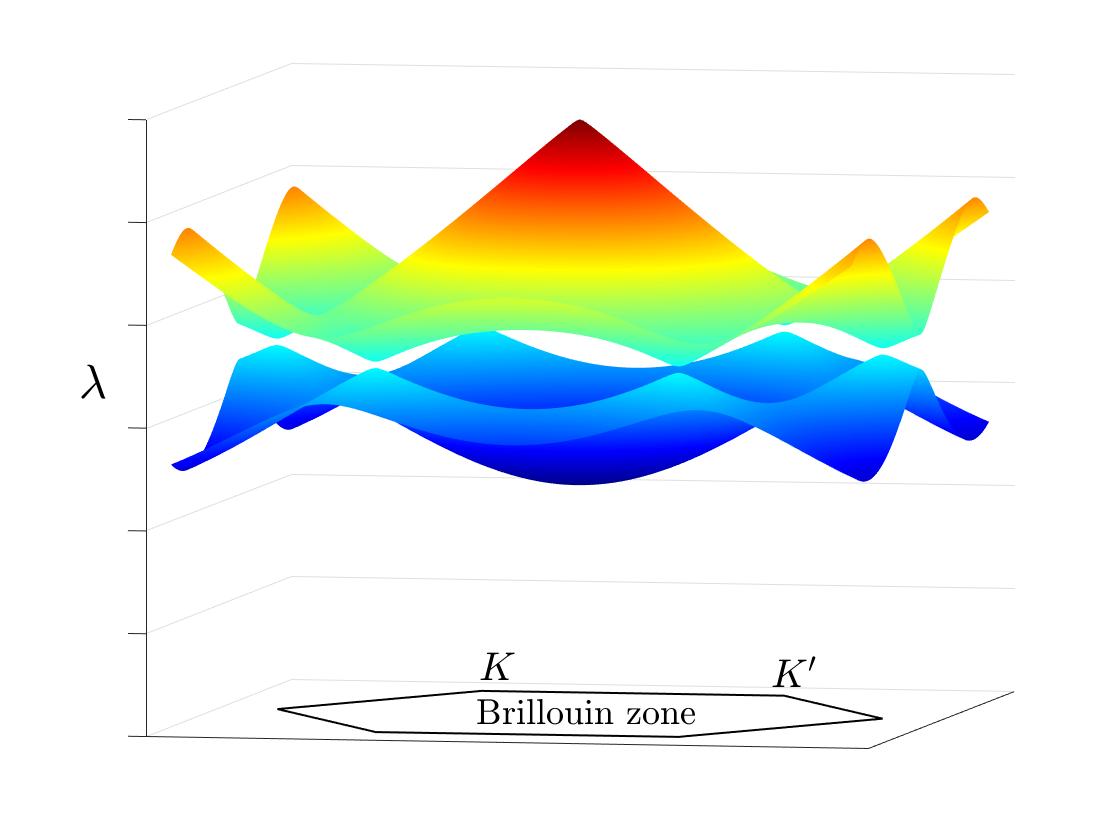}
\vspace*{-0.5cm}
\caption{The band structure of the two perturbed honeycomb lattices in Figure \ref{fig:honeycomb_perterb}. }
\label{fig:band_honeycomb_perterb}
\end{center}
\end{figure}

\begin{figure}[!htbp]
\begin{center}
\vspace*{-0.5cm}
\includegraphics[height=5cm]{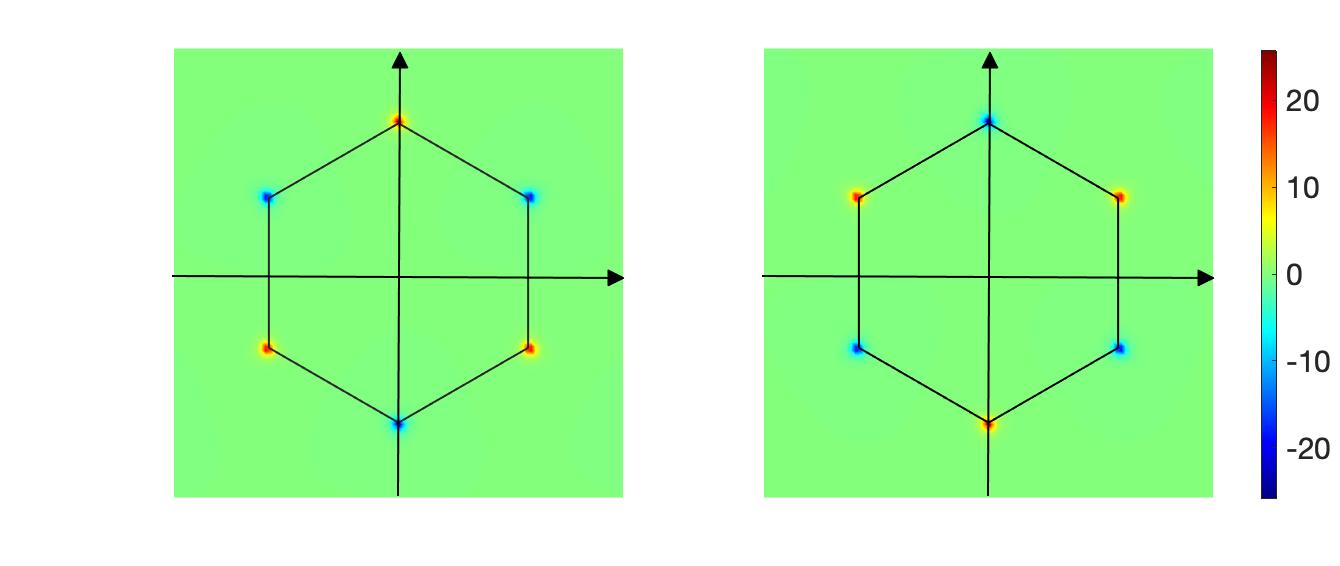}
\vspace*{-0.5cm}
\caption{The Berry curvature in the momentum space for the two perturbed honeycomb lattices in Figure \ref{fig:honeycomb_perterb}. }
\label{fig:Berry_curvature}
\end{center}
\end{figure}

\subsubsection{Interface modes in the joint honeycomb lattice}
Finally, we investigate the interface modes for the joint photonic structure formed by gluing the two perturbed honeycomb lattices together. The interface modes propagate parallel to the interface of the two media but decay along the direction perpendicular to the interface (cf. Figure \ref{fig:edge_mode}).
We consider the PDE operators for several configurations of joint photonic structures attaining different interface geometries, including the zigzag interface, the armchair interface, and the rational interfaces. The configurations of the joint structure with the zigzag and armchair interface are shown in Figure \ref{fig:honeycomb_joint}. We prove the existence of interface modes for the joint structures for each scenario, with the corresponding eigenfrequencies located in the common band gap of the two perturbed media enclosing the Dirac point. To address the sharp discontinuity of the medium coefficient across the interface, we set up a matching condition for the wave field at the interface using integral equations and investigate the characteristic values using the generalized Rouch\'e Theorem in Gohberg-Sigal theory~\cite{ts1971operator,Ammari-book}.
The method was applied to study the interface modes bifurcated from Dirac points in the topological waveguide structure recently \cite{Qiu-23} and can be employed to study interface modes in photonic structures with piecewise constant media in a general context.

\begin{figure}[!htbp]
\begin{center}
\hspace*{-30pt}
\includegraphics[height=4.5cm]{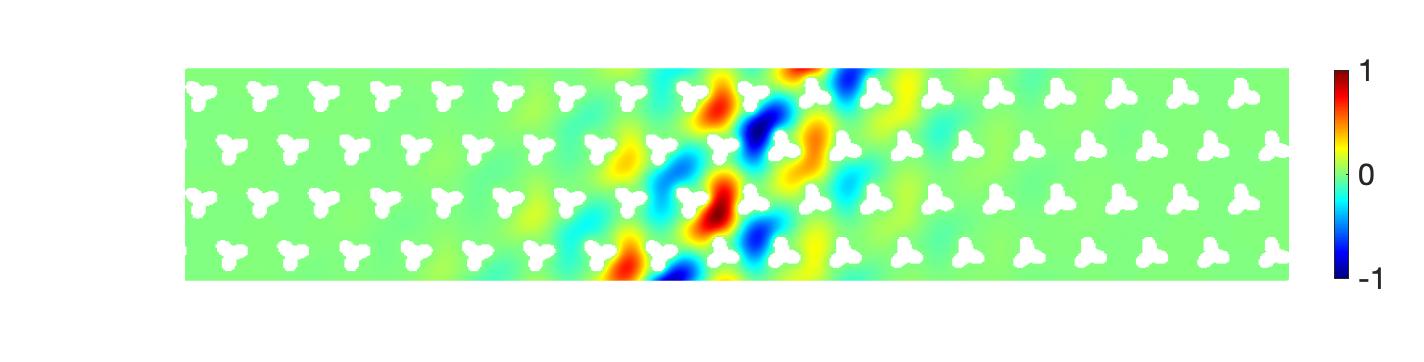}
\vspace*{-20pt}
\caption{Interface mode propagating along the interface of two perturbed honeycomb lattices. }
\label{fig:edge_mode}
\end{center}
\end{figure}

\begin{figure}[!htbp]
\begin{center}
\hspace*{-1.5cm}
\includegraphics[height=6cm]{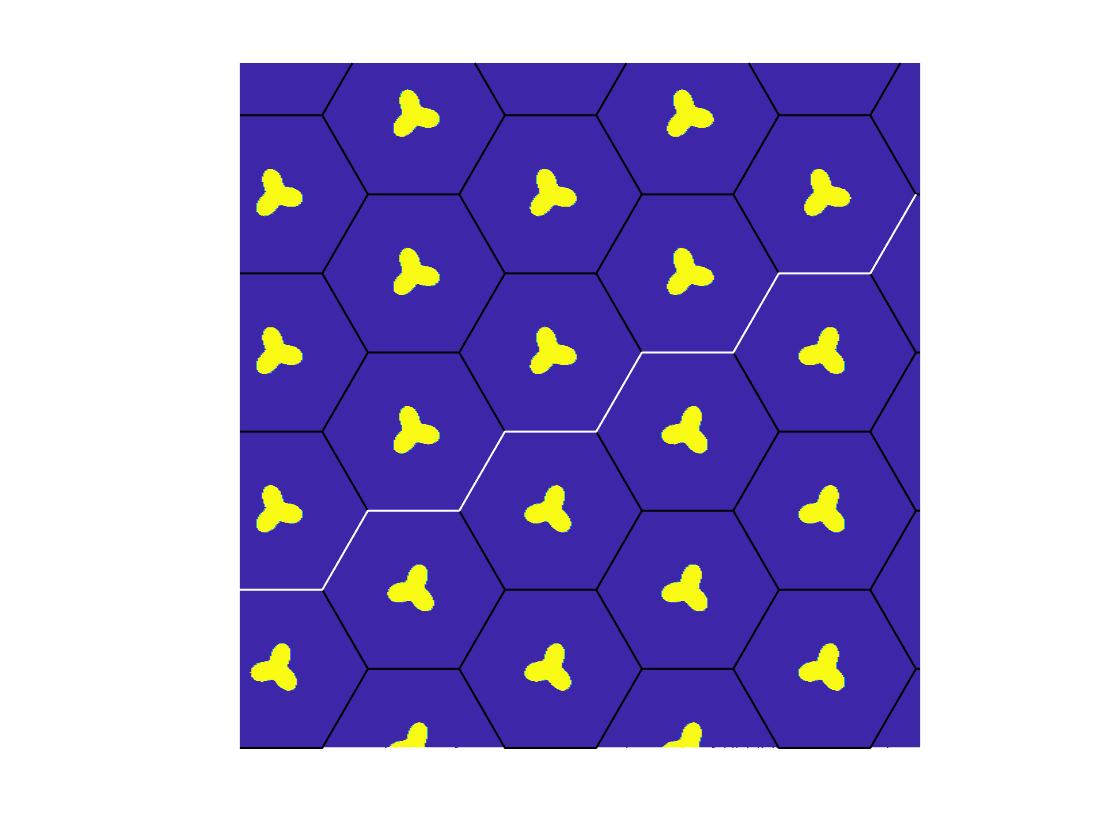}
\includegraphics[height=6cm]{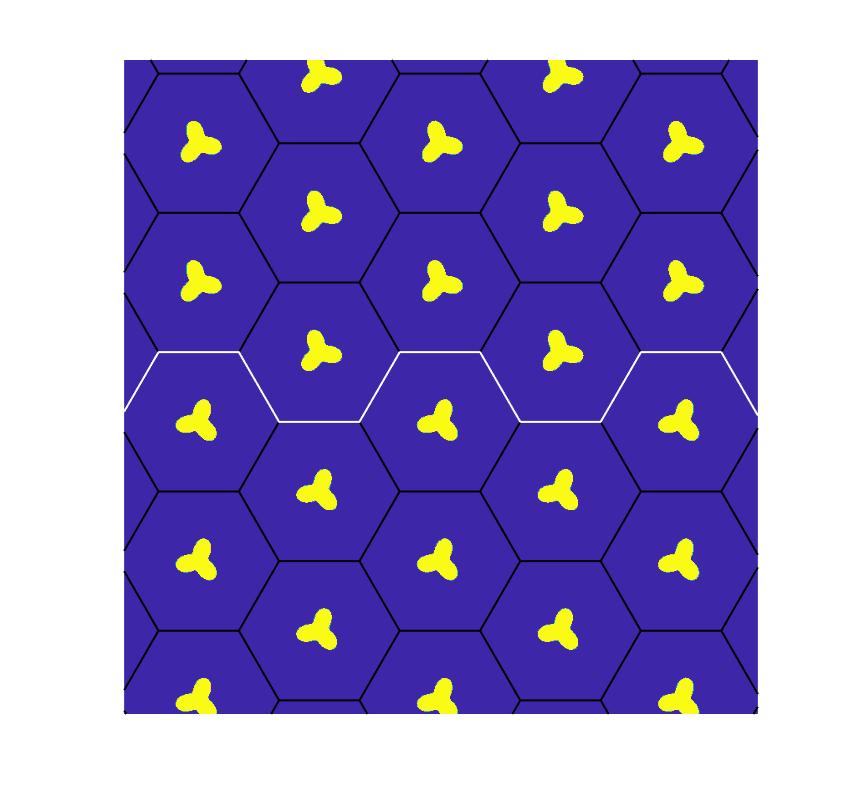}
\caption{Join photonic structures with a zigzag interface (left) and armchair interface (right). }
\label{fig:honeycomb_joint}
\end{center}
\end{figure}

\subsection{Notations}

\vskip2ex\noindent
{\bf Honeycomb lattice}\\
%
$\Lambda$, $\Lambda^*$: the honeycomb lattice and its dual lattice.\\
$\Kone$, $\Ktwo$: high symmetry points in the Brillouin zone.\\
$\tilde\Lambda^*:= \Kone + \Lambda^*$.\\
$\cC_z$: the fundamental cell of the honeycomb lattice for the zigzag interface.\\
$\be_1$, $\be_2$: generating vectors of the honeycomb lattice $\Lambda$ with the fundamental cell $\cC_z$.\\
 $\Gl$, $\Gr$, $\Gt$, $\Gb$, the left, right, top and bottom sides of $\cC_z$. See Figure \ref{fig:cellhomo}.\\
$\nu_1$, $\nu_2$: unit normal to $\Gl$ and $\Gb$. See Figure \ref{fig:cellhomo}.\\
$\bbeta_1$, $\bbeta_2$:  generating vectors of the dual lattice $\Lambda^*$ with $\bbeta_i \cdot \be_j = \delta_{ij}$.\\
$\cC_a$: the fundamental cell of the honeycomb lattice for the armchair interface.\\
$\be_1^a$, $\be_2^a$: generating vectors of the honeycomb lattice with the fundamental cell $\cC_a$.\\
$\bbeta_1^a$, $\bbeta_2^a$: generating vectors of the dual lattice with $\be_i^a \cdot \be_j^a = \delta_{ij}$.\\
$D_*$: the reference inclusion with required symmetries.\\
$D(\eta) := \eta D_*$.\\
$D := D(\eta_0)$ with a sufficiently small $\eta_0$.\\
$D^\eps$: the domain obtained by rotating $D$ by an angle of $\eps$ counterclockwise.

\vskip2ex\noindent
{\bf Layer potentials defined over the inclusion boundary}\\
$\mathcal S_0[\phi](\bx)$: single layer potential with Green function for free space Laplacian on $\partial D_*$, see \eqref{eq:SLap}. \\
$\mathcal S (\eta,\lambda,\bp)$: single layer potential with quasiperiodic Green function for Helmholtz equation on $\partial D_*$, see \eqref{eq:SHelm}.\\
$T (\eps,\lambda,\bp)$: single layer potential with quasiperiodic Green function for Helmholtz equation on $\partial D$, where $\eps$ represents the orientation, see \eqref{eq:Teps}.

\vskip2ex\noindent
{\bf Infinite strips for joint honeycomb structures}\\
$\displaystyle{\Omega^J:=\cup_{m\in\mathbb Z}\,(\cC_z +m\be_1})$: the infinite strip for the joint honeycomb structure with a zigzag interface.\\
$D^{J,\eps} := \left(\cup_{m\geq0} (D^{\eps}+m\be_1)\right) \cup \left(\cup_{m<0} (D^{-\eps}+m\be_1)\right)$: the domain of inclusions located in $\Omega^J$ for the joint honeycomb with the zigzag interface.\\
$\Omega^{J,\eps} := \Omega^J\backslash \overline{D^{J,\eps}}$. \\
$\Gamma$: the zigzag interface for the joint honeycomb structure.\\
$\Gamma_\pm:=\{\pm\frac{1}{2}\be_2+\ell\be_1, \ell\in\mathbb R\}$: the top and bottom boundaries of the domain $\Omega^J$.\\
$\displaystyle{\Omega_a^J:=\cup_{m\in\mathbb Z}\,(\cC_a +m\be_1^a})$: the infinite strip for the joint honeycomb structure with an armchair interface. \\
$D_a^{J,\eps}:= \left(\cup_{m\geq0} (D^{\eps}+m\be_1^a)\right) \cup \left(\cup_{m<0} (D^{-\eps}+m\be_1^a)\right)$: the domain of inclusions located in $\Omega_a^J$ for the joint honeycomb with the armchair interface.\\ 
$\Omega_a^{J,\eps}:= \Omega_a^J\backslash \overline{D_a^{J,\eps}}$.\\
$\Gamma^a$: the armchair interface restricted on $\Omega_a^J$.\\
$\Gamma_\pm^a:=\{\pm\frac{1}{2}\be_2^a+\ell\be_1^a, \ell\in\mathbb R\}$: the top and bottom boundaries of $\Omega_a^J$.\\
$\kp^*= \Kone\cdot\be_2 = \frac{4\pi}{3}$, ${\kp^{*}}' = \Ktwo\cdot\be_2 = -\frac{4\pi}{3}$, $\kp^{*,a}= \Kone\cdot\be_2^a = 2\pi$.\\

\vskip2ex\noindent
{\bf Energies and modes}\\
$w_i$, $i=1,2$: Bloch modes with quasimomentum $\Kone$ at the Dirac point energy $\lambda_*$. See Theorem~\ref{lem:Dirac}. \\
%
$\lambda_{n,\eps}(\bp(\ell))$: dispersion energies that are numbered from small to large.\\
$u_{n,\eps}(\bx;\bp)$: Bloch modes at energy $\lambda_{n,\eps}(\bp(\ell))$.\\
$\mu_{n,\eps}(\bp(\ell))$: dispersion energies that are smooth along $\bp(\ell)$.\\
$v_{n,\eps}(\bx;\bp(\ell))$: Bloch modes at energy $\mu_{n,\eps}(\bp)$ that are smooth along $\bp(\ell)$.\\
$\lambda_n$, $u_n$, $\mu_n$ and $v_n$: abbreviations of $\lambda_{n,0}$, $u_{n,0}$, $\mu_{n,0}$ and $v_{n,0}$. See
Section~\ref{sec:Floquet}.\\
$\vec v_i:=
\left(\begin{matrix}v_i|_\Gamma \\\partial_n v_i|_\Gamma\end{matrix}\right),\quad
i=1,2,$, see \eqref{eq:vecv}\\
$\mathfrak u_1:= \vec v_1+ \im\vec v_2$,  
$\mathfrak u_2:= \vec v_1- \im\vec v_2.$.

\vskip2ex\noindent
{\bf Green functions}\\
$G^f(\bx,\by;\lambda, \bp)$: quasi-periodic Green function for the Helmholtz equation in $\mathbb R^2$; see \eqref{eq:GQPhomo}.\\
$G^\eps(\bx,\by;\lambda)$:
Green function for the Helmholtz equation in $\mathbb R^2 \backslash\cup_{n_1,n_2\in\mathbb Z} (D^\eps + n_1\be_1+n_2\be_2)$ that is quasiperiodic in $\be_2$, see \eqref{eq:phyGeps}.

\vskip2ex\noindent
{\bf Layer potentials over the interface of the joint honeycomb lattice}\\
$\cS^\eps$, $\cD^\eps$, $\cK^{*,\eps}$, $\cK^\eps$, $\cN^\eps$: layer potentials over interface $\Gamma$, with the kernel $G^{\eps}(\bx,\by;\lambda)$.\\
$\mathbb T^\eps$, $\mathbb T_s^\eps$, $\mathbb T_t^\eps$, $\mathbb T_n^\eps$: matrix operators with layer potentials over $\Gamma$ with the kernel $G^{\eps}(\bx,\by;\lambda)$.\\
$\Omega^0$, 
$\tilde\cS^0(\lambda_*)$, $\tilde\cD^0(\lambda_*)$, $\tilde\cK^0(\lambda_*)$, $\tilde\cK^{*,0}(\lambda_*)$ and $\tilde\cN^0(\lambda_*)$: layer potentials over $\Gamma$, with kernel $\tilde{G}^0(\bx,\by;\lambda)$.\\
$\tilde{\mathbb T}^0(\lambda_*)$: matrix operator with layer potentials over $\Gamma$ with kernel $\tilde{G}^0(\bx,\by;\lambda)$.

\vskip2ex
\noindent{\bf Function spaces}\\
\eqref{eq:edgespace}
\begin{equation*}
\begin{aligned}
\cH^{J,\eps}:=\big\{&u\in H^1(\Omega^{J,\eps}): \Delta u\in L^2(\Omega^{J,\eps}), \quad u=0 \text{ on } \partial D^{J,\eps},\\
&u(\bx+ \be_2)= e^{\im \kp^*}u(\bx) \text{ for }\bx\in\Gamma_-, \quad \partial_{\nuGb} u(\bx+ \be_2) = e^{\im \kp^*}\partial_{\nuGb} u(\bx) \text{ for }\bx\in\Gamma_-\big\},
\end{aligned}
\end{equation*}
\eqref{eq:edgespace_armchair}
\begin{equation*}
\begin{aligned}
\cH_a^{J,\eps}:=\big\{&u\in H^1(\Omega_a^{J,\eps}): \Delta u\in L^2(\Omega_a^{J,\eps}), \quad u=0 \text{ on } \partial D_a^{J,\eps},\\
&u(\bx+ \be_2)= e^{\im \kp^{*,a}}u(\bx) \text{ for }\bx\in\Gamma_-^a, \quad \partial_{\nuGb} u(\bx+ \be_2^a) = e^{\im \kp^{*,a}}\partial_{\nuGb} u(\bx) \text{ for }\bx\in\Gamma_-^a\big\}.
\end{aligned}
\end{equation*}
\eqref{eq:HsD}
\begin{equation*}
H_{i}^{s}(\partial D):= \left\{ \phi\in H^s(\partial D);  R\phi(\bx):=\phi(R^{-1}\bx) = \tau^i \phi(\bx) \right\}, \quad i=0,1,2.
\end{equation*}
\eqref{eq:QP1pm}
\begin{equation*}
\begin{aligned}
\cH^{\eps}_{\text{loc}}:=\{&u\in H^1_{\text{loc}}(\Omega^{\eps}): \Delta u\in L^2_{\text{loc}}(\Omega^{\eps}), \quad u=0 \text{ on } \cup_{m\in\mathbb Z}(\partial D^{\eps} +m\be_1),\\
&u(\bx+ \be_2)= e^{\im \kp^*}u(\bx) \text{ for }\bx\in\Gamma_-, \quad \partial_{\nuGl}  u(\bx+ \be_2) =e^{\im \kp^*}\partial_{\nuGl}  u(\bx) \text{ for }\bx\in\Gamma_-\}.
\end{aligned}
\end{equation*}
\eqref{eq:QP20}
\begin{equation*}
\begin{aligned}
\cH^\eps(\ell):=\{&u\in H^1(\cunp): \Delta u\in L^2(\cpeps), \quad u=0 \text{ on } \partial D^\eps,\\
&u(\bx+ \be_2)= e^{\im \kp^*}u(\bx) \text{ for }\bx\in{\Gb}, \quad \partial_{\nuGb}  u(\bx+ \be_2) =e^{\im \kp^*}\partial_{\nuGb}  u(\bx) \text{ for }\bx\in\Gb\\
&u(\bx+ \be_1)= e^{\im (\Kone+\ell\bbeta_1)\cdot\be_1}u(\bx) \text{ for }\bx\in{\Gl},\quad
\partial_{\nuGl}  u(\bx+ \be_1) =e^{\im (\Kone+\ell\bbeta_1)\cdot\be_1}\partial_{\nuGl}  u(\bx) \text{ for }\bx\in\Gl \}.
\end{aligned}
\end{equation*}
\eqref{eq:H1Delta}
\begin{equation*}
\begin{aligned}
H^1(\Delta,\cpeps):= \{&u\in H^1(\cpeps): \Delta u\in L^2(\cpeps), \quad u=0 \text{ on } \partial D^\eps,\\
&u(\bx+ \be_2)= e^{\im \kp^*}u(\bx) \text{ for }\bx\in\Gb, \quad \partial_{\nuGb}  u(\bx+ \be_2) =e^{\im \kp^*}\partial_{\nuGb}  u(\bx) \text{ for }\bx\in\Gb
\}.
\end{aligned}
\end{equation*}
\eqref{eq:HsG2}
\begin{equation*}
\cH^s(\Gamma):= \left\{ u(\bx_0+ t\be_2) = \sum_{n\in\mathbb Z} a_n e^{\im \kp^* t} e^{\im 2\pi n t}: \|u\|_{\cH^s(\Gamma)}^2:=\sum_{n\in\mathbb Z} |a_n|^2 (1+|n|^2)^{s/2}\right\}.
\end{equation*}

\section{Main results}\label{sec:mainresults}
 The first main result of this work is the existence and asymptotic analysis of Dirac points for a family of honeycomb lattices of impenetrable obstacles with Dirichlet boundary conditions. We assume that the shape of each obstacle in the honeycomb lattice attains $120^\circ$-rotation symmetry and horizontal reflection symmetry.
The second main result is the existence and the number of interface modes for a joint photonic structure formed by gluing two lattices perturbed from the honeycomb lattice attaining Dirac points along an interface.
Our results cover the case of a zigzag interface, an armchair interface, and a rational interface. 
These interface modes are quasi-periodic along the direction of the interface but decay in the direction transverse the interface direction.
We also derive the dispersion relation of the interface modes with respect to the quasi-momentum along the interface.

\subsection{The honeycomb lattice of impenetrable obstacles}
As illustrated in Figure \ref{fig:honeycomb}, an infinite array of impenetrable obstacles are arranged periodically over
the honeycomb lattice $$ \Lambda := \bbZ \be_1 \oplus  \bbZ \be_2 := \{ \ell_1 \be_1 + \ell_2 \be_2 :  \ell_1, \ell_2 \in \bbZ \}, $$
wherein the lattice vectors
\begin{equation*}
\be_1=a(\frac{\sqrt3}{2},-\frac{1}{2})^T, \quad \be_2=a(\frac{\sqrt3}{2},\frac{1}{2})^T.
\end{equation*}
In what follows, without loss of generality, we assume that the lattice constant $a=1$. 
Let 
\begin{equation}\label{eq:cellz}
\cC_z := \{ \ell_1 \be_1 + \ell_2 \be_2 :  \ell_1, \ell_2\in [-1/2,1/2) \}
\end{equation}
be the fundamental cell of the lattice.
Let $ D_*\ssubset \cC_z$ be a connected smooth domain that is invariant under the $\frac{2\pi}{3}$-rotation transform $R$ and the horizontal reflection transform $\rflc$ given by
\begin{equation}\label{eq:RF}
R\bx := \left(\begin{matrix}-\frac{1}{2} & \frac{\sqrt3}{2} \\ -\frac{\sqrt3}{2} & -\frac{1}{2} \\ \end{matrix}\right)\bx, \quad
\rflc (x_1,x_2) = (-x_1,x_2).
\end{equation}
Denote the scaled inclusion 
$D(\eta):=\{\bx =\eta\bx', \bx'\in \partial D_*\}$ for $\eta\in (0,1)$. 

\begin{figure}[h]
\begin{center}
\includegraphics[height=4cm]{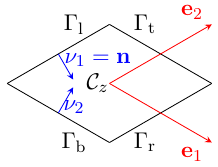} 
\includegraphics[height=5cm]{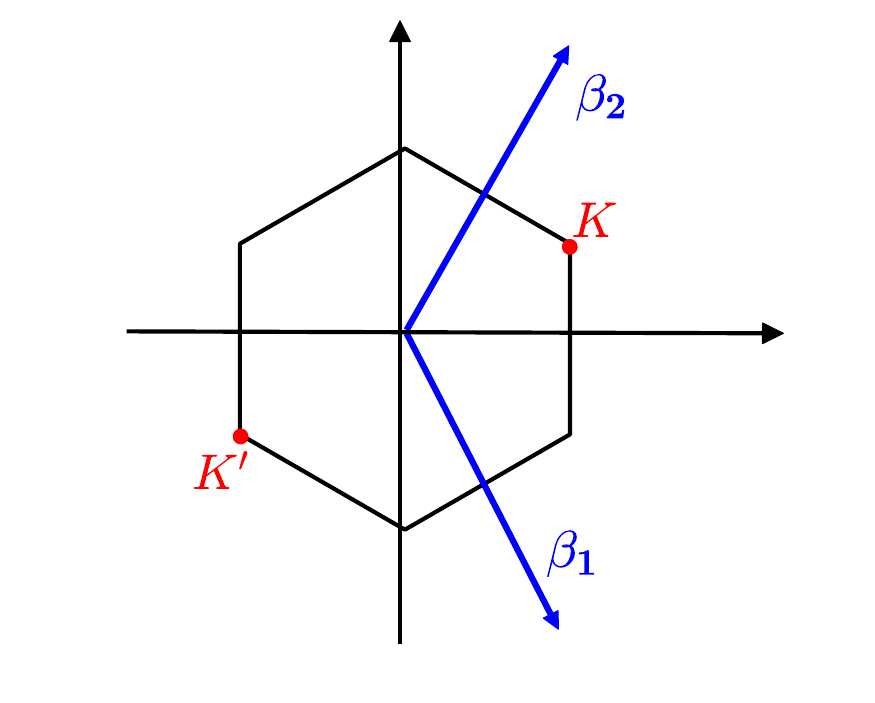}
\caption{The fundamental cell $\cC_z$ (left) and the Brillouin zone (right).}
\label{fig:cellhomo}
\end{center}
\end{figure}

Let
\begin{equation*}
\Lambda^* = \{ 2\pi\ell_1 \bbeta_1 + 2\pi\ell_2 \bbeta_2 :  \ell_1, \ell_2\in\mathbb Z\},
\end{equation*}
be the reciprocal lattice, where the reciprocal lattice vectors 
\begin{equation}
\bbeta_1=(\frac{1}{\sqrt3},-1)^T, \quad \bbeta_2=(\frac{1}{\sqrt3},1)^T
\end{equation}
satisfy $\be_i\cdot\bbeta_j=\delta_{ij}$, $i,j=1,2$.
The hexagon-shaped fundamental cell in $\Lambda^*$, or the Brillouin zone, is denoted by $\cB_z$
as shown in Figure \ref{fig:cellhomo} (right).
The high symmetry points located at the vertices of the Brillouin zone are given by 
$$ K:=2\pi(\frac{1}{\sqrt3},\frac{1}{3})=2\pi(\frac{1}{3}\bbeta_1+\frac{2}{3}\bbeta_2), \; K':=-K, \; RK, \;  R^2K, \; RK' \; \mbox{and}\; R^2K'. $$

\subsection{Dirac points for the honeycomb lattice}\label{subsec:Dirac_point}
Following the Floquet-Bloch theory \cite{Kuchment-12},
for each $\bp\in\cB_z$, we consider the following eigenvalue problem:
\begin{equation}\label{eq:bandu_zigzag}
\begin{aligned}
 - \Delta  u(\bx; \bp) - \lambda u(\bx; \bp)&=0,  \quad &\bx\in \chomo\backslash D(\eta) + \Lambda, \\
u(\bx; \bp)&=0,  \quad  &\bx\in \partial D(\eta) + \Lambda, \\
u(\bx+\be; \bp)&=e^{i \bp \cdot \be}u(\bx; \bp), & \mbox{for} \; \be \in \Lambda.
\end{aligned}
\end{equation}
For each $\bp$, the eigenvalues can be ordered by $\lambda_1(\bp) \le \lambda_2(\bp) \le \cdots \le \lambda_n(\bp) \le \cdots $. As $\bp$ varies in the Brillouin zone $\cB_z$, one obtains the band structure of the honeycomb lattice.

We define the $\eta^2$-vicinity of $|\Kone|$:
\begin{equation}\label{eq:vincexa}
U_{\eta}: = \left\{ \lambda \in \mathbb C:  \frac{(2\pi)^2}{ 3 |\cC_z|} \mathfrak a\eta^2 \le \left|\lambda-|\Kone|^2 \right| \le \frac{(2\pi)^2}{ |\cC_z|} \mathfrak a\eta^2 \right\}.
\end{equation}
where $\mathfrak a \neq 0$ is a complex number defined by \eqref{express-a}. 
The main results regarding the Dirac points at $\bp=K$ and $K'$ are stated below.

\begin{thm}\label{lem:Dirac}
If Assumption \ref{assumpS0} holds, then for $\eta$ sufficiently small but nonzero, there exists a Dirac point at $(\Kone,\lambda_*)$ in the band structure of the honeycomb lattice $D(\eta) + \Lambda$ with $\lambda_*\in U_\eta$. 
The dispersion surface near $(\Kone, \lambda_*)$ takes the form  
\begin{equation}
(\lambda-\lambda_*)^2 = m_*^2 \, |\bp-\Kone|^2+O(|\bp-\Kone|^3),\quad m_*\in\mathbb R, \quad m_*\geq0,
\end{equation}
where the slope of the Dirac cone is 
\begin{equation}
m_* = \frac{2}{3}(1+O(\eta)).
\end{equation}
In addition, the basis of the eigenspace at the Dirac point $(\Kone,\lambda_*)$ can be chosen as $w_1$ and $w_2$ that satisfy
\begin{equation}
Rw_1(\bx):=w_1(R^{-1}\bx) = \tau w_1(\bx),\quad Rw_2(\bx):=w_2(R^{-1}\bx) = \overline{\tau}  w_2(\bx),\quad w_2(\bx)=Fw_1(\bx):=w_1(\rflc\bx),
\end{equation}
in which $\tau = e^{\im\frac{2\pi}{3}}$.

\end{thm}

\begin{remark}
 Recall that $R$ and $\rflc$ defined in \eqref{eq:RF} denote the $\frac{2\pi}{3}$-rotation operation and the horizontal reflection operation acting on vectors in $\mathbb{R}^2$. Here and henceforth, for convenience of notation, we also use $R$ and $\rflc$ to denote the rotation and reflection operators acting on functions. The meaning of the notations should be clear in the context, depending on whether they are applied to vectors or functions. 
\end{remark}

\begin{remark}
In this work, we only prove the existence of Dirac points for the first two bands and when $\eta \ll 1$. It can be shown that Dirac points also exist for higher bands and for $\eta$ not small. This is not the focus of this work and will be reported separately elsewhere.
\end{remark}

\begin{coro}\label{lem:Ktwo}
For $\eta \ll 1$,  $(\Ktwo,\lambda_*)$ is also a Dirac point with the corresponding eigenspace spanned by
\begin{equation}
w_1'(\bx) := \bar w_2(\bx),\quad w_2'(\bx):=\bar w_1(\bx),
\end{equation}
which attain the following symmetry relations:
\begin{equation}
Rw_1'(\bx):=w_1'(R^{-1}\bx) = \tau w_1'(\bx),\quad Rw_2'(\bx):=w_2'(R^{-1}\bx) = \overline{\tau}w_2'(\bx),\quad w_2'(\bx)=w_1'(\rflc\bx).
\end{equation}
\end{coro}

\subsection{Band-gap opening at Dirac points}
The existence of Dirac points as established in the previous subsection is due to the $\frac{2\pi}{3}$-rotation symmetry and the horizontal reflection symmetry of the lattice structure. Under suitable perturbations that break one of these symmetries, the Dirac points will disappear and a bandgap can be opened therein. We show that this is indeed the case 
when the obstacles in the honeycomb lattice are rotated to their centers with an angle of $\pm\eps$ (cf. Figure \ref{fig:honeycomb_perterb}) so that the horizontal reflection symmetry of the lattice structure is broken. 

In subsequent analysis, we fix $D=D(\eta_0)$ by fixing a small enough $\eta_0>0$ such that Theorem~\ref{lem:Dirac} holds for $D(\eta_0)$. Denote by $D^{\pm\eps}$ the domain obtained by rotating $D$ by an angle of $\pm\eps$ counterclockwise.
We consider the following eigenvalue problem for each $\bp\in\cB_z$:
\begin{equation}\label{eq:bandu_eps}. 
\begin{aligned}
& - \Delta  u_{\pm\eps}(\bx; \bp) - \lambda u_{\pm\eps}(\bx; \bp)=0,  \quad &\bx\in \chomo\backslash D^{\pm\eps} + \Lambda, \\
&u_{\pm\eps}(\bx; \bp)=0,  \quad  &\bx\in \partial D^{\pm\eps} + \Lambda, \\
&u_{\pm\eps}(\bx+\be; \bp)=e^{i \bp \cdot \be}u_{\pm\eps}(\bx; \bp), & \mbox{for} \; \be \in \Lambda.
\end{aligned}
\end{equation}

\begin{thm}\label{thm:band_gap}
Let the constants $t_*$ and $\gamma_*$ be defined as in \eqref{eq:Tderiv} and assume $t_*>0$.
Then the following dispersion relations hold for $\bp$ near $K$ and $\lambda$ near $\lambda_*$:
\begin{equation}\label{eq:lambda_near_K}
\begin{aligned}
\lambda_{1,\pm\eps}(\bp) &= \lambda_*  - \frac{1}{|\gamma_*|}\sqrt{\eps^2 t_*^2 + m_*^2|\gamma_*|^2 | \bp - K |^2}\,\big(1+O(\eps,| \bp - K |)\big), \\
\lambda_{2,\pm\eps}(\bp) &= \lambda_*  + \frac{1}{|\gamma_*|}\sqrt{\eps^2 t_*^2 + m_*^2 |\gamma_*|^2 | \bp - K |^2}\,\big(1+O(\eps,| \bp - K |)\big).
\end{aligned}
\end{equation}
In addition, the corresponding Bloch modes at $\bp=K$ attain the following expansions:
\begin{equation}\label{eq:bloch_modes_K}
\begin{aligned}
& u_{1,\eps}(\bx;K) = w_1 + O(\eps), \quad
u_{2,\eps}(\bx;K) = w_2 + O(\eps)  \\
& u_{1,-\eps}(\bx;K) = w_2 + O(\eps), \quad
u_{2,-\eps}(\bx;K) = w_1 + O(\eps),
\end{aligned}
\end{equation}
in which $w_1$ and $w_2$ are defined in Theorem~\ref{lem:Dirac}.
\end{thm}
Theorem \eqref{thm:band_gap} can be concluded from Proposition~\ref{lem:pmlamgen} with $\ell=0$.
From the symmetry of the lattice,  similar expansions hold for $\bp$ near $K'$ and $\lambda$ near $\lambda_*$.
In view of \eqref{eq:lambda_near_K}, when $\eps\neq 0$, there holds $\lambda_{1,\pm\eps}(\bp)<\lambda_{2,\pm\eps}(\bp)$ for $\bp$ near $K$, thus a spectral gap is opened. The expansion \eqref{eq:bloch_modes_K} demonstrates the swap of the eigenspace at $\bp=K$ when the obstacles are rotated with the opposite rotation parameter $\pm\eps$.

\begin{remark}
 The assumption that $t_* \neq 0$ can be verified numerically for the structure considered in this work.
\end{remark}

\subsection{Interface modes for the joint photonic structure along a zigzag interface}\label{sec:geometryz} 
We investigate interface modes for the joint photonic structure with the zigzag interface shown in Figure \ref{fig:honeycomb_joint_zigzag}. The obstacles are rotated with an angle of $-\eps$ and $\eps$ about the origin respectively for the semi-infinite honeycomb lattice on the left and right side of the interface.   

Note that the direction of the interface is parallel to $\be_2$. Employing the Floquet theory along $\be_2$, we can restrict our studies to the infinite strip $\displaystyle{\Omega^J:=\cup_{m\in\mathbb Z}\,(\cC_z +m\be_1})$, which is a fundamental period of the joint photonic structure along the interface direction. Inside the strip $\displaystyle{\Omega^J}$, the region occupied by the inclusions is denoted by
\begin{equation*}
D^{J,\eps} := \left(\cup_{m\geq0} (D^{\eps}+m\be_1)\right) \cup \left(\cup_{m<0} (D^{-\eps}+m\be_1)\right),
\end{equation*}
and the region exterior to the inclusions is denoted by
$\Omega^{J,\eps} := \Omega^J\backslash\overline{D^{J,\eps}}$.
We also denote the lower boundary of the infinite strip $\displaystyle{\Omega^J}$ by $\Gamma_-:=\{-\frac{1}{2}\be_2+\ell\be_1, \ell\in\mathbb R\}$, then the upper boundary of the strip is $\Gamma_+=\be_2+\Gamma_-$. The normal direction on $\Gamma_\pm$ is $\nuGb=(\frac{1}{2},\frac{\sqrt3}{2})$.

\begin{figure}[!htbp]
\begin{center}
\includegraphics[height=7cm]{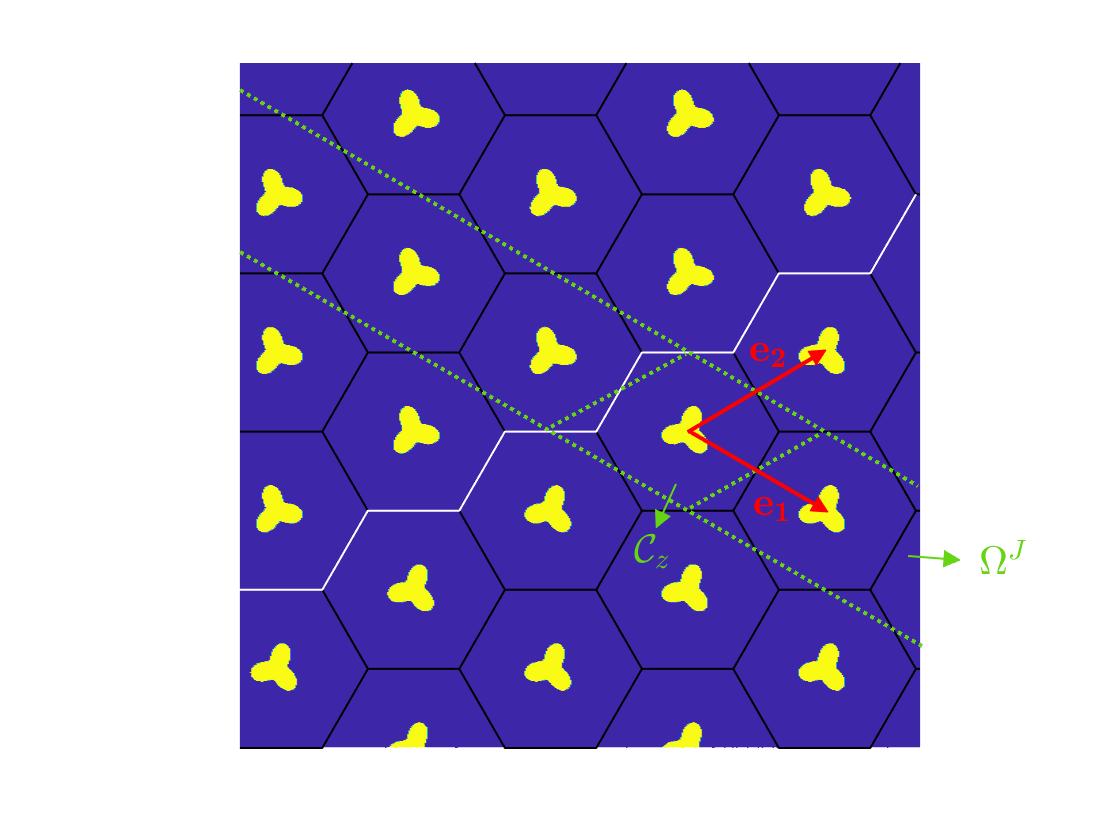}
\caption{Joined photonic structure with a zigzag interface.}
\label{fig:honeycomb_joint_zigzag}
\end{center}
\end{figure}

An interface mode $u\in L^2(\Omega^{J,\eps})$ for the joint photonic structure solves the following spectral problem:
\begin{equation}\label{eq:spectral_prob_zigzag}
\begin{aligned}
-\Delta u -\lambda u &= 0, \quad&\text{in } \Omega^{J,\eps}, \\
u&=0, \quad&\text{ on } \partial D^{J,\eps}, \\
u(\bx+\be_2) &= e^{\im \kp}u(\bx), \quad& \bx \in \Gamma_- \\
\partial_{\nuGb} u(\bx+ \be_2) &= e^{\im \kp}\partial_{\nuGb}u(\bx), \quad& \bx \in \Gamma_-.
\end{aligned}
\end{equation}
In the above, $\kp\in (0, 2\pi)$ is the quasi-momentum of the interface modes along the interface, and $\partial_{\nuGb}$ is normal derivative to $\Gamma_{-}$.

We first focus on interface modes with the quasi-momentum $\kp^* = \Kone\cdot\be_2 = \frac{4\pi}{3}$ by projecting the Bloch wave vector $K$ onto the direction of the interface $\be_2$. Namely, we investigate the interface modes bifurcated from the Dirac point $(K,\lambda^*)$.
The interface modes with other quasi-momenta will be discussed in Section \ref{sec:dispersion_edge_modes}.
To this end, we introduce the following function space 
\begin{equation}\label{eq:edgespace}
\begin{aligned}
\cH^{J,\eps}:=\big\{&u\in H^1(\Omega^{J,\eps}): \Delta u\in L^2(\Omega^{J,\eps}), \quad u=0 \text{ on } \partial D^{J,\eps},\\
&u(\bx+ \be_2)= e^{\im \kp^*}u(\bx) \text{ for }\bx\in\Gamma_-, \quad \partial_{\nuGb} u(\bx+ \be_2) = e^{\im \kp^*}\partial_{\nuGb} u(\bx) \text{ for }\bx\in\Gamma_-\big\}.
\end{aligned}
\end{equation}
Then an interface mode $u\in L^2(\Omega^{J,\eps})$ satisfies
\begin{equation}\label{eq:edgedef}
-\Delta u -\lambda u = 0 \quad\text{in } \Omega^{J,\eps} \; \mbox{and} \; u\in \cH^{J,\eps}.
\end{equation}

\begin{assumption}[The no-fold condition along the direction $\bbeta$]\label{lem:assNoFold}
Let $\bbeta\in\mathbb R^2$ be a fixed Bloch wave vector and $\lambda_*$ be the energy of the Dirac point at $\Kone$ and $\Ktwo$ introduced in Theorem~\ref{lem:Dirac}. For $\bp \in \{\Kone+\ell\bbeta$, $\ell\in \mathbb R\}$, the band energy of \eqref{eq:bandu_zigzag}
takes the value $\lambda_* $ only when 
$\bp\in (\Kone+\Lambda^*)\cup(\Ktwo+\Lambda^*)$. 
\end{assumption}

\begin{figure}[!htbp]
\begin{center}
\includegraphics[width=8cm]{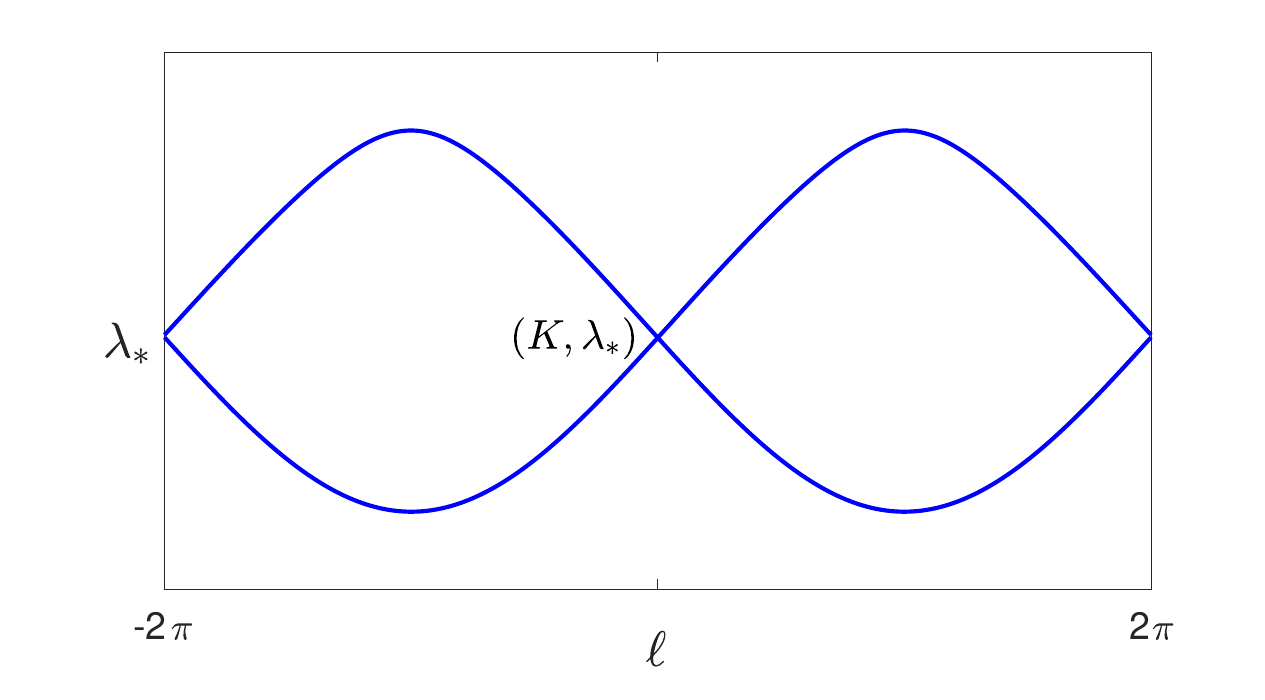}
\includegraphics[width=8cm]{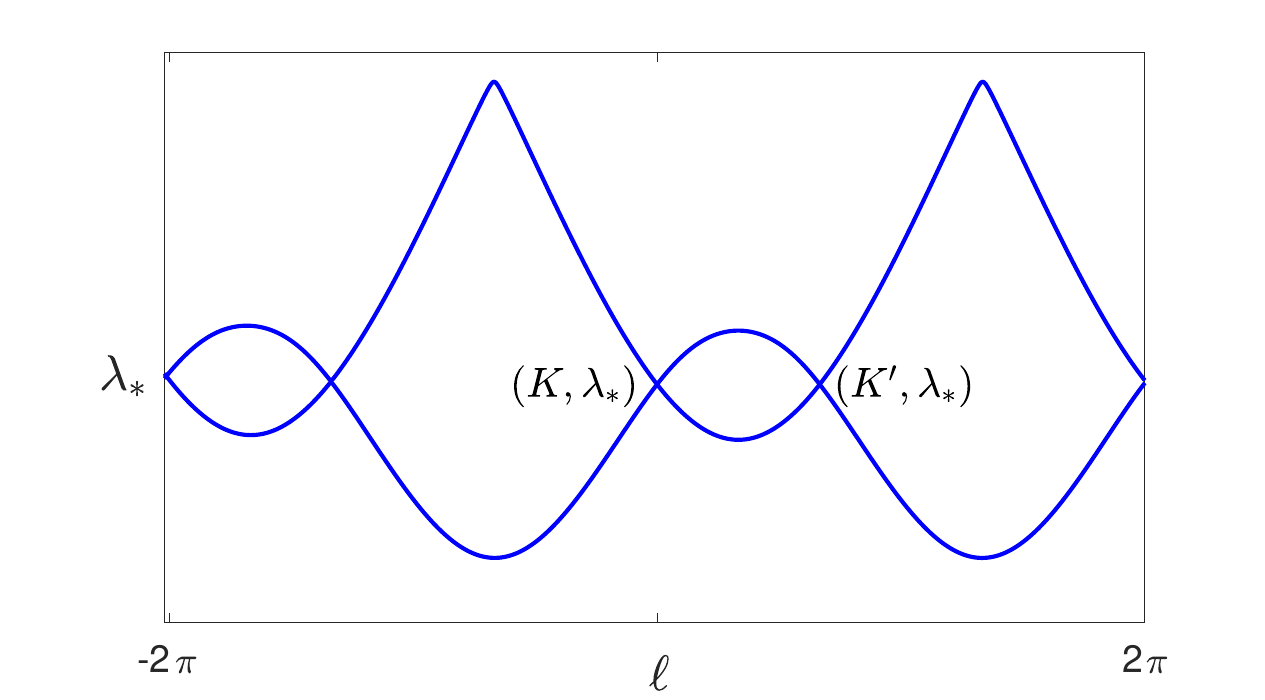}
\caption{The band structure of the spectral problem \eqref{eq:bandu_zigzag} for $\bp \in \{\Kone+\ell\bbeta$, $\ell\in [-2\pi, 2\pi]\}$. Left: $\bbeta=\bbeta_1:=(\frac{1}{\sqrt3},-1)^T$; Right: $\bbeta=\bbeta_1^a:=(0, -2)^T$. }
\label{fig:no_fold}
\end{center}
\end{figure}

\begin{remark}
     The above no-fold condition holds for the configuration of the periodic structures considered in this work. Indeed, the first two bands of the spectral problem \eqref{eq:bandu_zigzag} touch at the  Dirac point $(K,\lambda_*)$. Moreover, the energy  $\lambda_*$ is the maximum of the eigenvalues for the first band and the minimum of the eigenvalues for the second band.
     This can be rigorously proved when the inclusion size $\eta$ is small by using the layer potential technique and asymptotic analysis. The general case, when $\eta$ is not necessarily small, is beyond the scope of this work. Instead, we demonstrate numerically the no-fold conditions in Figure \ref{fig:no_fold} that are used in Theorems \ref{lem:edge} and \ref{lem:edgearm}, wherein $\bbeta=\bbeta_1$ and $\bbeta_1^a$ respectively.
\end{remark}

\begin{thm}\label{lem:edge} 
Let Assumption~\ref{lem:assNoFold} hold along the reciprocal lattice vector $\bbeta_1$.  Let $t_*$ and $\gamma_*$ be the two constants defined in \eqref{eq:Tderiv} and assume that $t_*\neq0$. Let $\mathfrak d$ be an arbitrary constant in $(0,1)$.
For sufficiently small positive $\eps$, there exists a unique interface mode $u\in L^2(\Omega^{J,\eps})$ satisfying \eqref{eq:edgedef} with the corresponding eigenvalue $\lambda\in (\lambda_* - \hrad\eps, \lambda_* + \hrad\eps)$.
In addition, the interface mode $u$ decays exponentially as $|\bx\cdot\be_1|\to \infty$.
\end{thm}

Define the quasi-momentum ${\kp^{*}}' = \Ktwo\cdot\be_2 = -\frac{4\pi}{3}$.
By the time-reversal symmetry of the differential operator, the following corollary is a direct consequence of Theorem~\ref{lem:edge}.
\begin{coro}
\label{lem:edgetwo} 
Under the same assumptions as in Theorem~\ref{lem:edge}, for sufficiently small positive $\eps$, there exists a unique interface mode  $u\in L^2(\Omega^{J,\eps})$ satisfying \eqref{eq:spectral_prob_zigzag} with ${\kp^{*}}' = - \frac{4\pi}{3}$
and the eigenvalue $\lambda\in (\lambda_* - \hrad\eps, \lambda_* + \hrad\eps)$. Furthermore, $u$ decays exponentially as $|\bx\cdot\be_1|\to \infty$.
\end{coro}

\begin{remark}
Numerical experiment demonstrates that the interface mode persists for $\eps$ not small. This will be analyzed rigorously in the future work.   
\end{remark}

\subsection{Interface modes along an armchair interface}
We consider interface modes for the joint photonic structure with an armchair interface as shown in Figure \ref{fig:honeycomb_joint_armchair}.  The inclusions above the interface are rotated to their centers with an angle of $-\eps$, while the ones below are roated with an angle of $\eps$. Note that the direction of the interface is along the $x_1$ axis, we rewrite the honeycomb lattice  equivalently as $$ \Lambda := \bbZ \be_1^a \oplus  \bbZ \be_2^a := \{ \ell_1 \be_1^a + \ell_2 \be_2^a :  \ell_1, \ell_2 \in \bbZ \}, $$
in which the lattice vectors are given by
\begin{equation}
\be_1^a=\be_1=(\frac{\sqrt3}{2},-\frac{1}{2})^T, \quad \be_2^a:=\be_1+\be_2=(\sqrt{3},0)^T.
\end{equation}
Correspondingly, the fundamental periodic cell is
\begin{equation}
\cC_a := \{ \ell_1 \be_1^a + \ell_2 \be_2^a :  \ell_1, \ell_2\in[-1/2,1/2) \},
\end{equation}
and the reciprocal lattice vectors are 
\begin{equation*}
\bbeta_1^a=(0,-2)^T, \quad \bbeta_2^a=(\frac{1}{\sqrt3},1)^T.
\end{equation*}
We assume that the inclusions $D$ and $D^{\pm\eps}$ are strictly included in the cell $\cC_a$. Similar to the zigzag interface, we introduce the infinite-strip domain $\displaystyle{\Omega_a^J:=\cup_{m\in\mathbb Z}\,(\cC_a +m\be_1^a})$ as the fundamental period for the joint photonic structure, which consists of the inclusions $D_a^{J,\eps}:= \left(\cup_{m\geq0} (D^{\eps}+m\be_1^a)\right) \cup \left(\cup_{m<0} (D^{-\eps}+m\be_1^a)\right)$ and their complement $\Omega_a^{J,\eps}:= \Omega_a^J\backslash D_a^{J,\eps}$. 

\begin{figure}[!htbp]
\begin{center}
\includegraphics[height=7cm]{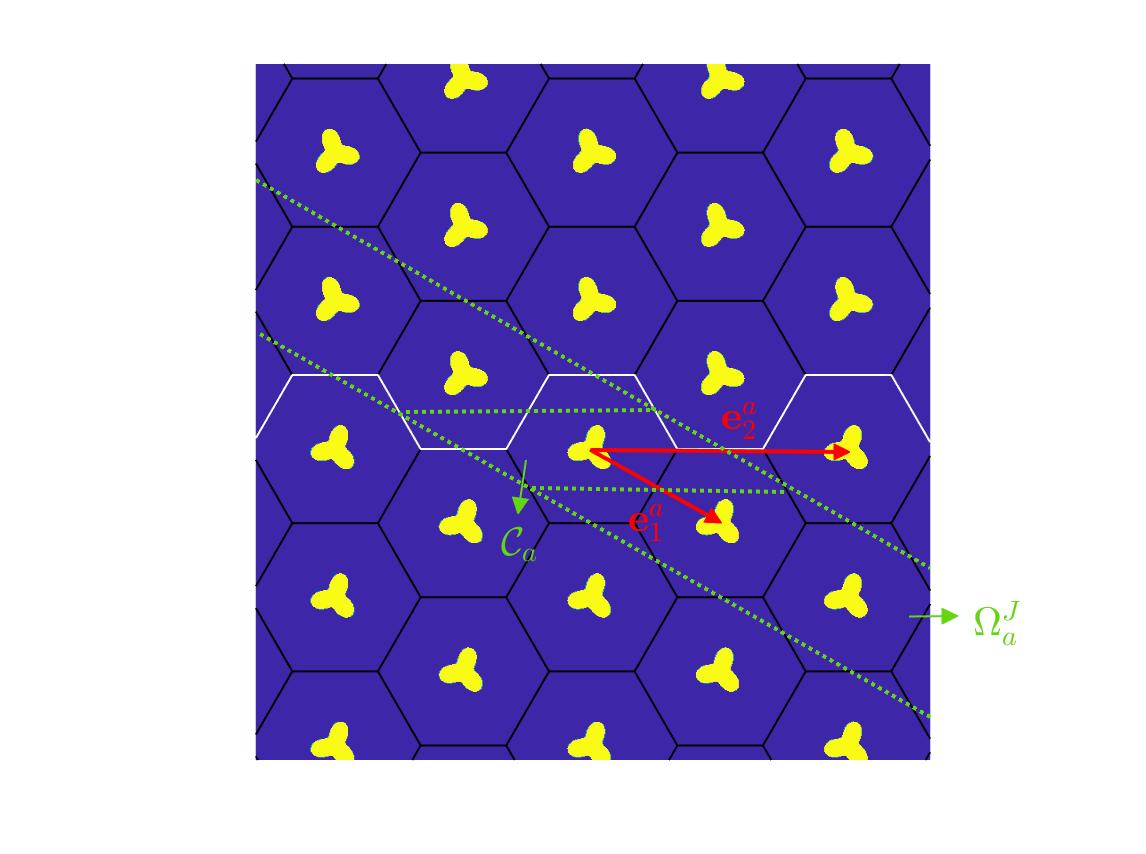}
\caption{Join photonic structure with an armchair interface.}
\label{fig:honeycomb_joint_armchair}
\end{center}
\end{figure}

We now investigate the interface modes bifurcated from the Dirac point $(K,\lambda^*)$ that propagate along the interface direction $\be_2^a$ with the quasi-momentum $\kp^{*,a} = \Kone\cdot\be_2^a = 2\pi$. Correspondingly, we define the function space
\begin{equation}\label{eq:edgespace_armchair}
\begin{aligned}
\cH_a^{J,\eps}:=\big\{&u\in H^1(\Omega_a^{J,\eps}): \Delta u\in L^2(\Omega_a^{J,\eps}), \quad u=0 \text{ on } \partial D_a^{J,\eps},\\
&u(\bx+ \be_2)= e^{\im \kp^{*,a}}u(\bx) \text{ for }\bx\in\Gamma_-^a, \quad \partial_{\nuGb} u(\bx+ \be_2^a) = e^{\im \kp^{*,a}}\partial_{\nuGb} u(\bx) \text{ for }\bx\in\Gamma_-^a\big\}.
\end{aligned}
\end{equation}
Here $\Gamma_-^a:=\{-\frac{1}{2}\be_2^a + \ell\be_1^a, \ell\in\mathbb R\}$ is the lower boundary of the strip $\Omega_a^J$.
Then an interface mode with the quasimomentum $\kp^{*,a} = 2\pi$ and energy $\lambda$ solves
\begin{equation}\label{eq:edgedef_armchair}
-\Delta u -\lambda u = 0 \quad\text{in } \Omega_a^{J,\eps} \; \mbox{for} \; u\in \cH_a^{J,\eps}.
\end{equation}

Note that the quasi-momenta $\bp$ satisfying $\bp\cdot \be_2^a = \Kone\cdot\be_2^a$ lies on the line $\bp(\ell) = \Kone + \ell\bbeta_1^a$ for $\ell\in \mathbb{R}$. We have the following results similar to Theorem~\ref{lem:edge} for the spectral problem \eqref{eq:edgedef_armchair}. 

\begin{thm}\label{lem:edgearm} 
Let Assumption~\ref{lem:assNoFold} hold along the reciprocal lattice vector $\bbeta_1^a$. 
Let $t_*$ and $\gamma_*$ be the two constants defined in \eqref{eq:Tderiv} and assume that $t_*\neq0$. Let $\mathfrak d$ be an arbitrary constant in $(0,1)$. 
For sufficiently small positive $\eps$, there exist exactly two interface modes with $\kp^{*,a} = 2\pi$, with the corresponding eigenvalues $\lambda_\pm \in (\lambda_* - \hrad\eps, \lambda_* + \hrad\eps)$. 
In addition, both interface modes decay exponentially as $|\bx\cdot\be_1^a|\to \infty$.
\end{thm}

\subsection{Dispersion relations of the interface modes}\label{sec:dispersion_edge_modes}

We consider interface modes with quasi-momentum $\kp$ near $\kp^*$ or $\kp^{*,a}$. In particular, we derive the leading order of the dispersion relation $\lambda(\kp)$ for the interface modes along a zigzag or armchair interface for $\kp$ near $\kp^*$ or $\kp^{*,a}$. The dispersion curve $\lambda(\kp)$ over the whole Bloch interval $[0, 2\pi]$ for both configurations are shown in Figure \ref{fig:edge_mode_dispersion}.

\begin{figure}[!htbp]
\begin{center}
\includegraphics[width=8cm]{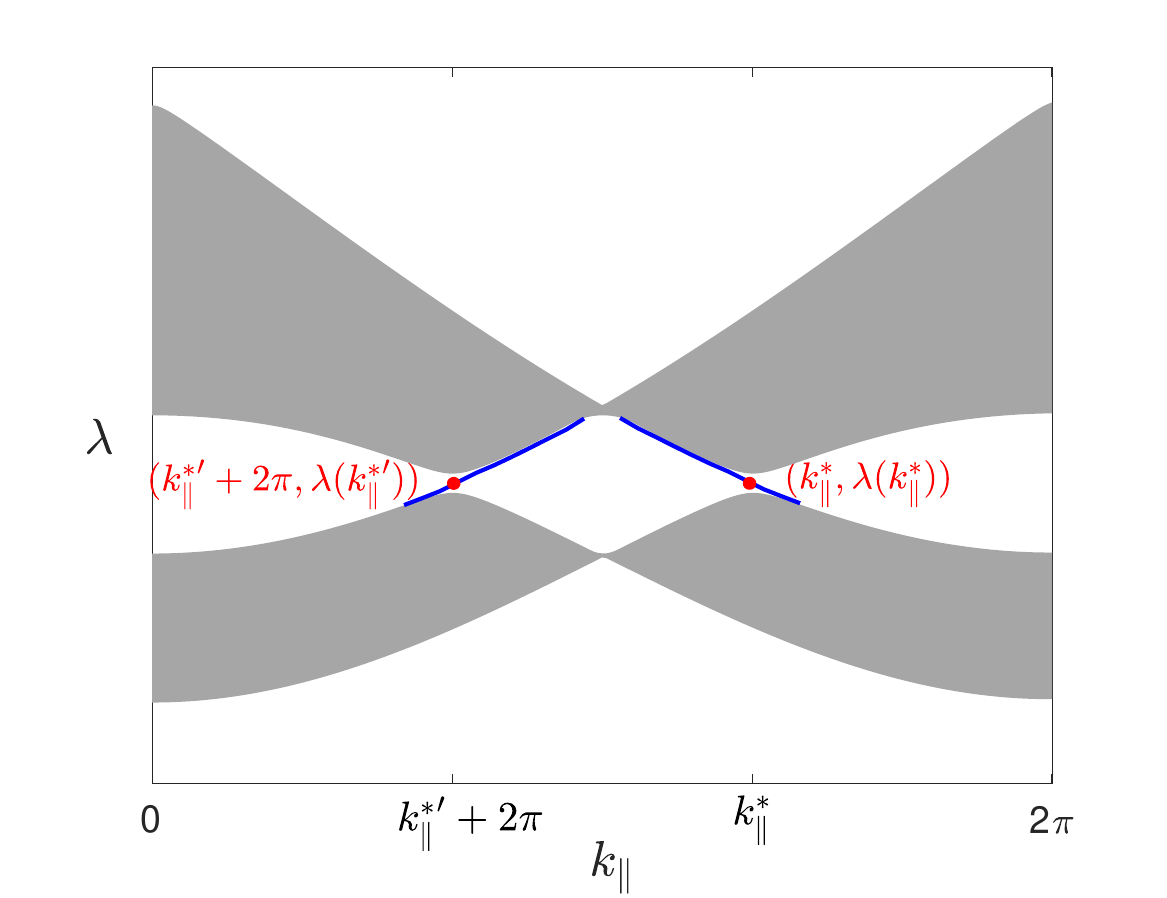}
\includegraphics[width=8cm]{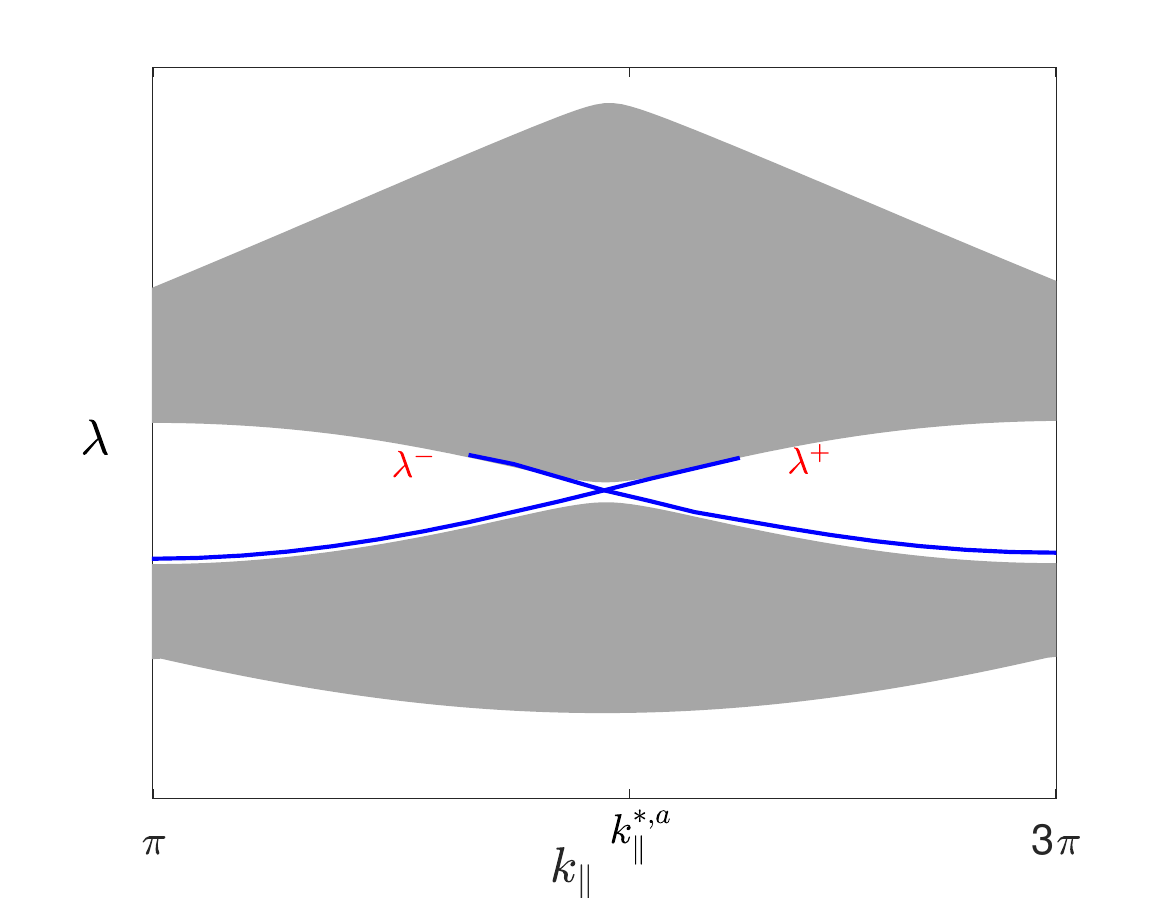}
\caption{The dispersion relations for the interface modes along the zigzag (left) and armchair interface (right). }
\label{fig:edge_mode_dispersion}
\end{center}
\end{figure}

\begin{thm}\label{lem:dispersion}
Let the assumptions in Theorem~\ref{lem:edge} hold and $\mathfrak d$ be an arbitrary constant in $(0,1)$.
If  $\eps>0$ is sufficiently small and $|\kp - \kp^*|< \mathfrak d \eps|\frac{\gamma_*}{t_* m_*}|$, the eigenvalue of the interface mode of the spectral problem \eqref{eq:spectral_prob_zigzag} is given by $\lambda - \lambda_* = \sgn(t_*) \cdot m_*(\kp - \kp^*) \cdot \big(1+ o(1))\big)$.
\end{thm}

\begin{thm}\label{lem:dispersionarm}
Let the assumptions in Theorem~\ref{lem:edgearm} hold and $\mathfrak d$ be an arbitrary constant in $(0,1)$.
If $\eps>0$ is sufficiently small and $|\kp - \kp^{*,a}|< \sqrt3 \mathfrak d\eps|\frac{\gamma_*}{t_* m_* }|$, the eigenvalues of the two interface modes along the armchair interface are $\lambda_\pm - \lambda_* = \pm \frac{1}{\sqrt3}m_*(\kp-\kp^{*,a})\cdot \big(1+ o(1))\big)$. 
\end{thm}

\subsection{Interface modes along rational interfaces}
\label{sec:rational}
  
We extend the previous studies to interface modes along a rational interface separating two honeycomb photonic structures. 
A rational interface is a line with a direction 
\begin{equation}
a\be_1 + b\be_2,
\end{equation}
where $a$ and $b$ are relatively prime integers.
When $a$ and $b$ are relatively prime, there exist $c,d\in\mathbb Z$, such that $bc-ad=1$ and 
\begin{equation*}
  \left(\begin{matrix} c & a\\ d & b\end{matrix}\right)^{-1} = \left(\begin{matrix} b & -a\\ -d & c\end{matrix}\right).  
\end{equation*}
Therefore, the vectors
\begin{equation}\label{eq:latvecr}
\be_1^r= c\be_1 + d \be_2, \quad \be_2^r= a\be_1 + b \be_2 
\end{equation}
generate the honeycomb lattice. Correspondingly, the reciprocal vectors
\begin{equation}\label{eq:dualvecr}
(\begin{matrix}\bbeta_1^r&\bbeta_2^r \end{matrix}) = (\begin{matrix}\bbeta_1&\bbeta_2 \end{matrix})\left(\left(\begin{matrix} c & a\\ b & d\end{matrix}\right)^{-1}\right)^T
=(\begin{matrix}\bbeta_1&\bbeta_2 \end{matrix})\left(\begin{matrix} b & -d\\ -a & c\end{matrix}\right)
\end{equation}
generate the dual lattice.

We call an interface direction $\be_2^r:=a\be_1 + b\be_2$ of the zigzag type if the dual slice $\{\bp(\ell)= \Kone+\ell \bbeta_1^r, \ell \in \mathbb R\}$ intersects with $\Kone+\Lambda^*$ but not with $\Ktwo+\Lambda^*$, and is of armchair if the slice intersects with both $\Kone+\Lambda^*$ and $\Ktwo+\Lambda^*$.
A straightforward calculation shows that  $\Kone+\ell \bbeta_1^r \in\Ktwo+\Lambda^*$ if and only if $a-b=3k$ for some $k\in\mathbb Z$.

\begin{defi}
The direction
$\be_2^r:=a\be_1 + b\be_2$ is called a rational if
$a$ and $b$ are relatively prime integers. The rational interface is of zigzag type if $a-b\neq3k$ for all $k\in\mathbb Z$, and is of armchair type if $a-b=3k$ for some $k\in\mathbb Z$.
\end{defi}

Let the inclusion $D$ and all of its rotations be compactly supported in the cell $\cC_r:=\{ \ell_1 \be_1^r + \ell_2 \be_2^r :  \ell_1, \ell_2\in[-1/2,1/2) \}$. 
We prove that the analog of Theorem~\ref{lem:edge} - Theorem~\ref{lem:dispersionarm} holds in the case of rational interfaces for the zigzag and armchair types. More specifically, 

Denote
\begin{equation}\label{eq:fr}
\mathfrak f^r = B- \frac{A}{|A|^2}\text{Re}(A\bar B),
\end{equation}
wherein
\begin{equation}
A = b-a\bar\tau,\quad B = -d+c\bar\tau.
\end{equation}
Then if the rational interface $\be_2^r$ is of zigzag type, we have the following theorem.
\begin{thm}\label{lem:ratzig}
Let $\be_2^r$ be a rational edge of zigzag type and let Assumption~\ref{lem:assNoFold} hold along the reciprocal lattice vector $\bbeta_1^r$. Let $t_*$ and $\gamma_*$ be the two constants defined in \eqref{eq:Tderiv} and assume that $t_*\neq0$. Let $\mathfrak d$ be an arbitrary constant in $(0,1)$.  If $\eps>0$ is sufficiently small, then
\begin{itemize}
    \item[(i)] There exists a unique interface mode along the interface, with the quasi-momentum $\kp^* = K\cdot\be_2^r$ and the eigenvalue $\lambda\in (\lambda_* - \hrad\eps, \lambda_* + \hrad\eps)$).
    \item[(ii)] For $|\kp - \kp^*|< \mathfrak d\eps|\frac{\gamma_*}{t_*}|/(|\mathfrak f^r|\frac{\sqrt3}{2}m_*)$, the dispersion relation for the interface mode adopts the expansion $\lambda - \lambda_* = \sgn(t_*)\cdot|\mathfrak f^r| \cdot \frac{\sqrt3}{2}m_* \cdot (\kp - \kp^*) \cdot \big(1+ o(1))\big)$. 
\end{itemize}
\end{thm}

If the rational interface $\be_2^r$ is of armchair type, we have the following theorem.

\begin{thm}\label{lem:ratarm}
Let $\be_2^r$ be a rational edge of armchair type and let Assumption~\ref{lem:assNoFold} hold along the reciprocal lattice vector $\bbeta_1^r$. Let $t_*$ and $\gamma_*$ be the two constants defined in \eqref{eq:Tderiv} and assume that $t_*\neq0$. Let $\mathfrak d$ be an arbitrary constant in $(0,1)$.  If $\eps>0$ is sufficiently small, then
\begin{itemize}
    \item[(i)] There exist exactly two interface modes along the interface, with the quasi-momentum $\kp^{*,r} = K\cdot\be_2^r$ and the eigenvalues $\lambda_\pm \in (\lambda_* - \hrad\eps, \lambda_* + \hrad\eps)$.
    \item[(ii)] For $|\kp - \kp^{*,r}|< \mathfrak d\eps|\frac{\gamma_*}{t_*}|/(|\mathfrak f^r|\frac{\sqrt3}{2} m_*)$, the dispersion relations for the interface modes adopt the expansions $\lambda - \lambda_* = \pm|\mathfrak  f^r| \cdot \frac{\sqrt3}{2} m_* \cdot (\kp - \kp^{*,r}) \cdot \big(1+ o(1))\big)$.
\end{itemize}
\end{thm}

\subsection{Extension of results to other settings}
We note that the method and framework developed in this paper can be extended to other settings:
\begin{enumerate}
     \item [(1)]
There are multiple inclusions in one periodic cell;
    \item [(2)] 
 The inclusions are penetrable such that the medium coefficient is piecewise constant;
   \item [(3)]
   The topological phase transition is induced by perturbations that break either the inversion symmetry or the time-reversal symmetry.

\end{enumerate}

\section{Dirac points for the honeycomb lattice}
In this section, we prove Theorem~\ref{lem:Dirac} regarding the  Dirac points by the layer potential technique. 

\subsection{Integral equation formulation}\label{sec:band}
In this subsection, we formulate the spectral problem for the honeycomb structure by using boundary integral equations. 
For each $\bp \in \cB_z$, let $G^f(\bx,\by; \bp,\lambda)$ be the quasi-periodic Green function over the honeycomb lattice that solves 
\begin{equation}\label{eq:GQPhomo}
(- \Delta - \lambda)  G^f(\bx,\by; \bp,\lambda)= \sum_{\be\in\Lambda}  e^{\im \bp \cdot \be} \delta(\bx-\by - \be)  \quad \mbox{for} \; \bx, \by \in \bbR^2.
\end{equation}

Define the single-layer potential 
\begin{equation*}
   u(\bx;\bp):= \int_{\partial D(\eta)}G^f(\bx,\by;\lambda, \bp)\tilde\phi(\by)\, ds_{\by}, 
\end{equation*}
wherein the density function $\tilde\phi\in H^{-1/2}(\partial D(\eta))$. Then it can be shown that $u$ solves the eigenvalue problem \eqref{eq:bandu_zigzag} if and only if $\tilde\phi\in H^{-1/2}(\partial D(\eta))$ solves the following boundary integral equation:
\begin{equation}\label{eq:banddens}
\int_{\partial D(\eta)}G^f(\bx,\by;\lambda, \bp)\tilde\phi(\by)\, ds_{\by} = 0,\quad \bx\in \partial D(\eta).
\end{equation}

Define $\phi(\bx):=\tilde\phi(\eta\bx)$. 
Then a point $(\bp,\lambda)$ belongs to the dispersion surface of the honeycomb lattice if and only if the triple $(\lambda, \bp, \phi) \in \mathbb R \times \cB_z \times \HmhalfpD$ solves the integral equation
\begin{equation}\label{eq:bandD}
\mathcal S (\eta,\lambda,\bp)[\phi] =0,
\end{equation}
where the single-layer integral operator
\begin{equation}\label{eq:SHelm}
\mathcal S (\eta,\lambda,\bp)[\phi](\bx) := \int_{\partial D_*}G^f(\eta\bx,\eta\by;\lambda, \bp)\phi(\by)\, ds_{\by} \quad \bx\in \partial D_*.
\end{equation}
In the rest of this section, we investigate the characteristic values of the integral operator $S (\eta,\lambda,\bp)$ when $\bp=K$.

\subsection{Symmetry of the integral operator}
In this subsection, we establish symmetry properties of the integral operator $S (\eta,\lambda,\bp)$. 
Note that
the Green function satisfying \eqref{eq:GQPhomo} can be represented by the lattice sum (cf. \cite{Ammari-book})
\begin{equation}\label{eq:GQPlattice}
G^f(\bx,\by;\lambda, \bp) = \frac{\im}{4} \sum_{\be \in \Lambda} e^{\im\bp\cdot \be} H_0^{(1)}(\omega |\bx- \by -\be|),
\end{equation}
where $H_0^{(1)}$ is the zero-order Hankel function of the first kind. More precisely, 
\begin{equation}
\frac{\im}{4}H_0^{(1)}(\omega;\bx) = - \frac{1}{2\pi}\left(\ln|\bx|+\ln\omega+\gamma_0+\ln(\omega|\bx|)\sum_{p\geq1}b_{p,1}(\omega|\bx|)^{2p}+\sum_{p\geq1}b_{p,2}(\omega|\bx|)^{2p}
 \right),
\end{equation}
 where
 \begin{equation*}
     b_{p,1} = \frac{(-1)^p}{2^{2p}(p!)^2},\quad
     b_{p,2} = \left(\gamma_0 - \sum_{s=1}^p\frac{1}{s}\right) b_{p,1}, \quad \gamma_0 = E_0 - \ln 2 - \frac{\im\pi}{2},
 \end{equation*} 
and $\displaystyle{E_0=\lim_{N\to\infty}\left(\sum_{p=1}^N \frac{1}{p}-\ln N\right)}$ is the Euler constant.
The Green function also attains the following spectral decomposition (cf. \cite{Ammari-book}):
\begin{equation}\label{eq:GQPfreq} 
G^f(\bx,\by;\lambda, \bp) =  - \frac{1}{ |\cC_z|} \sum_{\bq\in\Lambda^*} \frac{e^{\im (\bp +\bq)\cdot(\bx-\by)}}{\lambda - |\bp+\bq|^2},
\end{equation} 
wherein $|\cC_z|=\frac{\sqrt3}{2}$ represents the area of the fundamental cell $\cC_z$.

Recall that $H^s(\partial D_*)$ is the Sobolev space of order $s$ defined on $\partial D_*$. Note that the transformation $R\phi(\bx):=\phi(R^{-1}\bx)$ is unitary and it attains three eigenvalues $1$, $\tau$ and $\tau^2$, where  $\tau = e^{\im\frac{2\pi}{3}}$. Define 
\begin{equation}\label{eq:HsD}
H_{i}^{s}(\partial D_*):= \left\{ \phi\in H^s(\partial D_*):  R\phi(\bx):=\phi(R^{-1}\bx) = \tau^i \phi(\bx) \right\}, \quad i=0,1,2.
\end{equation}
These subspaces are pairwise orthogonal under the $L^2(\partial D_*)$ inner product and there holds
 $$
 H^s(\partial D_*) = H_{0}^{s}(\partial D_*) \bigoplus H_{1}^{s}(\partial D_*) \bigoplus H_{2}^{s}(\partial D_*).
 $$ In addition, using the relation $RF=FR^2$, we have $FH_{1}^{s}(\partial D)=H_{2}^{s}(\partial D)$. 

Define
\begin{equation}\label{}
\tilde\Lambda^*:= \Kone + \Lambda^*.
\end{equation} 
A straightforward calculation shows that
\begin{equation}\label{eq:dualsym}
 R \tilde\Lambda^*= \tilde\Lambda^*, \quad \rflc\tilde\Lambda^*=\tilde\Lambda^*.
\end{equation} 
Here we have used
\begin{equation}
R\bbeta_1=-\bbeta_1-\bbeta_2 \quad R\bbeta_2=\bbeta_1, \quad \rflc\bbeta_1=-\bbeta_2, \quad \rflc\bbeta_2=-\bbeta_1,
\end{equation} 
and
\begin{equation}
\Kone= 2\pi (\frac{2}{3}\bbeta_1 +\frac{1}{3}\bbeta_2),\quad  R\Kone= \Kone -\bbeta_2, \quad  \rflc\Kone = \Kone -\bbeta_1 -\bbeta_2.
\end{equation}

\begin{lemma}\label{lem:invsub}
Let $\bp=K$, then the following holds for the integral operator $\mathcal S (\eta,\lambda,\Kone)$:
\begin{itemize}
   \item[(i)] The operator $\mathcal S (\eta,\lambda,\Kone)$  commutes with $R$ and $\rflc$. That is, 
\begin{equation*}
    R\mathcal S (\eta,\lambda,\Kone) = \mathcal S (\eta,\lambda,\Kone) R \quad \mbox{and} \quad \rflc\mathcal S (\eta,\lambda,\Kone) = \mathcal S (\eta,\lambda,\Kone) \rflc.
\end{equation*}

  \item[(ii)] The operator $\mathcal S (\eta,\lambda,\Kone)$ is bounded from $H_{i}^{-1/2}(\partial D_*)$ to $H_{i}^{1/2}(\partial D_*)$ ($i=0,1,2$) for all $\eta$ and $\lambda$.

  \item[(iii)]  The triple $(\bp,\lambda, \phi)\in H_{1}^{-1/2}(\partial D_*)$ solves \eqref{eq:bandD} if and only if the triple point $(\bp,\lambda, \phi(\rflc(\cdot))\in H_{2}^{-1/2}(\partial D_*)$ solves \eqref{eq:bandD}.
\end{itemize}

\end{lemma}
\begin{proof}
For Statement (i), in light of \eqref{eq:GQPfreq} and \eqref{eq:dualsym}, we have
\begin{equation*}
\mathcal S (\eta,\lambda,\Kone) [R\phi] (\bx)=  - \frac{1}{ |\cC_z|}\int_{\partial D_*} \sum_{\bm\in \tilde\Lambda^*} \frac{1}{\lambda - |\bm|^2} e^{\im \bm \cdot(\bx-\by)} \phi(R^{-1}\by)\, ds_\by,
\end{equation*}

and
\begin{equation*}
\begin{aligned}
R\mathcal S (\eta,\lambda,\Kone) [\phi] (\bx) 
&=  - \frac{1}{ |\cC_z|}\int_{\partial D_*} \sum_{\bm\in\tilde\Lambda^*} \frac{1}{\lambda - |\bm|^2} e^{\im \bm \cdot (R^{-1}\bx-\by)}\phi(\by)\, ds_\by\\
&= - \frac{1}{ |\cC_z|}\int_{\partial D_*} \sum_{\bm\in\tilde\Lambda^*} \frac{1}{\lambda - |\bm|^2} e^{\im R\bm \cdot (\bx-\by')} \phi(R^{-1}\by')\, ds_{\by'}
=\mathcal S (\eta,\lambda,\Kone) [R\phi] (\bx).
\end{aligned}
\end{equation*}

In the above, we have used $R(\partial D_*) = \partial D_*$, $|R\bm|=|\bm|$ and $R\tilde\Lambda^* = \tilde\Lambda^*$.
The relation 
\begin{equation*}
\mathcal S (\eta,\lambda,\Kone)[\rflc\phi](\bx) = \rflc\mathcal S (\eta,\lambda,\Kone)[\phi](\bx) 
\end{equation*}
can be shown similarly using the relation $\rflc\tilde\Lambda^* = \tilde\Lambda^*$.

Statement (ii) follows from the standard layer potential theory; see for instance \cite{Ammari-book}.

Statement (iii) is a consequence of the relation $RF = FR^2$, which implies $\phi(\bx)\in H_{1}^{-1/2}(\partial D_*)$ if and only if $\phi(F\bx)\in H_{2}^{-1/2}(\partial D_*)$.
\end{proof}

\subsection{Dirac points in the lowest two bands}\label{sec:Dirac}
In this subsection, we establish the existence of Dirac points in the lowest two bands. In view of \eqref{eq:GQPlattice} and \eqref{eq:GQPfreq}, when $\bp = K$, the Green function $G^f(\bx,\by;\lambda, K)$ attains singularities around $|\bx-\by|=0$ and $\lambda = |\bm|^2$ for each $\bm\in\tilde\Lambda^*$. The singularity for the former arises naturally when the source point $\bf y$ and the target point $\bx$ overlap, while the latter occurs at special frequencies $\lambda = |\bm|^2$ when the spectral decomposition \eqref{eq:GQPfreq} is not well-defined.

As to be shown below, the Dirac point at $\Kone$ with the lowest energy $\lambda$ appears when $\lambda \approx  |\bm_1|^2$, where $\bm_1\in\tilde \Lambda^*$ attains the smallest norm among all lattice points in $\tilde \Lambda^*$. 
A straightforward calculation shows that
\begin{equation*}
|\bm_1|=|\Kone|, \quad
\{\bm\in \tilde \Lambda^*, |\bm|=|\bm_1|\} =\Kone+\left\{\bq_1,\bq_2,\bq_3\right\},
\end{equation*}
in which
\begin{equation*}
\bq_1=(0,0)^T, \quad \bq_2=2\pi(-\frac{2}{\sqrt3},0)^T, \quad \bq_3=2\pi(-\frac{1}{\sqrt3},-1)^T.
\end{equation*}

We now perform asymptotic expansion of the operator $\mathcal{S}(\eta,\lambda,\Kone)$ for $\lambda\approx|\bm_1|^2=|\Kone|^2$. 
To this end, we derive the expansion for Green's function $G^f(\eta\bx,\eta\by;\lambda, \Kone)$ when $\eta$ is small.
For simplicity we consider $G^f(\eta\bx,0;\lambda, \Kone)$ instead, since 
$G^f(\eta\bx,\eta\by;\lambda, \Kone) = G^f(\eta(\bx-\by), 0;\lambda, \Kone)$.
From the above discussions, the Green function $G^f(\eta\bx, 0;\lambda, \Kone)$ attains singularities when $\bx=0$ or  $\lambda = |\bm|^2$ for some $\bm\in \tilde\Lambda^*$, and those terms contributing to the singularities are the leading-order terms in the expansion of $G^f(\eta\bx, 0;\lambda, \Kone)$.

From the lattice sum \eqref{eq:GQPlattice}, the singularity at $\bx=0$ arises from the term $\frac{\im}{4}H_0^{(1)}(\eta\bx;\lambda)$ with $\be = (0,0)^T$. Using the expansion
\begin{equation*}
\frac{\im}{4}H_0^{(1)}(\eta\bx;\lambda) = - \frac{1}{2\pi}\left(\ln|\bx|+\ln\eta+\ln\sqrt{\lambda}+\gamma_0+\left(\ln(\sqrt{\lambda}|\bx|)+\ln\eta\right)\sum_{p\geq1}b_{p,1}(\sqrt{\lambda}\eta|\bx|)^{2p}+\sum_{p\geq1}b_{p,2}(\sqrt{\lambda}\eta|\bx|)^{2p}
 \right),
\end{equation*}
we define the leading-order term by $L_1 (\eta\bx;\lambda)$ and the remainder by $R_1 (\eta\bx;\lambda)$ as follows:
\begin{equation}\label{eq:pDlead1}
\begin{aligned}
L_1 (\eta\bx;\lambda)&:=
- \frac{1}{2\pi}\left(\ln|\bx|+\ln\eta+\ln\sqrt{\lambda}+\gamma_0\right), \\
R_1 (\eta\bx;\lambda)&:= - \frac{1}{2\pi}\left((\ln(\sqrt{\lambda}|\bx|)+\ln\eta)\sum_{p\geq1}b_{p,1}(\sqrt{\lambda}\eta|\bx|)^{2p}+\sum_{p\geq1}b_{p,2}(\sqrt{\lambda}\eta|\bx|)^{2p}
 \right).
\end{aligned}
\end{equation}
From the spectral decompostion \eqref{eq:GQPfreq}, the singularity at $\lambda \approx  |\bm_1|^2$ arises from the terms
\begin{equation*}
\begin{aligned}
- \frac{1}{ |\cC_z|} \sum_{\Kone+\bq\in[\bm_1]} \frac{e^{\im (\bp +\bq)\cdot\bx}}{\lambda - |\bp+\bq|^2} 
&= - \frac{1}{ |\cC_z|} \sum_{\Kone+\bq\in[\bm_1]} \frac{e^{\im (\bp +\bq)\cdot\bx}}{\lambda - |\bp+\bq|^2} \\ 
&= - \frac{1}{ |\cC_z|}\frac{1}{\lambda - |\bm_1|^2}\left(3 - \frac{1}{3}(2\pi)^2(\eta|\bx|)^2 + \sum_{j\geq 3, k=1,2,3}\frac{(\im (\Kone +\bq_k)\cdot\eta\bx)^j}{j!}\right).
\end{aligned}
\end{equation*}
We consider $\lambda$ in the $\eta^2$-neighborhood of $|\bm_1|^2$, namely, $\lambda\in U_\eta$ where $U_\eta$ is defined in \eqref{eq:vincexa}. Correspondingly, we define 
 the leading-order term by $L_2 (\eta\bx;\lambda)$ and the remainder by $R_2 (\eta\bx;\lambda)$ as follows:
\begin{equation}\label{eq:pDlead2}
\begin{aligned}
L_2 (\eta\bx;\lambda,\Kone)&:=
- \frac{1}{ |\cC_z|}\frac{1}{\lambda - |\bm_1|^2}\left(3 - \frac{1}{3}(2\pi)^2(\eta|\bx|)^2 \right),\\
R_2 (\eta\bx;\lambda,\Kone)&:= - \frac{1}{ |\cC_z|}\frac{1}{\lambda - |\bm_1|^2}\left(\sum_{j\geq 3, k=1,2,3}\frac{(\im (\Kone +\bq_k)\cdot\eta\bx)^j}{j!}\right).
\end{aligned}
\end{equation}
Finally, the smooth term in the Green's function is denoted by
\begin{equation}
R_0 (\eta\bx;\lambda):= G^f(\eta\bx,0;\lambda, \Kone) -\frac{\im}{4}H_0^{(1)}(\eta\bx;\lambda) + \frac{1}{ |\cC_z|} \sum_{\Kone+\bq\in[\bm_1]} \frac{e^{\im (\Kone +\bq)\cdot\bx}}{\lambda - |\Kone+\bq|^2}  .
\end{equation}

Using the above expansions for the Green's function, we obtain the decomposition for the integral operator $\mathcal S(\eta,\lambda,\Kone)$:
\begin{equation}
\mathcal S (\eta,\lambda,\Kone) =\mathcal L(\eta,\lambda,\Kone) + \mathcal R(\eta,\lambda,\Kone),
\end{equation}
where the leading-order integral operator is
\begin{equation}\label{eq:pDlead}
\mathcal L(\eta,\lambda,\Kone)\phi(\bx):= \int (L_1 (\eta(\bx-\by);\lambda)+L_2 (\eta(\bx-\by);\lambda,\Kone))\phi(\by)\, d_{s_\by}, 
\end{equation}
and the remainder operator is
\begin{equation}\label{eq:remainder}
\mathcal R(\eta,\lambda,\Kone)\phi(\bx):= \int (R_1 (\eta(\bx-\by);\lambda)+R_2 (\eta(\bx-\by);\lambda,\Kone)+R_0 (\eta(\bx-\by);\lambda,\Kone))\phi(\by)\, d_{s_\by}. 
\end{equation}

Let $\mathcal S_0:\HmhalfpD\to\HhalfpD$ be the single layer potential associated with the Laplace operator in free space defined by
\begin{equation}\label{eq:SLap}
\mathcal S_0[\phi](\bx) :=  \int_{\partial D_*} - \frac{1}{2\pi} \ln(|\bx -\by|) \phi(\by)\, ds_\by.
\end{equation}

\begin{assumption} \label{assumpS0}
 The operator  $\mathcal S_0: H^{-1/2}(\partial D_*)\to H^{1/2}(\partial D_*)$ is invertible. 
 \end{assumption}

\begin{remark}
According to ~\cite{perfekt2014spectral,verchota1984layer}, the above assumption holds generically for a given geometry of the inclusion $D_*$. Therefore, we assume that Assumption \ref{assumpS0} holds throughout the paper.
\end{remark}
 

\begin{remark}\label{lem:moresym}
The operator $\cS_0$ defined in \eqref{eq:SLap}, and the operators $\mathcal L(\eta,\lambda_0,\Kone)$ and $\mathcal R(\eta,\lambda_0,\Kone)$ commute with $R$ and $\rflc$.
\end{remark}

\begin{lemma}\label{lem:vansym}
When $\phi\in H_{1}^{-1/2}(\partial D_*)$,
\begin{subequations} 
\begin{align}
\int_{\partial D_*} \phi(\by)\, ds_\by &= 0, \label{eq:vansym1} \\
\int_{\partial D_*} |\by|^2 \phi(\by)\, ds_\by &= 0, \label{eq:vansym2}\\
\int_{\partial D_*} \by \phi(\by)\, ds_\by &\in\text{span}\{ (1,\im)\}, \label{eq:vansym3}
\end{align}
\end{subequations}
and
\begin{equation}\label{eq:leading_order_int}
 \int_{\partial D_*} L_2 (\eta(\bx-\by);\lambda,\Kone) \phi(\by)\, ds_\by \in\text{span}\{x_1+ \im x_2\}.
\end{equation}
\end{lemma}
\begin{proof}
Using $R\phi(\by)=\phi(R^{-1}\by)=\tau\phi(\by)$, we have
\begin{equation*}
\begin{aligned}
\int_{\partial D_*} \phi(\by)\, ds_\by &= \int_{\partial D_*}  \phi(R^{-1}\by')\, ds_{\by'} = \int_{\partial D_*} \tau \phi(\by')\, ds_{\by'} , \\
\int_{\partial D_*} |\by|^2 \phi(\by)\, ds_\by &= \int_{\partial D_*} |R^{-1}\by'|^2 \phi(R^{-1}\by')\, ds_{\by'} = \int_{\partial D_*}  |\by'|^2 \tau \phi(\by')\, ds_{\by'} . 
\end{aligned}
\end{equation*}
Since $\tau\neq 1$, we obtain \eqref{eq:vansym1} and \eqref{eq:vansym2}. Similarly,
\begin{equation*}
\int_{\partial D_*} \by \phi(\by)\, ds_\by = \int_{\partial D_*}R^{-1}\by' \phi(R^{-1}\by')\, ds_{\by'} = \int_{\partial D_*} R^{-1}\by' \tau \phi(\by')\, ds_{\by'} . 
\end{equation*}
Denoting $(a,b):=\int_{\partial D_*} \by \phi(\by)\, ds_\by$, the above relation reads $(a,b)=\tau R^{-1}(a,b)$, which implies ~\eqref{eq:vansym3}. Finally, \eqref{eq:leading_order_int} follows from \eqref{eq:vansym1} - \eqref{eq:vansym3}.
\end{proof}

\begin{lemma}\label{lem:defa}
There exists a unique function $\phi_*\in \HmhalfpD$ such that $\mathcal S_0[\phi_*](\bx) = x_1+\im \, x_2$. Moreover,  
$\displaystyle{\int_{\partial D_*} \by \phi_*(\by)\, ds_\by = \mathfrak a \, (1,\im)}$ for some $\mathfrak a\in\mathbb C\backslash\{0\}$.
\end{lemma}

\begin{proof}
Noting that $x_1+\im x_2 \in H_{1}^{-1/2}(\partial D_*)$,  we deduce that $\phi_*\in H_{1}^{-1/2}(\partial D_*)$ exists and is unique. 
Combining with Lemma~\ref{lem:vansym}, we have
\begin{equation} \label{express-a}
\mathfrak a:= \frac{\int_{\partial D_*} \bx\cdot \by \phi_*(\by)\, ds_\by}{x_1+\im x_2} \in\mathbb C.
\end{equation}
To show $\mathfrak a\neq0$, we notice that 
\begin{equation}
2\mathfrak a = (1,-\im)\cdot
\int_{\partial D_*} \by \phi_*(\by)\, ds_\by =\langle f, \cS_0^{-1}f \rangle_{\HhalfpD,\HmhalfpD}\neq0,
\end{equation}
where $f(\bx)=x_1+\im x_2$.
The inequality follows since $\langle \cdot, \cS_0^{-1}\cdot \rangle$ is an equivalent inner product on $\HhalfpD\times\HhalfpD$. 
\end{proof}

\begin{lemma}\label{lem:LeadD1}
When $\eta$ is sufficiently small, the following statements hold for the operator $\mathcal L(\eta,\lambda,\Kone):H_{1}^{-1/2}(\partial D_*) \to H_{1}^{1/2}(\partial D_*)$:
\begin{itemize}
    \item [(i)] $\mathcal L(\eta,\lambda,\Kone)$ is  analytic in $\lambda$ in a neighborhood of $U_{\eta}$.
    \item [(ii)] $\mathcal L(\eta,\lambda,\Kone)$ is a Fredholm operator of index zero for $\lambda\in U_{\eta}$.
    \item [(iii)] The only characteristic value of $\mathcal L(\eta,\lambda,\Kone)$ located in $U_{\eta}$ is given by
\begin{equation*}
    \lambda_0:= |\bm_1|^2 + \frac{1}{ |\cC_z|}\frac{2}{3}(2\pi)^2 \mathfrak a \eta^2. 
\end{equation*}
Moreover, 
\begin{equation*}
\ker\left(\mathcal L(\eta,\lambda_0,\Kone)\right) = \text{span}\{\phi_*\},
\end{equation*}
wherein $\phi_*$ is defined in \eqref{lem:defa}.
\item [(iv)]  The  multiplicity of $\lambda_0$ is $1$.
\item [(v)] For $\lambda\in\partial U_{\eta}$, $\mathcal L^{-1}(\eta,\lambda,\Kone)$ exists and the norm $||\mathcal L^{-1}(\eta,\lambda,\Kone)||$ is bounded by a constant indepent of $\eta$.

\end{itemize}
\end{lemma}
\begin{proof}
\noindent(i) is obvious from the definition of the operator in \eqref{eq:pDlead1}, \eqref{eq:pDlead2} and \eqref{eq:pDlead}.

\noindent (ii) From Lemma~\ref{lem:vansym}, $\mathcal L(\eta,\lambda,\Kone)$ is the sum of $\cS_0:H_{1}^{-1/2}(\partial D_*) \to H_{1}^{1/2}(\partial D_*)$, which is Fredholm of index zero~\cite{mclean2000strongly}, and a finite-rank operator whose range is in $\text{span}\{x_1+\im x_2\}$. 

\noindent(iii). Using Lemma~\ref{lem:vansym}, we see that $\mathcal L(\eta,\lambda,\Kone)[\phi](\bx) =0$ implies $\cS_0\phi = -L_2(\eta,\lambda,\Kone)[\phi](\bx)\in\text{span}\{x_1+\im x_2\}$, thus $\phi \in\text{span}\{ \phi_*\}$. 
In addition, a straightforward calculation shows that 
\begin{equation*}
\begin{aligned}
\mathcal L(\eta,\lambda,\bp)[\phi_*]
=(x_1+\im x_2)\left(1 - \frac{1}{ |\cC_z|}\frac{1}{\lambda - |\bm_1|^2} \frac{2}{3}(2\pi)^2 \eta^2 \mathfrak a\right).
\end{aligned}
\end{equation*}
Since $\mathfrak a\neq0$, $\mathcal L(\eta,\lambda,\bp)[\phi_*]=0$
 if and only if 
\begin{equation}
\lambda = \lambda_0:= \frac{1}{ |\cC_z|}\frac{2}{3}(2\pi)^2 \mathfrak a \eta^2 + |\bm_1|^2.
\end{equation} 

\noindent(iv). 
Following the definitions in Appendix~\ref{sec:GStheory}, we
assume that $\phi'\in H_{1}^{-1/2}(\partial D_*)$ satisfies
\begin{equation*}
\frac{d}{d\lambda}\mathcal L(\eta,\lambda_0,\Kone)[\phi_*] + \mathcal L(\eta,\lambda_0,\Kone)[\phi'] =0.
\end{equation*}
It can be shown that $\frac{d}{d\lambda}\mathcal L(\eta,\lambda_0,\Kone)[\phi_*] \in\text{span}\{ x_1+\im x_2\}$. Using $L_2(\eta,\lambda,\Kone)[\phi'](\bx)\in \text{span}\{ x_1+\im x_2\}$, we obtain $\cS_0\phi'\in \text{span}\{ x_1+\im x_2\}$,
which implies that $\phi'\propto\phi_*$. On the other hand, it follows from~(iii) that $\mathcal L(\eta,\lambda_0,\Kone)[\phi_*]=0$. Hence $\phi'$ does not exist, $\phi_*$ is of rank $1$, and the multiplicity of $\lambda_0$ is $1$.  \\

\noindent(v). For $\lambda\in\partial U_{\eta}$, $\mathcal L^{-1}(\eta,\lambda,\Kone)$ exists because   $\mathcal L(\eta,\lambda,\Kone)$ is Fredholm and has no characteristic values on $\partial U_{\eta}$.
When 
$\left|\lambda-|\Kone|^2 \right| = \frac{1}{ |\cC_z|}\frac{1}{3}(2\pi)^2 \mathfrak a\eta^2$, there holds
\begin{equation}
\mathcal L(\eta,\lambda,\Kone)\phi = \cS_0\phi +e^{\im\theta} \frac{2}{\mathfrak a}\bx\cdot\int_{\partial D_*}\by\phi(\by)\,ds_{\by},
\end{equation}
where $\theta\in\mathbb R$.
When 
$ \left|\lambda-|\Kone|^2 \right| = \frac{1}{ |\cC_z|}(2\pi)^2 \mathfrak a\eta^2$, there holds
\begin{equation}
\mathcal L(\eta,\lambda,\Kone)\phi = \cS_0\phi e^{\im\theta} \frac{2}{3\mathfrak a}\bx\cdot\int_{\partial D_*}\by\phi(\by)\,ds_{\by},
\end{equation}
where $\theta\in\mathbb R$.
The operators above do not depend on $\eta$. Thus the norm of $\mathcal L^{-1}(\eta,\lambda,\Kone)$ for $\lambda\in\partial U_\eta$ is bounded by a constant that does not depend on $\eta$.
\end{proof}
\begin{lemma}\label{lem:LeadD0}
When $\eta$ is sufficiently small, the following statements hold for the operator $\mathcal L(\eta,\lambda,\Kone):H_{0}^{-1/2}(\partial D_*) \to H_{0}^{1/2}(\partial D_*)$:
\begin{itemize}
    \item [(i)] $\mathcal L(\eta,\lambda,\Kone)$ is analytic in $\lambda$ in a neighborhood of $U_{\eta}$.
    \item[(ii)] For $\lambda\in\overline{U_{\eta}}$, $\mathcal L^{-1}(\eta,\lambda,\Kone)$ exists and the norm $\| \mathcal L^{-1}(\eta,\lambda,\Kone) \|_{H_{0}^{-1/2}(\partial D_*)\to H_0^{1/2}(\partial D_*)}$ is bounded by a positive constant  independent of $\eta$.
    \item [(iii)] $\mathcal L(\eta,\lambda,\Kone)$ is a Fredholm operator of index zero for $\lambda\in U_{\eta}$.
\end{itemize}
\end{lemma}
\begin{proof}
(i) follows similar lines as in Lemma~\ref{lem:LeadD1}. 
For (ii), let $f\in H_{0}^{-1/2}(\partial D_*)$ be the unique function that satisfies $\cS_0f = 1$. 
Since $\langle \phi, \cS_0 \phi \rangle_{\partial D_*}$ is equivalent to $\|\phi\|_{\HmhalfpD}^2$, we know
\begin{equation}
\int_{\partial D_*} f(\bx) \, ds_{\bx} = C_1 >0,
\end{equation}
where $C_1$ is a constant.
Thus for every $\phi\in  H_{0}^{-1/2}(\partial D_*)$, we have the decomposition
\begin{equation}
\phi = \frac{\bar\phi}{C_1} f + g,
\end{equation}
where $\bar\phi =\int_{\partial D_*} \phi(\bx) \, ds_{\bx}$, and $\int_{\partial D_*} g(\bx) \, ds_{\bx}=0$.

Since $L(\eta,\lambda,\Kone)$ is symmetric and $\langle \phi, \cS_0^{-1} \phi \rangle_{\partial D_*}$ is equivalent to $\|\phi\|_{\HhalfpD}^2$,
we calculate
\begin{equation}
\left|\frac{\langle\phi, \cL(\eta,\lambda,\Kone) \phi\rangle_{\partial D_*}}{\langle\cS_0^{-1}\cL(\eta,\lambda,\Kone)\phi, \cL(\eta,\lambda,\Kone) \phi\rangle_{\partial D_*}}\right| = \left|\frac{-\frac{\ln\eta}{2\pi} |\bar\phi|^2(1+o(1)) +\langle g,\cS_0 g\rangle_{\partial D_*} (1+o(1))}{ C_1(\frac{\ln\eta}{2\pi})^2|\bar\phi|^2(1+o(1)) +\langle g,\cS_0 g\rangle_{\partial D_*} (1+o(1)) }\right|. 
\end{equation}
Thus when $\eta$ is sufficiently small, 
$\| \mathcal L^{-1}(\eta,\lambda,\Kone) \|_{H_{0}^{-1/2}(\partial D_*)\to H_0^{1/2}(\partial D_*)} \leq \frac{1}{\min\{C_1,1\}}$ 
for $\lambda\in\overline{U_{\eta}}$. This finishes the proof of (ii).

(iii) is a direct corollary of (ii).
\end{proof}

\begin{lemma}\label{lem:RemD}
There exists $\eta_0>0$ such that for all $\eta\in(0,\eta_0)$ and $\lambda\in \overline{U_{\eta}}$, 
\begin{equation}\label{eq:remsmall}
\|\mathcal R(\eta,\lambda,\Kone)\|_{H^{-1/2}(\partial D_*) \to H^{1/2}(\partial D_*)} \leq C\eta 
\end{equation}
for some constant $C$ independent of $\eta$.
\end{lemma}
\begin{proof}
There exists $\eta_0>0$ and a constant $C$ such that for all $\eta\in(0,\eta_0)$ and $\lambda\in \overline{U_{\eta}}$, the following holds
\begin{subequations} 
\begin{align}
& |\partial^{\alpha_1}_\bx \partial^{\alpha_2}_\by R_1 (\eta(\bx-\by);\lambda)|<C\eta, \label{eq:bddder-1} \\
& |\partial^{\alpha_1}_\bx \partial^{\alpha_2}_\by R_2 (\eta(\bx-\by);\lambda,\bp)|<C\eta, \label{eq:bddder-2} \\
& |\partial^{\alpha_1}_\bx \partial^{\alpha_2}_\by R_0 (\eta(\bx-\by);\lambda,\bp)|<C\eta \label{eq:bddder-3}
\end{align}
\end{subequations}
for all multi-indices $|\alpha_1+\alpha_2|\leq2$, $\bx,\by\in \partial D_*$ and $\lambda\in U_\eta$.
The relations \eqref{eq:bddder-1} and \eqref{eq:bddder-2} can be shown by a direct calculation, and \eqref{eq:bddder-3} was shown in~\cite{Li2023}. 
An elementary calculation on the Fourier coefficients concludes the proof. 
\end{proof}
\begin{thm}\label{lem:DiracP}
When $\eta$ is sufficiently small, the following statements hold:
\begin{itemize}
    \item [(i)] The operator $\mathcal S(\eta,\lambda,\Kone):H_{i}^{-1/2}(\partial D_*) \to H_{i}^{1/2}(\partial D_*)$, $i=1, 2$, attains exactly one characteristic value $\lambda_*\in U_{\eta}$ of multiplicity $1$. More precisely, there exists exactly one pair $(\lambda_i,\rho_i)\in U_\eta\times H_{i}^{-1/2}(\partial D_*)$ such that $S(\eta,\lambda_i,\Kone)\rho_i=0$, $i=1,2$. In addition, $\lambda_1=\lambda_2=:\lambda_*\in\mathbb R$ and $\rho_2(\bx)=\rho_1(\rflc\bx)$.
    \item [(ii)] The operator $\mathcal S(\eta,\lambda,\Kone):H_{0}^{-1/2}(\partial D_*) \to H_{0}^{1/2}(\partial D_*)$ has no characteristic value in $U_{\eta}$. 
    \item[(iii)] 
    The function $\rho_1$ can be chosen such that \begin{equation}\label{eq:rhoest}
    \rho_1 = \phi_* + O(\eta).
    \end{equation}
    In addition, the function $\rho_2(\bx):=\rho_1(F\bx)$ spans the one-dimensional kernel space for $S(\eta,\lambda_*,\Kone)$ restricted to the subspace ${H_{2}^{-1/2}(\partial D_*)}$.
\end{itemize}

\end{thm}
\begin{proof}
For (i), we first find the multiplicity of the characteristic values for $\mathcal S(\eta,\lambda,\Kone):H_{i}^{-1/2}(\partial D_*) \to H_{i}^{1/2}(\partial D_*)$, $i=1,2$. To this end, we apply Theorem~\ref{lem:actGS} by setting $z=\lambda$, $X=H_{i}^{-1/2}(\partial D_*)$, $Y=H_{i}^{1/2}(\partial D_*)$, 
$V=U_\eta$, 
$A(z)=\mathcal L(\eta,\lambda,\Kone)$ 
and $B(z)=\mathcal R(\eta,\lambda,\Kone)$.
Recall from Lemma~\ref{lem:LeadD1} and Lemma~\ref{lem:RemD} that $A(z)$ and $B(z)$ are analytic on a neighborhood of $\overline U_\eta$, $A(z)$ is Fredholm of index zero on a neighborhood of $\overline U_\eta$, and the multiplicity of $A(z)$ in $U_\eta$ is $1$. When $\eta$ is sufficiently small, it follows that $\|A^{-1}(z)B(z)\|$ is small on $\partial U_\eta$ by the uniform boundedness of $A^{-1}(z)$ in $\partial U_\eta$ over $\eta$ and the smallness of $B(z)$ in $\bar V$ as $\eta\to0$. Thus the characteristic value of $\mathcal S(\eta,\lambda,\Kone):H_{i}^{-1/2}(\partial D_*) \to H_{i}^{1/2}(\partial D_*)$ attains multiplicity $1$ in $V$ for $i=1,2$.
Since the null multiplicity $1$ corresponds to exactly one eigenpair, we deduce that there exists exactly one pair $(\lambda_i,\rho_i)\in U_\eta\times H_{1}^{-1/2}(\partial D_*)$ such that $S(\eta,\lambda_i,\Kone)\rho_i=0$, $i=1,2$. The statement for $\lambda_1=\lambda_2$ and $\rho_2(\bx)=\rho_1(\rflc\bx)$ follows from $\phi\in H_1^{-1/2}$ solves $\mathcal S(\eta,\lambda,\Kone)\phi(\bx)=0$ if and only if $\phi(F(\bx))\in H_2^{-1/2}$ solves $\mathcal S(\eta,\lambda,\Kone)\phi(\rflc\bx)=0$. Finally $\lambda_i$ are real because $\overline{G^f(\bx,\by;\lambda, \bp)} = G^f(\by,\bx;\lambda, \bp)$, as can be seen from \eqref{eq:GQPfreq}.

For (ii), 
the argument is similar to that in (i), except that we identify $X=H_{0}^{-1/2}(\partial D_*)$, $Y=H_{0}^{1/2}(\partial D_*)$. By Lemma~\ref{lem:LeadD0} and Lemma~\ref{lem:RemD} and Theorem~\ref{lem:actGS}, we verify the statement.

For (iii), the correspondence between $\rho_1$ and $\rho_2$ follows similarly from Lemma~\ref{lem:invsub}. Finally, we show that \eqref{eq:rhoest} holds.
Let $T_0: H_{1}^{-1/2}(\partial D_*)\to H_{1}^{1/2}(\partial D_*) $ be defined by 
\begin{equation}
T_0\phi (\bx) := \mathcal L(\eta,\lambda_0,\Kone) \phi(\bx)= \cS_0\phi(\bx) - \frac{1}{\mathfrak a}\int_{\partial D_*}\bx\cdot\by \phi(\by)\, ds_{\by},
\end{equation}
where we have used Lemma~\ref{lem:vansym}. 
Note that $\ker(T_0) = \text{span}\{\phi_*\}$. We define $f(\bx):= x_1+\im x_2$, then it follows that
\begin{equation*}
\langle \phi_*,f  \rangle_{\partial D_*} = \langle \cS_0^{-1}f,f  \rangle_{\partial D_*} \neq0.
\end{equation*}
The inequality follows from the choice of the size of the inclusion stated after \eqref{eq:SLap},
which implies that the $\langle \cS_0^{-1} \cdot,\cdot  \rangle_{\partial D_*}$ pairing on $ H_{1}^{-1/2}(\partial D_*)$ is an inner product~\cite{perfekt2014spectral}. 
Since  $T_0: H_{1}^{-1/2}(\partial D_*) \to H_{1}^{1/2}(\partial D_*)$ is a Fredholm operator, the range of $T_0$ is perpendicular to $\ker(T_0)$ given by
\begin{equation*}
\text{Ran}\,T_0 = \{ \psi\in H_{1}^{1/2}(\partial D_*): \langle \phi_*, \psi  \rangle_{\partial D_*} = 0\}.
\end{equation*}
Define
\begin{equation}\label{eq:proj}
Q\psi :=\psi - \frac{\langle \phi_*,\psi  \rangle_{\partial D_*}}{\langle \phi_*,f  \rangle_{\partial D_*}} f.
\end{equation}
Then $Q$ is a projection and
\begin{equation*}
Q H_{1}^{1/2}(\partial D_*)= \text{Ran}\,T_0. 
\end{equation*}

Let the density $\rho_1\in \HmhalfpD$ be a solution to $\mathcal S(\eta,\lambda_*,\Kone)\rho_1=0$, where $\rho_1 = \phi_* + \phi^{(1)}$ for some $\phi^{1} \in (\ker(T_0))^{\perp}$. Applying $Q$ to the following equation
\begin{equation*}
\begin{aligned}
    0 = \mathcal S(\eta,\lambda_*,\Kone)\rho_1 &= (\mathcal L(\eta,\lambda_*,\Kone) + \mathcal R(\eta,\lambda_*,\Kone) )\rho_1 \\
&= (T_0 + \mathcal L(\eta,\lambda_*,\Kone) - \mathcal L(\eta,\lambda_0,\Kone)+ \mathcal R(\eta,\lambda_*,\Kone) )\rho_1,
\end{aligned}
\end{equation*}
we obtain
\begin{equation}\label{eq:rhoproj}
0 = Q(T_0 + \mathcal L(\eta,\lambda_*,\Kone) - \mathcal L(\eta,\lambda_0,\Kone)+ \mathcal R(\eta,\lambda_*,\Kone) )\rho_1 
= Q(T_0 + \mathcal R(\eta,\lambda_*,\Kone) )\rho_1.
\end{equation}
In the above, we have used the fact that
\begin{equation*}
 \text{Ran}\left( \mathcal L(\eta,\lambda_*,\Kone) - \mathcal L(\eta,\lambda_0,\Kone) \right) \subset \text{span}\{x_1+\im x_2\}, \quad 
Q(x_1+\im x_2) =0.
\end{equation*}
Thus \eqref{eq:rhoproj} implies that
\begin{equation}
Q(T_0 + \mathcal R(\eta,\lambda_*,\Kone) )  \phi^{(1)}=  - Q(T_0 + \mathcal R(\eta,\lambda_*,\Kone) )\phi_* = -Q \mathcal R(\eta,\lambda_*,\Kone) \phi_*.
\end{equation}
Let $A$ be the inverse of $T_0: (\ker(T_0))^\perp\to \text{Ran}(T_0)$, where the function space is perpendicular  with respect to the $H^{1/2}(\partial D_*)$ inner product. 
We obtain
\begin{equation*}
(I+ AQ\mathcal R(\eta,\lambda_*,\Kone) ) \phi^{(1)}=  - A Q \mathcal R(\eta,\lambda_*,\Kone) \phi_*.
\end{equation*}
Using the boundedness of $A$ and $Q$, which do not depend on $\eta$, and \eqref{eq:remsmall}, we obtain
\begin{equation*}
   \phi^{(1)} = - (I+ AQ\mathcal R(\eta,\lambda_*,\Kone) )^{-1}  A Q \mathcal R(\eta,\lambda_*,\Kone) \phi_*. 
\end{equation*} 
Therefore, $\|\phi^{(1)}\|_{\HmhalfpD}=O(\eta)$. 
\end{proof}
Note that the eigenmodes $w_i$ in Theorem~\ref{lem:Dirac} is expanded by the single layer potential
\begin{equation*}
w_i(\bx) = \int_{\partial D_*}G^f(\bx,\eta\by;\lambda_*, \Kone)\rho_i(\by)\, ds_{\by}, \quad \bx\in\chomo\backslash D(\eta),
 \end{equation*}
where $\rho_i$ is defined in Theorem~\ref{lem:DiracP}.

\subsection{Slope of the Dirac cone}
In this subsection, we establish the conical singularity of the dispersion surfaces near the point $(\Kone, \lambda_*)$ for the band structure. In particular,  we derive the slope value $m_*$ for the Dirac cone.
\begin{thm}\label{lem:slope}
When $\eta$ is sufficiently small,
the two dispersion surfaces around $(\lambda_*,\Kone)$ takes the form  
\begin{equation}\label{eq:conical}
(\lambda-\lambda_*)^2 = m_*^2 |\bp-\Kone|^2+O(|\bp-\Kone|^3),\quad m_*\in\mathbb R, \quad m_*\geq0.
\end{equation}
The coefficient $m_*$ represents the slope of the Dirac cone, and is given by
\begin{equation}
m_* = \frac{2}{3}\big(1+O(\eta)\big).
\end{equation}
\end{thm}
\begin{proof}
To prove \eqref{eq:conical}, we apply directly Proposition~\ref{lem:pmlamgen} that will be proved in Section \ref{sec:dispersion}, instead of presenting a similar proof here. Proposition~\ref{lem:pmlamgen} covers more general scenarios and can be applied to derive \eqref{eq:conical}. In more details, we
set $\eps=0$ in Proposition~\ref{lem:pmlamgen} and notice that the Bloch wave vector $\bp$ near $\Kone$ can be expanded as $\bp = \Kone + \ell\bbeta_1+\mu\bbeta_2$, then 
 the conical shape of the dispersion relation \eqref{eq:conical} follows by observing that $|\bp - \Kone|^2=| \ell+\mu\bar\tau |^2|\bbeta_1|^2$. Thus $m_* = \frac{1}{|\bbeta_1|}|\frac{\theta_*}{\gamma_*}|$, where $\theta_*$ and $\gamma_*$ are defined in \eqref{eq:Tderiv}.

Next we compute the slope $m_* :=  \frac{1}{|\bbeta_1|}|\frac{\theta_*}{\gamma_*}|$. 
Note that the ratio $\frac{\theta_*}{\gamma_*}$ is independent of a scaling of $\rho_i$'s. We can thus use $\rho_i$'s defined in Theorem~\ref{lem:DiracP}(iii) for the calculation. Since
\begin{equation*} 
\partial_\lambda G^f(\by, \bx;\lambda, \Kone) = \frac{1}{|\cC_z|}\sum_{\bm\in\tilde\Lambda^*}\frac{1}{(\lambda - |\bm|^2)^2} e^{\im\bm\cdot (\bx-\by)},
\end{equation*}
we have
\begin{equation} \label{eq:gammacalc}
\begin{aligned}
 &\langle \rho_1, \partial_\lambda \mathcal S(\eta,\lambda_*,\Kone) \rho_1\rangle_{\partial D_*}\\
=& \frac{1}{|\cC_z|}\int_{(\partial D_*)^2} \left(\frac{1}{(\lambda - |\bm_1|^2)^2} (3 -  \frac{1}{3}(2\pi)^2(\eta|\bx -\by|)^2 +O(\eta^3)) +O(1) \right) \overline{\rho_1(\bx)}\rho_1(\by) \, ds_\by ds_\bx\\
=& \frac{1}{|\cC_z|}2\mathfrak a^2 \frac{\eta^2}{(\lambda - |\bm_1|^2)^2}\frac{2}{3}(2\pi)^2(1 + O(\eta)).
\end{aligned}
\end{equation}
In addition,
\begin{equation*} 
\begin{aligned}
(1,0)^T\cdot\nabla_\bp G^f(\by, \bx;\lambda, \Kone) =& -\frac{1}{|\cC_z|}\sum_{\bm\in\tilde\Lambda^*}\frac{2}{(\lambda - |\bm|^2)^2} e^{\im \bm\cdot (\bx-\by)} \bm\cdot(1,0)^T\\
&- \frac{1}{|\cC_z|}\sum_{\bm\in\tilde\Lambda^*} \frac{\im}{\lambda - |\bm|^2}e^{\im \bm\cdot (\bx-\by)}(\bx-\by)\cdot(1,0)^T.
\end{aligned}
\end{equation*}
Therefore, 
\begin{equation} \label{eq:thetacalc}
\begin{aligned}
&\langle \rho_2, (1,0)^T\cdot\nabla_\bp G^f(\by, \bx;\lambda, \Kone)  \rho_1\rangle_{\partial D_*}\\
=& - \frac{1}{|\cC_z|}\iint_{\partial D_*\times \partial D_*} \left(
\frac{2}{(\lambda - |\bm_1|^2)^2} \left(-\eta^2\frac{2}{9}(x_1-y_1)(x_2-y_2) + O(\eta^3)\right)
+ \frac{\im}{\lambda - |\bm_1|^2}O(\eta)
+O(1) \right) \overline{\rho_2(\bx)}\rho_1(\by) \, ds_\by ds_\bx\\
=& \frac{1}{|\cC_z|}2\im \mathfrak a^2 \frac{\eta^2}{(\lambda - |\bm_1|^2)^2}\frac{4}{9}(2\pi)^2(1 + O(\eta)).
\end{aligned}
\end{equation}
Here we have used the fact that $\rho_1 = \phi_* + O(\eta)$ and $\rho_2(\bx)=\rho_1(\rflc\bx)$ in Theorem~\ref{lem:DiracP}.
Applying Proposition~\ref{lem:pmlamgen} with $\eps =0$ again, we obtain $m_*=\frac{1}{|\bbeta_1|} |\frac{\theta_*}{\gamma_*}|=\frac{2}{3}(1+O(\eta))$.
\end{proof}

Theorems~\ref{lem:DiracP} and \ref{lem:slope} combined give the existence and asymptotic analysis of the Dirac point when the Bloch wave vector $\bp = K$ as stated in Theorem~\ref{lem:Dirac}. The Dirac point at $K'$ follows similarly since 
\begin{equation}
G^f(\bx,\by;\lambda, \Ktwo)= \overline{G^f(\bx,\by;\lambda, \Kone)}.
\end{equation}


\section{Band-gap opening at the Dirac point for the perturbed lattices}\label{sec:bands}

In this section, we consider the bandgap opening near the Dirac points for the perturbed honeycomb lattices. More precisely, we consider the spectral problem \eqref{eq:bandu_eps}, in which the obstacle in each period is rotated by an angle of $\pm\eps$ for $\eps\in\mathbb R$. Here and henceforth, we denote $D=D(\eta_0)$ for a fixed $\eta_0>0$ for the ease of notation. We define the operator $T (\eps,\lambda,\bp): H^{-1/2}(\partial D) \to H^{1/2}(\partial D)$ as
\begin{equation}\label{eq:Teps}
T (\eps,\lambda,\bp)[\phi] (\bx):=  \int_{\partial D}G^f(R^{\eps} \bx, R^{\eps} \by;\lambda, \bp)\phi(\by)\, ds_{\by}, \bx\in \partial D.
\end{equation}
Note that the $\eps$ in $T (\eps,\lambda,\bp)$ represents the rotation angle, while the $\eta$ in $\mathcal S (\eta,\lambda,\bp)$ defined in \eqref{eq:bandD} represents the size of the obstacle. The dependence of $T (\eps,\lambda,\bp)$ on $\eta$ is suppressed since $\eta$ is fixed at $\eta_0$. 
Similar to the discussion in Section~\ref{sec:band}, $(\lambda,\bp)$ belongs to the dispersion surface of the honeycomb lattice if and only if the triple $(\lambda,\bp, \phi) \in \mathbb R \times \cB_z \times H^{-1/2}(\partial D)$ solves the integral equation
\begin{equation}\label{eq:bandE}
T (\eps,\lambda,\bp)[\phi] (\bx)=0,\quad \bx\in \partial D.
\end{equation}
The corresponding eigenmode is given by
\begin{equation}\label{eq:modegen}
u^{\eps}(\bx):=  \int_{\partial D}G^f(\bx, R^{\eps} \by;\lambda, \bp)\phi(\by)\, ds_{\by}, \quad \bx\in \cpeps.
\end{equation}

We extend $u^{\eps}(\bx)$ to $\cC_z$ by letting
\begin{equation}\label{eq:ext}
\tilde u^{\eps}(\bx) :=
\begin{cases}
u^{\eps}(\bx),\quad \bx\in \cpeps, \\
0,\quad \bx \in D^\eps.
\end{cases}
\end{equation}
It is clear that $\|\tilde u^{\eps}\|_{L^2(\chomo)} = \|u^{\eps}\|_{L^2(\cpeps)}$ and $\|\tilde u^{\eps}\|_{H^1(\chomo)} = \|u^{\eps}\|_{H^1(\cpeps)}$. In what follows, for convenience we will abuse the notations and denote both $u^{\eps}$ and $\tilde u^{\eps}$ by $u^{\eps}$.

\subsection{Band structure and Bloch modes for the perturbed honeycomb structures}

In this subsection, we compute the band structure of the perturbed honeycombs around $\Kone$ by a perturbation argument. 
Recall that at the Dirac point $(\lambda_*,K)$, there holds (see Proposition~\ref{lem:DiracP})
\begin{equation}\label{eq:Fredmodesker}
\ker \, T (0,\lambda_*,\Kone) =\text{span}\left\{ \rho_1, \rho_2 \right\},
\end{equation}
where $R\rho_1(\bx):=\rho_1(R^{-1}\bx)=\tau \rho_1(\bx)$, $R\rho_2(\bx):=\rho_2(R^{-1}\bx)=\bar\tau \rho_2(\bx)$, and $\rho_2(\bx)=\rho_1(F\bx)$. 
From now on, we normalize $\rho_n$ such that
\begin{equation}\label{eq:w_normalization}
\|w_n\|_{L^2(\cunp)} = 1,
\end{equation}
where $w_n$ is the single-layer potential defined as  
\begin{equation}\label{eq:wi}
w_n(\bx):=  \int_{\partial D}G^f(\bx,  \by;\lambda, \bp)\rho_n(\by)\, ds_{\by}, \quad \bx\in \Omega^0, \quad n=1,2.
\end{equation}

We first characterize the partial derivatives of the integral operator $T (\eps,\lambda,\bp)$ with respect to $\eps$, $\lambda$, and $\bp$ respectively.
\begin{prop}\label{lem:Tderiv}
Let $\brho=(\rho_1,\rho_2)$, where $\rho_i$ are normalized such that \eqref{eq:w_normalization} holds. 
Then the  partial derivatives of $\langle \brho,T (\eps,\lambda,\bp)\brho\rangle_{\partial D}$ at $\eps=0$, $\lambda=\lambda_*$ and $\bp=\Kone$ take the following forms:
\begin{equation}\label{eq:Tderiv}
\begin{aligned}
\langle \brho, \partial_{\lambda}T (0,\lambda_*,\Kone)\brho\rangle_{\partial D}|_{\lambda=\lambda_*}&= 
\left(\begin{matrix}
\gamma_*&0\\
0&\gamma_*
\end{matrix}\right), \\
\langle \brho,  \bbeta_1\cdot\nabla_{\bp}T(0,\lambda_*,\bp)\brho\rangle_{\partial D}|_{\bp=\Kone}&= 
\left(\begin{matrix}
0&\overline{\theta_*}\\
\theta_*&0
\end{matrix}\right),\\
\langle \brho,  \bbeta_2\cdot\nabla_{\bp}T(0,\lambda_*,\bp)\brho\rangle_{\partial D}|_{\bp=\Kone}&= 
\left(\begin{matrix}
0&\tau\overline{\theta_*}\\
\overline{\tau}\theta_*&0
\end{matrix}\right),\\
\langle \brho, \partial_{\eps}T (\eps,\lambda_*,\Kone)\brho\rangle_{\partial D}|_{\eps=0}&= 
\left(\begin{matrix}
t_*&0\\
0&-t_*
\end{matrix}\right), 
\end{aligned}
\end{equation}
where $t_*,\gamma_*\in\mathbb R$, $\theta_*\in\mathbb C$. 
\end{prop}
In the above, for $(\phi_1,\phi_2)\in (H^{-1/2}(\partial D))^2$ and $(\psi_1,\psi_2)\in (H^{1/2}(\partial D))^2$, we denote $$\langle (\phi_1,\phi_2),(\psi_1,\psi_2)\rangle_{\partial D}:= \langle \phi_1,\psi_1\rangle_{\partial D} + \langle \phi_2,\psi_2\rangle_{\partial D}, $$ where the symbol $\langle \cdot,\cdot\rangle_{\partial D}$ on the right hand side represents the regular $H^{-1/2}(\partial D)$-$H^{1/2}(\partial D)$ pairing.  The proof of Proposition~\ref{lem:Tderiv} is given in Appendix~\ref{sec:Tderiv}.
\begin{remark}
In \eqref{eq:gammacalc} and \eqref{eq:thetacalc}, we calculated the values of $\gamma_*$ and $\theta_*$ and showed that $\gamma_*\neq 0$ and $\theta_*\neq0$. 
For simplicity, in all calculations throughout the paper, we assume $\gamma_*>0$ and $t_*>0$. The cases when $\gamma_*<0$ or $t_*<0$ can be treated similarly. Therefore, the main results in Section~\ref{sec:mainresults} hold regardless of the signs of $\gamma_*$ and $t_*$.
\end{remark}

\begin{prop}\label{lem:pmlamgen} Assume that $t_*>0$ and parameterize the quasimomenta near $K$ by $\bp(\ell,\mu):=K + \ell \bbeta_1 + \mu\bbeta_2$. Let $\ell\in\mathbb R$, $\mu\in\mathbb R$ and $\eps\geq0$ be sufficiently small. 

\begin{itemize}
    \item [(i)] The dispersion relations for the spectral problem \eqref{eq:bandu_eps} attain the following expansions:
\begin{equation}\label{eq:energymu}
\begin{aligned}
\lambda_{1,\pm\eps}(\bp(\ell,\mu)) &= \lambda_*  - \frac{1}{|\gamma_*|}\sqrt{\eps^2 t_*^2 + |\theta_*|^2 | \ell+\mu\bar\tau |^2}(1+O(\eps,\ell,\mu)), \\
\lambda_{2,\pm\eps}(\bp(\ell,\mu)) &= \lambda_*  + \frac{1}{|\gamma_*|}\sqrt{\eps^2 t_*^2 + |\theta_*|^2 | \ell+\mu\bar\tau |^2}(1+O(\eps,\ell,\mu)). 
\end{aligned}
\end{equation}
\item[(ii)] 
The corresponding density functions for the integral equation \eqref{eq:bandE} attain the expansions
\begin{equation}\label{eq:pdens}
\begin{aligned}
\phi_{1,\eps}(\bx;\bp(\ell,\mu))&=w_1 + L(\eps,\ell,\mu)w_2  +O(\eps,\ell,\mu),\\
\phi_{2,\eps}(\bx;\bp(\ell,\mu))&= - \overline{L(\eps,\ell,\mu)}w_1 + w_2  +O(\eps,\ell,\mu),
\end{aligned}
\end{equation}
\begin{equation}\label{eq:mdens}
\begin{aligned}
\phi_{1,-\eps}(\bx;\bp(\ell,\mu))= \overline{L(\eps,\ell,\mu)}w_1 + w_2  +O(\eps,\ell,\mu),\\
\phi_{2,-\eps}(\bx;\bp(\ell,\mu))=w_1 - L(\eps,\ell,\mu)w_2  +O(\eps,\ell,\mu).
\end{aligned}
\end{equation}
In the above,
\begin{equation}\label{eq:Lmu}
L(\eps,\ell,\mu):= \frac{\theta_*(\ell+\mu\overline{\tau})}{\eps t_*+ \sqrt{\eps^2 t_*^2 + |\theta_*|^2  | \ell+\mu\bar\tau |^2}}. 
\end{equation}

\item[(iii)] 
The corresponding Bloch modes with unit $L^2(\cC_z\backslash D^{\pm})$ norm take the form
\begin{equation}\label{eq:pmodesLmu}
\begin{aligned}
u_{1,\eps}(\bx;\bp(\ell,\mu))&=\left(w_1 + L(\eps,\ell,\mu)w_2  +O(\eps,\ell,\mu) \right)\frac{1}{\sqrt{1+|L(\eps,\ell,\mu)|^2+O(\eps,\ell,\mu)}},\\
u_{2,\eps}(\bx;\bp(\ell,\mu))&=\left( - \overline{L(\eps,\ell,\mu)}w_1 + w_2  +O(\eps,\ell,\mu) \right)\frac{1}{\sqrt{1+|L(\eps,\ell,\mu)|^2+O(\eps,\ell,\mu)}},
\end{aligned}
\end{equation}
\begin{equation}\label{eq:mmodesLmu}
\begin{aligned}
u_{1,-\eps}(\bx;\bp(\ell,\mu))&=\left( \overline{L(\eps,\ell,\mu)}w_1 + w_2  +O(\eps,\ell,\mu)\right)\frac{1}{\sqrt{1+|L(\eps,\ell,\mu)|^2+O(\eps,\ell,\mu)}},\\
u_{2,-\eps}(\bx;\bp(\ell,\mu))&=\left(w_1 - L(\eps,\ell,\mu)w_2  +O(\eps,\ell,\mu) \right)\frac{1}{\sqrt{1+|L(\eps,\ell,\mu)|^2+O(\eps,\ell,\mu)}}.
\end{aligned}
\end{equation}

\end{itemize}
\end{prop}

Note that for all $\eps$, $\ell$ and $\mu$, the eigenvalues above satisfy $\lambda_{1,\pm\eps}(\bp(\ell,\mu)) < \lambda_{2,\pm\eps}(\bp(\ell,\mu))$. Hence, a band gap is opened near the Dirac point for the spectral problem \eqref{eq:bandu_eps}. 
Another observation is that 
$$|
\bp-K|^2 = |\ell\bbeta_1 +\mu\bbeta_2|^2=\frac{4}{3}(\ell^2+\mu^2-\ell\mu) =\frac{4}{3}| \ell+\mu\bar\tau |^2.
$$ 
This is used to relate  Theorem~\ref{thm:band_gap} to Proposition~\ref{lem:pmlamgen}.

\begin{proof}
Let $V$ be a sufficiently small neighborhood of $\lambda_*$. We first show that for $\eps$ and $|\bp-\Kone|$ sufficiently small, the characteristic value of $T (\eps,\lambda,\bp)$ in $V$ has multiplicity two. 

Note that for $\eps$ and $|\bp-K|$ being sufficiently small, $T (\eps,\lambda,\bp)$ is an analytic family of operators in the variable $\lambda$. When $V$ is sufficiently small, from Section \ref{sec:Dirac}, it is known that $\lambda=\lambda_*$ is the only characteristic value of $T (0,\lambda,\Kone)$ within $V$. Indeed, the multiplicity of the characteristic $\lambda_*$ of $T (0,\lambda,\Kone)$ is two. This follows  from~\eqref{eq:Fredmodesker} and  \eqref{eq:Tderiv}. That is, $\ker\, T (0,\lambda_*,\Kone) =\text{span}\left\{ \rho_1, \rho_2 \right\}$ and 
\begin{equation}
\langle \phi, \partial_{\lambda}T (0,\lambda,\Kone) \phi \rangle_{\partial D} \neq0 \quad \forall \phi\in\ker T (0,\lambda_*,\Kone).
\end{equation}
Thus $\partial_{\lambda}T (0,\lambda,\Kone)|_{\lambda=\lambda_*} \psi \notin \text{Ran}(T (0,\lambda_*,\Kone))$.
We conclude, $\phi$ is of rank one and the multiplicity of $T (0,\lambda_*,\Kone)$ is two. Since $T (\eps,\lambda,\bp)$ is a Fredholm operator ~\cite{mclean2000strongly} and it is continuous with respect to $\eps$ and $\bp$, by Theorem~\ref{lem:actGS}, we deduce that $T (\eps,\lambda,\bp)$ has multiplicity two in $V$.

Next, we use the perturbation argument to show that for $\eps$ and $|\bp-\Kone|$ being sufficiently small, $T (\eps,\lambda,\bp)$ attains two characteristic values in $V$, with multiplicity one each. This argument also gives rise to the asymptotic expansion of the characteristic values and the density functions.

Let $\eps,\ell,\mu\ll1$. We solve for $(\lambda,\phi)$ pairs in $V\times H^{-1/2}(\partial D)$ such that 
$T (\eps,\lambda,\bp)\phi=0$. Let us express 
\begin{equation}
\bp(\ell,\mu) = \Kone+\bp^{(1)},\quad \lambda=\lambda_*+\lambda^{(1)}, \quad T (\eps,\lambda,\bp(\ell,\mu)) = T^{(0)}+T^{(1)}, \quad \phi = \phi^{(0)} +  \phi^{(1)}.
\end{equation}
Here $|\bp^{(1)}| = |\ell\bbeta_1 +\mu\bbeta_2| \ll1$, $\lambda^{(1)}\ll1$,  $T^{(0)} = T (0,\lambda_*,\Kone)$, $T^{(1)} = T (\eps,\lambda,\bp)  - T^{(0)}$, 
$\phi^{(0)} \in \ker (T^{(0)})$, and $\phi^{(1)}\in  \left(\ker (T^{(0)})\right)^{\perp}$, where the perpendicular sign is with respect to the inner product of $H^{-1/2}(\partial D)$.
Using $T^{(0)}\phi^{(0)} =0$, the integral equation $T (\eps,\lambda,\bp)\phi=0$ boils down to
\begin{equation}\label{eq:bandasymp}
T^{(0)} \phi^{(1)}+T^{(1)} (\phi^{(0)} +  \phi^{(1)})=0.
\end{equation}

Note that $T^{(0)}:H^{-1/2}(\partial D) \to H^{1/2}(\partial D)$ is a Fredholm operator and
the range of $T_0$ is the space perpendicular to $\ker(T^{(0)})$ in the dual sense. That is,
\begin{equation}
\text{Ran}T^{(0)} = \{ \psi\in H^{1/2}(\partial D): \langle \rho_i, \psi  \rangle_{\partial D} = 0, i=1,2\}.
\end{equation}
Define
\begin{equation}\label{eq:proj1}
Q\psi :=\psi - \sum_{i=1,2}\frac{\langle \rho_i,\psi  \rangle_{\partial D}}{\langle \rho_i,f_ i \rangle_{\partial D}} f_i,
\end{equation}
where $f_i = S\rho_i$.
It is straightforward to check that  $QH^{1/2}(\partial D)=\text{Ran}(T^{(0)})$ and $Q$ is a projection.
Thus \eqref{eq:bandasymp} is equivalent to 
\begin{equation}\label{eq:LyapCoRan}
\langle\psi, T^{(0)} \phi^{(1)}+T^{(1)} (\phi^{(0)} +  \phi^{(1)})\rangle_{\partial D}=0 \quad\forall\psi\in\ker T^{(0)},
\end{equation}
and
\begin{equation}\label{eq:LyapRan}
T^{(0)} \phi^{(1)}+QT^{(1)} (\phi^{(0)} +  \phi^{(1)})=0.
\end{equation}
Here we have used $QT^{(0)}=T^{(0)}$.

Let $A$ be the inverse of $T^{(0)}|_{(\ker(T^{(0)}))^\perp}: (\ker(T^{(0)}))^\perp\to \text{Ran} (T^{(0)})$. It follows from \eqref{eq:LyapRan} and $\phi^{(1)}\in  (\ker(T^{(0)}))^\perp$ that
\begin{equation*}
(I + AQT^{(1)} ) \phi^{(1)}+ AQT^{(1)} \phi^{(0)} =0.
\end{equation*}
Here we have used the fact that $AT^{(0)} \phi^{(1)} = \phi^{1}$ since $\phi^{1}\in\text{Ran}(T^{(0)})$.
Since $T^{(1)} = O(|\eps|, |\ell|, |\mu|,|\lambda^{(1)}|)$, when $\eps$, $\ell$ and $\lambda^{(1)}$ are sufficiently small,  $(I + AQT^{(1)})$ is invertible with an inverse norm bounded by $\frac{1}{2}$. Thus, there holds 
\begin{equation}\label{eq:phi1band}
\phi^{(1)}= - (I + AQT^{(1)} )^{-1} AQT^{(1)} \phi^{(0)},
\end{equation}
where $(I + AQT^{(1)} )^{-1}$ is the inverse of $I + AQT^{(1)}:(\ker(T^{(0)}))^\perp\to \text{Ran} (T^{(0)})$. We obtain 
\begin{equation*}
   (I + AQT^{(1)} )^{-1} AQT^{(1)} =  O(|\eps|, |\ell|,|\mu| |\lambda^{(1)}|).
\end{equation*}

Using the expansions
\begin{equation}\label{eq:phi0}
\phi^{(0)}=a\rho_1+b\rho_2 \text{ for some constants }a,b\in\mathbb C,
\end{equation}
and
\begin{equation*}
T^{(1)} = \eps \partial_{\eps}T (0,\lambda_*,\Kone) + (\ell \bbeta_1 + \mu\bbeta_2)\cdot\nabla_{\bp} T(0,\lambda_*,\Kone) + \lambda^{(1)} \partial_{\lambda}T (0,\lambda,\Kone) +O(|\eps|^2, |\ell |^2,|\mu|^2, |\lambda^{(1)}|^2),
\end{equation*}
and applying Proposition~\ref{lem:Tderiv}, \eqref{eq:LyapCoRan} becomes
\begin{equation}\label{eq:lam1phi0}
\cM(\eps, \ell, \mu,\lambda^{(1)})
\left(\begin{matrix}
a\\
b
\end{matrix}\right)
=0,
\end{equation}
where
\begin{equation}\label{eq:dispM}
\cM(\eps, \ell,\mu, \lambda^{(1)})=
\left(\begin{matrix}
t_*\eps + \gamma_* \lambda^{(1)}& \ell\overline{\theta_*}+\mu\tau\overline{\theta_*}\\\
\ell\theta_*+\mu\overline{\tau}\theta_*&-t_*\eps + \gamma_* \lambda^{(1)}
\end{matrix}\right)
+O(|\eps|^2, |\ell |^2, |\mu|^2, |\lambda^{(1)}|^2).
\end{equation}
With the ansatz 
\begin{equation}
 \lambda^{(1)} = \frac{x}{|\gamma_*|}\sqrt{\eps^2 t_*^2 + |\theta_*|^2 |\ell+\mu\overline{\tau}|^2},
\end{equation}
the inverse function theorem implies that when $\eps$, $\ell$, $\mu$ and $\phi^{(1)}$ are sufficiently small,
there exist $x=1+O(\eps,\ell,\mu)$ and $x= - 1+O(\eps,\ell,\mu)$ such that $\lambda^{(1)}$
satisfies $\det(\cM(\eps, \ell,\mu, \lambda^{(1)}))=0$.
For each of these two values of $\lambda^{(1)}$,
by solving \eqref{eq:phi0}, we obtain $\phi^{(1)}$ from \eqref{eq:phi1band} and \eqref{eq:phi0}.


The expansion of normalized eigenmodes follows from Lemma~\ref{lem:normest} below.
\end{proof}

Before presenting Lemma~\ref{lem:normest}, we introduce the following auxiliary lemma whose proof is elementary.
\begin{lemma}\label{lem:invfact}
Let $X$ and $Y$ be two Banach spaces. 
Consider two operators $A_\eps, A_0: X\to Y$ and two functions $f_\eps, f_0\in Y$.
Suppose $A_\eps^{-1}$ and $A_0^{-1}$ exist. Then 
\begin{equation}
\begin{aligned}
\| A_\eps^{-1} f_\eps - A_0 ^{-1} f_0\|_{X} 
&\leq
\| A_\eps^{-1}\|_{Y\to X}  \| A_0 - A_\eps\|_{X\to Y}  \|A_0 ^{-1}\|_{Y\to X}  \|f_\eps\|_{Y}  + \| A_0 ^{-1}\|_{Y\to X}   \|f_\eps -  f_0\|_{Y} .
\end{aligned}
\end{equation}
\end{lemma}
\begin{lemma}\label{lem:normest}
Let
\begin{equation}
\tilde u_{1,\eps}(\bx;\bp(\ell,\mu)):= 
\int_{\partial D}G^f(\bx, R^{\eps} \by;\lambda, \bp(\ell,\mu))\phi_{1,\eps}(\by)\, ds_{\by}, \quad \bx\in \cpeps.
\end{equation}
There holds
\begin{equation}
\|\tilde u_{1,\eps}(\cdot;\bp(\ell,\mu)) - \big(w_1 + L(\eps,\ell,\mu)w_2 \big) \|_{H^{1}(\chomo)} = O(\eps,\ell),\quad n=1,2.
\end{equation}
Here $w_i$ are defined in \eqref{eq:wi} and $L(\eps,\ell,\mu)$ is defined in \eqref{eq:Lmu}. 
\end{lemma}
\begin{proof}
The function $\tilde u_{n,\eps}(\bx,\bp(\ell,\mu))$ attains the quasi-momentum $\bp(\ell,\mu)$ and it solves the differential equation
\begin{equation}
\big(-\Delta - \lambda_{1,\eps}(\bp(\ell,\mu))\big) \tilde u_{1,\eps}(\bx;\bp(\ell,\mu)) = \delta(\bx\in\partial D^\eps)\phi_{1,\eps}(R^{-\eps}\bx) \quad \text{in }\cC_z,
\end{equation}
The function $w_1 + L(\eps,\ell,\mu)w_2$ attains the quasi-momentum $\Kone$ and it solves
\begin{equation}
(-\Delta -  \lambda_*) \big(w_1 + L(\eps,\ell,\mu)w_2\big)= \delta(\bx\in\partial D)\big(\rho_1(\bx) + L(\eps,\ell,\mu)\rho_2(\bx)\big) \quad \text{in }\cC_z.
\end{equation}

Let us fix a pair of small $\ell$ and $\mu$.  Define the operators
\begin{equation*}
A_\eps := -\Delta - \lambda_{1,\eps}(\bp(\ell,\mu)) \quad \mbox{and} \quad f_\eps:= \delta(\bx\in\partial D^\eps)\phi_{1,\eps}(R^{-\eps}\bx), \quad \text{for }\eps\neq0
\end{equation*}
and
\begin{equation*}
A_0 := -\Delta -  \lambda_* \quad {and} \quad f_0:=\delta(\bx\in\partial D)\big(\rho_1(\bx) + L(\eps,\ell,\mu)\rho_2(\bx)\big).
\end{equation*}
Then $\tilde u_{1,\eps} - \big(w_1 + L(\eps,\ell,\mu)w_2)  = A_\eps^{-1} f_\eps - A_0 ^{-1} f_0$.
It is straightforward to verify that
\begin{equation}
\begin{aligned}
&|\bp(\ell,\mu)-\Kone|=O(\ell,\mu),\quad |\lambda_{1,\eps}(\bp(\ell,\mu)) -  \lambda_*| = O(\eps,\ell,\mu),\\
&\|\delta(\bx\in\partial D^\eps)\phi_{1,\eps}(R^{-\eps}\bx) - \delta(\bx\in\partial D)\big(\rho_1(\bx) + L(\eps,\ell,\mu)\rho_2(\bx)\big)\|_{H^{-1}(\chomo)} = O(\eps,\ell,\mu).
\end{aligned}
\end{equation}
By Theorem~\ref{lem:Dirac}, for $\eps$ sufficiently small, $\lambda_{1,\eps}$ are uniformly away from $\{|\bm|^2\}_{\bm\in \tilde \Lambda^*}$. Thus the inverses of $A_\eps$ and $A_0$ exist and are uniformly bounded. Applying Lemma~\ref{lem:invfact}, we finish the proof. 
\end{proof}

\section{Floquet theory and the Green functions in a periodic strip with a zigzag cross section}\label{sec:prep}
In this section and the subsequent two sections, we investigate the existence of interface modes for the joint photonic structure along a zigzag interface that solve \eqref{eq:spectral_prob_zigzag} for $\kp= \kp^* := \frac{4\pi}{3}$. 
The purpose of this section is to introduce the Floquet theory and Green functions in the following infinite strip
\begin{equation*}
\Omega^{\eps} := \Omega^J\backslash \cup_{m\in\mathbb Z} (D^{\eps}+m\be_1), \quad \eps\in\mathbb R.
\end{equation*}
We note that when $\eps=0$, $\Omega^0$ represents the unperturbed strip.
We define the following function spaces that are quasi-periodic along $\be_2$:
\begin{equation}\label{eq:QP1pm}
\begin{aligned}
\cH^{\eps}_{\text{loc}}:=\{&u\in H^1_{\text{loc}}(\Omega^{\eps}): \Delta u\in L^2_{\text{loc}}(\Omega^{\eps}), \quad u=0 \text{ on } \cup_{m\in\mathbb Z}(\partial D^{\eps} +m\be_1),\\
&u(\bx+ \be_2)= e^{\im \kp^*}u(\bx) \text{ for }\bx\in\Gamma_-, \quad \partial_{\nuGl}  u(\bx+ \be_2) =e^{\im \kp^*}\partial_{\nuGl}  u(\bx) \text{ for }\bx\in\Gamma_-\}.
\end{aligned}
\end{equation}
Then the analysis boils down to the spectrum of the operator $\Delta$ in $\cH^{\eps}_{\text{loc}}$,  that is, eigenpairs  $(\lambda,u)\in \mathbb R\times \cH^{\eps}_{\text{loc}}$ satisfying
\begin{equation}\label{eq:stripspec}
\begin{aligned}
-\Delta u - \lambda u &= 0 \quad&\text{on } \Omega^{\eps},\\
u &= 0 \quad&\text{on } \cup_{m\in\mathbb Z} (\partial D^\eps+m\be_1).
\end{aligned}
\end{equation}

\subsection{Floquet theory in a periodic strip  with a zigzag cross section}\label{sec:Floquet}
In this subsection, we decompose the operator $\Delta$ on $\cH^{\eps}_{\text{loc}}$ using the Floquet theory along the direction $\be_1$.
Let $\bp(\ell):=\Kone+\ell\bbeta_1$. For each $\ell\in\mathbb R$, we
denote $-\Delta^\eps(\ell)$ the restriction of 
 $-\Delta$ on the space $\cH^\eps(\ell)$ with the quasi-momentum $(\Kone+\ell\bbeta_1)\cdot\be_1$ along the direction $\be_1$, i.e. 
\begin{equation}\label{eq:QP20}
\begin{aligned}
\cH^\eps(\ell):=\{&u\in H^1(\cunp): \Delta u\in L^2(\cpeps), \quad u=0 \text{ on } \partial D^\eps,\\
&u(\bx+ \be_2)= e^{\im \kp^*}u(\bx) \text{ for }\bx\in\Gb, \quad \partial_{\nuGb}  u(\bx+ \be_2) =e^{\im \kp^*}\partial_{\nuGb}  u(\bx) \text{ for }\bx\in\Gb\\
&u(\bx+ \be_1)= e^{\im (\Kone+\ell\bbeta_1)\cdot\be_1}u(\bx) \text{ for }\bx\in\Gl,\quad
\partial_{\nuGl}  u(\bx+ \be_1) =e^{\im (\Kone+\ell\bbeta_1)\cdot\be_1}\partial_{\nuGl}  u(\bx) \text{ for }\bx\in\Gl \}.
\end{aligned}
\end{equation}
Here $\Gb$ and $\Gl$ are the bottom and left boundaries of $\chomo$ shown in Figure~\ref{fig:cellhomo},
the directional derivative $\partial_{\nuGb}$ is normal to $\Gt$ and $\Gb$ in the direction $\nuGb=(\frac{1}{2},\frac{\sqrt3}{2})$
and the directional derivative $\partial_{\nuGl}$ is normal to $\Gl$ and $\Gr$ in the direction $\nuGl=(\frac{1}{2},-\frac{\sqrt3}{2})$.
Equivalently, we solve for the $(\lambda,u)$ pair for each $\ell\in\mathbb R$ that satisfies
\begin{equation}\label{eq:band0} 
\begin{aligned}
\begin{cases}
-\Delta u - \lambda u = 0 \quad &\text{in } \cpeps,\\
u=0 \quad &\text{on }\partial D^\eps, \\
u(\bx+ \be_2)= e^{\im \kp^*}u(\bx) \quad &\text{for }\bx\in\Gb, \\
 \partial_{\nuGb}  u(\bx+ \be_2) =e^{\im \kp^*}\partial_{\nuGb}  u(\bx)  \quad &\text{ for }\bx\in\Gb,\\
u(\bx+ \be_1)= e^{\im (\Kone+\ell\bbeta_1)\cdot\be_1}u(\bx) \quad &\text{ for }\bx\in\Gl,\\
\partial_{\nuGl}  u(\bx+ \be_1) =e^{\im (\Kone+\ell\bbeta_1)\cdot\be_1}\partial_{\nuGl}  u(\bx)\quad& \text{ for }\bx\in\Gl.
\end{cases}
\end{aligned}
\end{equation}

For each fixed $\ell\in\mathbb R$, $-\Delta^{\eps}(\ell)$ is a self-adjoint 
positive operator with compact resolvent, thus its spectrum is real, discrete, and accumulates at $\infty$. The eigenvalues of $-\Delta^\eps(\ell)$ are labeled as $\lambda_n(\ell)$ in an increasing order
\begin{equation}
0\leq \lambda_{1,\eps}(\ell) \leq \lambda_{2,\eps}(\ell)\leq\cdots\leq \lambda_{n,\eps}(\ell) \leq\cdots.
\end{equation}
Note that $\lambda_{n,\eps}(\ell)$ are $2\pi$-periodic, continuous and piecewise differentiable functions in $\ell$. 
The corresponding eigenmodes $u_{n,\eps}(\bx,\bp(\ell))$ are chosen to be orthonormal with respect to the $L^2$ inner product in $\cpeps$. 
$\lambda_{n,\eps}(\ell)$ may not be differentiable at points $\ell$ where $\lambda_{n,\eps}(\ell)$ is not a simple eigenvalue, which only occurs at a finite number of $\ell$ values within a period for each $n$. 

The spectrum of $-\Delta^\eps(\ell)$ can alternatively be labeled as smooth branches as follows.
The smooth labeling enables a representation of the Green function using the Bloch modes to be introduced in Section~\ref{sec:dispstrip}. To be more precise, there exists a sequence of complex neighborhoods $D_{n,\eps}$ of $\mathbb R$, a sequence of analytic functions $\mu_{n,\eps}(\ell):D_{n,\eps}\to \mathbb C$, and a sequence of analytic functions $v_{n,\eps}(\bx,\bp(\ell)): D_{n,\eps}\to H^1(\Delta,\cpeps)$ such that
\begin{equation}\label{eq:muv}
\begin{aligned}
&v_{n,\eps}(\cdot,\bp(\ell))\in \cH^\eps(\ell), \quad -\Delta v_{n,\eps}(\cdot,\bp(\ell)) = \mu_{n,\eps}(\ell)v_{n,\eps}(\cdot,\bp(\ell)),\quad \ell\in \mathbb R, \\
&\left(v_{n,\eps}(\cdot,\bp(\ell)),v_{m,\eps}(\cdot,\bp(\ell))\right)_{L^2(\cpeps)} = \delta_{m,n},\quad \ell\in \mathbb R.
\end{aligned}
\end{equation}
In the above,
\begin{equation}\label{eq:H1Delta}
\begin{aligned}
H^1(\Delta,\cpeps):= \{&u\in H^1(\cpeps): \Delta u\in L^2(\cpeps), \quad u=0 \text{ on } \partial D^\eps,\\
&u(\bx+ \be_2)= e^{\im \kp^*}u(\bx) \text{ for }\bx\in\Gb, \quad \partial_{\nuGb}  u(\bx+ \be_2) =e^{\im \kp^*}\partial_{\nuGb}  u(\bx) \text{ for }\bx\in\Gb
\}.
\end{aligned}
\end{equation}
Moreover,
\begin{equation}
\forall \ell\in\mathbb R, \quad \{ \mu_{n,\eps}(\ell), n\geq1\} = \{ \lambda_{n,\eps}(\ell), n\geq1\},
\end{equation}
and the eigenmodes $u_{n,\eps}(\cdot,\ell)$ are chosen such that 
\begin{equation}\label{eq:unchoice}
\forall \ell\in\mathbb R, \quad \{ v_{n,\eps}(\cdot,\bp(\ell))), n\geq1\} = \{\alpha_{n,\eps} u_{n,\eps}(\cdot,\bp(\ell))), n\geq1\},
\end{equation}
where $\alpha_{n,\eps}$ is an $\ell$-dependent phase factor.
We extended the eigenmodes $u_{n,\eps}$ and $v_{n,\eps}$ to the whole strip $\Omega^\eps$ as quasi-periodic functions by letting
\begin{equation}
\begin{aligned}
u_{n,\eps}(\bx+m\be_1,\Kone+\ell\bbeta_1)=e^{\im (\Kone+\ell\bbeta_1)\cdot m\be_1}u_{n,\eps}(\bx,\Kone+\ell\bbeta_1),\quad \bx\in\cunp,\quad m\in\mathbb Z,\\
v_{n,\eps}(\bx+m\be_1,\Kone+\ell\bbeta_1)=e^{\im (\Kone+\ell\bbeta_1)\cdot m\be_1}v_{n,\eps}(\bx,\Kone+\ell\bbeta_1),\quad \bx\in\cunp,\quad m\in\mathbb Z.
\end{aligned}
\end{equation}

When $\eps=0$, for convenience we will  abbreviate $D_{n,0}$, $\lambda_{n,0}$, $u_{n,0}$, $\mu_{n,0}$ and $v_{n,0}$ as $D_n$, $\lambda_n$, $u_n$, $\mu_n$ and $v_n$, respectively.

\subsection{The band structure for the 
periodic strip $\Omega^{\eps}$ with a zigzag cross section near $\lambda^*$} \label{sec:dispstrip}

In this subsection, we derive the band structure for the periodic strip $\Omega^{\eps}$ with a zigzag cross section near $\lambda^*$. Note that the eigenvalues that solve \eqref{eq:band0} near the Dirac point $(\lambda_*,\Kone)$ can be obtained from Proposition~\ref{lem:pmlamgen} by letting $\mu=0$. Denoting $L(\eps,\ell) = L(\eps,\ell,0)$, where $L(\eps,\ell,\mu)$ is defined in \eqref{eq:Lmu}, we have

\begin{lemma}\label{lem:pmstrip}
Assume $t_*>0$. For sufficiently small $\ell\in\mathbb R$ and $\eps\geq0$, the eigenvalues for 
\eqref{eq:band0} are given by
\begin{equation}\label{eq:energ}
\begin{aligned}
\lambda_{1,\pm\eps}(\bp(\ell)) &= \lambda_*  - \frac{1}{|\gamma_*|}\sqrt{\eps^2 t_*^2 + |\theta_*|^2 \ell^2}(1+O(\eps,\ell)), \\
\lambda_{2,\pm\eps}(\bp(\ell)) &= \lambda_*  + \frac{1}{|\gamma_*|}\sqrt{\eps^2 t_*^2 + |\theta_*|^2 \ell^2}(1+O(\eps,\ell)). 
\end{aligned}
\end{equation}
The $L^2$-normalized Bloch modes for the first two bands on the $\pm\eps$-strips take the following forms
\begin{equation}\label{eq:pmodesL}
\begin{aligned}
u_{1,\eps}(\bx;\bp(\ell))&=\left(w_1 + L(\eps,\ell)w_2  +O(\eps,\ell) \right)\frac{1}{\sqrt{1+|L(\eps,\ell)|^2+O(\eps,\ell)}},\\
u_{2,\eps}(\bx;\bp(\ell))&=\left( - \overline{L(\eps,\ell)}w_1 + w_2  +O(\eps,\ell) \right)\frac{1}{\sqrt{1+|L(\eps,\ell)|^2+O(\eps,\ell)}},
\end{aligned}
\end{equation}
\begin{equation}\label{eq:mmodesL}
\begin{aligned}
u_{1,-\eps}(\bx;\bp(\ell))=\left( \overline{L(\eps,\ell)}w_1 + w_2  +O(\eps,\ell)\right)\frac{1}{\sqrt{1+|L(\eps,\ell)|^2+O(\eps,\ell)}},\\
u_{2,-\eps}(\bx;\bp(\ell))=\left(w_1 - L(\eps,\ell)w_2  +O(\eps,\ell) \right)\frac{1}{\sqrt{1+|L(\eps,\ell)|^2+O(\eps,\ell)}}.
\end{aligned}
\end{equation}
\end{lemma}
\begin{remark}\label{lem:commongap}
By Lemma~\ref{lem:pmstrip}, when Assumption~\ref{lem:assNoFold} holds in $\bbeta_1$, for an arbitrary fixed constant $\mathfrak d\in(0,1)$, when $\eps>0$ is sufficiently small, $-\Delta^{\pm\eps}$ on $\Omega^{\pm\eps}$ attain a common spectral band gap $(\lambda_* - \hrad\eps, \lambda_* + \hrad\eps)$.
\end{remark}
\begin{remark}
When $\eps>0$ in Lemma~\ref{lem:pmstrip}, observe that $\lambda_{n,\pm\eps}(\bp(\ell))$, $n=1,2$, are smooth in $\ell$. Thus 
$\mu_{n,\pm\eps}(\bp(\ell)) =\lambda_{n,\pm\eps}(\bp(\ell))$ and
$v_{n,\pm\eps}(\bx;\bp(\ell))= u_{n,\pm\eps}(\bx;\bp(\ell))$. 
\end{remark}
Setting $\eps=0$ in Lemma~\ref{lem:pmstrip}, we observe that $\lambda_n$, $n=1,2$, are not smooth at $\ell=0$. Thus $\mu_n(\bp(\ell))$ are obtained by matching different branches of $\lambda_n$ as shown in the following lemma and illustrated in Figure~\ref{fig:bandsplit}.
\begin{lemma}\label{lem:0strip}
For sufficiently small $\ell\in\mathbb R$,
\begin{equation}\label{eq:Diracband}
\begin{aligned}
\mu_1(\bp(\ell)) &= \lambda_*  + |\frac{\theta_*}{\gamma_*}|\ell(1+O(\ell)) \quad\text{(increasing in $\ell$)}, \\
\mu_2(\bp(\ell)) &= \lambda_*  - |\frac{ \theta_*}{\gamma_*}|\ell(1+O(\ell)) \quad\text{(decreasing in $\ell$)}. 
\end{aligned}
\end{equation}
The corresponding Bloch modes can be chosen as
\begin{equation}
\begin{aligned}
v_1(\bx;\bp(\ell))&=\left( \frac{\overline{\theta_*}}{|\theta_*|}w_1 - w_2  +O(\ell) \right)\frac{1}{\sqrt{2+O(\ell)}},\\
v_2(\bx;\bp(\ell))&=\left(  \frac{\overline{\theta_*}}{|\theta_*|}w_1 +  w_2  +O(\ell) \right)\frac{1}{\sqrt{2+O(\ell)}}.
\end{aligned}
\end{equation}
\end{lemma}

\begin{figure}
\begin{center}
\includegraphics[scale=0.4]{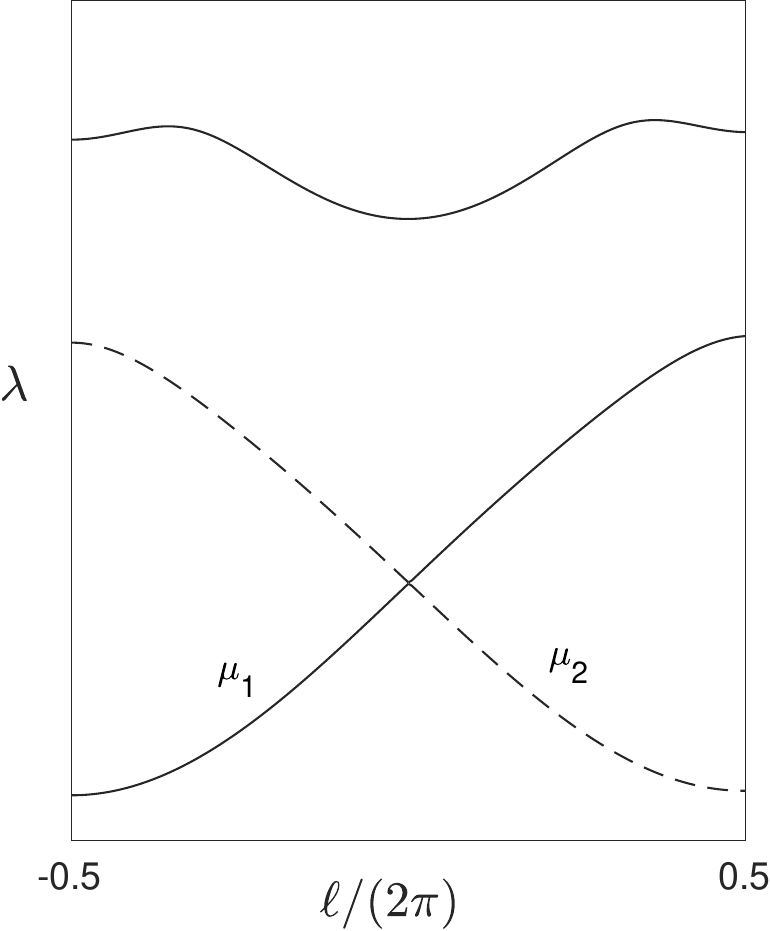}
\includegraphics[scale=0.4]{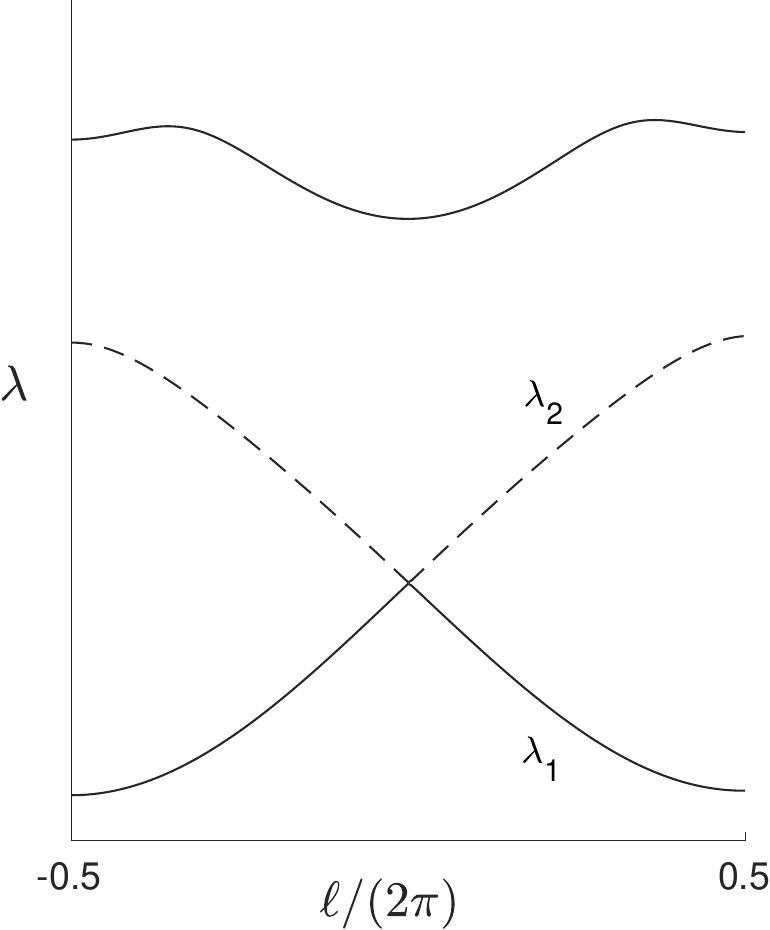}
\includegraphics[scale=0.4]{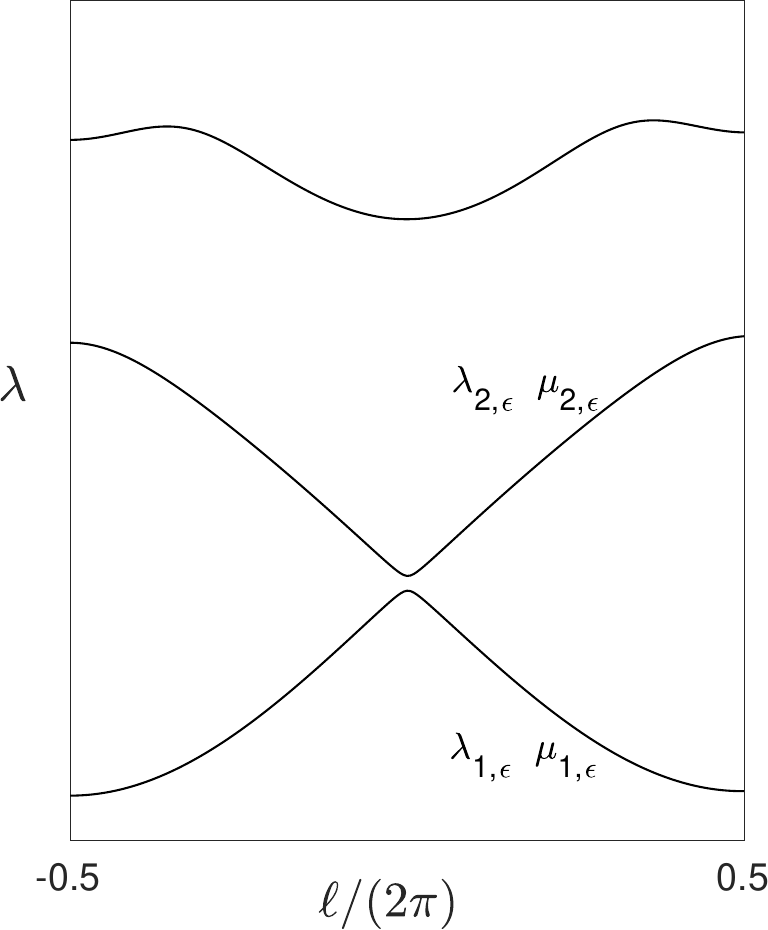}
\end{center}
\caption{The labelings of the band eigenvalues: $\mu_{n,\eps}$ are smooth branches, while $\lambda_{n,\eps}$ are piecewise smooth. When $\eps\neq0$, $\mu_{n,\eps}=\lambda_{n,\eps}$ for $n=1,2$.}
\label{fig:bandsplit}
\end{figure}

Let $v_i := v_i(\bx;\Kone)$. It follows that
\begin{equation}\label{eq:winv}
\begin{cases}
v_1=\frac{1}{\sqrt2}\left( \frac{\overline{\theta_*}}{|\theta_*|} w_1 - w_2\right) \\
v_2=\frac{1}{\sqrt2}\left(  \frac{\overline{\theta_*}}{|\theta_*|}w_1 +  w_2 \right)
\end{cases}
,\quad
\begin{cases}
w_1 = \frac{1}{\sqrt2}\frac{\theta_*}{|\theta_*|} (v_1  +  v_2)\\
w_2 = \frac{1}{\sqrt2} (- v_1 + v_2)
\end{cases}.
\end{equation}

\begin{remark}\label{lem:phase}
For $n=1,2$, there exist $\ell$-dependent phase factors $\alpha_n$ such that $\|u_n(\cdot,\bp(\ell))- \alpha_n u_{n,\eps}(\cdot,\bp(\ell))\|_{H^1(\chomo)} = O(\eps)$ uniformly for $\ell>0$ that are sufficiently small; and there exist $\ell$-dependent phase factors $\beta_n$ such that $\|u_n(\cdot,\bp(\ell))- \beta_n u_{n,\eps}(\cdot,\bp(\ell))\|_{H^1(\chomo)} = O(\eps)$ uniformly for $\ell<0$ that are sufficiently small. The same holds when $u_{n,\eps}$ is replaced by $u_{n,-\eps}$.
\end{remark}

\subsection{The $q$-sesquilinear form}

Define the quasi-periodic Sobolev space on $\Gamma$, $\cH^s(\Gamma)$, for $s\in\mathbb R$ by 
\begin{equation}\label{eq:HsG2}
\cH^s(\Gamma):= \left\{ u(\bx_0+ t\be_2) = \sum_{n\in\mathbb Z} a_n e^{\im \kp^* t} e^{\im 2\pi n t}: \|u\|_{\cH^s(\Gamma)}^2:=\sum_{n\in\mathbb Z} |a_n|^2 (1+|n|^2)^s\right\},
\end{equation}
Here $\bx_0 = -\frac{1}{2}\be_1 -\frac{1}{2}\be_2$ is the lower left corner of $\chomo$. 

Define the $q$-sesquilinear form on $\Gamma$ for functions $a,b$ in some neighborhood of $\Gamma$ with traces in $\HhalfG$ and normal derivatives in $\HmhalfG$:
\begin{equation}\label{eq:qses}
q(a,b):=\overline{\langle\partial_n a,b\rangle}_{\Gamma} - \langle\partial_n b,a\rangle_{\Gamma},
\end{equation}
where $\partial_n$ represents the normal derivative on $\Gamma$ in the direction $\bn = \nu_1=(\frac{1}{2},-\frac{\sqrt3}{2})$,  $\langle \phi,\psi\rangle_{\Gamma}$ represents the $\HmhalfG$-$\HhalfG$  pairing (basically $\int_\Gamma \overline{\phi}\psi\, ds$).
The $q$-sesquilinear form orthogonalized the modes with the same quasimomentum and same energy. That is, if $\mu_n(\bp(\ell_0)) = \mu_m(\bp(\ell_0))$, then
\begin{equation}
q(v_n(\cdot,\bp(\ell_0)),v_m(\cdot,\bp(\ell_0)))=0,\quad m,n\in\{1,2\}, \quad  m\neq n.
\end{equation}
\begin{equation}
q(v_n(\cdot,\bp(\ell_0)),v_n(\cdot,\bp(\ell_0)))=\im \frac{d\mu_n(\bp(\ell))}{d\ell}|_{\ell=\ell_0}, \quad n=1,2.
\end{equation}
On the unperturbed strip $\Omega^0$,  by Lemma~\ref{lem:0strip}, $\mu_1(\bp(0)) = \mu_m(\bp(0)) = \lambda_*$ and $\bp(0)=\Kone$, we know
\begin{equation}\label{eq:ortho0}
q(v_n(\cdot,\Kone),v_m(\cdot,\Kone))=0,\quad m,n\in\{1,2\}, \quad  m\neq n.
\end{equation}
\begin{equation}\label{eq:ortho1}
q(v_n(\cdot,\Kone),v_n(\cdot,\Kone))=\im \frac{d\mu_n(\bp(\ell))}{d\ell}|_{\ell=0}, \quad n=1,2.
\end{equation}
In addition, $\frac{d\mu_2(\bp(\ell))}{d\ell}|_{\ell=0} = - \frac{d\mu_1(\bp(\ell))}{d\ell}|_{\ell=0}$. We denote 
\begin{equation}\label{eq:alphadef}
\alpha_* := \left|\frac{d\mu_n(\bp(\ell))}{d\ell}|_{\ell=0}\right|,\quad n=1,2.
\end{equation}
\begin{remark}
\label{lem:alphabeta}
By Lemma~\ref{lem:0strip}, the derivative $\alpha_*$ defined in \eqref{eq:alphadef} is given by $\alpha_*=\left|\frac{\theta_*}{\gamma_*}\right|$.
\end{remark}

\subsection{The Green functions in the periodic strip $\Omega^{\eps}$ with a zigzag cross section}\label{sec:Green}
In this subsection, we introduce 
the Green functions in $\Omega^{\eps}$ with the quasi-periodic conditions using the limiting absorption principle and spectral representation in Section~\ref{sec:Floquet}. This result extends that in \cite{FlissJoly2016}.

Consider solving the following problem in $\Omega^\eps$:
\begin{equation}
\begin{aligned}
\begin{cases}
-\Delta u - (\lambda+ \im \sigma)  u  = f \quad &\text{in } \Omega^\eps,\\
u =0 \quad &\text{on }  \cup_{m\in\mathbb Z} (\partial D^\eps+m\be_1), \\
u(\bx+ \be_2)= e^{\im \kp^*}u(\bx) \quad &\text{for }\bx\in\Gamma_-, \\
\partial_{\nuGb}  u(\bx+ \be_2) =e^{\im \kp^*}\partial_{\nuGb}  u(\bx)  \quad &\text{ for }\bx\in\Gamma_-,\\
\end{cases}
\end{aligned}
\end{equation}
where $f\in L^2(\Omega^\eps)$, and $\sigma$ is a positive constant that converges to $0$. 
The corresponding Green function $G^\eps(\bx,\by; \lambda)$ satisfies
\begin{equation}\label{eq:phyGeps}
\begin{aligned}
\begin{cases}
(-\Delta_{\bx}  - \lambda) G^\eps(\bx,\by; \lambda)  = \delta(\bx-\by) \quad &\bx\in \Omega^\eps,\\
G^\eps(\bx,\by; \lambda) =0 \quad &\bx\in\cup_{m\in\mathbb Z} (\partial D^\eps+m\be_1), \\
G^\eps(\bx+ \be_2,\by; \lambda)= e^{\im \kp^*}G^\eps(\bx,\by; \lambda) \quad &\text{for }\bx\in\Gamma_-, \\
\partial_{\nuGb}  G^\eps(\bx+ \be_2,\by; \lambda) =e^{\im \kp^*}\partial_{\nuGb}  G^\eps(\bx,\by; \lambda)  \quad &\text{ for }\bx\in\Gamma_-. \\
\end{cases}
\end{aligned}
\end{equation}
In addition, the radiation conditions are imposed using the limiting absorption principle.

We will need the  Green functions on $\Omega^{0}$ at the energy $\lambda_*$. Recall that $\lambda_*$ is only an eigenvalue of \eqref{eq:band0} when $\eps=0$ and $\bp=\Kone$, where $\lambda_*$ is an eigenvalue of~\eqref{eq:band0} of multiplicity two. We have
\begin{equation}
\begin{aligned}\label{eq:G0}
G^0(\bx,\by; \lambda_*) = &
\sum_{n\geq3}\frac{1}{2\pi}\int_{[-\pi,\pi ]} 
\frac{\overline{v_{n}(\by;\bp(\ell))} v_{n}(\bx;\bp(\ell))}{ \mu_{n}(\bp(\ell)) - \lambda_* } \, \dpt +\sum_{n=1,2}\frac{1}{2\pi}\text{p.v.}\int_{[-\pi,\pi ]} 
\frac{\overline{v_{n}(\by;\bp(\ell))} v_{n}(\bx;\bp(\ell))}{ \mu_{n}(\bp(\ell)) - \lambda_* } \, \dpt \\
&+ \frac{\im}{2\alpha_*}\overline{v_{1}(\by;\Kone)} v_{1}(\bx;\Kone) + \frac{\im}{2\alpha_*}\overline{v_{2}(\by;\Kone)} v_{2}(\bx;\Kone),\quad \bx,\by\in\Omega^0,
\end{aligned}
\end{equation}
where $\mu_n$ and $v_n$ are the eigenvalues and eigenfunctions that are analytic in $\ell$ as introduced in \eqref{eq:muv}. For convenience, we denote the integral portion of the Green function by
\begin{equation}\label{eq:tildeG0}
\tilde{G^0}(\bx,\by; \lambda_*) := 
\sum_{n\geq3}\frac{1}{2\pi}\int_{[-\pi,\pi ]} 
\frac{\overline{v_{n}(\by;\bp(\ell))} v_{n}(\bx;\bp(\ell))}{ \mu_{n}(\bp(\ell)) - \lambda_* } \, \dpt  +\sum_{n=1,2}\frac{1}{2\pi}\text{p.v.}\int_{[-\pi,\pi ]} 
\frac{\overline{v_{n}(\by;\bp(\ell))} v_{n}(\bx;\bp(\ell))}{ \mu_{n}(\bp(\ell)) - \lambda_* } \, \dpt.
\end{equation}

When $\bx\cdot\be_1\to +\infty$, the terms in $G^0(\bx,\by; \lambda_*)$ can be regrouped as
\begin{equation}\label{eq:Gright}
G^0(\bx,\by; \lambda_*) = G^{0,+}(\bx,\by; \lambda_*)
+ \frac{\im}{\alpha_*}\overline{v_{1}(\by;\Kone)} v_{1}(\bx;\Kone), 
\end{equation}
where $G^{0,+}(\bx,\by; \lambda_*)$ decays exponentially as $\bx\cdot\be_1\to+\infty$, and is given by
\begin{equation}\label{eq:G0p}
G^{0,+}(\bx,\by; \lambda_*) :=\tilde{G^{0}}(\bx,\by; \lambda_*) - \frac{\im}{2\alpha_*}\overline{v_{1}(\by;\Kone)} v_{1}(\bx;\Kone) + \frac{\im}{2\alpha_*}\overline{v_{2}(\by;\Kone)} v_{2}(\bx;\Kone).
\end{equation}
When $\bx\cdot\be_1\to -\infty$, the terms in $G^0(\bx,\by; \lambda_*)$ can be regrouped as
\begin{equation}
G^0(\bx,\by; \lambda_*) = 
G^{0,-}(\bx,\by; \lambda_*) 
+ \frac{\im}{\alpha_*}\overline{v_{2}(\by;\Kone)} v_{2}(\bx;\Kone), 
\end{equation}
where $G^{0,-}(\bx,\by; \lambda_*)$ decays exponentially as $\bx\cdot\be_1\to-\infty$, and is given by
\begin{equation}\label{eq:G0m}
G^{0,-}(\bx,\by; \lambda_*) :=
\tilde{G^{0}}(\bx,\by; \lambda_*) 
+ \frac{\im}{2\alpha_*}\overline{v_{1}(\by;\Kone)} v_{1}(\bx;\Kone) - \frac{\im}{2\alpha_*}\overline{v_{2}(\by;\Kone)} v_{2}(\bx;\Kone).
\end{equation}

Denote the  Green functions in $\Omega^{\pm\eps}$ by $G^{\pm\eps}(\bx,\by;\lambda)$. For $\lambda\in(\lambda_* - \hrad\eps, \lambda_* + \hrad\eps)$, 
by~\cite{FlissJoly2016}, there holds
\begin{equation}\label{eq:pGreenpmeps}
G^{\pm\eps}(\bx,\by; \lambda) = 
 \sum_{n\geq1}\frac{1}{2\pi}\int_{[-\pi,\pi ]}
\frac{\overline{v_{n,\pm\eps}(\by;\bp(\ell))} v_{n,\pm\eps}(\bx;\bp(\ell))}{ \mu_{n,\pm\eps}(\bp(\ell)) - \lambda } \, \dpt ,\quad \bx,\by\in\Omega^0.
\end{equation}
Moreover, $G^{\pm\eps}(\bx,\by; \lambda)$ decays exponentially as $|\bx\cdot\be_1|\to\infty$.

\begin{remark}
Note that the $q$-sesquilinear form and the Green functions are independent of phase factors of the Floquet modes $u_{n}$, $v_{n}$, $u_{n,\pm\eps}$ and $v_{n,\pm\eps}$. 
\end{remark}

\section{Integral equations for the interface modes along a zigzag interface}\label{sec:characterization}

In this section, we establish the integral equations for the interface modes at a zigzag interface separating two honeycomb lattices using the layer potentials~\cite{colton2013integral,mclean2000strongly}. This is achieved by matching the Dirichlet and Neumann traces of the wave fields along the interface. 
Let $\Gamma$ be the interface two lattices as shown in Figure \ref{fig:honeycomb_joint_zigzag}
and $\bn=(\frac{1}{2},-\frac{\sqrt3}{2})$ be the unit normal vector of $\Gamma$ pointing to the right. Let $\eps>0$ and $\lambda\in (\lambda_* - \hrad\eps, \lambda_* + \hrad\eps)$, for $(\psi,\phi)\in \HhalfG\times\HmhalfG$, we define the single and double layer potentials:
\begin{equation}\label{eq:lpotentialeps}
\begin{aligned}
\cS^{\pm\eps}(\lambda)\phi(\bx)&:=\int_\Gamma  G^{\pm\eps}(\bx,\by;\lambda) \phi(\by)\, ds_{\by}, \quad \bx\notin \Gamma, \\
\cD^{\pm\eps}(\lambda)\psi(\bx)&:=\int_\Gamma  \partial_{n_\by} G^{\pm\eps}(\bx,\by;\lambda) \psi(\by)\, ds_{\by} \quad \bx\notin \Gamma,  
\end{aligned}
\end{equation}
where $G^{\pm\eps}(\bx,\by;\lambda)$ are the  Green functions on the $\pm\eps$-strip defined in \eqref{eq:pGreenpmeps}.
The single layer potential $\cS^{\pm\eps}(\lambda)\phi(\bx)$ can be continuously extended to $\Gamma$ and it defines an bounded integer operator from $\HmhalfG$ to $\HhalfG$, which we still denote by $\cS^{\pm\eps}$. Given $(\psi,\phi)\in \HhalfG\times\HmhalfG$, we also define the integral operators
\begin{equation}\label{eq:int_opt_K}
\begin{aligned}
\cK^{\pm\eps}(\lambda)\psi(\bx)&:=\int_\Gamma  \partial_{n_\by} G^{\pm\eps}(\bx,\by;\lambda) \psi(\by)\, ds_{\by} \quad \bx\in \Gamma,\\
\cK^{*,\pm\eps}(\lambda)\phi(\bx)&:=\int_\Gamma  \partial_{n_\bx} G^{\pm\eps}(\bx,\by;\lambda) \phi(\by)\, ds_{\by} \quad \bx\in \Gamma.\\
\end{aligned}
\end{equation}
It can be shown that $\cK^{\pm\eps}:\HhalfG\to\HhalfG$ and $\cK^*:\HmhalfG\to\HmhalfG$ are bounded.

By taking the limit of the layer potentials as $\bx\to\Gamma$,
the following jump relationship holds~\cite{colton2013integral}: 
\begin{equation}\label{eq:jumpeps}
\begin{aligned}
[\cS^\eps\psi (\lambda) ]_{\pm} &= \cS^\eps (\lambda) \psi, \\
[\partial_n \cS^\eps(\lambda)  \psi ]_{\pm} &= \mp\frac{1}{2}\psi+\cK^{*,\eps}(\lambda) \psi, \\
[ \cD^\eps\phi (\lambda) ]_{\pm} &= \pm\frac{1}{2}\phi + \cK^\eps (\lambda) \phi, \\
[\partial_n \cD^\eps (\lambda) \phi ]_{\pm}& =: \cN^\eps\phi. 
\end{aligned}
\end{equation}
In the above, the subscript $-$ and $+$ represent the limit of the layer potentials as $\bx\to\Gamma$ from the left and right side respectively. $\partial_n$ represents the normal derivative, and $\cN^{\pm\eps}: H^{1/2}(\Gamma)\to H^{-1/2}(\Gamma)$ are well-defined bounded operators.
In addition, it is clear that the jump relations \eqref{eq:jumpeps} hold when $\eps$ is replaced by $-\eps$.

Assume that $u(\bx)$ is an interface mode of \eqref{eq:spectral_prob_zigzag} with the eigenvalue $\lambda$. 
Let  $u|_\Gamma \in \HhalfG$ and $\partial_n u|_\Gamma \in \HmhalfG$ be the traces of $u$ and the normal derivatives of $u$ on $\Gamma$. Then by the Green's formula, it can be shown that $u$ attains the following representation in the infinite strip $\Omega^{J,\eps}$:
\begin{equation}\label{eq:tracetoedge}
u(\bx) = 
\begin{cases}
\big[\cD^\eps (\lambda)  u|_\Gamma \big](\bx) - \big[\cS^\eps(\lambda) \partial_n u|_\Gamma  \big] (\bx)  \quad \text{for}\; \bx \; \text{on the right of }\Gamma,\\
-\big[ \cD^{-\eps} (\lambda)  u|_\Gamma \big] (\bx) +  \big[\cS^{-\eps} \partial_n u|_\Gamma (\lambda)  \big](\bx)  \quad \text{for}\; \bx \; \text{on the left of }\Gamma.
\end{cases}
\end{equation}
Here we used the fact that  $u\in\cH^{J,\eps}$, especially the decay of $u$ when $|\bx\cdot\be_1| \to\infty$ when applying the Green's formula.
Taking the limit from either side of $\Gamma$, we obtain the following two systems of integral equations:
\begin{equation}\label{eq:densuright}
\left(\begin{matrix}
u|_\Gamma\\
\partial_n u|_\Gamma
\end{matrix}\right)
 =
\left(\begin{matrix}
\cK^{\eps}(\lambda)  + \frac{1}{2}\cI  & -\cS^{\eps}(\lambda)   \\
\cN^\eps(\lambda)  & -\cK^{*,\eps}(\lambda)  + \frac{1}{2}\cI
\end{matrix}\right) 
\left(\begin{matrix}
u|_\Gamma\\
\partial_n u|_\Gamma
\end{matrix}\right),
\end{equation}
and
\begin{equation}
\left(\begin{matrix}
u|_\Gamma\\
\partial_n u|_\Gamma
\end{matrix}\right)
=
\left(\begin{matrix}
- \cK^{-\eps}(\lambda)  + \frac{1}{2}\cI & \cS^{-\eps} (\lambda)    \\
 - \cN^{-\eps}(\lambda)  & \cK^{*,-\eps}(\lambda)  + \frac{1}{2}\cI
\end{matrix}\right) 
\left(\begin{matrix}
u|_\Gamma\\
\partial_n u|_\Gamma
\end{matrix}\right).
\end{equation}
The above is equivalent to the following two systems
\begin{equation}\label{eq:densuleft}
\begin{aligned}
&\left(\begin{matrix}
-(\cK^{\eps}(\lambda)+\cK^{-\eps}(\lambda)) & \cS^{\eps}(\lambda)+\cS^{-\eps}(\lambda) \\
 -(\cN^{\eps}(\lambda)+\cN^{-\eps}(\lambda)) & \cK^{*,\eps}(\lambda)+\cK^{*,-\eps}(\lambda)
\end{matrix}\right)
\left(\begin{matrix}
u|_\Gamma\\
\partial_n u|_\Gamma
\end{matrix}\right)
=0,\\
\text{and}\\
&\left(\begin{matrix}
-\cK^{\eps}(\lambda)+\cK^{-\eps}(\lambda)+\cI & \cS^{\eps}(\lambda)-\cS^{-\eps}(\lambda) \\
 -\cN^{\eps} (\lambda) +\cN^{-\eps}(\lambda) & \cK^{*,\eps}(\lambda)-\cK^{*,-\eps}(\lambda) + \cI 
\end{matrix}\right)
\left(\begin{matrix}
u|_\Gamma\\
\partial_n u|_\Gamma
\end{matrix}\right)
=0.
\end{aligned}
\end{equation}
It is obvious that $u$ is nontrivial only when $(u|_\Gamma,\partial_n u|_\Gamma)$ is nontrivial.

Conversely, assume $\lambda\in (\lambda_* - \hrad\eps, \lambda_* + \hrad\eps)$. Let $(\psi,\phi)\in \HhalfG\times\HmhalfG$, which is not necessarily the Cauchy data of an interface mode on the interface $\Gamma$. We define $u(\bx)$ in the infinite strip $\Omega^{J,\eps}$ as a combination of single and double layer potentials:
\begin{equation}\label{eq:denstoedge}
u(\bx)= 
\begin{cases}
  [\cD^\eps(\lambda)  \psi](\bx)  - [\cS^\eps(\lambda)  \phi](\bx)  \quad \text{on the right of }\Gamma, \\
-[\cD^{-\eps}(\lambda)  \psi](\bx)  + [\cS^{-\eps}(\lambda)  \phi](\bx)  \quad \text{on the left of }\Gamma.
\end{cases}
\end{equation}
Since the Green functions $G^{\pm\eps}(\bx,\by;\lambda)$ decay as $|\bx\cdot \be_1|\to\infty$  for $\lambda$ located in the gap, $u$ defined above is an interface mode if and only if it is nontrivial and its value and normal derivatives are continuous across the interface $\Gamma$. Using \eqref{eq:jumpeps}, taking the limit of the layer potentials and their normal derivatives as $\bx\to\Gamma$, we obtain the system of integral equations:
\begin{equation}\label{eq:dataright}
\left(\begin{matrix}
 \cK^{\eps}(\lambda) + \frac{1}{2}\cI & -\cS^{\eps}(\lambda) \\
 \cN^\eps(\lambda) & -\cK^{*,\eps}(\lambda) + \frac{1}{2}\cI 
\end{matrix}\right)  
\left(\begin{matrix}
\psi\\
\phi
\end{matrix}\right)
=
\left(\begin{matrix}
- \cK^{-\eps}(\lambda) + \frac{1}{2}\cI & \cS^{-\eps}(\lambda)   \\
- \cN^{-\eps}(\lambda) & \cK^{*,-\eps}(\lambda) + \frac{1}{2}\cI
\end{matrix}\right) 
\left(\begin{matrix}
\psi\\
\phi
\end{matrix}\right)\neq0.
\end{equation}
This is equivalent to
\begin{equation}
\left(\begin{matrix}
-(\cK^{\eps}(\lambda)+\cK^{-\eps}(\lambda)) & \cS^{\eps}(\lambda)+\cS^{-\eps}(\lambda) \\
 -(\cN^{\eps}(\lambda)+\cN^{-\eps}(\lambda)) & \cK^{*,\eps}(\lambda)+\cK^{*,-\eps}(\lambda) 
\end{matrix}\right)
(\lambda)
\left(\begin{matrix}
\psi\\
\phi
\end{matrix}\right)
=0,
\end{equation}
and 
\begin{equation}
\left(\begin{matrix}
\cK^{\eps}(\lambda) - \cK^{-\eps}(\lambda)+\cI & - \cS^{\eps}(\lambda) + \cS^{-\eps}(\lambda) \\
\cN^{\eps}(\lambda) - \cN^{-\eps}(\lambda) & - \cK^{*,\eps}(\lambda) + \cK^{*,-\eps}(\lambda) + \cI
\end{matrix}\right)
(\lambda)
\left(\begin{matrix}
\psi\\
\phi
\end{matrix}\right)
\neq0.
\end{equation}

Define the integral operators on $\HhalfG\times\HmhalfG$ 
\begin{equation}\label{eq:bTeps}
\mathbb T^\eps
(\lambda):=
\left(\begin{matrix}
-\cK^{\eps}(\lambda) & \cS^{\eps}(\lambda) \\
 -\cN^{\eps}(\lambda) & \cK^{*,\eps}(\lambda)
\end{matrix}\right),
\end{equation}
and
\begin{equation}\label{eq:bTs}
\mathbb T_s^\eps(\lambda):=\mathbb T^\eps + \mathbb T^{-\eps}, \quad
\mathbb T_t^\eps(\lambda):=- \mathbb T^\eps + \mathbb T^{-\eps} + \mathbb I, \quad
\mathbb T_n^\eps(\lambda):=\mathbb T^\eps - \mathbb T^{-\eps} + \mathbb I,
\end{equation}
where $\mathbb I$ is the identity operator.
Based on the above discussion, we obtain the following lemma for the characterization of interface modes.
\begin{lemma}\label{lem:edgestate}
Let $\lambda\in (\lambda_* - \hrad\eps, \lambda_* + \hrad\eps)$.

\begin{itemize}
    \item [(i)] There exists an interface mode $u$ satisfying \eqref{eq:edgedef} if and only if there exists $(\psi,\phi)\in \HhalfG\times\HmhalfG$ such that
\begin{equation}\label{eq:sufficient}
\mathbb T_s^\eps
(\lambda)
\left(\begin{matrix}
\psi\\
\phi
\end{matrix}\right)
=0,\quad  \mathbb T_t^\eps
(\lambda)
\left(\begin{matrix}
\psi\\
\phi
\end{matrix}\right)
\neq0.
\end{equation}
Furthermore, each solution to \eqref{eq:sufficient} yields an interface mode expressed by \eqref{eq:denstoedge}.
\item [(ii)] If $u$ is an interface mode satisfying \eqref{eq:edgedef}, then $0\neq( u|_\Gamma, \partial_n u|_\Gamma) \in \HhalfG\times \HmhalfG$ satisfies
\begin{equation}\label{eq:necessary}
\mathbb T_s^\eps
(\lambda)
\left(\begin{matrix}
u|_\Gamma\\
\partial_n u|_\Gamma
\end{matrix}\right)
=0,\quad
\mathbb T_n^\eps
(\lambda)
\left(\begin{matrix}
u|_\Gamma\\
\partial_n u|_\Gamma
\end{matrix}\right)
=0.
\end{equation}
\end{itemize}
\end{lemma}
\begin{remark}
First, suppose that $(\lambda,\psi_i,\phi_i)$ $i=1,\cdots,N$ satisfy \eqref{eq:sufficient} for some positive integer $N$. Let $u_i$ be defined by \eqref{eq:denstoedge} correspondingly. When $\{(\psi_i,\phi_i)\}_{i=1,\cdots,N}$ are linearly independent, $\{u_i\}_{i=1,\cdots,N}$ may be linearly dependent.

Second, the converse of Lemma~\ref{lem:edgestate} part (ii) does not hold. That is, a triple $(\lambda, \psi,\phi)$ satisfying \eqref{eq:necessary} may not produce an interface mode through \eqref{eq:denstoedge}. 

Third, the subscript for $\mathbb T^\eps_s$ represents ``sufficient", 
that for $\mathbb T^\eps_t$ represents ``nontrivial",
and that for $\mathbb T^\eps_n$ represents ``necessary".
\end{remark}

We introduce some notations similar to \eqref{eq:lpotentialeps}-\eqref{eq:jumpeps} and \eqref{eq:bTeps}. Specifically, for $\eps=0$, in the infinite strip $\Omega^0$, we define
$\cS^0(\lambda_*)$, $\cD^0(\lambda_*)$, $\cK^0(\lambda_*)$, $\cK^{*,0}(\lambda_*)$ and $\cN^0(\lambda_*)$ parallel to  \eqref{eq:lpotentialeps}-\eqref{eq:jumpeps} where the Green functions are replaced by $G^0(\bx,\by,\lambda_*)$ defined in \eqref{eq:G0}, and
$\tilde\cS^0(\lambda_*)$, $\tilde\cD^0(\lambda_*)$, $\tilde\cK^0(\lambda_*)$, $\tilde\cK^{*,0}(\lambda_*)$ and $\tilde\cN^0(\lambda_*)$,
where the Green functions are replaced by $\tilde G^0(\bx,\by,\lambda_*)$ defined in \eqref{eq:tildeG0}.
We also define 
$\cS^{0,\pm}(\lambda_*)$, $\cD^{0,\pm}(\lambda_*)$, $\cK^{0,\pm}(\lambda_*)$, $\cK^{*,0,\pm}(\lambda_*)$ and $\cN^{0,\pm}(\lambda_*)$, where the Green functions are replaced by $G^{0,\pm}(\bx,\by,\lambda_*)$ defined in \eqref{eq:G0p} and \eqref{eq:G0m}.
These layer potentials have the jump relations when $\eps$ is replaced by $0$ and $\lambda$ is replaced by $\lambda_*$ in \eqref{eq:jumpeps}.

Finally, define the integral operators on $\HhalfG\times\HmhalfG$ 
\begin{equation}\label{eq:bT0}
\mathbb T^0
(\lambda_*):=
\left(\begin{matrix}
-\cK^0(\lambda_*) & \cS^0(\lambda_*) \\
 -\cN^0(\lambda_*) & \cK^{0,*}(\lambda_*)
\end{matrix}\right)\quad\text{and}\quad
\tilde{\mathbb T}^0
(\lambda_*):=
\left(\begin{matrix}
-\tilde\cK^0(\lambda_*) & \tilde\cS^0(\lambda_*) \\
 -\tilde\cN^0(\lambda_*) & \tilde\cK^{0,*}(\lambda_*)
\end{matrix}\right).
\end{equation}

\section{The proof of Theorem \ref{lem:edge}}\label{sec:proofedge}
In this section, we investigate interface modes along a zigzag interface using the integral equation formulation in Lemma~\ref{lem:edgestate}.  We will first derive the limit of the integral operators, and then apply the generalized Rouch\'{e} theorem in Gohberg-Sigal theory to investigate the characteristic values of the integral operators.

\subsection{The limiting operators for $\mathbb T^\eps$, $\mathbb T_s^\eps$, $\mathbb T_t^\eps$, and $\mathbb T_n^\eps$ }
We derive asymptotic expansions for the integral operators $\mathbb T^\eps$, $\mathbb T_s^\eps$, $\mathbb T_t^\eps$, and $\mathbb T_n^\eps$ in this subsection. To this end, we first introduce several notations. For $\vec\phi = (\psi,\phi)\in\HhalfG\times\HmhalfG$, let
\begin{equation}\label{eq:ci}
c_i(\vec\phi):= \overline{\langle \phi,v_i\rangle}_\Gamma - \langle\partial_n v_i, \psi\rangle_\Gamma,
\end{equation}
where $v_i$ are defined in Remark~\ref{lem:phase}.
We also denote
\begin{equation}\label{eq:vecv}
\vec v_i:=
\left(\begin{matrix}v_i|_\Gamma \\\partial_n v_i|_\Gamma\end{matrix}\right),\quad
i=1,2,
\end{equation}
and define the operators
\begin{equation}\label{eq:P}
\mathbb P
\vec\phi : =
c_1(\vec\phi)\vec v_1
+c_2(\vec\phi)\vec v_2,
\end{equation}
\begin{equation}\label{eq:Q}
\mathbb Q
\vec\phi 
:=
c_2(\vec\phi)\vec v_1
+c_1(\vec\phi)\vec v_2.
\end{equation}
Let $\beta(h)$ and $\xi(h)$ be two functions given by
\begin{equation}\label{eq:betaxi}
\begin{aligned}
    & \beta(h):=   \frac{1}{2}\left|\frac{\gamma_*}{\theta_*}\right|\frac{h}{\sqrt{(\frac{t_*}{\gamma_*})^2-h^2}} = \frac{1}{2\alpha_*}\frac{h}{\sqrt{\beta_*^2-h^2}}, \\
    & \xi(h):=  \frac{t_*}{2|\theta_*|}\frac{1}{\sqrt{(\frac{t_*}{\gamma_*})^2-h^2}} = \frac{\beta_*}{2\alpha_*}\frac{1}{\sqrt{\beta_*^2-h^2}},
\end{aligned}
\end{equation}
where $\alpha_*$ is defined in Remark~\ref{lem:alphabeta} and  $\beta_*:= \frac{t_*}{|\theta_*|}$.
We have the following lemma for the limit of the integral operator $\mathbb T^{\pm\eps}$ as $\eps\to0$.
\begin{prop} \label{lem:oplim}
Let Assumption~\ref{lem:assNoFold} holds along $\bbeta_1$ and $t_*>0$. Let $\mathfrak d\in(0,1)$ be a constant. 
Then the following limit holds uniformly for $h\in\mathbb C$ that satisfy $|h|< \hrad$ as $\eps\to 0^+$: 
\begin{equation}\label{eq:oplim}
\mathbb T^{\pm\eps}(\lambda_*+ \eps h) 
\to \tilde{\mathbb T}^0 (\lambda_*) 
+ \beta(h)\mathbb P \mp \xi(h)\mathbb Q =:\mathbb U^{\pm}(h),
\end{equation}
where
the convergence is understood with the operator norm from $\HhalfG\times\HmhalfG$ to $\HhalfG\times\HmhalfG$. 
\end{prop}
The proof of the proposition is presented in Appendix~\ref{sec:oplim}. It is based on the representation of the  Green functions in the infinite strip in terms of the band modes~\cite{FlissJoly2016}: 
\begin{equation}\label{eq:Seps} 
\cS^{\pm\eps}(\lambda) \phi = 
\sum_{n\geq1}\frac{1}{2\pi}\int_{[-\pi,\pi ]} 
\frac{\overline{\langle \phi, v_{n,\pm\eps}(\cdot;\bp(\ell))\rangle}_\Gamma v_{n,\pm\eps}(\bx;\bp(\ell))}{ \mu_{n,\pm\eps}(\bp(\ell)) - \lambda } \, \dpt ,
\end{equation}
\begin{equation}\label{eq:Keps}
\cK^{\pm\eps}(\lambda) \psi =
\sum_{n\geq1}\frac{1}{2\pi}\int_{[-\pi,\pi ]} 
\frac{\langle \partial_n v_{n,\pm\eps}(\cdot;\bp(\ell)),\psi\rangle_\Gamma v_{n,\pm\eps}(\bx;\bp(\ell))}{ \mu_{n,\pm\eps}(\bp(\ell)) - \lambda} \, \dpt ,
\end{equation}
\begin{equation}\label{eq:Kstareps}
\cK^{*,\pm\eps}(\lambda) \phi=
\sum_{n\geq1}\frac{1}{2\pi}\int_{[-\pi,\pi ]} 
\frac{\overline{\langle \phi, v_{n,\pm\eps}(\cdot;\bp(\ell))\rangle}_\Gamma \partial_n v_{n,\pm\eps}(\bx;\bp(\ell))}{ \mu_{n,\pm\eps}(\bp(\ell)) - \lambda} \, \dpt ,
\end{equation}
\begin{equation}\label{eq:Neps}
\cN^{\pm\eps}(\lambda) \psi =
\partial_n\left( 
\sum_{n\geq1}\frac{1}{2\pi}\int_{[-\pi,\pi ]} 
\frac{\langle \partial_n v_{n,\pm\eps}(\cdot;\bp(\ell)),\psi\rangle_\Gamma  v_{n,\pm\eps}(\bx;\bp(\ell))}{ \mu_{n,\pm\eps}(\bp(\ell)) - \lambda} \, \dpt ,
\right)
\end{equation}
and
\begin{equation}\label{eq:tildeS0}
\tilde\cS^{0} (\lambda) \phi = 
\left(\sum_{n\geq3}+\sum_{n=1,2}\text{p.v.}\right)\frac{1}{2\pi}\int_{[-\pi,\pi ]} 
\frac{\overline{\langle \phi, v_{n}(\cdot;\bp(\ell))\rangle}_\Gamma v_{n}(\bx;\bp(\ell))}{ \mu_{n}(\bp(\ell)) - \lambda} \, \dpt ,
\end{equation}
\begin{equation}\label{eq:tildeK0}
\tilde\cK^{0}(\lambda)\psi =
\left(\sum_{n\geq3}+\sum_{n=1,2}\text{p.v.}\right)\frac{1}{2\pi}\int_{[-\pi,\pi ]} 
\frac{\langle \partial_n v_{n}(\cdot;\bp(\ell)),\psi\rangle_\Gamma v_{n}(\bx;\bp(\ell))}{ \mu_{n}(\bp(\ell))- \lambda} \, \dpt ,
\end{equation}
\begin{equation}\label{eq:tildeKstar0}
\tilde\cK^{*,0}(\lambda)\phi=
\left(\sum_{n\geq3}+\sum_{n=1,2}\text{p.v.}\right)\frac{1}{2\pi}\int_{[-\pi,\pi ]} 
\frac{\overline{\langle \phi, v_{n}(\cdot;\bp(\ell))\rangle}_\Gamma \partial_n v_{n}(\bx;\bp(\ell))}{ \mu_{n}(\bp(\ell))- \lambda} \, \dpt ,
\end{equation}
\begin{equation}\label{eq:tildeN0}
\tilde\cN^{0}(\lambda) \psi =
\partial_n\left( 
\left(\sum_{n\geq3}+\sum_{n=1,2}\text{p.v.}\right)\frac{1}{2\pi}\int_{[-\pi,\pi ]} 
\frac{\langle \partial_n v_{n}(\cdot;\bp(\ell)),\psi\rangle_\Gamma v_{n}(\bx;\bp(\ell))}{ \mu_{n}(\bp(\ell))- \lambda} \, \dpt
\right).
\end{equation}

\begin{coro}\label{lem:stnlimits}
Let Assumption~\ref{lem:assNoFold} hold along $\bbeta_1$ and $t_*>0$. Let $\mathfrak d\in(0,1)$ be a constant.
The following limits hold under the operator norm from $\HmhalfG\times\HhalfG$ to $\HhalfG\times\HmhalfG$ 
uniformly for $h\in\mathbb C$ that satisfy $|h|< \hrad$ 
as $\eps\to 0^+$: 
\begin{equation}\label{eq:sufficientlim}
\mathbb T_s^\eps (\lambda_*+\eps h)
= \mathbb U_s (h)+\mathbb R_1(h,\eps),
\end{equation}
\begin{equation}\label{eq:nontriviallim}
\mathbb T_t^\eps (\lambda_*+\eps h)
=\mathbb U_t(h)+\mathbb R_2(h,\eps),
\end{equation}
\begin{equation}\label{eq:necessarylim}
\mathbb T_n^\eps (\lambda_*+\eps h)
=\mathbb U_n(h)+\mathbb R_3(h,\eps).
\end{equation}
Here the limiting operators are
\begin{equation}\label{eq:ziglim}
\mathbb U_s (h) := 2\tilde{\mathbb T}^0(\lambda_*) 
+ 2\beta(h)\mathbb P,\quad
\mathbb U_t(h):=\mathbb I
+2\xi(h)\mathbb Q,\quad
\mathbb U_n(h):=\mathbb I
-2\xi(h)\mathbb Q,
\end{equation}
and the remainder terms have the estimate
 $\|\mathbb R_i(h,\eps)\|_{ \HhalfG\times\HmhalfG\to \HhalfG\times\HmhalfG} =o(1)$ as $\eps\to 0^+$ uniformly for $|h|<\hrad$, 
$i=1,2,3$. 

\end{coro}

\subsection{Properties of the limiting operators $\mathbb U_s$, $\mathbb U_t$ and $\mathbb U_n$}

Using the definition of $c_i$ in \eqref{eq:ci}, the definition of the q-sesquilinear form \eqref{eq:qses}, and the relations \eqref{eq:ortho0} and \eqref{eq:ortho1},
we obtain
\begin{equation}\label{eq:civj}
\begin{aligned}
c_1(\vec v_1) &= q(v_1,v_1) =  \overline{\langle\partial_n v_1,v_1\rangle}_\Gamma - \langle\partial_n v_1,v_1\rangle_\Gamma =\im \alpha_*, \\
c_2(\vec v_2) &= q(v_2,v_2) =  \overline{\langle\partial_n v_2,v_2\rangle}_\Gamma - \langle\partial_n v_2,v_2\rangle_\Gamma = -\im \alpha_*, \\
c_i(\vec v_j) &= q(v_j,v_i) =  \overline{\langle\partial_n v_j,v_i\rangle}_\Gamma - \langle\partial_n v_i,v_j\rangle_\Gamma= 0 \text{ for }i\neq j,
\end{aligned}
\end{equation}
where $\vec v_i$ are defined in \eqref{eq:vecv}.
Define the function spaces
\begin{equation}\label{eq:XY}
 X:=\text{span}\{\vec v_1, \vec v_2\},\quad Y:=\{\vec\phi \in \HhalfG\times\HmhalfG\}, c_i(\vec\phi)=0, i=1,2\}.
\end{equation}
Then $Y$ is the orthogonal complement of $X$ in the sense of dual spaces.
We let
\begin{equation}\label{eq:PY}
P_Y (\vec\phi) :=\vec\phi -
\frac{c_1(\vec\phi)}{i\alpha_*}
\vec v_1 -
\frac{c_2(\vec\phi)}{-i\alpha_*}
\vec v_2. 
\end{equation}
Since $c_i(P_Y (\vec\phi))=0$, $i=1,2$, we obtain
the following direct sum decomposition:
\begin{equation}\label{eq:directsum}
\HhalfG\times\HmhalfG = X \bigoplus Y.
\end{equation}

The following fact will be used repeatedly in the sequel. \eqref{eq:civj} implies that $\vec v_1$ and $\vec v_2$ are linearly independent.
For the operators defined in \eqref{eq:P} and \eqref{eq:Q}, there holds
\begin{equation}\label{eq:PQX}
\mathbb P \vec v_1= \im\alpha_* \vec v_1, \quad \mathbb P \vec v_2= - \im\alpha_* \vec v_2, \quad 
\mathbb Q \vec v_1= \im\alpha_*\vec v_2, \quad 
\mathbb Q \vec v_2= - \im\alpha_* \vec v_1,
\end{equation}
and 
\begin{equation}\label{eq:PQY}
\mathbb P Y= \mathbb Q Y=0.
\end{equation}

\begin{lemma}\label{lem:tildeT0}
The kernel and range of the operator $\tilde{\mathbb T}^0$ are
\begin{equation}
\ker \; \tilde{\mathbb T}^0 (\lambda_*) =X,\quad \text{Ran}  \; \tilde{\mathbb T}^0 (\lambda_*) =Y.
\end{equation}
\end{lemma}
\begin{proof}
We first show $X\subset \ker(\tilde{\mathbb T}^0 (\lambda_*))$. To this end, we establish the following relations:
\begin{equation}\label{eq:tildeT0vi}
\tilde{\mathbb T}^0 (\lambda_*) \vec v_1 = \left(\mathbb T^0 (\lambda_*) + \frac{1}{2}\mathbb I \right)\vec v_1,\quad \tilde{\mathbb T}^0(\lambda_*)\vec v_2 = \left(\mathbb T^0 (\lambda_*) - \frac{1}{2}\mathbb I \right)\vec v_2.
\end{equation}
We see that
\begin{equation}\label{eq:relatiion_tilde_T0}
\begin{aligned}
&- \cK^0(\lambda_*)( v_1|_{\Gamma})(\bx) + \cS^0(\lambda_*)(\partial_n v_1|_{\Gamma})(\bx) \\
=&  - \tilde\cK^0(\lambda_*)(v_1|_{\Gamma})(\bx) - \frac{\im}{2\alpha_*}(v_1(\bx)\langle \partial_n v_1, v_1\rangle +v_2(\bx)\langle \partial_n v_2, v_1\rangle)\\
&+ \tilde\cS^0(\lambda_*)(\partial_n v_1|_{\Gamma})(\bx) +  \frac{\im}{2\alpha_*}(v_1(\bx)\overline{\langle  v_1, \partial_n v_1\rangle} +v_2(\bx)\overline{\langle v_2, \partial_n v_1\rangle}), \quad \bx\in\Gamma.
\end{aligned}
\end{equation}
Using \eqref{eq:civj}, we obtain
\begin{equation}
\begin{aligned}
&- \cK^0(\lambda_*)( v_1|_{\Gamma})(\bx) + \cS^0(\lambda_*)(\partial_n v_1|_{\Gamma})(\bx) \\
=&  - \tilde\cK^0(\lambda_*)(v_1|_{\Gamma})(\bx) + \tilde\cS^0(\lambda_*)(\partial_n v_1|_{\Gamma})(\bx) - \frac{1}{2}v_1(\bx),\quad \bx\in\Gamma.
\end{aligned}
\end{equation}
Hence, the first relation in \eqref{eq:tildeT0vi} holds, and the other relation can be shown similarly.

Next we show that
\begin{equation}\label{eq:T0vi}
(\mathbb T^0 (\lambda_*) + \frac{1}{2})\vec v_1 =0,\quad  (\mathbb T^0 (\lambda_*) - \frac{1}{2})\vec v_2=0.
\end{equation}
For each constant $A\in\mathbb R$, define $\Gamma_A := \Gamma + A \be_1$.   
Integrating by parts, we obtain when $A>0$ and $\bx$ is between $\Gamma$ and $\Gamma_A$,
\begin{equation}\label{eq:viright}
v_i(\bx) = \cD^0(\lambda_*) (v_i|_\Gamma)(\bx) - \cS^0(\lambda_*) (\partial_n v_i|_\Gamma)(\bx) + \left(-\cD_A^0(\lambda_*) (v_i|_{\Gamma_A})(\bx) + \cS_A^0(\lambda_*) (\partial_n v_i|_{\Gamma_A})(\bx)\right),
\end{equation}
and when $A<0$ and $\bx$ is between $\Gamma$ and $\Gamma_A$,
\begin{equation}\label{eq:vileft}
v_i(\bx) = - \cD^0(\lambda_*) (v_i|_\Gamma)(\bx) + \cS^0(\lambda_*) (\partial_n v_i|_\Gamma)(\bx) + \left(\cD_A^0(\lambda_*) (v_i|_{\Gamma_A})(\bx) - \cS_A^0(\lambda_*) (\partial_n v_i|_{\Gamma_A})(\bx)\right).
\end{equation}
Here $\cS_A^0(\lambda_*)$ and $\cD_A^0(\lambda_*)$ represents the single and double layer potentials with kernel 
$G^0(\bx,\by; \lambda_*)$ as defined in \eqref{eq:G0} and normal derivative in the direction $\bn$. 
Set $i=2$ in \eqref{eq:viright}. By the relation \eqref{eq:Gright}, the decay of $G^{0,+}(\bx,\by; \lambda_*)$ and \eqref{eq:civj}, the limit of \eqref{eq:viright} when $A\rightarrow+\infty$ gives
\begin{equation}
v_2(\bx) = \cD^0 (v_2|_\Gamma)(\bx) - \cS^0 (\partial_n v_2|_\Gamma)(\bx) ,\quad \text{$\bx$ is to the right of $\Gamma$}.
\end{equation}
Taking the trace and normal derivative of $v_2$ from the right of $\Gamma$, we obtain the first relation in \eqref{eq:T0vi}. The second equation can be similarly obtained by setting $i=1$ in \eqref{eq:vileft}. 
Combining \eqref{eq:tildeT0vi} and \eqref{eq:T0vi}, we obtain $X\subset \ker(\tilde{\mathbb T}^0 (\lambda_*))$.

The relation $\ker \tilde{\mathbb T}^0 (\lambda_*)\subset X$
follows from Lemma~\ref{lem:bdstate} and the fact that $-\Delta^0$ does not have eigenvalues in the unperturbed strip.

For the range of  $\tilde{\mathbb T}^0$, we observe
\begin{equation}
\text{Ran}  \; \tilde{\mathbb T}^0 (\lambda_*) 
= \left\{(\psi_1,\phi_1)\in \HhalfG\times\HmhalfG, \overline{\langle\phi_1,\psi_2\rangle}_\Gamma + \langle\phi_2,\psi_1\rangle_\Gamma =0 \text{ for all }    (\psi_2,\phi_2) \in \text{Ker}\; \left(\tilde{\mathbb T}^0 (\lambda_*)\right)^* \right\}.
\end{equation}
It is straightfoward to verify that 
\begin{equation}\label{eq:tildeT0adj}
\left(\tilde{\mathbb T}^0 (\lambda_*)\right)^* = 
\left(\begin{matrix}
    0 &-1\\1&0
\end{matrix}\right)
\tilde{\mathbb T}^0 (\lambda_*)
\left(\begin{matrix}
    0 &-1\\1&0
\end{matrix}\right)
\end{equation}
Thus 
\begin{equation}
\text{Ran}  \; \tilde{\mathbb T}^0 (\lambda_*) 
= \left\{(\psi_1,\phi_1)\in \HhalfG\times\HmhalfG, \overline{\langle\phi_1,\psi_2\rangle}_\Gamma - \langle\phi_2,\psi_1\rangle_\Gamma =0 \text{ for all }    (\psi_2,\phi_2) \in \text{Ker}\; \tilde{\mathbb T}^0 (\lambda_*) \right\}=Y.
\end{equation}
\end{proof}

\begin{lemma}\label{lem:bdstate}
Suppose $\vec\phi\in \HhalfG\times\HmhalfG$ satisfies 
\begin{equation}\label{eq:bdstate}
\tilde{\mathbb T}^0 (\lambda_*)
\vec\phi=0,\quad
\vec\phi\in Y
\text{ (or } c_i(\vec\phi)=0, i=1,2\text{)}.
\end{equation}
Define
\begin{equation}\label{eq:denstoedge0}
u(\bx)= 
\begin{cases}
  [\cD^0(\lambda_*) \psi](\bx) - [\cS^0(\lambda_*) \phi](\bx)  \quad \text{on the right of }\Gamma,\\
- [\cD^0(\lambda_*) \psi](\bx)  + [\cS^0(\lambda_*) \phi](\bx)  \quad \text{on the left of }\Gamma.
\end{cases}
\end{equation}
Then $u\in H^1(\Omega^0)$ is an eigenfunction of $-\Delta$ in $\Omega^0$ with eigenvalue $\lambda_*$.
\end{lemma}
\begin{proof}
We only need to show that $u$ and $\partial_n u$ are continuous across $\Gamma$, $u$ is nonzero, and it decays  as $|\bx\cdot\be_1|\to\infty$. 

We first verify the continuity of $u$ and its normal derivatives across $\Gamma$.
Observe that 
\begin{equation}\label{eq:T0tildeT0}
\mathbb T^0(\lambda_*)
\vec\phi =
\tilde{\mathbb T}^0(\lambda_*)\vec\phi  + \frac{\im}{2\alpha_*}c_1(\vec\phi)\vec v_1+ \frac{\im}{2\alpha_*}c_2(\vec\phi)\vec v_2 = 0.
\end{equation}
Using the relations
\begin{equation}
\left(\begin{matrix}
u|_\Gamma^+\\
\partial_n u|_\Gamma^+
\end{matrix}\right) = 
\left( \frac{1}{2}\mathbb I - \frac{1}{2} \mathbb T^0(\lambda_*)\right)
\vec\phi =
\frac{1}{2}
\vec\phi, \quad 
\left(\begin{matrix}
u|_\Gamma^-\\
\partial_n u|_\Gamma^-
\end{matrix}\right)=
\left( \frac{1}{2}\mathbb I + \frac{1}{2} \mathbb T^0(\lambda_*)\right)
\vec\phi =
\frac{1}{2}
\vec\phi, 
\end{equation}
we obtain that the jumps in  $u$ and $\partial_n u$ are both $0$ across $\Gamma$ and  $u$ is nonzero. 

Using relations \eqref{eq:G0p} and \eqref{eq:G0m},
we obtain for $\bx\cdot\be_1\to+\infty$
\begin{equation}
\cD^0(\lambda_*) \psi - \cS^0(\lambda_*)\phi  = \cD^{0,+}(\lambda_*) \psi - \cS^{0,+}(\lambda_*)\phi  + \frac{\im}{\alpha_*}c_1(\vec\phi) v_1  =  \cD^{0,+}(\lambda_*) \psi - \cS^{0,+}(\lambda_*), 
\end{equation}
and for $\bx\cdot\be_1\to-\infty$
\begin{equation}
-\cD^0(\lambda_*) \psi + \cS^0(\lambda_*)\phi   = -\cD^{0,-}(\lambda_*) \psi + \cS^{0,-}(\lambda_*)\phi  + \frac{\im}{\alpha_*}c_2(\vec\phi) v_2 =  \cD^{0,-}(\lambda_*) \psi - \cS^{0,-}(\lambda_*). 
\end{equation}
Thus $u$ decays exponentially as $|\bx\cdot\be_1|\to\infty$. 
We conclude $u\in H^1(\Omega^0)$ and the proof is complete.
\end{proof}

\begin{prop}\label{lem:sufflimprop}
The following holds for $h\in\mathbb C$ that satisfy $|h|< \hrad$:
\begin{itemize}
    \item [(i)] The operator $\mathbb U_s(h)$ defined in \eqref{eq:sufficientlim}  is analytic in $h$ and it is a Fredholm operator with index zero. 
    \item[(ii)] The only characteristic value of $\mathbb U_s(h)$ is $h=0$, and the kernel of $\mathbb U_s(0)$ is given by 
\begin{equation}
\ker \;\mathbb U_s(0) =  \text{span} \left\{ \vec v_1,\vec v_2\right\} = X.
\end{equation}
The multiplicity of the characteristic $h=0$ is $2$.
\end{itemize}
\end{prop}

\begin{proof}
(i) 
The operator ${\mathbb T}^0(\lambda_*)$
is a Fredholm operator with index zero because $S^0$ and $N^0$ are Fredholm operators with index zero \cite{mclean2000strongly,saranen2001periodic} and the operators $K^0$ and $K^{*,0}$ are compact.
Therefore, in view of the relation between ${\mathbb T}^0(\lambda_*)$ and $\tilde{\mathbb T}^0(\lambda_*) $ in \eqref{eq:T0tildeT0} and the fact that the operator $\mathbb P$ is compact, we conclude that $\mathbb U_s(0)$ is a Fredholm operator with index zero. 

(ii) Note that $\beta(0)=0$,  we see that $h=0$ is a characteristic value of $\mathbb U_s(h)$, with $\ker \;\mathbb U_s(0) = \ker \;\tilde{\mathbb T}^0(\lambda_*) = X$ by Lemma~\ref{lem:tildeT0}. 
Now we assume $h\neq0$ is a characteristics of $\mathbb U_s(h)$. Then there exists a nontrivial $\vec\phi$ that satisfies
\begin{equation}\label{eq:kerUdh}
\left( 2\tilde{\mathbb T}^0 (\lambda_*)+ 2\beta(h)\mathbb P \right)
\vec\phi=0.
\end{equation}
Since $\text{Ran}\,\mathbb P = X$, $\text{Ran}\,\tilde{\mathbb T}^0 (\lambda_*) =Y$ and $X\cap Y=\emptyset$, it follows that $\mathbb P\vec\phi=0$, which in turn implies $\vec\phi=0$. Thus $h=0$ is the only characteristic.

The multiplicity of $h=0$ is at least two, since $\ker \;\mathbb U_s(0) =  \ker \; \tilde{\mathbb T}^0 (\lambda_*) = X$ is two dimensional. We next show that every $\phi\in X$ is an eigenfunction of $\mathbb U_s(0)$ of rank $1$ (see Appendix A for the definition of rank used here).
Let $\vec\phi(h)$ be a family of functions in $\HhalfG\times\HmhalfG$ that are analytic in a neighborhood of $h=0$, and $\vec\phi:= \vec\phi(0)\in\ker \; \mathbb U_s(0)$. 
We obtain
\begin{equation}\label{eq:Tnmult}
\frac{d \big(\mathbb U_s(h) \vec\phi(h) \big)}{dh}|_{h=0}
=\frac{d \mathbb U_s(h)}{dh}|_{h=0}
\vec\phi +\mathbb U_s(0) \vec\phi'(0) =
\frac{1}{\alpha_*\beta_*}\mathbb P \vec\phi  + \mathbb U_s(0) \vec\phi'(0),
\end{equation}
where we used \eqref{eq:ziglim} in the last equality above. 
Since $\text{Ran}\mathbb P \subset X $, $\text{Ran}\mathbb U_s(0) =\text{Ran} \tilde{\mathbb T}^{0} = Y$ and $X\cap Y=\emptyset$, 
we deduce that \eqref{eq:Tnmult} is nonzero unless $\mathbb P\vec\phi=0$, which in turn implies $\vec\phi=0$. 
That is, every $\vec\phi\in\ker \; \mathbb U_s(0)$ has rank $1$, thus the multiplicity is exactly two. 
\end{proof}

\begin{prop}
Let $h\in\mathbb C$ and $|h|< \hrad$. The operators $\mathbb U_n(h)$ and $\mathbb U_t(h)$ defined in \eqref{eq:necessarylim} and \eqref{eq:nontriviallim} are analytic in $h$ and are Fredholm operators with index zero. The only characteristic value of each operator is $h=0$ with a multiplicity of $2$. In addition, the kernel of $\mathbb U_n(0)$ and $\mathbb U_t(0)$ are given by 
\begin{equation}
\ker \;\mathbb U_n(0) =  \text{span} \left\{\vec v_1 + \im \vec v_2 \right\}\quad \text{and}\quad \ker \;\mathbb U_t(0) =  \text{span} \left\{\vec v_1 - \im \vec v_2 \right\}
\end{equation}
respectively.
\end{prop}

\begin{proof}
Since the structures of $\mathbb U_n(h)$ and $\mathbb U_t(h)$ are similar, we only give the proof of the claims for $\mathbb U_n(h)$.

It is obvious that the operator $\mathbb U_n(h)$ is a Fredholm operator since $\mathbb Q$ is a finite-rank operator.
Observe that for all $h$, $\mathbb U_n(h)|_Y=\mathbb I_Y$, thus  $\ker \; \mathbb U_n(h)\subset X$.
Under the basis $\{\vec v_1, \vec v_2\}$ of $X$, using \eqref{eq:PQX},
we obtain $\vec\phi =a\vec v_1+b\vec v_2 \in \ker \; \mathbb U_n(h)$ if and only if
\begin{equation}
\left(\begin{matrix}
1 &  2\im \xi(h) \alpha_* \\
- 2\im \xi(h)& 1
\end{matrix}\right)
\left(\begin{matrix}
a\\b
\end{matrix}\right) =0.
\end{equation}
It is easy to verify that the determinant of the matrix
$1-4\alpha_*^2(\xi(h))^2$
is zero if and only if $h=0$ by \eqref{eq:betaxi}. When $h=0$, $(a,b) = (1, \im)$ spans the kernel since $\xi(0)=\frac{1}{2\alpha_*}$. 

For the multiplicity of $h=0$, we observe
\begin{equation}
\mathbb U_n(h)|_{h=0}=\mathbb I - \frac{1}{\alpha_*}\mathbb Q, \quad
\frac{d \mathbb U_n(h)}{dh}|_{h=0}= 0, \quad
\frac{d^2 \mathbb U_n(h)}{dh^2}|_{h=0}= - \frac{1}{\alpha_*\beta_*^2}\mathbb Q.
\end{equation}
Let $\vec\phi(h)$ be a family of $\HhalfG\times\HmhalfG$ operators that is analytic in a neighborhood of $h=0$, and $\vec\phi:= \vec\phi(0)\in\ker \; \mathbb U_n(0)$.
Then
\begin{equation}
\frac{d \big( \mathbb U_n(h) \vec\phi(h) \big)}{dh}|_{h=0}
=\mathbb U_n(0) \phi'(0) =0
\end{equation}
can be satisfied by the choice  $\vec\phi'(0) = \vec\phi(0) $. 
The second derivative is given by
\begin{equation}\label{eq:2ndderiv}
\frac{d^2 \big( \mathbb U_n(h) \vec\phi(h) \big)}{dh}|_{h=0}
= - \frac{1}{\alpha_*\beta_*^2}\mathbb Q \vec\phi + (\mathbb I - \frac{1}{\alpha_*}\mathbb Q) \vec\phi''(0).
\end{equation}
To make the second derivative zero, $\vec\phi''(0)$ must be in $X\supset\text{Ran}(\mathbb Q)$. Let $\vec\phi''(0)=a \vec v_1 + b\vec v_2$. Then \eqref{eq:2ndderiv} is zero if and only if 
\begin{equation}
\left(\begin{matrix}
1 &  \im \\
- \im & 1
\end{matrix}\right)
\left(\begin{matrix}
a\\b
\end{matrix}\right) =
\frac{1}{\beta_*^2}\left(\begin{matrix}
1\\\im
\end{matrix}\right) .
\end{equation}
This equation has no solution (not even trivial) because the right-hand side is not in the range of the matrix on the left-hand side. Thus $h=0$ is a characteristic of multiplicity two.
\end{proof}

\subsection{The characteristic values for the operators $\mathbb T_s^\eps$, $\mathbb T_n^\eps$ and $\mathbb T_t^\eps$ }

\begin{lemma}\label{lem:suffphi}
Let Assumption~\ref{lem:assNoFold} hold along $\bbeta_1$ and $t_*>0$. Let $\mathfrak d\in(0,1)$ be a constant. 
For sufficiently small positive $\eps$ and $|h_0|<\hrad$, 
the following holds: 
\begin{itemize}
    \item [(i)]  The operator \begin{equation}\label{eq:opinphi1s}
    2\tilde{\mathbb T}^0(\lambda_*) +P_Y\mathbb R_1(h,\eps): Y\to Y
\end{equation}
is invertible, where $P_Y$ is the projection defined in \eqref{eq:PY}, 
and $\mathbb R_1(h,\eps)$ is the remainder defined in Corollary~\ref{lem:stnlimits}.

\item[(ii)] 
Denote the inverse of the operator in \eqref{eq:opinphi1s} by $A(h,\eps): Y\to Y$.
For each $\vec\phi_0\in X$, define
\begin{equation}\label{eq:tnphi1}
J_1(h,\eps)[\vec\phi_0]= -A(h_0,\eps)P_Y \mathbb R_1(h,\eps)\vec\phi_0.
\end{equation}
If $\vec\phi\in \ker \;(\mathbb T_s^\eps (\lambda_*+\eps h_0) )$, then
\begin{equation}\label{eq:sphi}
 \vec\phi=\vec\phi_0 + J_1(h_0,\eps)[\vec\phi_0],
\end{equation}
for some $\vec\phi_0\in X$.
Moreover, 
\begin{equation}
\|J_1(h,\eps)[\vec\phi_0]\|_{\HhalfG\times\HmhalfG} =  o(1) \cdot \|\vec\phi_0\|_{\HhalfG\times\HmhalfG}
\end{equation}
uniformly for $|h|<\hrad$ as $\eps\to0^+$.
\end{itemize}
\end{lemma}
\begin{proof}
The invertibility of \eqref{eq:opinphi1s} follows from the observation that
$2\tilde{\mathbb T}^0(\lambda_*): Y\to Y$ is invertible and 
$\|\mathbb R_1(h,\eps)\|_{ \HhalfG\times\HmhalfG\to \HhalfG\times\HmhalfG}$ is of order $o(1)$ uniformly for $|h|<\hrad$ as $\eps\to0^+$.
In addition, the norm for its inverse  $\|A(h,\eps)\|_{Y\to Y} = O(1)$ uniformly for $|h|<\hrad$ as $\eps\to 0^+$.

If $\vec\phi\in \ker \;(\mathbb T_s^\eps (\lambda_*+\eps h_0) )$, by \eqref{eq:directsum} and Corollary~\ref{lem:stnlimits}, we have the decomposition
\begin{equation}
\vec \phi = \vec\phi_0 + \vec\phi_1,\quad \vec\phi_0\in X,\quad \vec\phi_1\in Y,
\end{equation}
and
\begin{equation}\label{eq:remR1}
\mathbb T_s^\eps (\lambda_*+\eps h_0) = 2\tilde{\mathbb T}^0(\lambda_*) 
+ 2\beta(h_0)\mathbb P +\mathbb R_1(h_0,\eps).
\end{equation}
As such
\begin{equation}
\mathbb T_s^\eps (\lambda_*+\eps h_0) \vec \phi =(2\tilde{\mathbb T}^0(\lambda_*) 
+ 2\beta(h_0)\mathbb P+\mathbb R_1(h_0,\eps))(\vec\phi_0 + \vec\phi_1) =0.
\end{equation}
Projecting into $Y$ by $P_Y$, we obtain 
\begin{equation}
(2\tilde{\mathbb T}^0(\lambda_*) 
+P_Y\mathbb R_1(h_0,\eps))\vec\phi_1 +P_Y \mathbb R_1(h_0,\eps)\vec\phi_0 =0.
\end{equation}
Thus $\vec\phi_1 = J_1(h_0,\eps)[\vec\phi_0]$, and the proof is complete.

\end{proof}

\begin{lemma}\label{lem:rank1}
Let Assumption~\ref{lem:assNoFold} holds along $\bbeta_1$ and $t_*>0$. Let $\mathfrak d\in(0,1)$ be a constant. 
For sufficiently small positive $\eps$, and for $|h_0|<\hrad$, every nontrivial $\vec\phi\in \ker \;(\mathbb T_s^\eps (\lambda_*+\eps h_0) )$ is of rank $1$.
\end{lemma}
\begin{proof}
The goal is to prove that for all $\vec\phi(h)\in \HhalfG\times\HmhalfG$, analytic in $h$ with $\vec\phi(h_0) = \vec\phi$, there holds
\begin{equation}
\frac{d \big( \mathbb T_s^\eps (\lambda_*+\eps h) \vec\phi(h) \big)}{dh}|_{h=h_0}
\neq0.
\end{equation}
To this end, we only need to show that there is no $\vec\psi \in \HhalfG\times\HmhalfG$, such that
\begin{equation}
\mathbb T_s^\eps (\lambda_*+\eps h_0) \vec\psi + \frac{d \big( \mathbb T_s^\eps (\lambda_*+\eps h) \big)}{dh}|_{h=h_0} \vec\phi = 0.
\end{equation}
Taking the $\HmhalfG\times\HhalfG$-$\HhalfG\times\HmhalfG$ innerproduct with $\left(\begin{matrix} 0&-1\\1&0\end{matrix}\right)\vec\phi$, we obtain
\begin{equation}
 \langle \left(\begin{matrix} 0&-1\\1&0\end{matrix}\right)\vec\phi, (2\beta'(h_0)\mathbb P+\partial_h \mathbb R_1(h_0,\eps))\vec\phi\rangle_{\HmhalfG\times\HhalfG,\HhalfG\times\HmhalfG} =0,
\end{equation}
where we have used \eqref{eq:tildeT0adj}.
Using the Cauchy integral representation of the $h$ partial derivative
\begin{equation}
\partial_h \mathbb R_1(h,\eps) = \frac{1}{2\pi\im}\int_{|z|=\frac{1+\mathfrak d}{2}}\frac{R_1(z,\eps)}{(z-h)^2}\, dz.
\end{equation}
Thus
\begin{equation}
\|\partial_h \mathbb R_1(h,\eps)\|_{ \HhalfG\times\HmhalfG\to \HhalfG\times\HmhalfG} = o(1)
\end{equation}
uniformly for $|h|<\hrad$ as $\eps\to 0^+$, 
since $\|\mathbb R_1(h,\eps)\|_{ \HhalfG\times\HmhalfG\to \HhalfG\times\HmhalfG}$ is of $o(1)$ uniformly for $|h|<\frac{1+\mathfrak d}{2}|\frac{t_*}{\gamma_*}|$ when $\eps>0$ is sufficiently small. 
Writing $\vec\phi$ in the form of \eqref{eq:sphi} with $\vec\phi_0=a\vec v_1 + b\vec v_2$ for  $a,b\in\mathbb C$, we obtain
\begin{equation}
2\beta'(h_0)\im\alpha_*^2(|a|^2 + |b|^2) + o(1)=0.
\end{equation}
Since $|\beta'(h_0)|>|\frac{1}{2\alpha_*\beta_*}|$ for $|h|<\hrad$, the above equation never holds when $\eps>0$ is sufficiently small unless $a=b=0$. By  \eqref{eq:sphi}, $\vec\phi_0=0$, which contradicts that $\vec\phi$ is nontrivial. The proof is complete.
\end{proof}

To study the operators $\mathbb T_t$ and $\mathbb T_n$, we decompose the space $X$ further using the basis 
\begin{equation}
\mathfrak u_1:= \vec v_1+ \im\vec v_2,\quad 
\mathfrak u_2:= \vec v_1- \im\vec v_2.
\end{equation}
Define
\begin{equation}
X_i:=\text{span}\{\mathfrak u_i\},\quad i=1,2.
\end{equation}
Then
\begin{equation}\label{eq:XXY}
X=X_1\bigoplus X_2,\quad \HhalfG\times\HmhalfG = X_1\bigoplus X_2\bigoplus Y.
\end{equation}
It follows that
\begin{equation}\label{eq:PQu}
\mathbb P \mathfrak u_1= \im \alpha_* \mathfrak u_2, \quad 
\mathbb P \mathfrak u_2= \im \alpha_* \mathfrak u_1, \quad 
\mathbb Q \mathfrak u_1= \alpha_*\mathfrak u_1, \quad 
\mathbb Q \mathfrak u_2= -\alpha_*\mathfrak u_2,\quad \mathbb P Y= \mathbb Q Y=0.
\end{equation}
Similar to Lemma~\ref{lem:suffphi}, we have the following characterization of the root functions of $\mathbb T_t^\eps (\lambda_*+\eps h_0)$ and $\mathbb T_n^\eps (\lambda_*+\eps h_0)$. We will only give the proof of Lemma~\ref{lem:nphi} as that of Lemma~\ref{lem:tphi} is the same.

\begin{lemma}\label{lem:nphi}
Let Assumption~\ref{lem:assNoFold} holds along $\bbeta_1$ and $t_*>0$. Let $\mathfrak d\in(0,1)$ be a constant. 
The following holds for sufficiently small $\eps>0$ and $|h|<\hrad$,
\begin{itemize}
    \item [(i)] The operator 
    \begin{equation}\label{eq:invop}
    \mathbb I - \xi(h)P_{X_2\bigoplus Y} \mathbb Q + P_{X_2\bigoplus Y} \mathbb R_3(h,\eps) :  X_2\bigoplus Y\to X_2\bigoplus Y
\end{equation} 
is invertible, where
$\mathbb R_3(h,\eps)$ is the remainder defined in Corollary~\ref{lem:stnlimits}.
Denote the inverse by $C(h,\eps): X_2\bigoplus Y\to X_2\bigoplus Y$.

\item[(ii)] Define 
\begin{equation}\label{eq:nphi1}
 J_2(h,\eps)[\mathfrak u_1]= -C(h,\eps)P_{X_2\bigoplus Y} \mathbb R_3(h,\eps)\mathfrak u_1.
\end{equation}
Then $ J_2(h,\eps)[\mathfrak u_1]$ is analytic in $h$ and
\begin{equation}\label{eq:phi1small}
\| J_2(h,\eps)[\mathfrak u_1]\|_{\HhalfG\times\HmhalfG} =  o(1)\|\mathfrak u_1\|_{\HhalfG\times\HmhalfG},\quad \text{uniformly  for $|h|<\hrad$ as $\eps\to0^+$}.
\end{equation}
Moreover,
if $\vec\phi\in \ker \;(\mathbb T_n^\eps (\lambda_*+\eps h_0) )$, then up to a constant factor
\begin{equation}\label{eq:nphi}
 \vec\phi= \mathfrak u_1 + J_2(h_0,\eps)[\mathfrak u_1].
\end{equation}

\end{itemize}
\end{lemma}
\begin{proof}[Proof of Lemma~\ref{lem:nphi}]
The invertibility of \eqref{eq:invop} follows from $P_{X_2\bigoplus Y} \mathbb Q|_{X_2}  = -\alpha_*$, $\xi(h)>0$, $P_{X_2\bigoplus Y} \mathbb Q|_Y  = 0$ and the uniform smallness of $\mathbb R_3(h,\eps)$. 
The analyticity of  $J_2[\vec\phi_0](h,\eps)$ in $h$ and the smallness \eqref{eq:phi1small} follow straightforwardly. 

For the statement \eqref{eq:nphi}, by \eqref{eq:XXY} and Corollary~\ref{lem:stnlimits}, we have 
\begin{equation}
\vec\phi = \mathfrak u_1 + \vec\phi_1(h,\eps), \quad\text{for some } \vec\phi_1(h,\eps)\in X_2\bigoplus Y,
\end{equation}
and
\begin{equation}
\mathbb T_n^\eps (\lambda_*+\eps h_0)\vec\phi(h_0,\eps) = \big(\mathbb I - \xi(h_0)\mathbb Q+ \mathbb R_3(h_0,\eps)\big) (\mathfrak u_1 + \vec\phi_1(h_0,\eps)) =0.
\end{equation}
Projecting into $X_2\bigoplus Y$ using $P_{X_2\bigoplus Y}$, we have 
\begin{equation}
(\mathbb I - \xi(h_0)P_{X_2\bigoplus Y}\mathbb Q + P_{X_2\bigoplus Y} \mathbb R_3(h_0,\eps))  \vec\phi_1(h_0,\eps) + P_{X_2\bigoplus Y} \mathbb R_3(h_0,\eps) \mathfrak u_1 =0.
\end{equation}
By the invertibility of \eqref{eq:invop}, $\vec\phi_1 = \| J_2[\vec\phi_0](h,\eps)$. The proof is concluded.
\end{proof}
\begin{lemma}\label{lem:tphi}
Let Assumption~\ref{lem:assNoFold} holds along $\bbeta_1$ and $t_*>0$. Let $\mathfrak d\in(0,1)$ be a constant.
Then the following holds for sufficiently small $\eps>0$ and $|h|<\hrad$:
\begin{itemize}
    \item [(i)] The operator 
\begin{equation*}
\mathbb I + \xi(h)P_{X_1\bigoplus Y} \mathbb Q + P_{X_1\bigoplus Y} \mathbb R_2(h,\eps) :  X_1\bigoplus Y\to X_1\bigoplus Y.
\end{equation*} 
is invertible, where
$\mathbb R_2(h,\eps)$ is the remainder defined in Corollary~\ref{lem:stnlimits}. Denote the inverse by $B(h,\eps): X_1\bigoplus Y\to X_1\bigoplus Y$.

\item[(ii)] 
Let  
\begin{equation}\label{eq:tphi1}
 J_3(h,\eps)[\mathfrak u_2]= -B(h,\eps)P_{X_1\bigoplus Y} \mathbb R_2(h,\eps)\mathfrak u_2.
\end{equation}
Then $\vec\phi_1(h,\eps)$ is analytic in $h$ and
\begin{equation}
\|J_3(h,\eps)[\mathfrak u_2]\|_{\HhalfG\times\HmhalfG} =  o(1)\|\mathfrak u_2\|_{\HhalfG\times\HmhalfG},\quad \text{uniformly for $|h|<\hrad$ as $\eps\to0^+$}.
\end{equation}
Furthermore,
if $\vec\phi\in \ker \;(\mathbb T_t^\eps (\lambda_*+\eps h_0) )$, then
\begin{equation}\label{eq:tphi}
\vec\phi= \mathfrak u_2 +J_3(h_0,\eps)[\mathfrak u_2].
\end{equation}
\end{itemize}
\end{lemma}

\begin{prop}\label{lem:uniquestn}
Let Assumption~\ref{lem:assNoFold} holds along $\bbeta_1$ and $t_*>0$. Let $\mathfrak d\in(0,1)$ be a constant.
For sufficiently small $\eps>0$, the system
\begin{equation}\label{eq:simulsn}
\mathbb T_s^\eps (\lambda_*+\eps h)\vec\phi=0\quad\text{and} \quad\mathbb T_n^\eps (\lambda_*+\eps h)\vec\phi=0
\end{equation}
attains at most one pair of solution $(h,\vec\phi)$, with $|h|<\hrad$ and $\vec\phi\in\HhalfG\times\HmhalfG$.
The same holds for the system
\begin{equation}\label{eq:simulst}
\mathbb T_s^\eps (\lambda_*+\eps h)\vec\phi=0 \quad\text{and}\quad \mathbb T_t^\eps (\lambda_*+\eps h)\vec\phi=0.
\end{equation}
\end{prop}

\begin{proof}
We only display the proof on the system \eqref{eq:simulsn}.
Suppose $\vec\phi$ solves both equations in \eqref{eq:simulsn}. By Lemma~\ref{lem:nphi}, the solution to the second equation necessarily takes the form
 $\vec\phi= \mathfrak u_1 +J_2(h,\eps)[\mathfrak u_1]$ where $J_2(h,\eps)[\mathfrak u_1]$ is defined in \eqref{eq:nphi1}.
 Substituting $\vec\phi$ into the first equation in \eqref{eq:simulsn}, we obtain
\begin{equation}
(2\tilde{\mathbb T}^0(\lambda_*) + 2\beta(h)\mathbb P + \mathbb R_1(h,\eps))(\mathfrak u_1 + J_2(h,\eps)[\mathfrak u_1])= 0.
\end{equation}
Projecting into $X_2$ using $P_{X_2}$, we obtain
\begin{equation}\label{eq:uniq1D}
2\im\alpha_* \beta(h) \mathfrak u_2  + [P_{X_2}\mathbb R_1(h,\eps)\mathfrak u_1 + P_{X_2}\Big(2\beta(h)\mathbb P + \mathbb R_1(h,\eps)\Big)J_2(h,\eps)[\mathfrak u_1]]= 0.
\end{equation}

Since $X_2$ is a one-dimensional space, \eqref{eq:uniq1D} becomes
\begin{equation}\label{eq:uniqscalar}
2\im\alpha_* \beta(h) + r(h,\eps) = 0.
\end{equation}
Note that for each $\eps$,  $r(h,\eps)$ is analytic in $h$, and $|r(h,\eps)|\to 0$ as $\eps\to 0^+$ uniformly in $|h|<\hrad$. Thus the Rouch\'e Theorem for single-variable complex functions implies that there is exactly one $h_0$ that solves \eqref{eq:uniqscalar} in the region $|h|<\hrad$. 
There is also at most one root function at $h_0$, since $\vec\phi$ is determined by \eqref{eq:nphi} when $h=h_0$.
\end{proof}

\subsection{Proof of Theorem~\ref{lem:edge}}
%
\begin{proof}[Proof of Theorem~\ref{lem:edge}] 

We claim that $\mathbb T_s^\eps (\lambda_*+\eps h)$ is of multiplicity $2$ in $|h|<\hrad$ when $\eps>0$ is sufficiently small.
In Theorem~\ref{lem:actGS}, we identify $z=h$, $X=Y=\HhalfG\times\HmhalfG$, 
$V=\{h,|h|<\hrad\}$, 
$A(z)=\mathbb U_s (h)$ 
and $B(z)=\mathbb R_1(h,\eps)$.
In Proposition~\ref{lem:sufflimprop}, we have shown that $A(z)$ is analytic and Fredholm of index zero on a neighborhood of $\overline V$, and the multiplicity of $A(z)$ in $V$ is $2$. 
On $\partial V$, we know $A(z)$ is invertible and is independent of $\eps$, and $B(z)$ converges uniformly to $0$ as $\eps\to0^+$
by Corollary~\ref{lem:stnlimits}. 
Since $B(z)$ is analytic on a neighborhood of $\overline V$, by Theorem~\ref{lem:actGS} and the relation $\mathbb T_s^\eps (\lambda_*+\eps h)=A(z)+B(z)$, we conclude the claim.

By Lemma~\ref{lem:rank1}, 
there are two pairs $(h_i,\vec\phi_i)$ solving $\mathbb T_s^\eps (\lambda_*+\eps h)\vec\phi=0$, $i=1,2$. 
We first show that there is at least one interface mode.
Assume $u_i$ generated by $\vec\phi_i$ through the expression \eqref{eq:denstoedge} are both equal to zero. Then $(h_i,\vec\phi_i)$, $i=1,2$ are both solutions to the system \eqref{eq:simulst}. 
This contradicts the uniqueness of simultaneous solutions to \eqref{eq:simulst} established in Proposition~\ref{lem:uniquestn}.
Next, we show that there is at most one interface mode. Suppose there are two linearly independent interface modes $u_i$ at $h_i$ respectively, $i=1,2$. Denote  $\vec\phi_i = (u_i|_\Gamma,\partial_n u_i|_\Gamma)$. Then $(h_i,\vec\phi_i)$, $i=1,2$ are the solutions to the system \eqref{eq:simulsn}. 
This contradicts the uniqueness of solutions to the system \eqref{eq:simulsn} established in Proposition~\ref{lem:uniquestn}.

\end{proof}

\section{The interface modes along an armchair interface}\label{sec:geometrya}

In this section, we study interface modes along the armchair interface as stated in the eigenvalue problem \eqref{eq:edgedef_armchair} and prove Theorem \ref{lem:edgearm}. To this end, we extend in parallel the mathematical framework developed in the previous sections for the zigzag interface.

For ease of notation, we will abuse notations $v_i$, $\vec v_i$, $c_i$, $\mathbb P$, $\mathbb Q$, $X$, $\mathfrak u_i$, $X_i$ that are introduced in Sections~\ref{sec:bands}, ~\ref{sec:prep}, \ref{sec:characterization}, and \ref{sec:proofedge} to represent the quantities relevant to the zigzag interface. In this section, these notations represent the quantities relevant to the armchair interface.

For the eigenvalue problem \eqref{eq:edgedef_armchair}, the Floquet theory on the strip $\Omega_a^J = \cup_{n_1\in\mathbb Z}(\cC_a +n_1\be_1^a)$ corresponds to the slice of quasimomenta $\bp(\ell):= \Kone + \ell\bbeta_1^a$. This slice intersects with both $\Kone+\Lambda$ and $\Ktwo+\Lambda$, as $\bp(0)=\Kone$, and $\bp(-\frac{2\pi}{3}) = \Ktwo+\bbeta_1 - \bbeta_2$. We have the following proposition whose proof is given in Appendix~\ref{sec:Tderiv}.

\begin{prop}\label{lem:TderivA}
Let $\brho := (\rho_1,\rho_2)$, where $\rho_i$ are the functions defined in Proposition~\ref{lem:DiracP} with the normalization such that \eqref{eq:w_normalization} holds. 
Let $t_*,\gamma_*\in\mathbb R$ and $\theta_*\in\mathbb C$ be the constants defined in Proposition~\ref{lem:Tderiv}.
There holds
\begin{equation}\label{eq:TderivA}
\langle \brho,  \bbeta_1^a\cdot\nabla_{\bp}T(0,\lambda_*,\bp)\brho\rangle_{\partial D}|_{\bp=\Kone}= 
\left(\begin{matrix}
0&\sqrt3\im\,\overline{\tau}\overline{\theta_*}\\
-\sqrt3\im\tau\theta_*&0
\end{matrix}\right).
\end{equation}
Let $\brho':=(\rho_1',\rho_2')$,  where 
\begin{equation}
\rho_1'(\bx) := \bar \rho_2(\bx),\quad \rho_2'(\bx):=\bar \rho_1(\bx).
\end{equation}
In particular,
\begin{equation}
R\rho_1'(\bx):=\rho_1'(R^{-1}\bx) = \tau \rho_1'(\bx),\quad R\rho_2'(x):=\rho_2'(R^{-1}\bx) = \overline{\tau},\quad \rho_2'(\bx)=\rho_1'(\rflc\bx).
\end{equation}
The partial derivatives of $\langle \brho',T (\eps,\lambda,\bp)\brho'\rangle_{\partial D}$ at $\eps=0$, $\lambda=\lambda_*$ and $\bp=\Ktwo$ take the forms
\begin{equation}\label{eq:TderivKtwo}
\begin{aligned}
\langle \brho', \partial_{\lambda}T (0,\lambda,\Ktwo)\brho'\rangle_{\partial D}|_{\lambda=\lambda_*}&= 
\left(\begin{matrix}
\gamma_*&0\\
0&\gamma_*
\end{matrix}\right), \\
\langle \brho',  \bbeta_1\cdot\nabla_{\bp}T(0,\lambda_*,\bp)\brho'\rangle_{\partial D}|_{\bp=\Ktwo}&= 
\left(\begin{matrix}
0&-\overline{\theta_*}\\
-\theta_*&0
\end{matrix}\right),\\
\langle \brho', \partial_{\eps}T (\eps,\lambda_*,\Ktwo)\brho'\rangle_{\partial D}|_{\eps=0}&= 
\left(\begin{matrix}
-t_*&0\\
0&t_*
\end{matrix}\right). 
\end{aligned}
\end{equation}
In particular,
\begin{equation}\label{eq:TderivAKtwo}
\langle \brho',  \bbeta_1^a\cdot\nabla_{\bp}T(0,\lambda_*,\bp)\brho'\rangle_{\partial D}|_{\bp=\Ktwo}= 
\left(\begin{matrix}
0&-\sqrt3\im\,\overline{\tau}\overline{\theta_*}\\
\sqrt3\im\tau\theta_*&0
\end{matrix}\right).
\end{equation}
Here $\langle\cdot,\cdot\rangle_{\partial D}$ represents the $(H^{-1/2}(\partial D))^2$-$(H^{1/2}(\partial D))^2$ pairing. 
\end{prop}

Define 
\begin{equation}\label{eq:newalpha}
\alpha_*^a := |\frac{\sqrt3\im\,\overline{\tau}\overline{\theta_*}}{\gamma_*}| = \sqrt{3}\alpha_*.
\end{equation} 

\subsection{Band structure of the periodic strip $\Omega^{\eps}_a$ near the Dirac points}
We start with the following strip region 
\begin{equation}
\Omega^{\eps}_a := \Omega^J_a\backslash \cup_{m\geq0} (D^{\eps}+m\be_1^a).
\end{equation}
When $\eps=0$, $\Omega^0_a$ represents the region when the rotation angle $\eps=0$.
For $\eps\in\mathbb R$, define
\begin{equation}\label{eq:QP1pmA}
\begin{aligned}
\cH^{\eps,a}_{\text{loc}}:=\{&u\in H^1_{\text{loc}}(\Omega^{\eps}_a): \Delta u\in L^2_{\text{loc}}(\Omega^{\eps}_a), \quad u=0 \text{ on } \cup_{m\in\mathbb Z}(\partial D^{\eps} +m\be_1^a),\\
&u(\bx+ \be_2)= e^{\im \kp^*}u(\bx) \text{ for }\bx\in\Gamma_-^a, \quad \partial_{\nuGb}  u(\bx+ \be_2) =e^{\im \kp^*}\partial_{\nuGb}  u(\bx) \text{ for }\bx\in\Gamma_-^a\}.
\end{aligned}
\end{equation}
We consider the dispersion relation of the operator $\Delta$ on $\cH^{\eps,a}_{\text{loc}}$, that is, we find  $(\lambda,u)\in \mathbb R\times \cH^{\eps,a}_{\text{loc}}$, such that
\begin{equation}
\begin{aligned}
-\Delta u - \lambda u &= 0 \quad&\text{on } \Omega^{\eps}_a,\\
u &= 0 \quad&\text{on } \cup_{m\in\mathbb Z} (\partial D^{\eps}+m\be_1^a).
\end{aligned}
\end{equation}

At $\ell=0$ and $\bp(0) = \Kone$, comparing  \eqref{eq:TderivA} and the second equation in \eqref{eq:Tderiv}, $-\sqrt{3}\im\tau\theta_*$ is in the place of $\theta_*$. At $\ell=-\frac{2\pi}{3}$ and $\bp(-\frac{2\pi}{3}) = \Ktwo$, comparing \eqref{eq:TderivAKtwo} the second equation in \eqref{eq:Tderiv}, $\sqrt{3}\im\tau\theta_*$ is in the place of $\theta_*$.
Thus we obtain the following remark and lemma parallel to Remark~\ref{lem:commongap}
and Lemma~\ref{lem:0strip}, respectively.

\begin{remark}\label{lem:commongapA}
If the Assumption~\ref{lem:assNoFold} holds along $\bbeta_1^a$, then for an arbitrary fixed constant $\mathfrak d\in(0,1)$, when $\eps>0$ is sufficiently small, the operator $\Delta$ on $\Omega^{\pm\eps}_a$ have a common band gap $(\lambda_* - \sqrt3\hrad\eps, \lambda_* + \sqrt3\hrad\eps)$.
\end{remark}

\begin{lemma}\label{lem:0stripA}
Let $\bp(\ell):= \Kone + \bbeta_1^a$.
For $|\ell| \ll 1 $,
\begin{equation}\label{eq:DiracbandA}
\begin{aligned}
\mu_1(\bp(\ell)) &= \lambda_*  + \sqrt3|\frac{\theta_*}{\gamma_*}|\ell(1+O(\ell)) \quad\text{(increasing in $\ell$)}, \\
\mu_2(\bp(\ell)) &= \lambda_*  - \sqrt3|\frac{ \theta_*}{\gamma_*}|\ell(1+O(\ell)) \quad\text{(decreasing in $\ell$)}. 
\end{aligned}
\end{equation}
The corresponding Bloch modes can be chosen as
\begin{equation}
\begin{aligned}
v_1(\bx;\bp(\ell))&=\left( \frac{-\im \, \overline{\tau\theta_*}}{|\theta_*|}w_1 - w_2  +O(\ell) \right)\frac{1}{\sqrt{2+O(\ell)}},\\
v_2(\bx;\bp(\ell))&=\left(  \frac{-\im \, \overline{\tau\theta_*}}{|\theta_*|}w_1 +  w_2  +O(\ell) \right)\frac{1}{\sqrt{2+O(\ell)}}.
\end{aligned}
\end{equation}
For $|\ell + \frac{2\pi}{3}| \ll 1$,
\begin{equation}\label{eq:DiracbandA2}
\begin{aligned}
\mu_1(\bp(\ell)) &= \lambda_*  + \sqrt3|\frac{\theta_*}{\gamma_*}|(\ell + \frac{2\pi}{3})(1+O(\ell + \frac{2\pi}{3})) \quad\text{(increasing in $\ell$)}, \\
\mu_2(\bp(\ell)) &= \lambda_*  - \sqrt3|\frac{ \theta_*}{\gamma_*}|(\ell + \frac{2\pi}{3}))(1+O(\ell + \frac{2\pi}{3})) \quad\text{(decreasing in $\ell$)}. 
\end{aligned}
\end{equation}
The corresponding Bloch modes can be chosen as
\begin{equation}\label{eq:armKtwov}
\begin{aligned}
\mathfrak b_1 v_1(\bx;\bp(\ell))&=\left( \frac{\im\,\overline{\tau\theta_*}}{|\theta_*|}w_1 - w_2  +O(\ell) \right)\frac{1}{\sqrt{2+O(\ell + \frac{2\pi}{3})}},\\
\mathfrak b_2 
v_2(\bx;\bp(\ell))&=\left(  \frac{\im\,\overline{\tau\theta_*}}{|\theta_*|}w_1 +  w_2  +O(\ell) \right)\frac{1}{\sqrt{2+O(\ell + \frac{2\pi}{3})}}.
\end{aligned}
\end{equation}
Here $\mathfrak b_i$, $i=1,2$, are two phase factors.
\end{lemma}
In the above, the phase factors show up in  \eqref{eq:armKtwov} because $v_n(\bp(\ell))$ needs to be smooth for  $\ell\in(-\pi,\pi)$ as explained in Section~\ref{sec:Floquet}.
Define $v_i := v_i(\bx;\bp(0))=v_i(\bx;\Kone)$, and $v_i' :=\mathfrak b_i v_i(\bx;\bp(-\frac{2\pi}{3}))= \mathfrak b_i v_i(\bx;\Ktwo)$. We obtain the relations
\begin{equation}
\begin{cases}
w_1 = \frac{1}{\sqrt2}\frac{-\im\tau\theta_*}{|\theta_*|} (v_1  +  v_2)\\
w_2 = \frac{1}{\sqrt2} (- v_1 + v_2)
\end{cases}, \quad
\begin{cases}
w_1' = \frac{1}{\sqrt2}\frac{\im\tau\theta_*}{|\theta_*|} (v_1'  +  v_2')\\
w_2' = \frac{1}{\sqrt2} (- v_1' + v_2')
\end{cases}.
\end{equation}


\subsection{Green's function in the infinite strip and the integral equation formulation}
Compared to the Green function \eqref{eq:G0} for the infinite strip $\Omega^0$ considered in Section \ref{sec:prep},
the Green function in the strip $\Omega_a^0$  contains four propagating modes: $v_{i}(\bx;\Kone)$ and $v_{i}(\bx;\Ktwo)$ ($i=1,2$), due to the fact that $\bp(\ell):= \Kone + \bbeta_1^a\ell$ intersects with both $\Kone+\Lambda$ and $\Ktwo+\Lambda$. As a result, it attains the following spectral representation:
\begin{equation}
\begin{aligned}\label{eq:G0A}
G^{0,a}(\bx,\by; \lambda_*) = &
\sum_{n\geq3}\frac{1}{2\pi}\int_{[-\pi,\pi ]} 
\frac{\overline{v_{n}(\by;\bp(\ell))} v_{n}(\bx;\bp(\ell))}{ \mu_{n}(\bp(\ell)) - \lambda_* } \, \dpt +\sum_{n=1,2}\frac{1}{2\pi}\text{p.v.}\int_{[-\pi,\pi ]} 
\frac{\overline{v_{n}(\by;\bp(\ell))} v_{n}(\bx;\bp(\ell))}{ \mu_{n}(\bp(\ell)) - \lambda_* } \, \dpt \\
&+\sum_{n=1,2}\frac{\im}{2\alpha_*^a}\overline{v_{i}(\by;\Kone)} v_{i}(\bx;\Kone) +\sum_{n=1,2}\frac{\im}{2\alpha_*^a}\overline{v_{i}(\by;\Ktwo)} v_{i}(\bx;\Ktwo) ,\quad \bx,\by\in\Omega^0,
\end{aligned}
\end{equation}
where its integral part takes the form
\begin{equation}
\begin{aligned}\label{eq:tildeG0A}
\tilde G^{0,a}(\bx,\by; \lambda_*) = 
\sum_{n\geq3}\frac{1}{2\pi}\int_{[-\pi,\pi ]} 
\frac{\overline{v_{n}(\by;\bp(\ell))} v_{n}(\bx;\bp(\ell))}{ \mu_{n}(\bp(\ell)) - \lambda_* } \, \dpt +\sum_{n=1,2}\frac{1}{2\pi}\text{p.v.}\int_{[-\pi,\pi ]} 
\frac{\overline{v_{n}(\by;\bp(\ell))} v_{n}(\bx;\bp(\ell))}{ \mu_{n}(\bp(\ell)) - \lambda_* } \, \dpt.
\end{aligned}
\end{equation}
The Green function in $\Omega^{\pm\eps}_a$ attains the following spectral representation:
\begin{equation}\label{eq:pGreenpmepsA}
G^{\pm\eps,a}(\bx,\by; \lambda) = 
 \sum_{n\geq1}\frac{1}{2\pi}\int_{[-\pi,\pi ]}
\frac{\overline{v_{n,\pm\eps}(\by;\bp(\ell))} v_{n,\pm\eps}(\bx;\bp(\ell))}{ \mu_{n,\pm\eps}(\bp(\ell)) - \lambda } \, \dpt ,\quad \bx,\by\in\Omega^0.
\end{equation}

We now investigate the interface modes along the armchair interface.  For $s\in\mathbb R$, we define the following quasi-periodic Sobolev space on $\Gamma^a$:
\begin{equation}\label{eq:HsG1a}
\cH^{s,a}(\Gamma^a):= \left\{ u(\bx_0+ t\be_2^a) = \sum_{n\in\mathbb Z} a_n e^{\im \Kone\cdot\be_2^a t} e^{\im 2\pi n t}: \|u\|_{\cH^{s,a}(\Gamma^a)}^2:=\sum_{n\in\mathbb Z} |a_n|^2 (1+n^2)^s\right\}.
\end{equation}
Here $x_0 = -\frac{1}{2}\be_1^a-\frac{1}{2}\be_2^a$.

We define the layer potentials 
$\cS^{\pm\eps,a}(\lambda)$, $\cD^{\pm\eps,a}(\lambda)$, $\cK^{\pm\eps,a}(\lambda)$, $\cK^{\pm\eps,a}(\lambda)$ and $\cN^{\pm\eps,a}(\lambda)$ parallel to 
\eqref{eq:lpotentialeps},  
 where the Green functions are replaced by $G^{\pm\eps,a}(\bx,\by,\lambda)$, the integral region is replaced by $\Gamma^a$, the integral operators on $\HhalfGa\times\HmhalfGa$ are defined by
\begin{equation}\label{eq:bTaeps}
\mathbb T^{\eps,a}
(\lambda):=
\left(\begin{matrix}
-\cK^{\eps,a}(\lambda) & \cS^{\eps,a}(\lambda) \\
 -\cN^{\eps,a}(\lambda) & \cK^{*,\eps,a}(\lambda)
\end{matrix}\right),
\end{equation}
and
\begin{equation}\label{eq:bTas}
\mathbb T_s^{\eps,a}(\lambda):=\mathbb T^{\eps,a} + \mathbb T^{-\eps,a}, \quad
\mathbb T_t^{\eps,a}(\lambda):=- \mathbb T^{\eps,a} + \mathbb T^{-\eps,a} + \mathbb I, \quad
\mathbb T_n^{\eps,a}(\lambda):=\mathbb T^{\eps,a} - \mathbb T^{-\eps,a} + \mathbb I.
\end{equation}
We characterize the interface modes by using boundary integral operators as follows. 
\begin{lemma}\label{lem:edgestate_armchair}
Let $\lambda \in (\lambda_* - \sqrt3\hrad\eps, \lambda_* + \sqrt3\hrad\eps)$.
\begin{itemize}
    \item [(i)] There exists an interface mode $u$ satisfying \eqref{eq:edgedef_armchair} if and only if there exists $(\psi,\phi)\in \HhalfGa\times\HmhalfGa$ such that
\begin{equation}\label{eq:sufficient_armchair}
\mathbb T_s^{\eps,a}
(\lambda)
\left(\begin{matrix}
\psi\\
\phi
\end{matrix}\right)
=0,\quad  \mathbb T_t^{\eps,a}
(\lambda)
\left(\begin{matrix}
\psi\\
\phi
\end{matrix}\right)
\neq0.
\end{equation}
\item [(ii)] If $u$ is an interface mode satisfying \eqref{eq:edgedef_armchair}, then $0\neq( u|_{\Gamma^a}, \partial_n u|_{\Gamma^a}) \in \HhalfGa\times \HmhalfGa$ satisfies
\begin{equation}\label{eq:necessary_armchair}
\mathbb T_s^{\eps,a}
(\lambda)
\left(\begin{matrix}
u|_{\Gamma^a}\\
\partial_n u|_{\Gamma^a}
\end{matrix}\right)
=0,\quad
\mathbb T_n^{\eps,a}
(\lambda)
\left(\begin{matrix}
u|_{\Gamma^a}\\
\partial_n u|_{\Gamma^a}
\end{matrix}\right)
=0.
\end{equation}
\end{itemize}
\end{lemma}

\subsection{Limiting operators}

Define
\begin{equation}
\begin{aligned}
c_i(\vec\phi)&= \overline{\langle \phi,v_i\rangle}_{\Gamma^a} - \langle\partial_n v_i, \psi\rangle_{\Gamma^a},\quad
c_i'(\vec\phi)&= \overline{\langle \phi,v_i'\rangle}_{\Gamma^a} - \langle\partial_n v_i', \psi\rangle_{\Gamma^a}, \quad i=1,2.
\end{aligned}
\end{equation}
Here $\langle\cdot,\cdot\rangle_{\Gamma^a}$ represents the $\cH^{-1/2,a}(\Gamma^a)$-$\cH^{1/2,a}(\Gamma^a)$ pairing on the armchair interface $\Gamma^a$.  
Define  $\vec v_i:= (v_i|_{\Gamma^a}, \partial_n v_i|_{\Gamma^a})$ and $\vec v_i':= (v_i'|_{\Gamma^a}, \partial_n v_i'|_{\Gamma^a})$. Similar to the argument for \eqref{eq:civj}, we have
\begin{equation}\label{eq:civjA}
\begin{aligned}
&c_1(\vec v_1) =\im \alpha_*^a, \quad c_2(\vec v_2) = -\im \alpha_*^a,  \quad c_i(\vec v_j) = 0 &\text{ for }i\neq j,\\
&c_1'(\vec v_1') =\im \alpha_*^a, \quad c_2'(\vec v_2') = -\im \alpha_*^a, \quad c_i'(\vec v_j') = 0 &\text{ for }i\neq j,  \\
&c_i(\vec v_j')= c_i'(\vec v_j)=0 &\text{ for }i,j=1,2.
\end{aligned}
\end{equation}
Define
\begin{equation}\label{eq:Pp}
\mathbb P\vec\phi : =
c_1(\vec\phi)\vec v_1+c_2(\vec\phi)\vec v_2, \quad 
\mathbb P'\vec\phi : =
c_1'(\vec\phi)\vec v_1'+c_2'(\vec\phi)\vec v_2',
\end{equation}
and
\begin{equation}\label{eq:Qp}
\mathbb Q
\vec\phi :=
c_2(\vec\phi)\vec v_1
+c_1(\vec\phi)\vec v_2,\quad
\mathbb Q'
\vec\phi :=
c_2'(\vec\phi)\vec v_1'
+c_1'(\vec\phi)\vec v_2'.
\end{equation}
Let
\begin{equation}\label{eq:betaxi_a}
    \beta^a(h): = \frac{1}{2\alpha_*^a}\frac{h}{\sqrt{\beta_*^2-h^2}}, \quad
    \xi^a(h):= \frac{\beta_*}{2\alpha_*^a}\frac{1}{\sqrt{\beta_*^2-h^2}},
\end{equation}
where $\alpha_*^a$ is defined in \eqref{eq:newalpha}.

\begin{prop} \label{lem:oplimA}
Let Assumption~\ref{lem:assNoFold} hold along $\bbeta_1^a$ and $t_*>0$. Let $\mathfrak d\in(0,1)$ be a constant.
The following limit holds uniformly for $h\in\mathbb C$ satisfying $|h|< \hrad$ as $\eps\to 0^+$ under the operator norm on $\HhalfGa\times\HmhalfGa$: 
\begin{equation}\label{eq:oplimarm}
\mathbb T^{\pm\eps,a}(\lambda_*+ \eps h) 
\to \mathbb U^{\pm}(h),
\end{equation}
where
\begin{equation}\label{eq:Upm}
\mathbb U^{\pm,a}(h):=
\tilde{\mathbb T}^{0,a} (\lambda_*) 
+ \beta^a(h)(\mathbb P+\mathbb P') \mp \xi^a(h)(\mathbb Q - \mathbb Q'),
\end{equation}
and 
\begin{equation}
\tilde{\mathbb T}^{0,a}
(\lambda):=
\left(\begin{matrix}
-\tilde\cK^{0,a}(\lambda) & \tilde\cS^{0,a}(\lambda) \\
 -\tilde\cN^{0,a}(\lambda) & \tilde\cK^{*,0,a}(\lambda)
\end{matrix}\right).
\end{equation}
Here the layer potentials 
$\tilde\cS^{0,a}(\lambda)$, $\tilde\cD^{0,a}(\lambda)$, $\tilde\cK^{0,a}(\lambda)$, $\tilde\cK^{0,a}(\lambda)$ and $\tilde\cN^{0,a}(\lambda)$ are defined parallel to 
\eqref{eq:lpotentialeps},  
 where the Green functions are replaced by $\tilde G^{0,a}(\bx,\by,\lambda)$ as defined in \eqref{eq:tildeG0A}, and the integral region is replaced by $\Gamma^a$.
\end{prop}

\begin{coro}\label{lem:stnlimitsarm}
Let Assumption~\ref{lem:assNoFold} hold along $\bbeta_1^a$ and $t_*>0$. Let $\mathfrak d\in(0,1)$ be a constant.
We have
\begin{equation}\label{eq:sufficientlimarm}
\mathbb T_s^{\eps,a} (\lambda_*+\eps h)
 =\mathbb U^a_s (h)+\mathbb R_1^a(h,\eps),
\end{equation}
\begin{equation}\label{eq:nontriviallimarm}
\mathbb T_t^{\eps,a} (\lambda_*+\eps h)
=\mathbb U^a_t(h)+\mathbb R_2^a(h,\eps),
\end{equation}
\begin{equation}\label{eq:necessarylimarm}
\mathbb T_n^{\eps,a} (\lambda_*+\eps h)
=\mathbb U^a_n(h)+\mathbb R_3^a(h,\eps),
\end{equation}
where
\begin{equation}
\mathbb U^a_s (h):=2\tilde{\mathbb T}^{0,a}(\lambda_*) 
+ 2\beta^a(h)(\mathbb P + \mathbb P'), \quad
\mathbb U^a_t(h):=\mathbb I
+2\xi^a(h)(\mathbb Q - \mathbb Q'), \quad
\mathbb U^a_n(h):=\mathbb I
-2\xi^a(h)(\mathbb Q - \mathbb Q'),
\end{equation}
and $\|R_i^a(h,\eps)\|_{ \HhalfGa\times\HmhalfGa\to \HhalfGa\times\HmhalfGa} =o(1)$ as $\eps\to 0^+$ uniformly for $|h|<\hrad$, 
$i=1,2,3$. 
\end{coro}

Define the function spaces
\begin{equation}\label{eq:XYA}
 X:=\text{span}\{\vec v_1, \vec v_2\},\quad  X':=\text{span}\{\vec v_1', \vec v_2'\},\quad Z:=\{\vec\phi \in \HhalfGa\times\HmhalfGa, c_i(\vec\phi)=0, i=1,2\}.
\end{equation}
\begin{equation}\label{eq:PZ}
P_Z (\vec\phi) :=\vec\phi -
\frac{c_1(\vec\phi)}{i\alpha_*^a}
\vec v_1 -
\frac{c_2(\vec\phi)}{-i\alpha_*^a}
\vec v_2-
\frac{c_1'(\vec\phi)}{i\alpha_*^a}
\vec v_1' -
\frac{c_2'(\vec\phi)}{-i\alpha_*^a}
\vec v_2'. 
\end{equation}
Since $c_i(P_Z (\vec\phi))=c_i'(P_Z (\vec\phi))=0$, $i=1,2$, we obtain
the direct sum decomposition:
\begin{equation}\label{eq:directsumA}
\HhalfGa\times\HmhalfGa = X \bigoplus X'\bigoplus Z.
\end{equation}

The following fact will be used repeatedly in the proofs. 
The relations in \eqref{eq:civjA} imply that $\vec v_1$, $\vec v_2$, $\vec v_1'$ and $\vec v_2'$ are linearly independent.
For the operators defined in \eqref{eq:P} and \eqref{eq:Q}, 
\begin{equation}\label{}
\begin{aligned}
\mathbb P \vec v_1= \im\alpha_*^a \vec v_1, \quad \mathbb P \vec v_2= - \im\alpha_*^a \vec v_2, \quad 
\mathbb Q \vec v_1= \im \vec v_2, \quad 
\mathbb Q \vec v_2= - \im\alpha_*^a\vec v_1,\\
\mathbb P' \vec v_1'= \im\alpha_*^a \vec v_1', \quad \mathbb P' \vec v_2'= - \im\alpha_*^a \vec v_2', \quad 
\mathbb Q' \vec v_1'= \im \vec v_2', \quad 
\mathbb Q' \vec v_2'= - \im\alpha_*^a\vec v_1',
\end{aligned}
\end{equation}
and 
\begin{equation}\label{}
\mathbb P X'= \mathbb Q X'=0,\quad \mathbb P' X= \mathbb Q' X=0, \quad \mathbb P Z= \mathbb Q Z=\mathbb P' Z= \mathbb Q' Z=0.
\end{equation}

\begin{prop}\label{lem:sufflimpropA}
The following holds for $h\in\mathbb C$ satisfying $|h|< \hrad$. The operators $\mathbb U^a_s(h)$, $\mathbb U^a_s(h)$, and $\mathbb U^a_t(h)$ are analytic in $h$ and are Fredholm operators with index zero. The only characteristic value of each operator is $h=0$ and the multiplicity of the characeteric value is 4. In addition, the kernels are $\mathbb U^a_s(0)$, $\mathbb U^a_t(0)$ and $\mathbb U^a_n(0)$ are given by 
\begin{equation*}
\begin{aligned}
& \ker \;\mathbb U^a_s(0) =  \text{span} \left\{ \vec v_1,\vec v_2, \vec v_1',\vec v_2'\right\} = X\bigoplus X', \\
& \ker \;\mathbb U^a_t(0) =  \text{span} \left\{ \vec v_1 - \im\vec v_2,  \vec v_1' + \im\vec v_2' \right\}, \\
& \ker \;\mathbb U^a_n(0) =  \text{span} \left\{ \vec v_1 + \im\vec v_2,  \vec v_1' -  \im\vec v_2' \right\}.
\end{aligned}
\end{equation*}
\end{prop}

\subsection{The characteristic values of the integral operators}
Define
\begin{equation}
\mathfrak u_1:= \vec v_1+ \im\vec v_2,\quad 
\mathfrak u_2:= \vec v_1- \im\vec v_2,\quad
\mathfrak u_1':= \vec v_1'+ \im\vec v_2',\quad 
\mathfrak u_2':= \vec v_1'- \im\vec v_2',
\end{equation}
and
\begin{equation}
X_i:=\text{span}\{\mathfrak u_i\},\quad X_i':=\text{span}\{\mathfrak u_i'\},\quad i=1,2.
\end{equation}
Then
\begin{equation}\label{eq:XXYarm}
X=X_1\bigoplus X_2,\quad X'= X_1'\bigoplus X_2',\quad \HhalfGa\times\HmhalfGa = X_1\bigoplus X_2\bigoplus X_1'\bigoplus X_2'\bigoplus Z.
\end{equation}
We have 
\begin{equation}\label{}
\begin{aligned}
&\mathbb P \mathfrak u_1= \im\alpha_*^a\mathfrak u_2, \quad 
\mathbb P \mathfrak u_2= \im\alpha_*^a\mathfrak u_1, \quad 
\mathbb Q \mathfrak u_1= \alpha_*^a\mathfrak u_1, \quad 
\mathbb Q \mathfrak u_2= -\alpha_*^a\mathfrak u_2,\quad \mathbb P X'= \mathbb Q X'=0,\\
&\mathbb P' \mathfrak u_1'= \im\alpha_*^a\mathfrak u_2', \quad 
\mathbb P' \mathfrak u_2'= \im\alpha_*^a\mathfrak u_1', \quad 
\mathbb Q' \mathfrak u_1'= \alpha_*^a\mathfrak u_1', \quad 
\mathbb Q' \mathfrak u_2'= -\alpha_*^a\mathfrak u_2',\quad \mathbb P' X= \mathbb Q' X=0,\\
&\mathbb P Z = \mathbb Q Z= \mathbb P' Z = \mathbb Q' Z =0
\end{aligned}
\end{equation}

Lemmas~\ref{lem:rank1arm}-\ref{lem:nphiarm} are  parallel to Lemmas~\ref{lem:rank1}- \ref{lem:tphi}. We state them below without proof.
\begin{lemma}\label{lem:rank1arm}
Let Assumption~\ref{lem:assNoFold} hold along $\bbeta_1^a$ and $t_*>0$. Let $\mathfrak d\in(0,1)$ be a constant.
Suppose $|h_0|<\hrad$ and $0\neq\vec\phi\in\ker \; \mathbb T_s^{\eps,a} (\lambda_*+\eps h_0)$. When $|\eps| \ll 1$, the rank of $\vec\phi$ is $1$.
\end{lemma}
\begin{lemma}\label{lem:tphiarm}
Let Assumption~\ref{lem:assNoFold} hold along $\bbeta_1^a$ and $t_*>0$. Let $\mathfrak d\in(0,1)$ be a constant.
For sufficiently small positive $\eps$ and $|h|<\hrad$,
the operator below is invertible
\begin{equation}
\mathbb I + \xi(h)P_{X_1\bigoplus X_2'\bigoplus Y} (\mathbb Q - \mathbb Q') + P_{X_1\bigoplus X_2'\bigoplus Y} \mathbb R_2^a(h,\eps) :  X_1\bigoplus X_2'\bigoplus Y\to X_1\bigoplus X_2'\bigoplus Y.
\end{equation} 
Here $P_{X_1\bigoplus X_2'\bigoplus Y}$ is the projection onto $X_1\bigoplus X_2'\bigoplus Y$ associated to the direct sum \eqref{eq:XXYarm}, and $\mathbb R_2^a(h,\eps)$ is the remainder defined in Corollary~\ref{lem:stnlimitsarm}.
Denote the inverse by $B^a(h,\eps): X_1\bigoplus X_2'\bigoplus Y\to X_1\bigoplus X_2'\bigoplus Y$, and
define 
\begin{equation}\label{eq:tphi1arm}
J_4(h,\eps)[\vec\phi_0]= -B^a(h,\eps)P_{X_1\bigoplus X_2'\bigoplus Y} \mathbb R_2^a(h,\eps)\vec\phi_0,
\end{equation}
for each given $\vec\phi_0\in X_2\bigoplus X_1'$.
We have $J_4(h,\eps)[\vec\phi_0]$ is analytic in $h$ and
\begin{equation}
\|J_4(h,\eps)[\vec\phi_0]\|_{\HhalfGa\times\HmhalfGa} =  o(1)\|\vec\phi_0\|_{\HhalfGa\times\HmhalfGa},\quad \text{uniformly in $h$}.
\end{equation}

Moreover,
if $\vec\phi\in \ker \;(\mathbb T_t^\eps (\lambda_*+\eps h_0) )$, then
\begin{equation}\label{eq:tphiarm}
\vec\phi= \vec\phi_0 +J_4(h,\eps)[\vec\phi_0]
\end{equation}
for some $\vec\phi_0\in X_2\bigoplus X_1'$.
\end{lemma}
\begin{lemma}\label{lem:nphiarm}
Let Assumption~\ref{lem:assNoFold} hold along $\bbeta_1^a$ and $t_*>0$. Let $\mathfrak d\in(0,1)$ be a constant.
For sufficiently small positive $\eps$ and $|h|<\hrad$,
the operator below is invertible
\begin{equation}
\mathbb I - \xi(h)P_{X_2\bigoplus X_1'\bigoplus Y} (\mathbb Q -\mathbb Q') + P_{X_2\bigoplus X_1'\bigoplus Y} \mathbb R_3^a(h,\eps) :  X_2\bigoplus X_1'\bigoplus Y\to X_2\bigoplus X_1'\bigoplus Y.
\end{equation} 
Here $P_{X_2\bigoplus X_1'\bigoplus Y}$ is the projection onto $X_2\bigoplus X_1'\bigoplus Y$ associated to the direct sum \eqref{eq:XXYarm}, and $\mathbb R_3^a(h,\eps)$ is the remainder defined in Corollary~\ref{lem:stnlimitsarm}.
Denote the inverse by $C^a(h,\eps): X_2\bigoplus X_1'\bigoplus Y\to X_2\bigoplus X_1'\bigoplus Y$, and
define 
\begin{equation}\label{eq:nphi1arm}
 J_5(h,\eps)[\vec\phi_0]= -C^a(h,\eps)P_{X_2\bigoplus X_1'\bigoplus Y} \mathbb R_3^a(h,\eps)\vec\phi_0,
\end{equation}
for each given $\vec\phi_0\in X_1\bigoplus X_2'$.
We have $ J_5(h,\eps)[\vec\phi_0]$ is analytic in $h$ and
\begin{equation}
\| J_5(h,\eps)[\vec\phi_0]\|_{\HhalfGa\times\HmhalfGa} =  o(1)\|\vec\phi_0\|_{\HhalfGa\times\HmhalfGa},\quad \text{uniformly in $h$}.
\end{equation}

Moreover,
if $\vec\phi\in \ker \;(\mathbb T_t^\eps (\lambda_*+\eps h_0) )$, then
\begin{equation}\label{eq:nphiarm}
\vec\phi= \vec\phi_0 + J_5(h,\eps)[\vec\phi_0]
\end{equation}
for some $\vec\phi_0\in X_1\bigoplus X_2'$.
\end{lemma}

\begin{prop}\label{lem:uniquestnA}
Let Assumption~\ref{lem:assNoFold} hold along $\bbeta_1^a$ and $t_*>0$. Let $\mathfrak d\in(0,1)$ be a constant. 
For sufficiently small $\eps>0$, the system
\begin{equation}\label{eq:simulsnarm}
\mathbb T_s^{\eps,a} (\lambda_*+\eps h)\vec\phi=0 \quad\text{and } \quad \mathbb T_n^{\eps,a} (\lambda_*+\eps h)\vec\phi=0
\end{equation}
attains at most two pairs of solutions $(h,\vec\phi)$, with $|h|<\hrad$ and $\vec\phi\in\HhalfGa\times\HmhalfGa$.
Moreover, if $h_1=h_2$, then $\vec\phi_1$ and $\vec\phi_2$ are linearly independent.
The same holds for the system
\begin{equation}\label{eq:simulstarm}
\mathbb T_s^{\eps,a} (\lambda_*+\eps h)\vec\phi=0 \quad\text{and } \quad \mathbb T_t^{\eps,a} (\lambda_*+\eps h) \vec\phi=0.
\end{equation}
\end{prop}


\begin{proof}

Suppose $\vec\phi$ solves both equations in \eqref{eq:simulsnarm}. By Lemma~\ref{lem:nphiarm}, the solution to the second equation necessarily takes the form
 $\vec\phi= \vec\phi_0 +\vec \phi_1(h,\eps)$, where $\vec\phi_1(h,\eps) = J_5(h,\eps)[\vec\phi_0]$ as defined in \eqref{eq:nphi1arm}, with $\vec\phi_0 = a\mathfrak u_1 + b\mathfrak u_2' $ for some $a,b\in\mathbb C$.
 Substituting $\vec\phi$ into the first equation, we obtain
 \begin{equation*}
(2\tilde{\mathbb T}^{0,a}(\lambda_*) + 2\beta^a(h)(\mathbb P+\mathbb P') + \mathbb R_1^a(h,\eps))(a \mathfrak u_1 + b \mathfrak u_2' + \vec\phi_1(h,\eps))= 0.
\end{equation*}
Projecting the above onto the space $X_2\bigoplus X_1'$ using $P_{X_2\bigoplus X_1'}$, we obtain
\begin{equation*}
P_{X_2\bigoplus X_1'} 2\beta^a(h)(\mathbb P+\mathbb P')(a \mathfrak u_1 + b \mathfrak u_2') +
P_{X_2\bigoplus X_1'} \left((2\beta^a(h)(\mathbb P+\mathbb P') + \mathbb R_1^a(h,\eps))  
\vec\phi_1(h_\eps,\eps) 
+ \mathbb R_1^a(h,\eps) 
(a \mathfrak u_1 + b \mathfrak u_2')\right) =0.
\end{equation*}
The projections onto $X_2$ and $X_1'$ give
\begin{equation}\label{eq:uniquearm}
\left(D(h)+ E(h,\eps)\right)
\left(\begin{matrix}a\\b\end{matrix}\right)=0,
\end{equation}
where $D(h)$ is defined by
\begin{equation}
\left(\begin{matrix}
2\im\beta^a(h)\alpha_*^a & 0\\
0& 2\im\beta^a(h)\alpha_*^a
\end{matrix}\right)
\end{equation}
and $E(h,\eps)$ is of higher order in $\eps$.

It is obvious that  $D(h)$ is analytic in $h$ in a neighborhood of $\{h\in \cC, |h|<\hrad$\}. Furthermore, $h=0$ is the unique characteristic of $D(h)$ in $|h|<\hrad$, and the multiplicity of $h=0$ is two.
Note that $E(h,\eps)$ is analytic and its matrix norm is of order $o(1)$ uniformly in $h$ as $\eps\to 0^+$. 
Thus the generalized Rouch\'e Theorem implies that there are two pairs of $(h_i,(a_i,b_i))$ $(i=1, 2)$ solving \eqref{eq:uniquearm}. When $h_1=h_2$, $(a_1,b_1)$ and $(a_2,b_2)$ are linearly independent. By the independence of $\mathfrak u_2$ and $\mathfrak u_1'$, we complete the proof.
\end{proof}

\subsection{Proof of Theorem~\ref{lem:edgearm}}
\begin{proof}[Proof of Theorem~\ref{lem:edgearm}]

Using the same argument as that in the proof of Theorem~\ref{lem:edge}, by Propositions~\ref{lem:sufflimpropA}  and Theorem~\ref{lem:actGS}, we conclude that $\mathbb T_s^{\eps,a} (\lambda_*+\eps h)$ is of multiplicity four in $|h|<\hrad$ when $\eps>0$ is sufficiently small.
By Lemma~\ref{lem:rank1arm}, there are four pairs $(h_i,\vec\phi_i)$ solving $\mathbb T_s^{\eps,a} (\lambda_*+\eps h)\vec\phi=0$, $i=1,\cdots,4$. Moreover, if any of the $h_i$'s coincide, the corresponding $\vec\phi_i$'s form a linearly independent set.

We first show that there are at least two interface modes.
Let $u_i$ be generated by $\vec\phi_i$ \eqref{eq:denstoedge}, where the Green function is replaced by 
$G^{\pm\eps,a}(\bx,\by; \lambda)$ as defined in \eqref{eq:pGreenpmepsA}, and the integral domain is replaced by $\Gamma^a$. 
Assume $u_i$, $i=1,\cdots,4$ represents fewer than two interface modes. 
Then we have the following two cases:
\begin{itemize}
    \item [(i)] All $u_i$, $i=1,\cdots,4$ are zero. Then $(h_i,\vec\phi_i)$ are four solutions to \eqref{eq:simulstarm}, which contradicts with Proposition~\ref{lem:uniquestnA}.
\item [(ii)] After rearranging, for some $m\in\{2,3,4\}$, $u_1,u_2,\cdots,u_m$ are nonzero and span a one-dimensional space, and the rest $(4-m)$ of $u_i$'s are zero. 
Then it follows that $h_1=\cdots = h_m$. Also, using $(u_i|_{\Gamma^a},\partial_n u_i|_{\Gamma^a})$, $i=1,\cdots,m$, we can construct  densities $\vec\phi'_i$, $i= 1,\cdots,m-1$, which are linearly independent and generates zero modes.
Thus we have $(4-m)+(m-1)=3$ simultaneous solutions to \eqref{eq:simulstarm}, which contradicts with  Proposition~\ref{lem:uniquestnA}.
\end{itemize}

Next, we show that there are at most two interface modes. Suppose there are more than two linearly independent interface modes $u_i$ at $h_i$ respectively for  $i=1,2,3$. Denote $\vec\phi_i = (u_i|_{\Gamma^a},\partial_n u_i|_{\Gamma^a})$. Then $(h_i,\vec\phi_i)$ with $i=1,2,3$ are solutions to the system \eqref{eq:simulsnarm}, and if $h_i=h_j$, for some $i\neq j$, then $\vec\phi_i$ and $\vec\phi_j$ are linearly independent. This contradicts with Proposition~\ref{lem:uniquestnA}.

\end{proof}

\section{Dispersion relations of interface modes}
\label{sec:dispersion}
In this section, we investigate the dispersion relation along the zigzag interface as stated in Theorem~\ref{lem:dispersion}. The dispersion relations along the armchair interface and arbitrary rational interfaces as stated in Theorems~\ref{lem:dispersionarm}, \ref{lem:ratzig} and \ref{lem:ratarm} are proved in Appendix~\ref{sec:disprat}.

For $\eps$ in a neighborhood of $0$, define the Green function on the periodic zigzag strips with quasimomentum $\kp^* + \mu$ by
\begin{equation}
\begin{aligned}
\begin{cases}
(-\Delta_{\bx}  - \lambda) G^\eps[\mu](\bx,\by; \lambda)  = \delta(\bx-\by) \quad &\bx\in \Omega^\eps,\\
G^\eps[\mu](\bx,\by; \lambda) =0 \quad &\bx\in\cup_{m\in\mathbb Z} (\partial D^\eps+m\be_1), \\
G^\eps[\mu](\bx+ \be_2,\by; \lambda)= e^{\im \kp^*+\mu}G^\eps[\mu](\bx,\by; \lambda) \quad &\text{for }\bx\in\Gamma_-, \\
\partial_{\nuGb}  G^\eps[\mu](\bx+ \be_2,\by; \lambda) =e^{\im \kp^*+\mu}\partial_{\nuGb}  G^\eps[\mu](\bx,\by; \lambda)  \quad &\text{ for }\bx\in\Gamma_-,\\
G^\eps[\mu](\bx,\by; \lambda) \,\,\text{satisfies the radiation conditions when $|\bx\cdot\be_1| \to \infty$}.
\end{cases}
\end{aligned}
\end{equation}

Similar to \eqref{eq:HsG2} and \eqref{eq:HsG1a}, we define the following quasi-periodic Sobolev space on $\Gamma$ for $s\in\mathbb R$ 
\begin{equation}\label{eq:HsG1eta}
\cH^s(\Gamma,\mu):= \left\{ u(\bx_0+ t\be_2) = \sum_{n\in\mathbb Z} a_n e^{\im (\Kone + \mu\bbeta_2)\cdot\be_2 t} e^{\im 2\pi n t}: \|u\|_{\cH^s(\Gamma,\eta)}^2:=\sum_{n\in\mathbb Z} |a_n|^2 (1+ n^2)^s\right\}.
\end{equation}
The functions in $\cH^s(\Gamma,\mu)$ attain the quasimomentum $\kp^*+\mu$ along the zigzag edge $\be_2$. In particular, $\cH^s(\Gamma,0) = \cH^s(\Gamma)$ in \eqref{eq:HsG2}.

Define the layer potentials 
$\cS^{\pm\eps}(\lambda,\mu)$, $\cD^{\pm\eps}(\lambda,\mu)$, $\cK^{\pm\eps}(\lambda,\mu)$, $\cK^{\pm\eps}(\lambda,\mu)$ and $\cN^{\pm\eps}(\lambda,\mu)$ in parallel to 
\eqref{eq:lpotentialeps},  
 where the Green functions are replaced by $G^{\pm\eps}[\mu](\bx,\by,\lambda)$ above. 
We also define the integral operators on $\HhalfG\times\HmhalfG$ as in \eqref{eq:bTeps} and \eqref{eq:bTs} by
\begin{equation}\label{eq:bTeps_mu}
\mathbb T^\eps
(\lambda, \mu):=
\left(\begin{matrix}
-\cK^{\eps}(\lambda,\mu) & \cS^{\eps}(\lambda,\mu) \\
 -\cN^{\eps}(\lambda,\mu) & \cK^{*,\eps}(\lambda,\mu)
\end{matrix}\right),
\end{equation}
and
\begin{equation}\label{eq:bTs_mu}
\mathbb T_s^\eps(\lambda,\mu):=\mathbb T^\eps + \mathbb T^{-\eps}, \quad
\mathbb T_t^\eps(\lambda,\mu):=- \mathbb T^\eps + \mathbb T^{-\eps} + \mathbb I, \quad
\mathbb T_n^\eps(\lambda,\mu):=\mathbb T^\eps - \mathbb T^{-\eps} + \mathbb I.
\end{equation}
Let $\mathbb M(\mu)$ be the operator of multiplication by the factor $e^{-\im \mu\bbeta_2\cdot\bx}$. 
Since $\phi\in \cH^s(\Gamma,\mu)$ if and only if $\mathbb M(\mu)\phi\in \cH^s(\Gamma)$, we have the following characterization of edge states with quasimomentum $\kp^*+\mu$.
\begin{lemma}\label{lem:dispoplim}
Let $\mu=\eps\zeta$. 
There exists an interface mode of quasimomentum $\kp^*+\mu$ along $\be_2$ if and only if there exists $(\psi,\phi)\in \HhalfG\times\HmhalfG$, such that
\begin{equation}\label{eq:sufficientmu}
\mathbb M^{-1}(\eps\zeta)\mathbb T^{\eps}_s(\lambda_*+ \eps h,\eps\zeta)\mathbb M(\eps\zeta)
\left(\begin{matrix}
\psi\\
\phi
\end{matrix}\right)
=0,\quad  \mathbb M^{-1}(\eps\zeta)\mathbb T^{\eps}_t(\lambda_*+ \eps h,\eps\zeta)\mathbb M(\eps\zeta)
\left(\begin{matrix}
\psi\\
\phi
\end{matrix}\right)
\neq0.
\end{equation}
Moreover, if $u$ is an interface mode with quasimomentum $\kp^*+\mu$ along $\be_2$, then $0\neq( u|_\Gamma, \partial_n u|_\Gamma) \in \HhalfG\times \HmhalfG$ satisfies
\begin{equation}\label{eq:necessarymu}
\mathbb M^{-1}(\eps\zeta)\mathbb T^{\eps}_s(\lambda_*+ \eps h,\eps\zeta)\mathbb M(\eps\zeta)
\left(\begin{matrix}
u|_\Gamma\\
\partial_n u|_\Gamma
\end{matrix}\right)
=0,\quad
\mathbb M^{-1}(\eps\zeta)\mathbb T^{\eps}_n(\lambda_*+ \eps h,\eps\zeta)\mathbb M(\eps\zeta)
\left(\begin{matrix}
u|_\Gamma\\
\partial_n u|_\Gamma
\end{matrix}\right)
=0.
\end{equation}
\end{lemma}

Define
\begin{equation}\label{eq:betaxisig}
\begin{aligned}
\beta(h,\zeta)&:=   \frac{1}{2} \left|\frac{\gamma_*}{\theta_*}\right|\frac{h}{\sqrt{(\frac{t_*}{\gamma_*})^2+ \frac{3}{4}|\frac{\theta_*}{\gamma_*}|^2\zeta^2-h^2}} =  \frac{1}{2\alpha_*}\frac{h}{\sqrt{\beta_*^2 + \frac{3}{4}\alpha_*^2\zeta^2 -h^2}},\\
\xi(h,\zeta) &:=  \frac{t_*}{2|\theta_*|}\frac{1}{\sqrt{(\frac{t_*}{\gamma_*})^2+ \frac{3}{4}|\frac{\theta_*}{\gamma_*}|^2\zeta^2-h^2}}=  \frac{\beta_*}{2\alpha_*}\frac{1}{\sqrt{\beta_*^2 + \frac{3}{4}\alpha_*^2\zeta^2 -h^2}},\\
\sigma(h,\zeta)& := \im \frac{\sqrt3}{2}\zeta\frac{1}{\sqrt{(\frac{t_*}{\gamma_*})^2+ \frac{3}{4}|\frac{\theta_*}{\gamma_*}|^2\zeta^2-h^2}}  =   \im \frac{\sqrt3}{2}\zeta\frac{1}{\sqrt{\beta_*^2 + \frac{3}{4}\alpha_*^2\zeta^2 -h^2}}.
\end{aligned}
\end{equation}
Similarly to Proposition~\ref{lem:oplim}, we obtain the following operator limit.
\begin{prop}\label{lem:chadisp}
Let Assumption~\ref{lem:assNoFold}  hold along $\bbeta_1$ and $t_*>0$. Let $\mathfrak d\in(0,1)$ be a constant.
For each constant $\zeta$, the following limit holds uniformly for $|h|<\hrad$ as $\eps\to 0^+$ in the operator norm from $\HhalfG\times\HmhalfG$ to $\HhalfG\times\HmhalfG$: 
\begin{equation}\label{eq:oplimeta}
\mathbb M^{-1}(\eps\zeta)\mathbb T^{\pm\eps}(\lambda_*+ \eps h,\eps\zeta) \mathbb M(\eps\zeta)
\to 
\tilde{\mathbb T}^0 (\lambda_*) 
+ \beta(h,\zeta)\mathbb P \mp \xi(h,\zeta)\mathbb Q + \sigma(h,\zeta)\mathbb O,
\end{equation}
where 
\begin{equation}
\mathbb O\vec\phi:= - c_2(\vec\phi) \vec v_1 + c_1(\vec\phi)\vec v_2 .
\end{equation}
In addition, there hold the uniform convergences
\begin{equation}\label{eq:etasuff}
\mathbb M^{-1}(\eps\zeta)\mathbb T^{\eps}_s(\lambda_*+ \eps h,\eps\zeta)\mathbb M(\eps\zeta) 
\to 
2\tilde{\mathbb T}^0 (\lambda_*) 
+2 \beta(h,\zeta)\mathbb P  + 2\sigma(h,\zeta)\mathbb O=: \mathbb U_s (h,\zeta),
\end{equation}
\begin{equation}\label{eq:etanont}
\mathbb M^{-1}(\eps\zeta)\mathbb T^{\eps}_t(\lambda_*+ \eps h,\eps\zeta)\mathbb M(\eps\zeta) 
\to 
I + 2 \xi(h,\zeta)\mathbb Q =: \mathbb U_t (h,\zeta),
\end{equation}
\begin{equation}\label{eq:etanece}
\mathbb M^{-1}(\eps\zeta)\mathbb T^{\eps}_n(\lambda_*+ \eps h,\eps\zeta)\mathbb M(\eps\zeta)
\to 
I - 2 \xi(h,\zeta)\mathbb Q =: \mathbb U_n (h,\zeta),
\end{equation}
\end{prop}
The proof is in Appendix~\ref{sec:disprat}.

Now we are ready to prove Theorem~\ref{lem:dispersion}.
\begin{proof}[Proof of Theorem~\ref{lem:dispersion}]
Assume $t_*>0$. 
We fix an $\zeta\neq0$ and let $\vec\phi = a \vec v_1 + b\vec v_2$.
First, $\mathbb U_s (h,\zeta)\vec\phi = 0$ yields
\begin{equation}\label{eq:etaslim}
\im\alpha_*\left(\begin{matrix}
\beta(h,\zeta) &  \sigma(h,\zeta) \\
\sigma(h,\zeta)& - \beta(h,\zeta)
\end{matrix}\right)
\left(\begin{matrix}
a\\b
\end{matrix}\right) =0. 
\end{equation}
Solving \eqref{eq:etaslim}, we obtain 
two characteristic values for $h$. More specifically,   
$\sigma(h,\zeta) = \im\beta(h,\zeta)$ gives that $h= \frac{\sqrt3}{2}\alpha_*\zeta $, $(a,b) = (1,\im)$; and $\sigma(h,\zeta) = - \im\beta(h,\zeta)$ gives that $h= - \frac{\sqrt3}{2}\alpha_*\zeta $, $(a,b) = (1, - \im)$.

Next, $\mathbb U_t (h,\zeta)\vec\phi = 0$  yields
\begin{equation}\label{eq:etatlim}
\left(\begin{matrix}
1 &  -2\im\alpha_*\xi(h,\zeta) \\
2\im\alpha_*\xi(h,\zeta)& 1
\end{matrix}\right)
\left(\begin{matrix}
a\\b
\end{matrix}\right) =0.
\end{equation}
Solving \eqref{eq:etanlim}, we obtain two characteristic values for $h$. Indeed, the condition $1 - 4\alpha_*^2 (\xi(h,\zeta))^2=0$ is achieved if and only if $|2\alpha_*\xi(h,\zeta)| = 1$, which is equivalent to $\frac{3}{4}\alpha_*^2\zeta^2-h^2 = 0$.
For both $h = \pm \frac{\sqrt3}{2}\alpha_*\zeta$, $(a,b) = (1,-\im)$.

Similarly, $\mathbb U_n (h,\zeta)\vec\phi = 0$  yields that
\begin{equation}\label{eq:etanlim}
\left(\begin{matrix}
1 &  2\im\alpha_*\xi(h,\zeta) \\
-2\im\alpha_*\xi(h,\zeta)& 1
\end{matrix}\right)
\left(\begin{matrix}
a\\b
\end{matrix}\right) =0,
\end{equation}
which also gives two characteristic values for $h$. 
For both $h = \pm \frac{\sqrt3}{2}\alpha_*\zeta$, $(a,b) = (1,\im)$. 

Thus for a fixed $\zeta\neq0$, $\mathbb T^{\eps}_s(\lambda_*+ \eps h,\zeta) $ has one characteristic value $h\approx\frac{\sqrt3}{2}\alpha_*\zeta$, with root function $\vec\phi \approx \vec v_1 + \im\vec v_2$. This $\vec\phi$ generates a nonzero interface mode since $\mathbb U_t (h,\zeta)( \vec v_1 + \im\vec v_2)\neq 0$. 
The theorem for $t_*>0$ follows by noting the following relation
$\lambda = \lambda_* + \eps h$, $\kp = (K+\eps\zeta\bbeta_2+\ell\bbeta_1)\cdot\be_2 = \kp^* + \eps \zeta$ and $\alpha_* = m_*|\bbeta_1|$.

When $t_*<0$, doing the parallel calculations, we obtain the theorem when $\sgn(t_*)=-1$.
\end{proof}

\begin{appendices}

\section{Gohberg and Sigal theory}
\label{sec:GStheory}
We briefly introduce the Gohberg and Sigal theory. We refer to Chapter 1.5 of \cite{Ammari-book} for a thorough exposition of the topic.

Let $X$ and $Y$ be two Banach spaces. 
Let $\mathfrak{U}(z_0)$ be the set of all operator-valued functions with values in $\mathcal{B}(X,Y)$, which are holomorphic in some neighborhood of $z_0$, except possibly at $z_0$. Then the point $z_0$ is called a \textbf{characteristic value} of $A(z)\in\mathfrak{U}(z_0)$ if there exists a vector-valued function $\phi(z)$ with values in $X$ such that
\begin{enumerate}
    \item $\phi(z)$ is holomorphic at $z_0$ and $\phi(z_0)\neq 0$,
    \item $A(z)\phi(z)$ is holomorphic at $z_0$ and vanishes at this point.
\end{enumerate}
Here $\phi(z)$ is called a \textbf{root function} of $A(z)$ associated with the characteristic value $z_0$, and $\phi(z_0)$ is called an \textbf{eigenvector}. By this definition, there exists an integer $m(\phi)\geq 1$ and a vector-valued function $\psi(z)\in Y$, holomorphic at $z_0$, such that
\begin{equation*}
    A(z)\phi(z)=(z-z_0)^{m(\phi)}\psi(z),\quad \psi(z_0)\neq 0.
\end{equation*}
The number $m(\phi)$ is called the \textbf{multiplicity} of the root function $\phi(z)$. For $\phi_0\in\text{Ker}A(z_0)$, the \textbf{rank} of $\phi_0$, which is denoted by $\text{rank}(\phi_0)$, is defined as the maximum of the multiplicities of all root functions $\phi(z)$ with $\phi(z_0)=\phi_0$.

Suppose that $n=\text{dim Ker}A(z_0)<+\infty$ and the ranks of all vectors in $\text{Ker}A(z_0)$ are finite. A system of eigenvectors $\phi_0^j$ ($j=1,2,\cdots,n$) is called a \textbf{canonical system of eigenvectors} of $A(z)$ associated to $z_0$ if for $j=1,2,\cdots,n$, $\text{rank}(\phi_0^j)$ is the maximum of the ranks of all eigenvectors in the direct complement in $\text{Ker}A(z_0)$ of the linear span of the vectors $\phi_0^{1},\cdots,\phi_0^{j-1}$. We call
\begin{equation*}
    N(A(z_0)):=\sum_{j=1}^{n}\text{rank}(\phi_0^{j})
\end{equation*}
the \textbf{null multiplicity} of the characteristic value $z_0$ of $A(z)$. 
Suppose that $A^{-1}(z)$ exists and is holomorphic in some neighborhood of $z_0$, except possibly at $z_0$. Then the number
\begin{equation*}
    M(A(z_0)):= N(A(z_0))- N(A^{-1}(z_0))
\end{equation*}
is called the \textbf{multiplicity} of $z_0$.
Note if $A(z)$ is bounded in a neighborhood of $z_0$, then $N(A^{-1}(z_0))=0$.

Now, let $V$ be a simply connected bounded domain with a rectifiable boundary $\partial V$. Let $A(z)$ be an operator-valued  function that is analytic in a neighborhood of $\overline{V}$, is Fredholm of index zero in $V$ and is invertible in $\overline{V}$ except at possibly a finite number of points.
For such a function $A(z)$, the full multiplicity $\mathcal{M}(A(z);\partial V)$ counts the number of characteristic values of $A(z)$ in $V$ (computed with their multiplicities). Namely,
\begin{equation*}
    \mathcal{M}(A(z);\partial V):=\sum_{i=1}^{\sigma}M(A(z_i))=\sum_{i=1}^{\sigma}N(A(z_i)),
\end{equation*}
where $z_i$ ($i=1,2,\cdots,\sigma$) are all characteristic values of $A(z)$ lying in $V$.
The generalized Rouch\'{e} theorem for analytic operator-valued functions is stated as follows. This is a special case of the generalized Rouche\'{e} theorem for finitely meromorphic operator-valued functions~\cite[Theorem 1.5]{Ammari-book}.


\begin{thm}\label{lem:actGS}
The operator-valued function $A(z)$ is analytic and Fredholm of index zero in a neighborhood of $\overline{V}$, $A^{-1}(z)$ exists except at a finite number of points in $\overline{V}$. The operator-valued function $B(z)$ is analytic in a neighborhood of $\overline{V}$. Suppose 
\begin{equation*}
    \|A^{-1}(z)B(z)\|_{\mathcal{B}(X,Y)}<1,\quad z\in \partial V.
\end{equation*}
Then the multiplicity of $A(z)$ in $V$ equals the multiplicity of $A(z)+B(z)$ in $V$. That is,
\begin{equation*}
    \mathcal{M}(A(z);\partial V)=\mathcal{M}(A(z)+B(z);\partial V).
\end{equation*}
In addition the multiplicities of $A(z)$ and $A(z)+B(z)$ do not involve their inverses. That is,
\begin{equation*}
    \mathcal{M}(A(z);\partial V)=\sum_{i=1}^{\sigma}N(A(z_i)),
\end{equation*}
\begin{equation*}
    \mathcal{M}((A+B)(z);\partial V)=\sum_{i=i}^{\sigma'}N((A+B)(z_i')),
\end{equation*}
where $z_i$, $i=1,\cdots,\sigma$ are all characteristic values of $A(z)$ in $V$, and $z_i'$, $i=1,\cdots,\sigma'$ are all characteristic values of $(A+B)(z)$ in $V$.
\end{thm}

\section{Proof of Propositions~\ref{lem:Tderiv} and~\ref{lem:TderivA}}
\label{sec:Tderiv}

We first display a set of facts that can be easily checked for the readers' convenience.
Define 
\begin{equation}
R_s := \left(\begin{matrix} \cos s & \sin s \\ -\sin s&\cos s\end{matrix}\right), \quad s\in\mathbb R, \quad R:=R_{\frac{2\pi}{3}}\quad\text{and}\quad J:= \left(\begin{matrix}0 &1\\-1&0 \end{matrix}\right).
\end{equation}
Acting on vectors in $\mathbb R^2$, it holds
\begin{equation}\label{eq:oprelat}
\begin{aligned}
&R^3=I, \quad I+R+R^2=0, \quad R_aR_b=R_bR_a,\quad \frac{\partial }{\partial_\eps}R_\eps|_{\eps=0}=J, \\ &\rflc^2=I, \quad \rflc R_\theta=(R_\theta)^{-1}\rflc, \quad
J^{-1}=-J,\quad JR=RJ,\quad \rflc J = - J \rflc.
\end{aligned}
\end{equation} 
For all vectors $\bx,\by\in\mathbb R^2$,
\begin{equation}
\quad R\bx\cdot R\by=\bx\cdot\by,  \quad \rflc\bx\cdot \rflc\by=\bx\cdot\by.
\end{equation}
The set $\tilde{\Lambda}^*:=\{\Kone+\bq,\bq\in\Lambda^*\}$ is invariant under $R$ and $\rflc$ as shown in \eqref{eq:dualsym}. The invariance of $\partial D$ under $R$ guarantees a partition 
\begin{equation}
\partial D = \sqcup_{n=1,2,3}C_n, \quad \bx\in C_n \text{ iff } \bx= R^{n-1}\bx' \text{ for some }\bx'\in C_1. 
\end{equation} 
We also have the relations
\begin{equation}\label{eq:Gxy}
\overline{G^f(\bx,\by;\lambda, \bp)} =G^f(\by,\bx;\lambda, \bp) \quad\forall \lambda\in\mathbb R, \forall\bp\in\mathbb R^2,
\end{equation} 
and
\begin{equation}\label{eq:RFmodes}
\begin{aligned}
&\overline{\rho_i(R^{-1}\bx)}\rho_i(R^{-1}\by) =\overline{\rho_i(\bx)}\rho_i(\by),\quad  
i=1,2,\quad 
\overline{\rho_1(R^{-1}\bx)}\rho_2(R^{-1}\by) =\overline{\tau}^2\overline{\rho_1(\bx)}\rho_2(\by),\\
&\overline{\rho_1(\rflc\bx)}\rho_1(\rflc\by) =\overline{\rho_2(\bx)}\rho_2(\by),\quad  
\overline{\rho_1(\rflc\bx)}\rho_2(\rflc\by) =\overline{\rho_2(\bx)}\rho_1(\by).
\end{aligned}
\end{equation}
Since $D=RD = \rflc D$, it is straightforward to verify 
\begin{equation}\label{eq:RFkernels}
\begin{aligned}
\int_{\partial D}\int_{\partial D} K(\bx,\by) \overline{f(\by)}g(\bx)\,ds_\bx \,ds_\by &= \int_{\partial D}\int_{\partial D} K(R^{-1}\bx,R^{-1}\by) \overline{f(R^{-1}\bx)}g(R^{-1}\by)\,ds_\bx \,ds_\by\\
&=\int_{\partial D}\int_{\partial D} K(\rflc\bx,\rflc\by) \overline{f(\rflc\bx)}g(\rflc\by)\,ds_\bx \,ds_\by.
\end{aligned}
\end{equation}

Recall
\begin{equation}
\langle f, T(\eps,\lambda,\bp)g\rangle_{\partial D} := \int_{\partial D}\int_{\partial D} \overline{f(\bx)} G^f(R^{\eps} \bx, R^{\eps} \by;\lambda, \bp)g(\by)\,ds_\bx \,ds_\by,
\end{equation}
where $\langle \cdot,\cdot\rangle_{\partial D}$ represents the $H^{-1/2}(\partial D)$-$H^{1/2}(\partial D)$ pairing. 
\begin{lemma}\label{lem:Tpair}
For all  $\lambda\in\mathbb R$, all quasimomenta $\bp$ with $\ell\in\mathbb R$, $\eps\in\mathbb R$, and $f,g\in H^{-1/2}(\partial D)$,
\begin{equation}
\overline{\langle f, T(\eps,\lambda,\bp)g\rangle}_{\partial D} =\langle g, T(\eps,\lambda,\bp)f\rangle_{\partial D}, \quad \text{and}\quad \langle f, T(\eps,\lambda,\bp)f\rangle_{\partial D}\in\mathbb R.
\end{equation}
\end{lemma}
\begin{proof} 
By \eqref{eq:Gxy}, 
\begin{equation}
\overline{ \int_{\partial D}\int_{\partial D} \overline{f(\bx)} G^f(R^{\eps} \bx, R^{\eps} \by;\lambda, \bp)g(\by)\,ds_\bx \,ds_\by } = \int_{\partial D}\int_{\partial D} f(\bx) G^f(R^{\eps} \by, R^{\eps} \bx;\lambda, \bp) \overline{g(\by)}\,ds_\bx \,ds_\by.
\end{equation}
\end{proof}
Now we are ready to prove Proposition~\ref{lem:Tderiv}.
\begin{proof}[Proposition~\ref{lem:Tderiv}]
Since the derivatives in Proposition~\ref{lem:Tderiv} are all with respect to real variables, by Lemma~\ref{lem:Tpair}, the four matrices in Proposition~\ref{lem:Tderiv} all take the form
\begin{equation} \label{eq:derivgeneral}
\left(\begin{matrix}a&c\\c^*&b \end{matrix}\right),\quad a,b\in\mathbb R, \quad c\in\mathbb C.
\end{equation}

For the first equation, consider the $\lambda$ derivative when $\eps=0$ and $\bp=\Kone$.
We have
\begin{equation} \label{eq:Gderlam}
\begin{aligned}
\partial_\lambda G^f(\bx, \by;\lambda, \Kone) = \frac{1}{|\cC_z|}\sum_{\bm\in\tilde\Lambda^*}\frac{1}{(\lambda - |\bm|^2)^2} e^{\im\bm\cdot (\bx-\by)}. 
\end{aligned}
\end{equation}
We observe that
\begin{equation} \label{eq:lambdaR}
\begin{aligned}
\partial_\lambda G^f(R^{-1}\bx, R^{-1}\by;\lambda, \Kone) 
&= \frac{1}{|\cC_z|}\sum_{\bm\in\tilde\Lambda^*}\frac{1}{(\lambda - |\bm|^2)^2} e^{\im\bm\cdot (\bx-\by)}
= \frac{1}{|\cC_z|}\sum_{\bm\in\tilde\Lambda^*}\frac{1}{(\lambda - |\bm|^2)^2} e^{\im\bm\cdot R^{-1}(\bx-\by)} \\
&= \frac{1}{|\cC_z|}\sum_{\bm\in\tilde\Lambda^*}\frac{1}{(\lambda - |\bm|^2)^2} e^{\im R\bm\cdot (\bx-\by)} 
= \partial_\lambda G^f(\bx, \by;\lambda, \Kone)
\end{aligned}
\end{equation}
and
\begin{equation}\label{eq:lambdaF}
\begin{aligned}
\partial_\lambda G^f(\rflc\bx, \rflc\by;\lambda, \Kone) 
&= \frac{1}{|\cC_z|}\sum_{\bm\in\tilde\Lambda^*}\frac{1}{(\lambda - |\bm|^2)^2} e^{\im\bm\cdot \rflc(\bx-\by)}\\
&= \frac{1}{|\cC_z|}\sum_{\bm\in\tilde\Lambda^*}\frac{1}{(\lambda - |\bm|^2)^2} e^{\im\rflc\bm\cdot (\bx-\by)}
= \partial_\lambda G^f(\bx, \by;\lambda, \Kone).
\end{aligned}
\end{equation}
Here we have used \eqref{eq:dualsym}
 and $|\bm| = |R\bm| = |\rflc\bm|$. 
The diagonal terms are equal because
\begin{equation} 
\begin{aligned}
\int_{\partial D}\int_{\partial D} \partial_\lambda G^f(\bx, \by;\lambda, \Kone)\overline{\rho_1(\bx)}\rho_1(\by)\,ds_\bx \,ds_\by
&=\int_{\partial D}\int_{\partial D} \partial_\lambda G^f(\rflc\bx, \rflc\by;\lambda, \Kone)\overline{\rho_2(\bx)}\rho_2(\by)\,ds_\bx \,ds_\by\\
&=\int_{\partial D}\int_{\partial D} \partial_\lambda G^f(\bx,\by;\lambda, \Kone)\overline{\rho_2(\bx)}\rho_2(\by)\,ds_\bx \,ds_\by.
\end{aligned}
\end{equation}
Here the first equality follows from \eqref{eq:RFmodes} and \eqref{eq:RFkernels}, and the second equaltiy follows from \eqref{eq:lambdaF}.  
The off-diagonal terms are zero follows from the relation below and the fact $\overline{\tau}^2\neq1$:
\begin{equation} 
\begin{aligned}
\int_{\partial D}\int_{\partial D} \partial_\lambda G^f(\bx, \by;\lambda, \Kone)\overline{\rho_1(\bx)}\rho_2(\by)\,ds_\bx \,ds_\by
&=\overline{\tau}^2\int_{\partial D}\int_{\partial D} \partial_\lambda G^f(R^{-1}\bx, R^{-1}\by;\lambda, \Kone)\overline{\rho_1(\bx)}\rho_2(\by)\,ds_\bx \,ds_\by\\
&=\overline{\tau}^2\int_{\partial D}\int_{\partial D} \partial_\lambda G^f(\bx,\by;\lambda, \Kone)\overline{\rho_1(\bx)}\rho_2(\by)\,ds_\bx \,ds_\by.
\end{aligned}
\end{equation}
Here we have used \eqref{eq:RFmodes}, \eqref{eq:RFkernels} and \eqref{eq:lambdaR}.

For the second equation, consider the $\nabla_\bp$ derivative when $\eps=0$ around $\Kone$ 
\begin{equation} \label{eq:Gderbeta1}
\begin{aligned}
\nabla_\bp G^f(\bx, \by;\lambda, \bp)|_{\bp=\Kone} =& -\frac{1}{|\cC_z|}\sum_{[\bm]\in[\tilde\Lambda^*]}\frac{2}{(\lambda - |\bm|^2)^2} e^{\im \bm\cdot (\bx-\by)} \bm-\frac{1}{|\cC_z|}\sum_{[\bm]\in[\tilde\Lambda^*]} \frac{\im}{\lambda - |\bm|^2}e^{\im \bm\cdot (\bx-\by)})(\bx-\by)
\end{aligned}
\end{equation}
We only need to show that the diagonal terms are zero. 
Decompose the integral domain into
\begin{equation}
\begin{aligned}
\partial D\times\partial D =&\left((C_1\times C_1)\sqcup (C_2\times C_2) \sqcup (C_3\times C_3)\right) \\ 
&\sqcup \left((C_1\times C_2)\sqcup (C_2\times C_3) \sqcup (C_3\times C_2)\right) \\
&\sqcup \left((C_1\times C_3)\sqcup (C_2\times C_1) \sqcup (C_3\times C_2)\right).
\end{aligned}
\end{equation}
Define the vectors 
\begin{equation}
I_{i,j}:=\int_{C_i}\int_{C_j} \overline{\rho_1(\bx)} \nabla_\bp G^f(R^{\eps} \bx, R^{\eps} \by;\lambda, \bp)\rho_1(\by)\,ds_\bx \,ds_\by,\quad i,j=1,2,3.
\end{equation}
We have 
\begin{equation} 
\begin{aligned}
&I_{2,2} = \\
&-\frac{1}{|\cC_z|}\sum_{\bm\in\tilde\Lambda^*}\int_{C_2\times C_2} \left( \frac{2}{(\lambda - |\bm|^2)^2} e^{\im \bm\cdot (\bx-\by)} \bm 
+ \frac{\im}{\lambda - |\bm|^2}e^{\im \bm\cdot (\bx-\by)}(\bx-\by)1 \right) \overline{\rho_1(\bx)}\rho_1(\by) \,ds_\bx \,ds_\by\\
&= -\frac{1}{|\cC_z|}\sum_{\bm\in\tilde\Lambda^*}\int_{C_1\times C_1} \left( \frac{2}{(\lambda - |\bm|^2)^2} e^{\im \bm\cdot R(\bx'-\by')} \bm
+ \frac{\im}{\lambda - |\bm|^2}e^{\im \bm\cdot R(\bx'-\by')}R(\bx'-\by')\right) \overline{\rho_1(R\bx')}\rho_1(R\by') \,ds_{\bx'} \,ds_{\by'}\\
&= -\frac{1}{|\cC_z|}\sum_{\bm\in\tilde\Lambda^*}\int_{C_1\times C_1} \left(\frac{2}{(\lambda - |\bm|^2)^2} e^{\im R^2\bm\cdot (\bx'-\by')} \bm + \frac{\im}{\lambda - |\bm|^2}e^{\im R^2\bm\cdot (\bx'-\by')}R(\bx'-\by') \right)\overline{\rho_1(\bx')}\rho_1(\by') \,ds_{\bx'} \,ds_{\by'}\\
&= -\frac{1}{|\cC_z|}\sum_{\bm'\in\tilde\Lambda^*}\int_{C_1\times C_1} \left(\frac{2}{(\lambda - |\bm'|^2)^2} e^{\im \bm'\cdot (\bx-\by)} R\bm'  + \frac{\im}{\lambda - |\bm'|^2}e^{\im \bm'\cdot (\bx-\by)}R(\bx-\by) \right)\overline{\rho_1(\bx')}\rho_1(\by') \,ds_\bx \,ds_\by\\
&=RI_{1,1}.
\end{aligned}
\end{equation}
Similarly we have $I_{3,3} = R^2I_{1,1}$ and thus $I_{1,1}+I_{2,2}+I_{3,3}=0$ by  
$I+R+R^2=0$ on vectors as stated in \eqref{eq:oprelat}. 
Similarly, $I_{1,2}+I_{2,3}+I_{3,1}=I_{1,3}+I_{2,1}+I_{3,2}=0$. Thus $\langle \rho_1, \nabla_\bp G^f(\bx, \by;\lambda, \bp)|_{\bp=\Kone}\rho_1\rangle_{\partial D=0}$. The same method gives  $\langle \rho_2, \nabla_\bp G^f(\bx, \by;\lambda, \bp)|_{\bp=\Kone}\rho_2\rangle_{\partial D=0}$.

For the third equation in \eqref{eq:Tderiv}, we only need to show 
\begin{equation}
\bbeta_2\cdot\langle \rho_1, \nabla_\bp G^f(\bx, \by;\lambda, \bp)|_{\bp=\Kone}\rho_2\rangle_{\partial D=0}
=\tau \bbeta_1\cdot\langle \rho_1, \nabla_\bp G^f(\bx, \by;\lambda, \bp)|_{\bp=\Kone}\rho_2\rangle_{\partial D=0}
\end{equation}
This is true since
\begin{equation} 
\begin{aligned}
& \sum_{\bm\in\tilde\Lambda^*}\int_{\partial D\times \partial D} \left( \frac{2}{(\lambda - |\bm|^2)^2} e^{\im \bm\cdot (\bx-\by)} \bm \cdot\bbeta_2 
+ \frac{\im}{\lambda - |\bm|^2}e^{\im \bm\cdot (\bx-\by)}(\bx-\by)\cdot\bbeta_2 \right) \overline{\rho_1(\bx)}\rho_2(\by) \,ds_\bx \,ds_\by\\
& =\sum_{\bm\in\tilde\Lambda^*}\int_{\partial D\times \partial D} \left( \frac{2}{(\lambda - |\bm|^2)^2} e^{\im \bm\cdot (\bx-\by)} \bm \cdot R^{-1}\bbeta_1
+ \frac{\im}{\lambda - |\bm|^2}e^{\im \bm\cdot (\bx-\by)}(\bx-\by)\cdot R^{-1}\bbeta_1\right) \overline{\rho_1(\bx)}\rho_2(\by) \,ds_\bx \,ds_\by\\
& = \sum_{\bm\in\tilde\Lambda^*}\int_{\partial D\times \partial D} \left( \frac{2}{(\lambda - |\bm|^2)^2} e^{\im R\bm\cdot R(\bx-\by)} R\bm \cdot\bbeta_1
+ \frac{\im}{\lambda - |\bm|^2}e^{\im R\bm\cdot R(\bx-\by)}R(\bx-\by)\cdot R\bbeta_1\right) \overline{\rho_1(R^{-1}R\bx)}\rho_2(R^{-1}R\by) \,ds_\bx \,ds_\by\\
& = \sum_{\bm'\in\tilde\Lambda^*}\int_{\partial D\times \partial D} \left( \frac{2}{(\lambda - |\bm'|^2)^2} e^{\im \bm'\cdot (\bx'-\by')} \bm' \cdot\bbeta_1
+ \frac{\im}{\lambda - |\bm'|^2}e^{\im \bm'\cdot (\bx'-\by')}(\bx'-\by')\cdot R\bbeta_1\right) \overline{\tau\rho_1(\bx')}\bar\tau\rho_2(\by') \,ds_{\bx'} \,ds_{\by'}.
\end{aligned}
\end{equation}
Here we have used \eqref{eq:Gderbeta1}, $\bbeta_2 = R^{-1}\bbeta_1$ and $(\bar\tau)^2=\tau$.

For the fourth equation, consider the  $\eps$ derivative when $\bp=\Kone$ around $\eps=0$. Define
\begin{equation} \label{eq:Gdereps}
\begin{aligned}
K(\bx,\by):=\partial_\eps G^f(R^\eps\by, R^\eps\bx;\lambda, \Kone)|_{\eps=0} =-\frac{1}{|\cC_z|}\sum_{\bm\in\tilde\Lambda^*} \frac{\im}{\lambda - |\bm|^2}e^{\im \bm\cdot (\bx-\by)}\bm\cdot J(\bx-\by).
\end{aligned}
\end{equation}
We verify 
\begin{equation} \label{eq:epsR}
\begin{aligned}
K(R^{-1}\bx, R^{-1}\by) 
&= -\frac{1}{|\cC_z|}\sum_{\bm\in\tilde\Lambda^*} \frac{\im}{\lambda - |\bm|^2}e^{\im \bm\cdot R^{-1}(\bx-\by)}\bm\cdot JR^{-1}(\bx-\by)\\
&=-\frac{1}{|\cC_z|}\sum_{\bm\in\tilde\Lambda^*} \frac{\im}{\lambda - |\bm|^2}e^{\im R\bm\cdot (\bx-\by)}R\bm\cdot J(\bx-\by)
= K(\bx,\by) 
\end{aligned}
\end{equation}
and
\begin{equation}\label{eq:epsF}
\begin{aligned}
K(\rflc\bx, \rflc\by)
&= -\frac{1}{|\cC_z|}\sum_{\bm\in\tilde\Lambda^*} \frac{\im}{\lambda - |\bm|^2}e^{\im \bm\cdot \rflc(\bx-\by)}\bm\cdot J\rflc(\bx-\by)\\
&= \frac{1}{|\cC_z|} \sum_{\bm\in\tilde\Lambda^*} \frac{\im}{\lambda - |\bm|^2}e^{\im \rflc\bm\cdot (\bx-\by)}\rflc\bm\cdot J(\bx-\by)
= -K(\bx,\by).
\end{aligned}
\end{equation}
Here we have used \eqref{eq:oprelat}, \eqref{eq:dualsym} and $|\bm| = |R\bm| = |\rflc\bm|$.
Thus a similar argument as that for the first equation gives opposite diagonal terms are zero off-diagonal terms.
\end{proof}

\begin{proof}[Proof of Proposition~\ref{lem:TderivA}]
A linear combination of the second and third equations in \eqref{eq:Tderiv} gives \eqref{eq:TderivA}, since $\tilde\bbeta_1 = \bbeta_1 - \bbeta_2$.

Observe
\begin{equation}\label{}
\Ktwo + \Lambda^* = -\Kone +\Lambda^*.
\end{equation} 
Using \eqref{eq:Gderlam}, \eqref{eq:Gderbeta1} and \eqref{eq:Gdereps}, we know that 
\begin{equation}\label{eq:TderivZKK}
\begin{aligned}
\partial_\lambda G^f(\by, \bx;\lambda, \Ktwo) &= \overline{\partial_\lambda G^f(\by, \bx;\lambda, \Kone)}, \\
\nabla_\bp G^f(\by, \bx;\lambda, \bp)|_{\bp=\Ktwo} &= -\overline{ \nabla_\bp G^f(\by, \bx;\lambda, \bp)|_{\bp=\Kone} }\\
\partial_\eps G^f(R^\eps\by, R^\eps\bx;\lambda, \Ktwo)|_{\eps=0} &=
\overline{\partial_\eps G^f(R^\eps\by, R^\eps\bx;\lambda, \Kone)|_{\eps=0}}.
\end{aligned}
\end{equation}
Using $\rho_1'(\bx)= \overline{\rho_2(\bx)}$ and $\rho_2'(\bx) = \overline{\rho_1(\bx)}$, we obtain 
\begin{equation}
\langle \rho_1', \nabla_\bp G^f(\bx, \by;\lambda, \bp)|_{\bp=\Ktwo}\rho_2'\rangle_{\partial D=0}
= - \overline{\langle \rho_2, \nabla_\bp G^f(\bx, \by;\lambda, \bp)|_{\bp=\Kone}\rho_1\rangle_{\partial D=0}} = -\overline{\theta_*}.
\end{equation}
This finishes the proof of \eqref{eq:TderivKtwo}.
\end{proof}

\section{Proof of Proposition~\ref{lem:oplim}}
\label{sec:oplim}
In this subsection, all $\langle \cdot, \cdot\rangle$ pairings represent the $\HhalfG$-$\HmhalfG$ pairing.
We prove Proposition~\ref{lem:oplim} in this appendix. Recall the definitions of $\mathbb T^\eps(\lambda)$ in \eqref{eq:bTeps} and $\mathbb U_{\pm}(h)$ in \eqref{eq:ci} to \eqref{eq:Upm}. We are claiming that
\begin{equation}\label{eq:Sv}
\begin{aligned}
\left\|\cS^{\pm\eps}(\lambda_*+\eps h)\phi - \left(\tilde\cS^0(\lambda_*)\phi 
+\beta(h)\overline{\langle \phi, v_1\rangle} v_1
+\right.\right.&\beta(h)\overline{\langle \phi, v_2\rangle} v_2
\mp \xi(h)\overline{\langle \phi, v_1\rangle} v_2 \\
&\left.\left.
\mp \xi(h)\overline{\langle \phi, v_2\rangle} v_1
\right) \right\|_{\HhalfG} /\|\phi\|_{\HmhalfG} \to 0.
\end{aligned}
\end{equation}
\begin{equation}
\begin{aligned}
\left\|\cK^{\pm\eps}(\lambda_*+\eps h)\psi - \left(\tilde\cK^0(\lambda_*)\psi 
+\right.\right.&\beta(h)\langle \partial_n v_1, \psi \rangle v_1
+\beta(h)\langle \partial_n v_2, \psi \rangle v_2\\
&\left.\left.
\mp \xi(h)\langle \partial_n v_1, \psi \rangle v_2
\mp \xi(h)\langle \partial_n v_2, \psi \rangle v_1
\right)
\right\|_{\HhalfG} /\|\psi\|_{\HhalfG} \to 0,
\end{aligned}
\end{equation}
\begin{equation}
\begin{aligned}
\left\|\cK^{*,\pm\eps}(\lambda_*+\eps h)\phi - \left(\tilde\cK^{*,0}(\lambda_*)\phi 
+\right.\right.&\beta(h)\overline{\langle \phi, v_1\rangle} \partial_n v_1
+\beta(h)\overline{\langle \phi, v_2\rangle} \partial_n v_2\\
&\left.\left.
\mp \xi(h)\overline{\langle \phi, v_1\rangle} \partial_n v_2
\mp \xi(h)\overline{\langle \phi, v_2\rangle} \partial_n v_1
\right) \right\|_{\HmhalfG} /\|\phi\|_{\HmhalfG} \to 0,
\end{aligned}
\end{equation}
\begin{equation}
\begin{aligned}
\left\|\cN^{\pm\eps}(\lambda_*+\eps h)\psi - \left(\tilde\cN^0(\lambda_*)\psi 
+\right.\right.&\beta(h)\langle \partial_n v_1, \psi \rangle \partial_n v_1
+\beta(h)\langle \partial_n v_2, \psi \rangle \partial_n v_2\\
&\left.\left.
\mp \xi(h)\langle \partial_n v_1, \psi \rangle \partial_n v_2
\mp \xi(h)\langle \partial_n v_2, \psi \rangle \partial_n v_1
\right) \right\|_{\HmhalfG} /\|\psi\|_{\HmhalfG} \to 0,
\end{aligned}
\end{equation}

In subsections~\ref{sec:nearDirac}-\ref{sec:Sn3}, we focus on $\cS^{\pm\eps}(\lambda_*+\eps h)$ in \eqref{eq:Sv}.
Using the representations of
$\cS^{\pm\eps}(\lambda_*+\eps h)$ in \eqref{eq:Seps} and $\tilde\cS^{0} (\lambda)$ in \eqref{eq:tildeS0}, we will break the integral in $\cS^{\pm\eps}(\lambda_*+\eps h)$ into three parts: (a) near the Dirac point given by \eqref{eq:nearDiracF},  (b) the first two bands away from the Dirac point given by \eqref{eq:Sn12}, and (c) higher bands given by \eqref{eq:Sn3}.
The convergence of the other three operators are verified similarly, and 
are given by Lemma~\ref{lem:NearD}, Lemma~\ref{lem:n12}, Corollary~\ref{lem:Kstarn3} and Subsection~\ref{sec:Dn3}.

\subsection{Near the Dirac point}\label{sec:nearDirac}
\begin{lemma}
The following convergence holds in the operator norm from $\HmhalfG$ to $\HhalfG$ 
uniformly for $h\in\mathbb C$ that satisfy $|h|< \hrad$, 
as $\eps\to 0^+$: 
\begin{equation}\label{eq:nearDiracF}
\begin{aligned}
\sum_{n=1,2}&\frac{1}{2\pi}\int_{[-\eps^{1/3},\eps^{1/3} ]} 
\frac{\overline{\langle \phi, v_{n,\pm\eps}(\cdot;\bp(\ell))\rangle} v_{n,\pm\eps}(\bx;\bp(\ell))}{ \mu_{n,\pm\eps}(\bp(\ell)) - (\lambda_*+\eps h)} \, \dpt\\
&\to 
\beta(h)\overline{\langle \phi, v_1\rangle} v_1
+\beta(h)\overline{\langle \phi, v_2\rangle} v_2
\mp \xi(h)\overline{\langle \phi, v_1\rangle} v_2
\mp \xi(h)\overline{\langle \phi, v_2\rangle} v_1.
\end{aligned}
\end{equation}
\end{lemma}
We prove this lemma by proving the next two lemmas. 
\begin{lemma}
The following convergence holds in the operator norm from $\HmhalfG$ to $\HhalfG$ 
uniformly for $h\in\mathbb C$ that satisfy $|h|< \hrad$, 
as $\eps\to 0^+$: 
\begin{equation}\label{eq:nearDirac}
\sum_{n=1,2}\frac{1}{2\pi}\int_{[-\eps^{1/3},\eps^{1/3} ]} 
\frac{\overline{\langle \phi, v_{n,\pm\eps}(\cdot;\bp(\ell))\rangle} v_{n,\pm\eps}(\bx;\bp(\ell))}{ \mu_{n,\pm\eps}(\bp(\ell)) - (\lambda_*+\eps h)} \, \dpt
\to a_{\pm}(h) \overline{\langle \phi, w_1\rangle} w_1 + b_{\pm}(h)   \overline{\langle \phi, w_2\rangle} w_2. 
\end{equation}
where 
\begin{equation}\label{eq:abpm}
a_+(h) = b_-(h)=  f_1(h) + f_4(h), \quad 
b_+(h) = a_-(h)= f_2(h) + f_3(h).
\end{equation}
and 
\begin{equation}\label{eq:fgpm}
\begin{aligned}
f_1(h)&=\frac{1}{2\pi}\int_{\mathbb R} \frac{1}{- \frac{1}{|\gamma_*|}\sqrt{ t_*^2 + |\theta_*|^2 \ell^2} - h} \cdot \frac{1}{1+|L(1,\ell)|^2}\, d\ell,\\
f_2(h) &= \frac{1}{2\pi}\int_{\mathbb R} \frac{1}{-\frac{1}{|\gamma_*|}\sqrt{t_*^2 + |\theta_*|^2 \ell^2}-h} \cdot \frac{|L(1,\ell)|^2}{1+|L(1,\ell)|^2}\, d\ell,\\
f_3(h) &=\frac{1}{2\pi}\int_{\mathbb R} \frac{1}{\frac{1}{|\gamma_*|}\sqrt{t_*^2 + |\theta_*|^2 \ell^2} - h} \cdot \frac{1}{1+|L(1,\ell)|^2}\, d\ell.\\
f_4(h) &= \frac{1}{2\pi}\int_{\mathbb R} \frac{1}{\frac{1}{|\gamma_*|}\sqrt{ t_*^2 + |\theta_*|^2 \ell^2} - h } \cdot \frac{|L(1,\ell)|^2}{1+|L(1,\ell)|^2}\, d\ell,
\end{aligned}
\end{equation}
and $L(eps,\ell) = L(\eps,\ell,0)$ and $L(\eps,\ell,\mu)$ is defined in \eqref{eq:Lmu}. 
Note that the individual integrals in \eqref{eq:fgpm} are all divergent, but $f_1(h) + f_4(h)$ and $f_2(h) +f_3(h)$ are convergent. 
\end{lemma}

\begin{proof}
Define
\begin{equation}
I_i^{\pm\eps}(h):=\frac{1}{2\pi}\int_{[-\eps^{1/3},\eps^{1/3} ]} 
\frac{\overline{\langle \phi, v_{i,\pm\eps}(\cdot;\bp(\ell))\rangle} v_{i,\pm\eps}(\bx;\bp(\ell))}{ \mu_{i,\pm\eps}(\bp(\ell)) - (\lambda_*+\eps h)} \, \dpt. 
\end{equation}
Using \eqref{eq:energ}, \eqref{eq:pmodesL}, Remark~\ref{lem:phase} and
\begin{equation}
\begin{aligned}
0\leq L(\eps,\ell)&\leq 1,\\
\frac{1}{|\gamma_*|}\sqrt{\eps^2 t_*^2 + |\theta_*|^2 \ell^2}(1+O(\eps,\ell))\pm \eps h &= \left(\frac{1}{|\gamma_*|}\sqrt{\eps^2 t_*^2 + |\theta_*|^2 \ell^2}\pm \eps h\right)(1+O(\eps,\ell)), \quad |h|< \hrad,\\
\frac{1}{|1+|L(\eps,\ell)|^2+O(\eps,\ell)|} &= \frac{1}{|1+|L(\eps,\ell)|^2|} (1+O(\eps,\ell)),
\end{aligned}
\end{equation}
we obtain
\begin{equation}
\begin{aligned}
&I_1^\eps(h) =
-\frac{1}{2\pi}\int_{[-\eps^{1/3},\eps^{1/3} ]} \\
&\frac{
\overline{\langle \phi, w_1\rangle}w_1 + \overline{L(\eps,\ell)} \overline{\langle \phi, w_2 \rangle} w_1 +  L(\eps,\ell)\overline{\langle \phi, w_1\rangle}w_2 + |L(\eps,\ell)|^2 \overline{\langle \phi, w_2 \rangle} w_2
+ O(\eps,\ell)\|\phi\|_{\HmhalfG}}{ (\frac{1}{|\gamma_*|}\sqrt{\eps^2 t_*^2 + |\theta_*|^2 \ell^2}+\eps h)(|1+|L(\eps,\ell)|^2|)}(1+O(\eps,\ell)) \, \dpt. 
\end{aligned}
\end{equation}
The fact that 
$\frac{L(\eps,\ell)}{ (\frac{1}{|\gamma_*|}\sqrt{\eps^2 t_*^2 + |\theta_*|^2 \ell^2}\pm \eps h)}$ is odd in $\ell$ implies that
\begin{equation}
\begin{aligned}
&-\frac{1}{2\pi}\int_{[-\eps^{1/3},\eps^{1/3} ]} 
\frac{
\overline{\langle \phi, w_1\rangle}w_1 + \overline{L(\eps,\ell)} \overline{\langle \phi, w_2 \rangle} w_1 +  L(\eps,\ell)\overline{\langle \phi, w_1\rangle}w_2 + |L(\eps,\ell)|^2 \overline{\langle \phi, w_2 \rangle} w_2
}{ (\frac{1}{|\gamma_*|}\sqrt{\eps^2 t_*^2 + |\theta_*|^2 \ell^2}+\eps h)(|1+|L(\eps,\ell)|^2|)} \, \dpt\\
&=-\frac{1}{2\pi}\int_{[-\eps^{1/3},\eps^{1/3} ]} 
\frac{
\overline{\langle \phi, w_1\rangle}w_1 + |L(\eps,\ell)|^2 \overline{\langle \phi, w_2 \rangle} w_2
}{ (\frac{1}{|\gamma_*|}\sqrt{\eps^2 t_*^2 + |\theta_*|^2 \ell^2}+\eps h)(|1+|L(\eps,\ell)|^2|)} \, \dpt\\
\end{aligned}
\end{equation}
Using 
\begin{equation}
L(\eps,\ell)=L(1,\ell/\eps),\quad 
\frac{1}{(\frac{1}{|\gamma_*|}\sqrt{\eps^2 t_*^2 + |\theta_*|^2 \ell^2}+\eps h)}\dpt = \frac{1}{(\frac{1}{|\gamma_*|}\sqrt{t_*^2 + |\theta_*|^2 (\ell/\eps)^2}+h)}\dpt/\eps,
\end{equation}
we see that $f_1(h)$ and $f_2(h)$ emerge as the coefficients of $\overline{\langle \phi, w_1\rangle}w_1$ and $\overline{\langle \phi, w_2\rangle}w_2$ as~$\eps\to0$.
A similar argument for $I_2^\eps(h)$ gives that
\begin{equation}
I_1^\eps(h) + I_2^\eps(h) =I^{\eps}(h)+ \int_{[-\eps^{1/3},\eps^{1/3} ]} 
\frac{O(\eps,\ell)\|\phi\|_{\HmhalfG}}{ (\frac{1}{|\gamma_*|}\sqrt{\eps^2 t_*^2 + |\theta_*|^2 \ell^2}+\eps h)(|1+|L(\eps,\ell)|^2|)} \, \dpt, 
\end{equation}
where 
\begin{equation}
I^{\eps}(h) \to  f_1(h)\overline{\langle \phi, w_1\rangle}w_1 + f_2(h) \overline{\langle \phi, w_2 \rangle} w_2 + f_4(h) \overline{\langle \phi, w_1\rangle}w_1 + f_3(h) \overline{\langle \phi, w_2 \rangle} w_2,\quad \text{as }\eps\to0.
\end{equation}
Thus to establish the $+\eps$ identity in \eqref{eq:nearDirac}, we only need to show
\begin{equation}
\int_{[-\eps^{1/3},\eps^{1/3} ]} 
\frac{O(\eps,\ell)}{ (\frac{1}{|\gamma_*|}\sqrt{\eps^2 t_*^2 + |\theta_*|^2 \ell^2}+\eps h)(|1+|L(\eps,\ell)|^2|)} \, \dpt \to 0, \text{ as }\eps\to 0. 
\end{equation}
This is true because $O(\eps,\ell) = O(\eps^{1/3})$ within the integral domain, the integral domain is of size $\eps^{1/3}$, and 
\begin{equation}
\int_{[-\eps^{1/3},\eps^{1/3} ]} \frac{1}{\frac{1}{|\gamma_*|}\sqrt{\eps^2 t_*^2 + |\theta_*|^2 \ell^2}} \, \dpt =|\frac{\gamma_*}{\theta_*}|\int_{[-|\frac{\theta_*}{t_*}|\eps^{-2/3},|\frac{\theta_*}{t_*}|\eps^{-2/3} ]} 
\frac{1}{\sqrt{1+x^2}} \, dx
= O(\ln \eps).
\end{equation}
Here we have used $\int\frac{1}{\sqrt{1+x^2}} \, dx = \ln(x+\sqrt{x^2+1})$.

The $-\eps$ identity in \eqref{eq:nearDirac} can be shown similarly.
\end{proof}

\begin{lemma}
The coefficients defined in \eqref{eq:abpm} are equal to
\begin{equation}\label{eq:onesidecoeff}
a_+(h)=b_-(h){=  \beta(h) - \xi(h)} ,\quad
b_+(h) =a_-(h){ = \beta(h) + \xi(h) },
\end{equation}
where $\beta(h)$ and $\xi(h)$ are defined in \eqref{eq:betaxi}. There holds the identity
\begin{equation}\label{eq:nearDiracre}
a_{\pm}(h) \overline{\langle \phi, w_1\rangle} w_1 + b_{\pm}(h)   \overline{\langle \phi, w_2\rangle} w_2
=\beta(h)\overline{\langle \phi, v_1\rangle} v_1
+\beta(h)\overline{\langle \phi, v_2\rangle} v_2
\mp \xi(h)\overline{\langle \phi, v_1\rangle} v_2
\mp \xi(h)\overline{\langle \phi, v_2\rangle} v_1
\end{equation}
\end{lemma}
\begin{proof}
We first derive the simplified forms of the coefficients $a_{\pm}(h)$ and $b_{\pm}(h)$.
Direct calculation shows
\begin{equation}
\begin{aligned}
2\pi(f_1(h)+f_4(h))&=\int_{\mathbb R} \frac{1}{ -\frac{1}{|\gamma_*|}\sqrt{ t_*^2 + |\theta_*|^2 \ell^2} - h} \cdot \frac{1}{1+|L(1,\ell)|^2}\, d\ell
+\int_{\mathbb R} \frac{1}{\frac{1}{\gamma_*}\sqrt{ t_*^2 + |\theta_*|^2 \ell^2} - h } \cdot \frac{|L(1,\ell)|^2}{1+|L(1,\ell)|^2}\, d\ell,\\
&= - \int_{\mathbb R} \frac{ 1- |L(1,\ell)|^2 }{ 1+ |L(1,\ell)|^2 } \cdot  \frac{\frac{1}{\gamma_*}\sqrt{ t_*^2 + |\theta_*|^2 \ell^2}}{ \frac{1}{(\gamma_*)^2}(t_*^2 + |\theta_*|^2 \ell^2) - h^2} d\ell
+ \int_{\mathbb R} \frac{h}{ \frac{1}{(\gamma_*)^2}(t_*^2 + |\theta_*|^2 \ell^2) - h^2} d\ell
\end{aligned}
\end{equation}
Using 
\begin{equation}
\begin{aligned}
1- |L(1,\ell)|^2 &= \frac{2t_*}{t+\sqrt{t_*^2+|\theta_*|^2\ell^2}},\\
1+ |L(1,\ell)|^2 &= \frac{2\sqrt{t_*^2+|\theta_*|^2\ell^2}}{t+\sqrt{t_*^2+|\theta_*|^2\ell^2}},\\
\frac{ 1- |L(1,\ell)|^2 }{ 1+ |L(1,\ell)|^2 } &= \frac{t_*}{\sqrt{t_*^2+|\theta_*|^2\ell^2}}, \\
\int_{\mathbb R}\frac{1}{ \frac{1}{(\gamma_*)^2}(t_*^2 + |\theta_*|^2 \ell^2) - h^2} d\ell &= \frac{|\gamma_*|}{|\theta_*|}\frac{1}{\sqrt{(\frac{t_*}{\gamma_*})^2-h^2}} \pi,
\end{aligned}
\end{equation}
we obtain the first equation in \eqref{eq:onesidecoeff}. The second equation in \eqref{eq:onesidecoeff} can be similarly obtained.

The relation \eqref{eq:nearDiracre} can be obtained using the relation between $w_i$ and $v_i$ \eqref{eq:winv}.
\end{proof}

Using the same method, we obtain the convergence of the first two bands close to the Dirac point in the other operators.
\begin{lemma}\label{lem:NearD}
The following convergences hold in the operator norm  from $\HmhalfG$ to $\HhalfG$, from $\HhalfG$ to $\HhalfG$ and from $\HhalfG$ to $\HmhalfG$ respectively
uniformly for $h\in\mathbb C$ that satisfy $|h|<\hrad$, 
as $\eps\to 0^+$: 
\begin{equation}\label{eq:Kstarn3}
\begin{aligned}
\sum_{n=1,2}\frac{1}{2\pi}\int_{[-\eps^{1/3},\eps^{1/3}]}&
\frac{\overline{ \langle v_{n,\pm\eps}(\cdot;\bp(\ell)),\psi\rangle} \partial_n v_{n,\pm\eps}(\bx;\bp(\ell))}{ \mu_{n,\pm\eps}(\bp(\ell)) - \lambda_*- \eps h} \, \dpt \\
&\to
\beta(h)\langle \partial_n v_1, \psi \rangle v_1
+\beta(h)\langle \partial_n v_2, \psi \rangle v_2
\mp \xi(h)\langle \partial_n v_1, \psi \rangle v_2
\mp \xi(h)\langle \partial_n v_2, \psi \rangle v_1
\end{aligned}
\end{equation}
\begin{equation}
\begin{aligned}
\sum_{n=1,2}\frac{1}{2\pi}\int_{[-\eps^{1/3},\eps^{1/3}]}&
\frac{\langle \partial_n v_{n,\pm\eps}(\cdot;\bp(\ell),\psi)\rangle v_{n,\pm\eps}(\bx;\bp(\ell))}{ \mu_{n,\pm\eps}(\bp(\ell)) - \lambda_*- \eps h} \, \dpt
\to\\
&\beta(h)\overline{\langle \phi, v_1\rangle} \partial_n v_1
+\beta(h)\overline{\langle \phi, v_2\rangle} \partial_n v_2
\mp \xi(h)\overline{\langle \phi, v_1\rangle} \partial_n v_2
\mp \xi(h)\overline{\langle \phi, v_2\rangle} \partial_n v_1
\end{aligned}
\end{equation}
\begin{equation}
\begin{aligned}
\partial_n\sum_{n=1,2}\frac{1}{2\pi}\int_{[-\eps^{1/3},\eps^{1/3}]}&
\frac{\langle \partial_n v_{n,\pm\eps}(\cdot;\bp(\ell),\psi)\rangle v_{n,\pm\eps}(\bx;\bp(\ell))}{ \mu_{n,\pm\eps}(\bp(\ell)) - \lambda_*- \eps h} \, \dpt
\to\\
&\beta(h)\langle \partial_n v_1, \psi \rangle \partial_n v_1
+\beta(h)\langle \partial_n v_2, \psi \rangle \partial_n v_2
\mp \xi(h)\langle \partial_n v_1, \psi \rangle \partial_n v_2
\mp \xi(h)\langle \partial_n v_2, \psi \rangle \partial_n v_1
\end{aligned}
\end{equation}
\end{lemma}

\subsection{The first two bands away from the Dirac point}\label{sec:Sn12}
\begin{lemma}
The following convergence holds in the operator norm from $\HmhalfG$ to $\HhalfG$ 
uniformly for $h\in\mathbb C$ that satisfy $|h|< \hrad$, 
as $\eps\to 0$: 
\begin{equation}\label{eq:Sn12}
\sum_{n=1,2}\frac{1}{2\pi}\int_{[-\pi,-\eps^{1/3}]\cup[\eps^{1/3},\pi]}
\frac{\overline{ \langle \phi, v_{n,\pm\eps}(\cdot;\bp(\ell))\rangle} v_{n,\pm\eps}(\bx;\bp(\ell))}{ \mu_{n,\pm\eps}(\bp(\ell)) - \lambda_* - \eps h} \, \dpt
\to
\sum_{n=1,2}\frac{1}{2\pi}\text{p.v.}\int_{[-\pi,\pi]} \frac{\overline{ \langle \phi, v_{n}(\cdot;\bp(\ell))\rangle} v_{n}(\bx;\bp(\ell))}{ \mu_{n}(\bp(\ell)) - \lambda_* } \, \dpt.
\end{equation}
\end{lemma}
\begin{proof}
We show that among the following five quantities, the difference between each adjacent pair converges to zero in operator norm as $\eps\to 0$. The five quantities are 
\begin{equation}
\begin{aligned}
I_1\phi&:=\sum_{n=1,2}\lim_{\delta\to0^+}\frac{1}{2\pi}\int_{[-\pi,-\delta]\cup[\delta,\pi]} \frac{\overline{ \langle \phi, v_{n}(\cdot;\bp(\ell))\rangle} v_{n}(\bx;\bp(\ell))}{\mu_{n}(\bp(\ell)) - \lambda_*} \, \dpt,\\
I_2\phi&:=\sum_{n=1,2}\frac{1}{2\pi}\int_{[-\pi,-\eps^{1/3}]\cup[\eps^{1/3},\pi]} \frac{\overline{ \langle \phi, v_{n}(\cdot;\bp(\ell))\rangle} v_{n}(\bx;\bp(\ell))}{\mu_{n}(\bp(\ell)) - \lambda_*} \, \dpt,\\
I_3\phi&:=\sum_{n=1,2}\frac{1}{2\pi}\int_{[-\pi,-\eps^{1/3}]\cup[\eps^{1/3},\pi]} \frac{\overline{ \langle \phi, u_{n}(\cdot;\bp(\ell))\rangle} u_{n}(\bx;\bp(\ell))}{\lambda_{n}(\bp(\ell)) - \lambda_*} \, \dpt,\\
I_4\phi&:=\sum_{n=1,2}\frac{1}{2\pi}\int_{[-\pi,-\eps^{1/3}]\cup[\eps^{1/3},\pi]}
\frac{\overline{ \langle \phi, u_{n,\pm\eps}(\cdot;\bp(\ell))\rangle} u_{n,\pm\eps}(\bx;\bp(\ell))}{\lambda_{n,\pm\eps}(\bp(\ell)) - \lambda_*- \eps h} \, \dpt,\\
I_5\phi&:=\sum_{n=1,2}\frac{1}{2\pi}\int_{[-\pi,-\eps^{1/3}]\cup[\eps^{1/3},\pi]}
\frac{\overline{ \langle \phi, v_{n,\pm\eps}(\cdot;\bp(\ell))\rangle} v_{n,\pm\eps}(\bx;\bp(\ell))}{\mu_{n,\pm\eps}(\bp(\ell)) - \lambda_*- \eps h} \, \dpt.
\end{aligned}
\end{equation}

For $I_1-I_2\to 0$, we define $f(\bx,\ell):= \overline{ \langle \phi, v_1(\cdot;\bp(\ell))\rangle} v_1(\bx;\bp(\ell))$. From the analyticity of $v_1(\cdot;\bp(\ell))$ in $\ell$ in a neighborhood of $\mathbb R$ as $H^1(\cunp)$ functions as stated above \eqref{eq:muv}, we know $\| \frac{d}{d\ell} v_1(\cdot;\bp(\ell))\|_{H^1(\cunp)}$ is bounded on $\ell\in[0,1]$, thus
\begin{equation}
\begin{aligned}
\| f(\cdot,\ell)|\|_{\HhalfG} &\leq 
 \|\phi\|_{\HmhalfG}\max_{|\ell|\leq1} (\| v_1(\cdot,\ell)\|_{H^1(\cunp)} \|v_1(\cdot,\ell)|\|_{H^1(\cunp)}),\\
\|\frac{d}{d\ell}f(\cdot, \ell) \|_{\HhalfG} &= \| \overline{ \langle \phi,  \frac{d}{d\ell} v_1(\cdot;\bp(\ell))\rangle} v_1(\bx;\bp(\ell)) + \overline{ \langle \phi, v_1(\cdot;\bp(\ell))\rangle} \frac{d}{d\ell}  v_1(\bx;\bp(\ell))\|_{\HhalfG}\\
&\leq 2\|\phi\|_{\HmhalfG}\max_{|\ell|\leq\eps^{1/3}} (\| \frac{d}{d\ell}v_1(\cdot,\ell)|\|_{\HhalfG} \|v_1(\cdot,\ell)|\|_{\HhalfG})\\
&\leq  2\|\phi\|_{\HmhalfG}\max_{|\ell|\leq1} (\| \frac{d}{d\ell}v_1(\cdot,\ell)\|_{H^1(\cunp)} \|v_1(\cdot,\ell)|\|_{H^1(\cunp)}).
\end{aligned}
\end{equation}
Thus
\begin{equation}
\begin{aligned}
&\left\|\lim_{\delta\to0^+}\frac{1}{2\pi}\int_{[-\eps^{1/3},-\delta]\cup[\delta,\eps^{1/3}]} \frac{\overline{ \langle \phi, v_1(\cdot;\bp(\ell))\rangle} v_1(\bx;\bp(\ell))}{\mu_1(\bp(\ell)) - \lambda_*} \, \dpt \right\|_{\HhalfG}/ \|\phi\|_{\HmhalfG} \\
= &\left\|\lim_{\delta\to0^+}\frac{1}{2\pi}\int_{[-\eps^{1/3},-\delta]\cup[\delta,\eps^{1/3}]} \frac{f(\bx,\ell)}{|\alpha_*|\ell}(1+O(\ell)) \, \dpt \right\|_{\HhalfG}/ \|\phi\|_{\HmhalfG}\\
=&\left\|\lim_{\delta\to0^+}\frac{1}{2\pi}\int_{[\delta,\eps^{1/3}]} \frac{f(\bx,\ell) - f(\bx,-\ell)}{|\alpha_*|\ell}
+O(\ell)\frac{f(\bx,\ell)}{|\alpha_*|\ell} +O(\ell)\frac{f(\bx, - \ell)}{|\alpha_*|\ell} \, \dpt \right\|_{\HhalfG}/ \|\phi\|_{\HmhalfG}\\
=&O(\eps^{1/3}) \left(  \max_{|\ell|\leq\eps^{1/3}}\| \frac{d}{d\ell}f(\cdot,\ell)|\|_{\HhalfG} +  \max_{|\ell|\leq\eps^{1/3}}\| f(\cdot,\ell)|\|_{\HhalfG} \right)/ \|\phi\|_{\HmhalfG}
\to 0. 
\end{aligned}
\end{equation}
From Fig.~\ref{fig:bandsplit}, we observe $I_2$ and $I_3$ are the same, and $I_4$ and $I_5$ are the same.
Finally,  on the integral domain, 
\begin{equation}\label{eq:error}
\begin{aligned}
|\lambda_{n}(\bp(\ell)) - \lambda_*|\gtrapprox \eps^{1/3}, \quad |\lambda_{n,\pm\eps}(\bp(\ell)) - \lambda_*- \eps h| \gtrapprox \eps^{1/3}, \\
\|u_{n,\pm\eps}(\cdot;\bp(\ell)) - \alpha u_n(\cdot;\bp(\ell))\|_{H^1(\chomo)} = O(\eps),\\
\lambda_{n,\pm\eps}(\bp(\ell)) - \eps h - \lambda_{n}(\bp(\ell)) = O(\eps).
\end{aligned}
\end{equation}
Here $\alpha$ is a phase factor that depends on $\pm\eps$, $n$ and $\ell$ as remarked on Remark~\ref{lem:phase}.
Thus by elementary insertion and grouping, $I_3-I_4=O(\eps^{1/3})\to 0$. 
\end{proof}

Using the same method, we obtain the convergence of the first two bands away from the Dirac point in the other operators.
\begin{lemma}\label{lem:n12}
The following convergences hold in the operator norm  from $\HmhalfG$ to $\HhalfG$, from $\HhalfG$ to $\HhalfG$ and from $\HhalfG$ to $\HmhalfG$ respectively
uniformly for $h\in\mathbb C$ that satisfy $|h|< \hrad$, 
as $\eps\to 0$: 
\begin{equation}
\begin{aligned}
\sum_{n=1,2}\frac{1}{2\pi}\int_{[-\pi,-\eps^{1/3}]\cup[\eps^{1/3},\pi]}&
\frac{\overline{ \langle v_{n,\pm\eps}(\cdot;\bp(\ell)),\psi\rangle} \partial_n v_{n,\pm\eps}(\bx;\bp(\ell))}{ \mu_{n,\pm\eps}(\bp(\ell)) - \lambda_*- \eps h} \, \dpt \\
&\to
\sum_{n=1,2}\frac{1}{2\pi}\text{p.v.}\int_{[-\pi,\pi ]}  \frac{\overline{ \langle \phi, v_{n}(\cdot;\bp(\ell))\rangle} \partial_n v_{n}(\bx;\bp(\ell))}{ \mu_{n}(\bp(\ell)) - \lambda_*} \, \dpt. 
\end{aligned}
\end{equation}
\begin{equation}
\begin{aligned}
\sum_{n=1,2}\frac{1}{2\pi}\int_{[-\pi,-\eps^{1/3}]\cup[\eps^{1/3},\pi]}&
\frac{\langle \partial_n v_{n,\pm\eps}(\cdot;\bp(\ell),\psi)\rangle v_{n,\pm\eps}(\bx;\bp(\ell))}{ \mu_{n,\pm\eps}(\bp(\ell)) - \lambda_*- \eps h} \, \dpt
\to\\
&\sum_{n=1,2}\frac{1}{2\pi}\text{p.v.}\int_{[-\pi,\pi ]}  \frac{ \langle \partial_n  v_{n}(\cdot;\bp(\ell)) ,\psi\rangle v_{n}(\bx;\bp(\ell))}{ \mu_{n}(\bp(\ell)) - \lambda_*} \, \dpt, 
\end{aligned}
\end{equation}
\begin{equation}
\begin{aligned}
\partial_n\sum_{n=1,2}\frac{1}{2\pi}\int_{[-\pi,-\eps^{1/3}]\cup[\eps^{1/3},\pi]}&
\frac{\langle \partial_n v_{n,\pm\eps}(\cdot;\bp(\ell),\psi)\rangle v_{n,\pm\eps}(\bx;\bp(\ell))}{ \mu_{n,\pm\eps}(\bp(\ell)) - \lambda_*- \eps h} \, \dpt
\to\\
&\partial_n\sum_{n=1,2}\frac{1}{2\pi}\text{p.v.}\int_{[-\pi,\pi ]}  \frac{ \langle \partial_n  v_{n}(\cdot;\bp(\ell)) ,\psi\rangle v_{n}(\bx;\bp(\ell))}{ \mu_{n}(\bp(\ell)) - \lambda_*} \, \dpt. 
\end{aligned}
\end{equation}
\end{lemma}

\subsection{Higher bands}\label{sec:Sn3}

In this section, we will prove Lemma~\ref{lem:Sn3} and Lemma~\ref{lem:Kstarn3}.
\begin{lemma}\label{lem:Sn3}
The following convergence holds in the operator norm from $\HmhalfG$ to $\HhalfG$
uniformly for $h\in\mathbb C$ that satisfy $|h|<\hrad $, 
as $\eps\to 0$: 
\begin{equation}\label{eq:Sn3}
\begin{aligned}
\sum_{n\geq3}\frac{1}{2\pi}\int_{[-\pi,\pi ]} 
\frac{\overline{ \langle \phi, v_{n,\pm\eps}(\cdot;\bp(\ell))\rangle} v_{n,\pm\eps}(\bx;\bp(\ell))}{ \mu_{n,\pm\eps}(\bp(\ell)) - \lambda_*- \eps h} \, \dpt
\to
\sum_{n\geq3}\frac{1}{2\pi}\int_{[-\pi,\pi ]}  \frac{\overline{ \langle \phi, v_{n}(\cdot;\bp(\ell))\rangle} v_{n}(\bx;\bp(\ell))}{ \mu_{n}(\bp(\ell)) - \lambda_*} \, \dpt. 
\end{aligned}
\end{equation}
\eqref{eq:Sn3} is equivalently be represented by
\begin{equation}\label{eq:Sn3u}
\begin{aligned}
\sum_{n\geq3}\frac{1}{2\pi}\int_{[-\pi,\pi ]} 
\frac{\overline{ \langle \phi, u_{n,\pm\eps}(\cdot;\bp(\ell))\rangle} u_{n,\pm\eps}(\bx;\bp(\ell))}{ \lambda_{n,\pm\eps}(\bp(\ell)) - \lambda_*- \eps h} \, \dpt
\to
\sum_{n\geq3}\frac{1}{2\pi}\int_{[-\pi,\pi ]}  \frac{\overline{ \langle \phi, u_{n}(\cdot;\bp(\ell))\rangle} u_{n}(\bx;\bp(\ell))}{ \lambda_{n}(\bp(\ell)) - \lambda_*} \, \dpt, 
\end{aligned}
\end{equation}
where $\lambda_{n,\pm\eps}$ and $\lambda_n$ are ranked increasingly as introduced in Section~\ref{sec:Floquet}.
\end{lemma}
\begin{remark}
The limits in Lemma~\ref{lem:Sn3} do not depend on the sign of $\eps$. Thus we will work with $\eps$ with sufficiently small absolute values. 
We also use $\lambda_n = \lambda_{n,0}$ and $u_n=u_{n,0}$ when convenient.
\end{remark}

Notice that $u_{n,\eps}$ are supported on different domains $\cpeps$ as $\eps$ vaires. Instead of extending them by $0$ into $D^\eps$ as done in \eqref{eq:ext}, we convert them to the same support $\cunp$ through diffeomorphisms, where results in \cite[p.423]{kato2013perturbation} and~\cite{Kuchmet2022} can be applied. 
Let $\bx$ and $\by^\eps$ be the Euclidean coordinates of $\cunp$ and $\cpeps$. Fix an open set $\cO$ compactly supported in $\chomo$ and containing $D^\eps$ for all $\eps$.
There is a smooth bijective map from $\cunp$ to $\cpeps$ that is analytic in $\eps$, denoted by $y^\eps=y^\eps(x)$.
Moreover, we may require that $\by^\eps(\bx)$ satisfy the following conditions:(i) $\by^\eps(\bx)=\bx$ for $\bx\in\cO$, 
(ii) $|\by^\eps(\bx)-\bx|\to0$ uniformly in $\cunp$ as $\eps\to0$, and (iii) the Jacobian $|\frac{\partial \by^\eps(\bx)}{\partial \bx}|\to0$ uniformly in $\cunp$ as $\eps\to0$. 
Every function $u(\by^\eps)$ on $\cpeps$ can be treated as a function $u (\by^\eps(\bx))$ on $\cunp$.

Since $\by^\eps(\bx)=\bx$ in a neighborhood of $\Gamma$, Lemma~\ref{lem:Sn3} follows from the following lemma and taking trace to $\Gamma$.
\begin{lemma}\label{lem:Sn3diffeo}
As $\eps\to 0$, the following convergence holds uniformly for $h\in\mathbb C$ that satisfy $|h|<\hrad $
\begin{equation}\label{eq:S3dom}
\| \Sepsevan (\lambda_* + \eps h) - \Sevan (\lambda_*) \|_{\HmhalfG\to H^1(\cunp)} \to 0,
\end{equation}
where
\begin{equation}
\begin{aligned}
\Sepsevan (\lambda_* + \eps h) \phi (\bx) &:= \sum_{n\geq3}\frac{1}{2\pi}\int_{[-\pi,\pi ]}
\frac{\overline{ \langle \phi, u_{n,\eps}(\cdot;\bp(\ell))\rangle} u_{n,\eps}(\by^\eps(\bx);\bp(\ell))}{ \lambda_{n,\eps}(\bp(\ell)) - \lambda_*- \eps h} \, \dpt,\\
\Sevan (\lambda_*)\phi (\bx) &:=\sum_{n\geq3}\frac{1}{2\pi}\int_{[-\pi,\pi ]}  \frac{\overline{ \langle \phi, u_{n}(\cdot;\bp(\ell))\rangle} u_{n}(\bx;\bp(\ell))}{ \lambda_{n}(\bp(\ell)) - \lambda_*} \, \dpt .
\end{aligned}
\end{equation}
\end{lemma}

We will prove Lemma~\ref{lem:Sn3diffeo} using the dominant convergence theorem and the following results.
\begin{lemma} \label{lem:DTC}
Let $\Sepsevan (\lambda_* + \eps h ,\bp(\ell))$ and $\Sevan (\lambda_* ,\bp(\ell))$ be operators from $\HmhalfG$ to $H^1(\cunp)$ that are defined by
\begin{equation}\label{eq:Sn3p}
\begin{aligned}
\Sepsevan (\lambda_* + \eps h ,\bp) \phi(\bx)  &:=
\sum_{n\geq3}\frac{\overline{ \langle \phi, v_{n,\eps}(\cdot;\bp(\ell))\rangle} v_{n,\eps}(\by^\eps(\bx);\bp(\ell))}{ \mu_{n,\eps}(\bp(\ell)) - \lambda_*- \eps h}, \\
\Sevan (\lambda_* ,\bp)\phi(\bx)  &:=\sum_{n\geq3}\frac{\overline{ \langle \phi, v_{n}(\cdot;\bp(\ell))\rangle} v_{n}(\bx;\bp(\ell))}{ \mu_{n}(\bp(\ell)) - \lambda_*}.
\end{aligned}
\end{equation} 
There is a constant $\eps_1$, such that the following three statements hold uniformly for $h\in\mathbb C$ that satisfy $|h|<\hrad$.
\begin{enumerate}
\item For each $|\eps|<\eps_1$, $\Sepsevan (\lambda_* + \eps h ,\bp(\ell))$ are continuous in $\ell$ in the operator norm.
\item $\Sepsevan (\lambda_* + \eps h ,\bp(\ell))$ are uniformly bounded in the operator norm over $\ell$ and over $|\eps|<\eps_1$.
\item For each $\ell\neq 0$, $\Sepsevan (\lambda_* + \eps h ,\bp(\ell))$ converge to $\Sevan (\lambda_* ,\bp(\ell))$ in the operator norm as $\eps\to0$.
\end{enumerate}
\end{lemma}

To prepare for the proof of Lemma~\ref{lem:DTC}, we introduce the following functions for each fixed $\ell$:
\begin{equation}\label{eq:singlev}
w^\eps(\bx):=\sum_{n\geq3}\frac{\overline{ \langle \phi, u_{n,\eps}(\cdot;\bp(\ell))\rangle} u_{n,\eps}(\bx;\bp(\ell))}{ \lambda_{n,\eps}(\bp(\ell)) - \lambda_*- \eps h}
\end{equation}
and
\begin{equation}\label{eq:singlevdiff}
\hat{w}^\eps(\bx):=\sum_{n\geq3}\frac{\overline{ \langle \phi, u_{n,\eps}(\cdot;\bp(\ell))\rangle} u_{n,\eps}(\by^\eps(\bx);\bp(\ell))}{ \lambda_{n,\eps}(\bp(\ell)) - \lambda_*- \eps h}.
\end{equation} 
We introduce the following function spaces, whose more detailed properties can be found in~\cite{mclean2000strongly,saranen2001periodic}. 
Let $H^{s}(\ell)$ be the Sobolev space on $\mathbb R^2$ of order $s$ that is quasiperiodic with quasimomenta $\bp(\ell)\cdot\be_i$ in $\be_i$, $i=1,2$.
Let $L^{2,\eps}(\ell)=H^{0,\eps}(\ell)$ be $H^0(\ell)$ functions that are supported on $\mathbb R^2\backslash\cup_{n_1,n_2\in\mathbb Z}(D^\eps+n_1\be_1+n_2\be_2)$, 
$H^{1,\eps}(\ell)$ be $H^1(\ell)$ functions that are supported on $\overline{\mathbb R^2\backslash\cup_{n_1,n_2\in\mathbb Z}(D^\eps+n_1\be_1+n_2\be_2)}$, and $H^{-1,\eps}(\ell)$ distributions on $\mathbb R^2\backslash\cup_{n_1,n_2\in\mathbb Z}(\overline{D^\eps}+n_1\be_1+n_2\be_2)$ who can be extended to distributions in $H^{-1}(\ell)$.
The dual space of $H^{1,\eps}(\ell)$ is $H^{-1,\eps}(\ell)$. 

The pairings on these spaces are defined through their quasiperiodic Fourier expansions. Denote the $L^{2,\eps}(\ell)$-$L^{2,\eps}(\ell)$ pairing by $\langle \cdot, \cdot \rangle_{\eps}$, and the $H^{-1,\eps}(\ell)$-$H^{1,\eps}(\ell)$ pairing by $(\cdot,\cdot)_\eps$.
The innerproduct on $H^{1,\eps}(\ell)$ is given by $\langle \nabla u, \nabla v\rangle_{\eps} + \langle u, v\rangle_{\eps} $.

Define the operators $( - \Delta_\eps(\ell))^{-1}: H^{-1,\eps}(\ell) \to H^{1,\eps}(\ell)$ by 
\begin{equation}
( - \Delta_\eps(\ell))^{-1} f = u \text{ if and only if }  \langle\nabla u, \nabla v \rangle_{\eps} = (f,v)_\eps  \quad \forall v\in H^{1,\eps}(\ell).
\end{equation}
Notice that when restricted on $H^{1,\eps}(\ell)$, the operator $ (- \Delta_\eps(\ell))^{-1}: H^{1,\eps}(\ell) \to H^{1,\eps}(\ell)$ is bounded, selfadjoint, compact and positive. Thus there exists an orthogonal eigensystem of $(- \Delta_\eps(\ell))^{-1}$ that is complete in $H^{1,\eps}(\ell)$. This eigensystem coinsides with  $\lambda_{n,\eps}:=\lambda_{n,\eps}(\bp(\ell))$ and $u_{n,\eps}:=u_{n,\eps}(\bx,\bp(\ell))$ that are defined in Section~\ref{sec:Floquet}. We have
\begin{equation}\label{eq:lapinv}
\begin{aligned}
(- \Delta_\eps(\ell))^{-1} u_{n,\eps}(\ell) = \frac{1}{\lambda_{n,\eps}(\ell)} u_n(\ell), \quad  0<\lambda_{1,n}(\ell) \leq \lambda_{2,\eps}(\ell) \leq \cdots \to \infty, \\
\langle \nabla u_{n,\eps}(\ell), \nabla u_{m,\eps}(\ell)\rangle_{\eps} + \langle u_{n,\eps}(\ell),  u_{m,\eps}(\ell)\rangle_{\eps}= (1+\lambda_{n,\eps}(\ell))\delta_{m,n}.
\end{aligned}
\end{equation}
We have the expansion
\begin{equation}\label{eq:H1basis}
\forall f\in H^{1,\eps}(\ell), \quad f = \sum_{n\geq1} \langle u_{n,\eps}(\ell),f\rangle_{\eps} u_{n,\eps}(\ell), \text{ which converges in }H^{1,\eps}(\ell).
\end{equation}
By the density of $H^{1,\eps}(\ell)$ in $L^2(C^\eps)$ and $H^{-1,\eps}(\ell)$, we also have
\begin{equation}\label{eq:L2basis}
 \forall f\in L^2(C^\eps),\quad  f = \sum_{n\geq1} \langle u_{n,\eps}(\ell),f\rangle_{\eps} u_{n,\eps}(\ell),  \text{ which converges in }L^2(\cpeps),
\end{equation}
and
\begin{equation}\label{eq:Hm1basis}
\forall f\in H^{-1,\eps}(\ell) ,\quad f =  \sum_{n\geq1} \overline{(f, u_{n,\eps}(\ell))_\eps} u_{n,\eps}(\ell),  \text{ which converges in }H^{-1,\eps}(\ell).
\end{equation}
For $A\subset\mathbb Z^{+}$, define the projection operator $P_{n\in A,\eps,\ell}:H^{1,\eps}(\ell) \to H^{1,\eps}(\ell)$ by
\begin{equation}
P_{n\in A,\eps,\ell} f = \sum_{n\in A} \langle u_{n,\eps}(\ell),f\rangle_{\eps} u_{n,\eps}(\ell).
\end{equation}
It is straightforward to verify that the function $w^\eps$ in \eqref{eq:singlev} solves the following problem
\begin{equation}\label{eq:Smeaning}
\begin{aligned}
(I - (\lambda_*+ \eps h)( - \Delta_\eps(\ell))^{-1})  w(\bx) 
&=  ( - \Delta_\eps(\ell))^{-1} \left(  \phi\delta(\Gamma) -  \sum_{n\geq1,2}
\overline{ \langle \phi, u_{n,\eps}(\cdot;\bp(\ell))\rangle} u_{n,\eps}(\bx;\bp(\ell)) \right)\\
&= P_{n\geq3,\eps,\ell}  ( - \Delta_\eps(\ell))^{-1}   \left( \phi\delta(\Gamma) \right).
\end{aligned}
\end{equation}
This can be equivalently represented as
\begin{equation}\label{eq:Ssource}
(-\Delta - (\lambda_*+ \eps h)) w(\bx) 
= \phi\delta(\Gamma) -  \sum_{n\geq1,2}
\overline{ \langle \phi, u_{n,\eps}(\cdot;\bp(\ell))\rangle} u_{n,\eps}(\bx;\bp(\ell)).
\end{equation}
Here we have used $\phi\delta(\Gamma)\in H^{-1,\eps}(\ell)$ and 
\begin{equation}
\overline{(\phi\delta(\Gamma),u_{n,\eps})_\eps} = 
\overline{ \langle \phi, u_{n,\eps}(\cdot;\bp(\ell))\rangle}.
\end{equation}

The diffeomorphism $\by^\eps(\bx)$ converts \eqref{eq:singlev} to \eqref{eq:singlevdiff}. We have the following change of variables formulas
\begin{equation}
\begin{aligned}
\int_{\cpeps} \overline{u}(\by^\eps) v(\by^\eps)\, d\by^\eps &=  \int_{\cunp} \overline{u(\by^\eps(\bx))} v(\by^\eps(\bx)) \left|\left(\frac{ \partial \by^\eps}{\partial \bx}\right)\right|\, d\bx ,\\
\int_{\cpeps} \nabla_{\by^\eps} \overline{u(\by^\eps)} \cdot \nabla_{\by^\eps} v(\by^\eps)\, d\by^\eps &=  \int_{\cunp} \left(\frac{\partial \bx}{ \partial \by^\eps}\right)\nabla_{\bx} \overline{u(\by^\eps(\bx))} \cdot \left(\frac{\partial \bx}{ \partial \by^\eps}\right)\nabla_{\bx} v(\by^\eps(\bx)) \left|\left(\frac{ \partial \by^\eps}{\partial \bx}\right)\right| \, d\bx.
\end{aligned}
\end{equation}
Thus 
\begin{equation}
\| u(\by^\eps(\cdot)) \|_{H^{k,0}(\ell)} = \| u(\cdot) \|_{H^{k,\eps}(\ell)} (1+O(\eps)), \quad k=-1,0,1.
\end{equation}
Define $(- \tilde\Delta_\eps)^{-1}: H^{-1,0}(\ell) \to H^{1,0}(\ell)$ by
\begin{equation}\label{eq:LapDiff}
\begin{aligned}
&( - \tilde\Delta_\eps(\ell))^{-1} f =  u \text{ if and only if }  \\
&\int_{\cunp} \left(\frac{\partial x}{ \partial y^\eps}\right)\nabla_{x} \overline{u(x)} \cdot \left(\frac{\partial x}{ \partial y^\eps}\right)\nabla_{x} v(x) \left|\left(\frac{ \partial y^\eps}{\partial x}\right)\right| \, dx
= \int_{\cunp} \overline{f(x)} v(x) \left|\left(\frac{ \partial y^\eps}{\partial x}\right)\right|\, dx.
\end{aligned}
\end{equation}
Then we have
\begin{equation}
( - \tilde\Delta_\eps(\ell))^{-1} f(\by^\eps(\cdot)) =  u(\by^\eps(\cdot))
\text{ if and only if } ( - \Delta_\eps(\ell))^{-1} f(\cdot) = u(\cdot).
\end{equation}
Combining with \eqref{eq:Ssource}, we obtain that \eqref{eq:singlevdiff} solves
\begin{equation}\label{eq:Ssourcediff}
(-\tilde \Delta_\eps - (\lambda_*+ \eps h)) \hat w(\bx)  
= \phi\delta(\Gamma) -  \sum_{n\geq1,2}
\overline{ \langle \phi, u_{n,\eps}(\cdot;\bp(\ell))\rangle} u_{n,\eps}(\by^\eps(\bx);\bp(\ell)).
\end{equation}

We are ready to prove Lemma~\ref{lem:DTC}.
\begin{proof}[Proof of Lemma~\ref{lem:DTC}]
We know that there are constants $\eps_1$, $\mathfrak c_1,\mathfrak c_2>0$ and a family of constants $\mathfrak c_3(\ell)>0$ depending on $\ell$ such that 
\begin{equation}
\begin{aligned}
&\lambda_{n,\eps}(\ell) >\mathfrak c_1\quad \text{for all } |\eps|<\eps_1, \quad n\geq1, \quad  \ell\in[-\pi,\pi]; \\
&\lambda_{n,\eps}(\ell) >\lambda_*+\frac{1+\mathfrak d}{2}|\frac{t_*}{\gamma_*}|, \quad \text{for all } |\eps|<\eps_1, \quad n\geq3, \quad  \ell\in[-\pi,\pi]; \\
&\left|\lambda_{n,\eps}(\ell) - \lambda  \right| > \mathfrak c_2 \quad \text{for all } |\eps|<\eps_1, 
|\lambda - \lambda_*|<\hrad\eps,\quad n\geq3,\quad \ell\in[-\pi,\pi] ;\\
\text{for each }\ell\neq0,\quad &\left|1- \frac{\lambda}{\lambda_{n,\eps}(\ell)}  \right| > \mathfrak c_3(\ell) \quad \text{for all } |\eps|<\eps_1, |\lambda - \lambda_*|<\hrad\eps, \quad n\geq1.
\end{aligned}
\end{equation}

Statement 1 follows from the analyticity in $\bp(\ell)$ through gauge transformations~\cite{Qiu-23,FlissJoly2016,kato2013perturbation}.

For statement 2, before the diffeomorphism, treated as functions on $\cpeps$, \eqref{eq:Smeaning} gives
\begin{equation}
v = (I - (\lambda_*+ \eps h)( - \Delta_\eps(\ell))^{-1})^{-1} P_{n\geq3,\eps,\ell} ( - \Delta_\eps(\ell))^{-1}  \left( \phi\delta(\Gamma) \right).
\end{equation}
We have 
\begin{equation}
\begin{aligned}
\|\phi\delta(\Gamma)\|_{H^{-1,\eps}(\ell)}&\leq C\|\phi\|_{\HmhalfG},\\ 
\left\| ( - \Delta_\eps(\ell))^{-1} \right\|_{H^{-1,\eps}(\ell) \to H^{1,\eps}(\ell)} 
&=\sup_n \left| \frac{1+\lambda_{n,\eps}(\ell)}{\lambda_{n,\eps}(\ell)} \right| 
\leq \frac{1+\mathfrak c_1}{\mathfrak c_1},\\
\left\| P_{n\geq3,\eps,\ell}  \right\|_{H^{-1,\eps}(\ell) \to H^{1,\eps}(\ell)} &\leq 1,\\
\left\| (I - (\lambda_*+ \eps h)( - \Delta_\eps(\ell))^{-1})^{-1} \right\|_{P_{n\geq3,\eps,\ell} H^{1,\eps}(\ell) \to H^{1,\eps}(\ell)} &=  \sup_{n\geq3} \left|\left(1- \frac{\lambda_*+ \eps h}{\lambda_{n,\eps}(\ell)} \right)^{-1} \right| <\frac{1}{1-(\lambda_*+\hrad)/(\lambda_*+\frac{1+\mathfrak d}{2}|\frac{t_*}{\gamma_*}) }.
\end{aligned}
\end{equation}
Since the operator norms are related by a factor of $(1+O(\eps))$ uniformly when treated as maps between functions on $\cunp$ through the diffeomorphisms, the proof of Statement 2 is complete.

For Statement 3, for each fixed $\ell$, we apply Lemma~\ref{lem:invfact} to
\begin{equation}
\begin{aligned}
A_\eps &= - \tilde \Delta_\eps - \lambda_* - \eps h,\\
A_0 &= - \tilde\Delta_0 - \lambda_*,\\
f_\eps &= \phi\delta(\Gamma) -  \sum_{n\geq1,2}
\overline{ \langle \phi, u_{n,\eps}(\cdot;\bp(\ell))\rangle} u_{n,\eps}(\by^\eps(\bx);\bp(\ell)),\\
f_0 &= \phi\delta(\Gamma) -  \sum_{n\geq1,2}
\overline{ \langle \phi, u_{n,\eps}(\cdot;\bp(\ell))\rangle} u_{n,\eps}(\bx;\bp(\ell)).
\end{aligned}
\end{equation}
The forms defined in \eqref{eq:LapDiff} are analytic in $\eps$. Thus by~\cite{kato2013perturbation}, each pair of eigenvalue and eigeneigenfunction $\lambda_{n,\eps}$, $u_{n,\eps}$ converges to $\lambda_n$, $u_n$ in $\mathbb R \times H^1(\cunp)$ at a rate of $O(\eps)$. Thus
it is straight forward to verify that when $\eps_1$ is sufficiently small, $|\eps|<\eps_1$, $|h|<\hrad$, for each $\ell\neq0$, there exists a constant $C(\ell)$, such that
\begin{equation}
\begin{aligned}
&\| A_0^{-1}\|_{H^{-1,0}(\ell) \to H^{1,0}(\ell)}, \| A_\eps^{-1}\|_{H^{-1,0}(\ell) \to H^{1,0}(\ell)} \leq C(\ell),\\
&\|f_0\|_{H^{-1,0}(\ell)}, \|f_\eps\|_{H^{-1,0}(\ell)} \leq C(\ell)\|\phi\|_{\HmhalfG},\\
&\|f_\eps - f_0\|_{H^{-1,0}(\ell)} \leq  \eps C(\ell)\|\phi\|_{\HmhalfG},\\
&\| A_\eps - A_0\|_{H^{1,0}(\ell)\to H^{-1,0}(\ell)} \to 0.
\end{aligned}
\end{equation}
Thus we obtain Statement 3.
\end{proof}

Since we showed convergence in $H^1$, taking the normal derivatives and using the jump relations, we obtain the following corollary. 
\begin{coro}\label{lem:Kstarn3}
The following convergence holds in the operator norm from $\HmhalfG$ to $\HmhalfG$
uniformly for $h\in\mathbb C$ that satisfy $|h|<\hrad $, 
as $\eps\to 0$: 
\begin{equation}
\begin{aligned}
\sum_{n\geq3}\frac{1}{2\pi}\int_{[-\pi,\pi ]} 
\frac{\overline{ \langle v_{n,\pm\eps}(\cdot;\bp(\ell)),\psi\rangle} \partial_n v_{n,\pm\eps}(\bx;\bp(\ell))}{ \mu_{n,\pm\eps}(\bp(\ell)) - \lambda_*- \eps h} \, \dpt
\to
\sum_{n\geq3}\frac{1}{2\pi}\int_{[-\pi,\pi ]}  \frac{\overline{ \langle \phi, v_{n}(\cdot;\bp(\ell))\rangle} \partial_n v_{n}(\bx;\bp(\ell))}{ \mu_{n}(\bp(\ell)) - \lambda_*} \, \dpt. 
\end{aligned}
\end{equation}
\end{coro}

\subsection{Double layer potential and the convergence of higher bands}\label{sec:Dn3}

Let $\psi\in\HhalfG$. 
Denote the part of $\cpeps$ to the left of $\Gamma$ by $\cL^\eps$ and that to the right of $\Gamma$ by $\cR^\eps$. Using \eqref{eq:lapinv}, we know that when $\eps\neq0$ or $\ell\neq0$, $ \lambda_*+ \eps h \neq \lambda_{n,\eps}(\ell)$ for all $n$,
\begin{equation}
v(\bx):= \sum_{n\geq3}
\frac{\langle \partial_n u_{n,\eps}(\cdot;\bp(\ell)), \psi\rangle u_{n,\eps}(\bx;\bp(\ell))}{ \lambda_{n,\eps}(\bp(\ell)) - \lambda_*- \eps h}
\end{equation} 
is the unique function that satisfies
\begin{equation}
\begin{aligned}
&[v]=\psi, \quad [\partial_n v]=0,\\
&(-\Delta -\lambda_*-\eps h) v = - \sum_{n=1,2}
\langle \partial_n u_{n,\eps}(\cdot;\bp(\ell)), \psi\rangle u_{n,\eps}(\bx;\bp(\ell))\quad\bx\in\cL^\eps \cup \cR^\eps. 
\end{aligned}
\end{equation}
Here $[\cdot]$ represents the jump of the quantity across $\Gamma$. The uniqueness follows from that the difference between two such functions is a Floquet mode of quasimomentum $\bp(\ell)$ and energy $\lambda_*+\eps h$ on $\cpeps$.

We next decompose $v$ as the sum of two functions $\tilde{u}$ and $w$ which are introduced below. Let $E:\HhalfG \to H^{1,\eps}(\ell)$ be an extension operator which is a right inverse of the trace operator, for which functions in its image are supported in a neighborhood $\cO_1$ of $\Gamma$ that does not intersect $\cO$. 
Thus this extension operator stays the same for all $\ell$. Define 
\begin{equation}
u:= \chi_{\cR^\eps} E\psi,
\end{equation}
where $\chi_{\cR^\eps}$ is the characteristic function of $\cR^\eps$. 
Then $u$ satisfies
\begin{equation}
\begin{aligned}
&[u]=\psi, \quad [\partial_n u]= - \partial_n E\psi,\\
&(-\Delta - \lambda_*-\eps h) u = \chi_{\cR^\eps}( -\Delta - \lambda_*-\eps h) E\psi \quad\bx\in\cL^\eps \cup \cR^\eps. 
\end{aligned}
\end{equation}
We have that there exists a constant $C$, such that for all $\eps$, $|h|<\hrad$, and $\ell$,
\begin{equation}
\|u\|_{H^1(\cR^\eps)}, \|u\|_{H^1(\cL^\eps)}\leq C\|\psi\|_{\HhalfG}, \quad
\|(-\Delta - \lambda_*-\eps h) u\|_{H^{-1}(\cL^\eps \cup \cR^\eps)}\leq C\|\psi\|_{\HhalfG}.
\end{equation}

Next we shift $u$ along $u_{n,\eps}(\ell)$, $n=1,2$, to make sure the source for the rest of $v$ does not have components in these two directions. Define 
\begin{equation}
\tilde u:= \chi_{\cR^\eps} E\psi - \sum_{n=1,2} \langle u_{n,\eps}(\ell), \chi_{\cR^\eps} E\psi \rangle_\eps u_{n,\eps}(\ell)
\end{equation}
Then $\tilde u$ satisfies
\begin{equation}
\begin{aligned}
&[\tilde u]=\psi, \quad [\partial_n \tilde u]= - \partial_n E\psi,\\
&(-\Delta - \lambda_*-\eps h) \tilde u = \chi_{\cR^\eps}( -\Delta - \lambda_*-\eps h) E\psi 
- \sum_{n=1,2}(\lambda_{n,\eps} - \lambda_*-\eps h) \langle u_{n,\eps}(\ell), \chi_{\cR^\eps} E\psi \rangle_\eps u_{n,\eps}(\ell), \quad\bx\in\cL^\eps \cup \cR^\eps.  
\end{aligned}
\end{equation}

Let $w\in H^{1,\eps}(\ell)$ be the solution to
\begin{equation}
\begin{aligned}
 [w]=0, \quad& [\partial_n w]=  \partial_n E\psi,\\
(-\Delta - \lambda_*-\eps h) w =& - \sum_{n=1,2}
\overline{\langle \partial_n E\psi,u_{n,\eps}(\cdot;\bp(\ell))\rangle} u_{n,\eps}(\bx;\bp(\ell)) - \chi_{\cR^\eps}( -\Delta - \lambda_*-\eps h) E\psi \\
&+ \sum_{n=1,2}(\lambda_{n,\eps} - \lambda_*-\eps h) \langle u_{n,\eps}(\ell), \chi_{\cR^\eps} E\psi \rangle_\eps u_{n,\eps}(\ell), 
\quad\bx\in\cL^\eps \cup \cR^\eps. 
\end{aligned}
\end{equation}
Then $w\in H^{1,\eps}(\ell)$ is the solution to the PDE below on the entire $\cpeps$
\begin{equation}\label{eq:Dn3p}
\begin{aligned}
(-\Delta - \lambda_*-\eps h) w = (\partial_n E\psi)\delta(\Gamma)
- \sum_{n=1,2}
\overline{\langle \partial_n E\psi,u_{n,\eps}(\cdot;\bp(\ell))\rangle} u_{n,\eps}(\bx;\bp(\ell)) - \chi_{\cR^\eps}( -\Delta - \lambda_*-\eps h) E\psi\\
+ \sum_{n=1,2}(\lambda_{n,\eps} - \lambda_*-\eps h) \langle u_{n,\eps}(\ell), \chi_{\cR^\eps} E\psi \rangle_\eps u_{n,\eps}(\ell), \quad\bx\in\cpeps. 
\end{aligned}
\end{equation}
It can be checked that the right hand side of \eqref{eq:Dn3p} is orthogonal to $u_{n,\eps}(\ell)$, $n=1,2$, 
thus $w$ can be treated using the procedure shown in Section~\ref{sec:Sn3}.

Integration by parts gives
\begin{equation}
\overline{\langle \partial_n E\psi,u_{n,\eps}(\cdot;\bp(\ell))\rangle}
- \langle \partial_n u_{n,\eps}(\cdot;\bp(\ell)), \psi\rangle =
\langle u_{n,\eps}, (-\Delta-\lambda_*-\eps h)\chi_{\cR^\eps}E\psi \rangle
-(\lambda_{n,\eps}-\lambda_*-\eps h)\langle u_{n,\eps}, \chi_{\cR^\eps}E\psi \rangle
\end{equation}
Thus we see
\begin{equation*}
v= \tilde u+w.
\end{equation*}
Treating $v - \chi_{\cR^\eps}E\psi$ using the procedures in Section~\ref{sec:Sn3}, we conclude the convergence of 
\begin{equation}
\frac{1}{2\pi} \sum_{n\geq3}\int_{-\pi,\pi}
\frac{\langle \partial_n u_{n,\eps}(\cdot;\bp(\ell)), \psi\rangle u_{n,\eps}(\bx;\bp(\ell))}{ \lambda_{n,\eps}(\bp(\ell)) - \lambda_*- \eps h} \, \dpt \to 
\frac{1}{2\pi} \sum_{n\geq3}\int_{-\pi,\pi}
\frac{\langle \partial_n u_{n}(\cdot;\bp(\ell)), \psi\rangle u_{n}(\bx;\bp(\ell))}{ \lambda_{n}(\bp(\ell)) - \lambda_*} \, \dpt 
\end{equation} 
in $H^1$ on the right of $\Gamma$ in the proper sense through diffeomorphism. Taking the traces and the normal derivatives, we obtain the following convergences.
\begin{lemma}\label{lem:Dn3}
The following convergences hold in the operator norm from $\HhalfG$ to $\HhalfG$ and from $\HhalfG$ to $\HmhalfG$ respectively
uniformly for $h\in\mathbb C$ that satisfy $|h|<\hrad $, 
as $\eps\to 0$: 
\begin{equation}\label{eq:Kn3}
\begin{aligned}
\sum_{n\geq3}\frac{1}{2\pi}\int_{[-\pi,\pi ]} 
\frac{\langle \partial_n v_{n,\pm\eps}(\cdot;\bp(\ell),\psi)\rangle v_{n,\pm\eps}(\bx;\bp(\ell))}{ \mu_{n,\pm\eps}(\bp(\ell)) - \lambda_*- \eps h} \, \dpt
\to
\sum_{n\geq3}\frac{1}{2\pi}\int_{[-\pi,\pi ]}  \frac{ \langle \partial_n  v_{n}(\cdot;\bp(\ell)) ,\psi\rangle v_{n}(\bx;\bp(\ell))}{ \mu_{n}(\bp(\ell)) - \lambda_*} \, \dpt, 
\end{aligned}
\end{equation}
\begin{equation}\label{eq:Nn3}
\begin{aligned}
\partial_n\sum_{n\geq3}\frac{1}{2\pi}\int_{[-\pi,\pi ]} 
\frac{\langle \partial_n v_{n,\pm\eps}(\cdot;\bp(\ell),\psi)\rangle v_{n,\pm\eps}(\bx;\bp(\ell))}{ \mu_{n,\pm\eps}(\bp(\ell)) - \lambda_*- \eps h} \, \dpt
\to
\partial_n\sum_{n\geq3}\frac{1}{2\pi}\int_{[-\pi,\pi ]}  \frac{ \langle \partial_n  v_{n}(\cdot;\bp(\ell)) ,\psi\rangle v_{n}(\bx;\bp(\ell))}{ \mu_{n}(\bp(\ell)) - \lambda_*} \, \dpt. 
\end{aligned}
\end{equation}
\end{lemma}

\section{Proof of Propositions~\ref{lem:chadisp}, and Theorems~\ref{lem:ratzig}, \ref{lem:ratarm} and \ref{lem:dispersionarm}.}
\label{sec:disprat}
\begin{proof}[Proof of Proposition~\ref{lem:chadisp}]
The proof is similar to that of Proposition~\ref{lem:oplim}. The leading order term in the limit comes from the integrals close to the Dirac points. We only display this calculation below.

Using \eqref{eq:energymu}, \eqref{eq:pmodesLmu} and \eqref{eq:mmodesLmu},
\begin{equation}\label{eq:nearDiracdisp}
\begin{aligned}
e^{\im \eps\zeta\bbeta_2\cdot\bx}\sum_{n=1,2}\frac{1}{2\pi}&\int_{[-\eps^{1/3},\eps^{1/3} ]} 
\frac{ u_{n,\pm\eps}(\bx;\bp(\ell,\eps\zeta))\overline{u_{n,\pm\eps}(\by;\bp(\ell,\eps\zeta))}}{ \lambda_{n,\pm\eps}(\bp(\ell,\eps\zeta)) - (\lambda_*+\eps h)} \, \dpt \,e^{-\im \eps\zeta\bbeta_2\cdot\by}
\\
\to&
F^{\pm}_1(h,\zeta)w_1(\bx)\overline{w_1(\by)}
+ F^{\pm}_2(h,\zeta) w_2(\bx)\overline{w_2(\by)}
+ F^{\pm}_3(h,\zeta) w_2(\bx)\overline{w_1(\by)}\\
&+F^{\pm}_4(h,\zeta) w_1(\bx)\overline{w_2(\by)}
+o(1)\|\phi\|_{\HhalfG}.
\end{aligned}
\end{equation}
Here 
\begin{equation}\label{eq:dispcoeffp}
\begin{aligned}
F^{+}_1(h,\zeta)&= \frac{1}{2\pi}\int_{\mathbb R}\frac{ 1- |L(\eps,\ell,\eps\zeta)|^2 }{ 1+ |L(\eps,\ell,\eps\zeta)|^2 } \cdot
\frac{J(\eps,\ell,\eps\zeta)}{-(J(\eps,\ell,\eps\zeta))^2 +\eps^2 h^2} 
+ \frac{-\eps h}{-(J(\eps,\ell,\eps\zeta))^2 + \eps^2h^2} \, d\ell,\\
F^{+}_2(h,\zeta)&= \frac{1}{2\pi}\int_{\mathbb R} - \frac{ 1- |L(\eps,\ell,\eps\zeta)|^2 }{ 1+ |L(\eps,\ell,\eps\zeta)|^2 }\cdot
\frac{J(\eps,\ell,\eps\zeta)}{-(J(\eps,\ell,\eps\zeta))^2 + \eps^2h^2} 
+ \frac{-\eps h}{-(J(\eps,\ell,\eps\zeta))^2 + \eps^2h^2} \, d\ell,\\
F^{+}_3(h,\zeta)&=  \frac{1}{2\pi}\int_{\mathbb R}\frac{ L(\eps,\ell,\eps\zeta)}{ 1+ |L(\eps,\ell,\eps\zeta)|^2 }\cdot 
\frac{2J(\eps,\ell,\eps\zeta)}{-(J(\eps,\ell,\eps\zeta))^2 + \eps^2h^2}  \, d\ell,\\
F^{+}_4(h,\zeta)&=\overline{F^{+}_3(h,\zeta)},
\end{aligned}
\end{equation}
\begin{equation}\label{eq:dispcoeffm}
\begin{aligned}
F^{-}_1(h,\zeta)= F^{+}_2(h,\zeta), \quad
F^{-}_2(h,\zeta)= F^{+}_1(h,\zeta), \quad
F^{-}_3(h,\zeta)= F^{+}_3(h,\zeta),\quad
F^{-}_4(h,\zeta)= F^{+}_4(h,\zeta), 
\end{aligned}
\end{equation}
$L(\eps,\ell,\eps\zeta)$ is defined in \eqref{eq:Lmu} and
\begin{equation}\label{eq:dispJ}
J(\eps,\ell,\mu)=\frac{1}{|\gamma_*|}\sqrt{\eps^2t_*^2+ |\ell+\mu\bar\tau|^2|\theta_*|^2}.
\end{equation}
Note it turns out that the functions $F^{\pm}_i$ are independent of $\eps>0$. This is because
\begin{equation}
\begin{aligned}
\frac{L(\eps,\ell,\eps\zeta)}{ 1+ |L(\eps,\ell,\eps\zeta)|^2 } &= \frac{\theta_*(\ell+\eps\zeta\overline{\tau})}{2\gamma_*(\eps,\ell,\eps\zeta)}, \\
\frac{ 1- |L(\eps,\ell,\eps\zeta)|^2 }{ 1+ |L(\eps,\ell,\eps\zeta)|^2 } &= \frac{\eps t_*}{\gamma_*J(\eps,\ell,\eps\zeta)}.
\end{aligned}
\end{equation}

Observing
\begin{equation}
|\ell+\eps\zeta\tau|^2 = (\ell-\frac{1}{2}\eps\zeta)^2 + \frac{3}{4}(\eps\zeta)^2,
\end{equation}
We set $\tilde\ell = \ell-\frac{1}{2}\eps\zeta$.
Using
\begin{equation}
\frac{1}{\pi}\int_{\mathbb R}\frac{1}{ \frac{1}{(\gamma_*)^2}(t_*^2+ \frac{3}{4}|\theta_*|^2\zeta^2 + |\theta_*|^2 \tilde\ell^2) - h^2} d\tilde\ell = \left|\frac{\gamma_*}{\theta_*}\right|\frac{1}{\sqrt{(\frac{t_*}{\gamma_*})^2+ \frac{3}{4}|\frac{\theta_*}{\gamma_*}|^2\zeta^2-h^2}},
\end{equation}
we obtain
\begin{equation}
F^{+}_1(h,\zeta) = - \xi(h,\zeta) + \beta(h,\zeta), \quad F^{+}_2(h,\zeta) = \xi(h,\zeta) + \beta(h,\zeta),  \quad F^{+}_3(h,\zeta) = \frac{\theta_*}{|\theta_*|}\sigma(h,\zeta),  \quad F^{+}_4(h,\zeta) = \frac{\overline{\theta_*}}{|\theta_*|}\overline{\sigma(h,\zeta)}.
\end{equation}
Here $\xi(h,\zeta)$, $\beta(h,\zeta)$ and $\sigma(h,\zeta)$ are defined in \eqref{eq:betaxisig}. A coincidence is that $\sigma(h,\zeta)\in\mathbb C$, which implies $\overline{\sigma(h,\zeta)} = -\sigma(h,\zeta)$.

Finally using
\begin{equation}
\begin{aligned}
w_1\bar w_1 &= \frac{1}{2}( v_1\bar v_1 + v_2\bar v_2 + v_2\bar v_1 + v_1\bar v_2 ),\\
w_2\bar w_2 &= \frac{1}{2}( v_1\bar v_1 + v_2\bar v_2 - v_2\bar v_1 - v_1\bar v_2 ),\\
w_2\bar w_1 &= \frac{1}{2}\frac{\overline{\theta_*}}{|\theta_*|}( - v_1\bar v_1 + v_2\bar v_2 + v_2\bar v_1 - v_1\bar v_2 ),\\
w_1\bar w_2 &= \frac{1}{2}\frac{\theta_*}{|\theta_*|}( - v_1\bar v_1 + v_2\bar v_2 - v_2\bar v_1 + v_1\bar v_2 ),
\end{aligned}
\end{equation}
we finish the proof.

\end{proof}

\begin{proof}[Proof of Theorem~\ref{lem:ratzig}]
Along the rational edge, we need to consider the quasimomenta 
\begin{equation}
\ell \bbeta_1^r + \mu \bbeta_2^r = \ell(b \bbeta_1 - a\bbeta_2) + \mu(-d\bbeta_1 + c \bbeta_2)
\end{equation}
The leading order term in the counter part of matrix \eqref{eq:dispM} is
\begin{equation}
\left(\begin{matrix}
t_*\eps + \gamma_* \lambda^{(1)}& (\ell+\mu\tau)\overline{\theta_*}\\\
(\ell+\mu\bar\tau)\theta_*&-t_*\eps + \gamma_* \lambda^{(1)}
\end{matrix}\right).
\end{equation}
Thus in the counter parts of $L$ defined in \eqref{eq:Lmu} and $J$ defined in \eqref{eq:dispJ},  we have the replacement 
\begin{equation}
\ell+\mu\bar\tau \to \ell(b-a\bar\tau)+\mu( -d +c\bar\tau) =: A\ell+B\mu,
\end{equation}
where $A:=b-a\bar\tau$ and $B:=  -d +c\bar\tau$.
Since 
\begin{equation}
|A\ell+B\mu|^2 =|A|^2(\ell + \frac{\text{Re}(A\bar B)}{|A|^2}\mu)^2 +\left(|B|^2 - |A|^2\left(\frac{\text{Re}(A\bar B)}{|A|^2}\right)^2\right)\mu^2,
\end{equation} 
we rewrite
\begin{equation}
A\ell+B\mu = A(\ell + \frac{\text{Re}(A\bar B)}{|A|^2}\mu) +\mathfrak f^r\mu,
\end{equation} 
where $\mathfrak f^r$ is defined in \eqref{eq:fr}. 
So we set $\tilde\ell: = \ell + \frac{\text{Re}(A\bar B)}{|A|^2}\eps\zeta$ when $\mu=\eps\zeta$. 
It can be shown that $|\mathfrak f^r|^2 = |B|^2 - |A|^2\left(\frac{\text{Re}(A\bar B)}{|A|^2}\right)^2$.
Define
\begin{equation}
C^r(h,\zeta):=\frac{1}{\pi}
\int_{\mathbb R}\frac{1}{ \frac{1}{(\gamma_*)^2}\left(t_*^2+ |\theta_*|^2|A|^2\tilde\ell^2 +|\theta_*|^2|\mathfrak f^r|^2\zeta^2\right) - h^2} d\tilde\ell 
= \frac{1}{2}\left|\frac{\gamma_*}{\theta_*}\right|\frac{1}{\sqrt{(\frac{t_*}{\gamma_*})^2+ \frac{3}{4}|\frac{\theta_*A}{\gamma_*}|^2\zeta^2-h^2}}.
\end{equation} 
We use the superscript $r$ to denote operators corresponding to $(\mathbb M(\eps\zeta))^{-1}\mathbb T^{\pm\eps}(\lambda_*+ \eps h,\eps\zeta) (\mathbb M)(\eps\zeta)$ defined in \eqref{eq:oplimeta}. The limits are
\begin{equation}\label{eq:oplimetarat}
(\mathbb M^r(\eps\zeta))^{-1}\mathbb T^{\pm\eps,r}(\lambda_*+ \eps h,\eps\zeta) (\mathbb M^r)(\eps\zeta)
\to 
\tilde{\mathbb T}^{0,r} (\lambda_*) 
+ \beta^r(h,\zeta)\mathbb P \mp \xi^r(h,\zeta)\mathbb Q + \sigma_1^r(h,\zeta)\mathbb O_1 + \sigma_2^r(h,\zeta)\mathbb O_2 ,
\end{equation}
where 
\begin{equation}
\begin{aligned}
&\mathbb O_1\vec\phi:= c_1(\vec\phi)\vec v_1 - c_2(\vec\phi) \vec v_2,\quad
\mathbb O_2\vec\phi:= - c_2(\vec\phi) \vec v_1 + c_1(\vec\phi)\vec v_2 ,\\
&\beta^r(h,\zeta) =  hC^r(h,\zeta), \quad 
\beta^r(h,\zeta) = \frac{t_*}{|\gamma_*|} C^r(h,\zeta),\\
&\sigma_1^r(h,\zeta) = \text{Re}(\mathfrak f^r)\zeta|\frac{\theta_*}{\gamma_*}| C^r(h,\zeta), \quad 
\sigma_2^r(h,\zeta) = \text{Im}(\mathfrak f^r)\zeta|\frac{\theta_*}{\gamma_*}|C^r(h,\zeta).
\end{aligned}
\end{equation}
A calculation similar to the proof of Proposition~\ref{lem:dispersion} produces the result.
\end{proof}

\begin{proof}[Proof of Theorem~\ref{lem:ratarm}]
Comparing to \eqref{eq:Tderiv} and \eqref{eq:TderivKtwo}, we observe that at $\Kone$ and $\Ktwo$, the signs of $\theta_*$ and $t_*$ are opposite. Thus the edge states bifurcating from $\Kone$ and $\Ktwo$ have opposite dispersion slopes.
\end{proof}

\begin{proof}[Proof of Theorem~\ref{lem:dispersionarm}]
The dispersion relations on the zigzag interface and the armchair interface in Theorems~\ref{lem:dispersion} and \ref{lem:dispersionarm} are special cases of Propositions~\ref{lem:ratzig} and \ref{lem:ratarm}.

For the zigzag interface $a=0, b=1, c=1, d=0$ and $A=1, B=\bar\tau$, $\mathfrak f^r =\frac{1}{2}+\bar\tau $ and $|\mathfrak f^r| = \frac{\sqrt{3}}{2}$.

For the armchair interface, $a=1, b=1, c=1, d=0$ and $A=1-\bar\tau, B=\bar\tau$, $\mathfrak f^r = \frac{1+\bar\tau}{2}$ and $|\mathfrak f^r| = \frac{1}{2}$.

\end{proof}

\end{appendices}

\bibliographystyle{siam}

\end{document}